\documentclass[12pt]{article}

\usepackage[x11names]{xcolor}
\usepackage{tikz}
\usepackage{tikz-cd}
\tikzcdset{scale cd/.style={every label/.append style={scale=#1},    cells={nodes={scale=#1}}}}
\usepackage{cite}
\usepackage{hyperref}
\usetikzlibrary{calc,patterns,positioning,calc,math,shapes}
\usetikzlibrary{decorations.pathreplacing,decorations.markings}
\usetikzlibrary{positioning}
\usetikzlibrary{cd}
\usetikzlibrary{external} %%% Externalize tikz graphics
\usepackage{graphicx}
\usepackage{amsmath}
\usepackage{mathtools}
\usepackage{amssymb}
\usepackage{subsupscripts}
\usepackage{braket}
\usepackage{amsthm}
\usepackage{xcolor}
\usepackage{subcaption}
\usepackage{bbm}
\usepackage{tabularray}

\makeatletter

\@addtoreset{equation}{section}
\makeatother
\newcommand{\ba}{\begin{eqnarray}}
\newcommand{\ea}{\end{eqnarray}}

\usetikzlibrary{decorations.pathreplacing,decorations.markings}
\tikzset{
  on each segment/.style={
    decorate,
    decoration={
      show path construction,
      moveto code={},
      lineto code={
        \path [#1]
        (\tikzinputsegmentfirst) -- (\tikzinputsegmentlast);
      },
      curveto code={
        \path [#1] (\tikzinputsegmentfirst)
        .. controls
        (\tikzinputsegmentsupporta) and (\tikzinputsegmentsupportb)
        ..
        (\tikzinputsegmentlast);
      },
      closepath code={
        \path [#1]
        (\tikzinputsegmentfirst) -- (\tikzinputsegmentlast);
      },
    },
  },
  mid arrow/.style={postaction={decorate,decoration={
        markings,
        mark=at position .5 with {\arrow[#1]{stealth}}
      }}},
}

\DeclareMathOperator*{\Span}{\text{Span }}

\newcommand{\dket}[1]{\ket{#1}\!\rangle}
\newcommand{\dbra}[1]{\langle\!\bra{#1}}
\newcommand{\superp}[2]{\genfrac{}{}{0pt}{}{#1}{#2}}

 \def\d{\delta}

 \def\p{\partial}
 
 \def\a{\alpha}
 \def\b{\beta}
 \def\g{\gamma}
 \def\d{\delta}
 \def\e{\varepsilon}

 \def\k{\kappa}
 \def\l{\lambda}

 \def\r{\rho}
 
 \def\s{\sigma}

 \def\G{\Gamma}
 \def\D{\Delta}

 \def\L{\Lambda}

\def\CB{{\mathcal{B}}}
\def\CC{{\mathcal{C}}}

\def\CE{{\mathcal{E}}}
\def\CF{{\mathcal{F}}}
\def\CG{{\mathcal{G}}}

\def\CK{{\mathcal{K}}}

\def\CM{{\mathcal{M}}}
\def\CN{{\mathcal{N}}}
\def\CO{{\mathcal{O}}}
\def\CP{{\mathcal{P}}}

\def\CS{{\mathcal{S}}}
\def\CT{{\mathcal{T}}}

\def\CW{{\mathcal{W}}}

\def\CY{{\mathcal{Y}}}

\def\equskip{\!\!\!\!\!\!\!\!} %Huit espaces négatifs
\def\la{\left\langle}
\def\ra{\right\rangle}

\def\hf{\dfrac{1}{2}}
\def\bc{{\bar{c}}}

\def\implies{\quad\Rightarrow\quad}

\def\CS{\mathcal{S}}

\def\hg{{\hat\gamma}}

\def\bd{{\bar d}}

\def\vac{\emptyset}
\def\res{\mathop{\text{Res}}}

\def\bl{{\boldsymbol{\lambda}}}
\def\bmu{{\boldsymbol{\mu}}}

\def\bnu{{\boldsymbol{\nu}}}

\def\mZ{\mathbb{Z}}

\def\mC{\mathbb{C}}

\def\bsn{\boldsymbol{n}}

\def\gl{\mathfrak{gl}}
\def\sl{\mathfrak{sl}}

\def\Uqsl{U_q(\widehat{\mathfrak{sl}(2)})}

\def\Abox{{\tikz[scale=0.007cm] \draw (0,0) rectangle (1,1);}}
\def\sAbox{{\tikz[scale=0.005cm] \draw (0,0) rectangle (1,1);}}

\def\AboxB{{\tikz[scale=0.007cm] \draw[fill=black] (0,0) rectangle (1,1);}}
\def\sAboxB{{\tikz[scale=0.005cm] \draw[fill=black] (0,0) rectangle (1,1);}}

\def\bell{\bar\ell}

\def\ad{\text{ad}}

\def\tN{\tilde{N}}

\def\bsx{{\boldsymbol{x}}}
\def\bsy{{\boldsymbol{y}}}

\def\bsw{{\boldsymbol{w}}}
\def\bsu{{\boldsymbol{u}}}
\def\bsv{{\boldsymbol{v}}}
\def\bsa{{\boldsymbol{a}}}
\def\bsb{{\boldsymbol{b}}}

\def\tP{\tilde{P}}

\def\mukT{\mathsf{T}}
\def\bCK{\bar{\mathcal{K}}}
\def\FTW{\mathsf{W}}

\newcommand{\mathd}{\mathrm{d}}
\newcommand{\mathe}{\mathrm{e}}
\newcommand{\mathi}{\mathrm{i}}

\newcommand{\BC}{\mathbb{C}}

\newcommand{\BQ}{\mathbb{Q}}
\newcommand{\BZ}{\mathbb{Z}}
\newcommand{\lam}{\lambda}
\newcommand{\GHT}{\mathsf{V}}
\newcommand{\me}{\xi}
\renewcommand{\sp}{\varepsilon}
\newcommand{\blam}{{\boldsymbol{\lam}}}
\newcommand{\Id}{\mathrm{Id}}
\newcommand{\Whit}{\mathring{W}}

\newcommand{\be}{\begin{equation}}
\newcommand{\ee}{\end{equation}}
\let\emptyset\varnothing

\newtheorem{proposition}{Proposition}[section]
\newtheorem{conjecture}{Conjecture}[section]
\newtheorem{theorem}{Theorem}[section]
\newtheorem{corollary}{Corollary}[proposition]
\newtheorem{lemma}[theorem]{Lemma}
\theoremstyle{definition}
\newtheorem{definition}{Definition}[section]
\theoremstyle{remark}
\newtheorem*{remark}{Remark}

\begin{document}
\begin{titlepage}
% \vspace*{-1cm}
% \begin{flushright}
% \jobname .pdf\\ \today\\
% \end{flushright}

\begin{center}
{\Huge Generalized Macdonald functions\\\vspace{10pt} and quantum toroidal $\gl(1)$ algebra}%

\vskip 1cm
{\Large Jean-Emile Bourgine\footnote{bourgine@simis.cn}$^\dagger$, Luca Cassia\footnote{luca.cassia@unimelb.edu.au}$^\circ$, Artem Stoyan\footnote{astoian@student.unimelb.edu.au}$^\circ$}\\

\vskip 1cm
{\it $\dagger$ Shanghai Institute for Mathematics and Interdisciplinary Sciences (SIMIS)}\\
{\it Fudan University}
{\it Block A, International Innovation Plaza}\\
{\it No. 657 Songhu Road, Yangpu District Shanghai, China}

\vskip 1cm
{\it $^\circ$ School of Mathematics and Statistics}\\
{\it University of Melbourne}\\
{\it Parkville, Victoria 3010, Australia}\\

\end{center}
\vfill
\begin{abstract}
The Macdonald operator is known to coincide with a certain element of the quantum toroidal $\mathfrak{gl}(1)$ algebra in the Fock representation of levels $(1,0)$. A generalization of this operator to higher levels $(r,0)$ can be built using the coproduct structure, it is diagonalized by the generalized Macdonald symmetric functions, indexed by $r$-tuple partitions and depending on $r$ alphabets. In this paper, we extend to the generalized case some of the known formulas obeyed by ordinary Macdonald symmetric functions, such as the $e_1$-Pieri rule or the identity relating them to Whittaker vectors obtained by Garsia, Haiman, and Tesler. We also propose a generalization of the five-term relation, and the Fourier/Hopf pairing. In addition, we prove the factorized expression of the generalized Macdonald kernel conjectured previously by Zenkevich.
\end{abstract}
\vfill
\end{titlepage}

\setcounter{footnote}{0}

\newpage

\tableofcontents

\newpage

\section{Introduction}
Generalized Macdonald functions have been introduced in \cite{Awata2011} in the context of the 5d AGT correspondence relating $q$-deformed $W$-algebras and $\CN=1$ supersymmetric gauge theories in the five-dimensional $\Omega$-deformed background $\mC_{q_1}\times\mC_{q_2}\times S^1$. In this context, they play the role of the AFLT basis constructed in \cite{Alba:2010qc}. Their definition follows from the identification of the Macdonald operator in the large $N$ limit (where $N$ is the number of variables) with the generator $x_0^+$ of the quantum toroidal $\gl(1)$ algebra evaluated in the Fock representation of levels $(1,0)$. This identification is a consequence of the well-known ring isomorphism between the bosonic Fock space and the ring of symmetric functions $\L[\bsx]$, which sends the PBW basis of the Heisenberg algebra to the power sum basis. To define generalized Macdonald functions at level $r$, this isomorphism is extended to tensor product of Fock spaces. It leads to interpreting the generator $x_0^+$ evaluated in the Fock representation of levels $(r,0)$ (defined in terms of $r$ bosons) as a generalized Macdonald operator acting on symmetric functions of $r$ alphabets $\L[\bsx^{(1)},\dots,\bsx^{(r)}]$. Generalized Macdonald symmetric functions $P_\bl(\bsx^{(1)},\dots,\bsx^{(r)}|u_1,\dots,u_r)$ are then defined as the eigenfunctions of this operator. They are indexed by $r$-tuple partitions $\bl=(\l^{(1)},\dots,\l^{(r)})$ and depend on $r-1$ extra complex parameters corresponding to ratios of the weights of the representation. For short, we will refer to these generalized Macdonald symmetric functions by the acronym GMP.

Beyond the 5d AGT correspondence, there are several reasons for studying this generalization of Macdonald symmetric functions. One of our main motivation comes from the $(q,t)$-deformed integrable hierarchies recently introduced in \cite{Bourgine2023}. The main idea behind this deformation of standard KP/Toda hierarchies is to promote the underlying $\gl(\infty)$/$q$-$W_{1+\infty}$ symmetry algebra to the quantum toroidal $\gl(1)$ algebra.\footnote{This is motivated by correspondences between integrable hierarchies and topological string theory on Calabi--Yau threefolds \cite{Aganagic2003,Nakatsu2007,Bourgine2021b}. A similar deformation of the symmetry algebra underpins the definition of the refined topological vertex \cite{AFS}.} In doing so, the algebra becomes non-linear and acquires a non-trivial coproduct structure. As a result, techniques based on determinants, free fermions and Schur functions no longer apply. In this context, generalized Macdonald symmetric functions of level $r=2$ are expected to provide solutions of the $(q,t)$-deformed hierarchies that deform the well-known Schur polynomial solutions. Besides, the generalized Macdonald functions encode the information of the co-product structure, and offer the perspective of deforming the notions of Plucker coordinates on Sato's Grassmaniann. However, it should be mentioned, that the use of GMP in this context requires us to consider special values of their parameter $Q=u_1/u_2$, which corresponds to the ratio of weights in the horizontal representation $(2,0)$. In this paper, we restrict ourselves to $Q$ generic, we will examine the case the degenerate case $Q=q_1^a q_2^b$ in a future publication.

Another motivation comes from the connection between the 5d/K-theoretic extension of the AFLT approach to the AGT correspondence and the geometry of moduli spaces of $U(r)$ instantons on $\BC^2$. These moduli spaces can be explicitly described as certain Nakajima quiver varieties \cite{Nakajima1994} (a generalization of the ADHM construction) and it is a well-known fact that there exists an action of the quantum toroidal $\gl(1)$ algebra (also known as Elliptic Hall algebra) on the localized equivariant K-theory of these spaces \cite{FT-shuffle,SV-Hall,Schiffmann2012}.
More precisely, let $\mathcal{M}(r,n)$ be the moduli space of rank $r$ instantons of topological charge $n$ on $\BC^2$. There exists an action of $T=(\BC^\times)^{r+2}$ corresponding to the maximal torus of $GL(r)\times GL(2)$, such that the $T$-fixed points are isolated and labeled by $r$-tuples of partitions with a total number of $n$ boxes. After localization in $T$-equivariant K-theory, one obtains an isomorphism
\be
 \Lambda[\bsx^{(1)},\dots,\bsx^{(r)}] \xrightarrow{\cong}
 \bigoplus_{n\geq0} K_T(\mathcal{M}(r,n))_{\rm loc} 
\ee
which identifies generalized Macdonald functions with the equivariant classes of fixed points in K-theory. Moreover, the quantum parameters $(q_1,q_2)$ and the weights $(u_1,\dots,u_r)$ are identified with equivariant parameters for the maximal tori of $GL(2)$ and $GL(r)$, respectively.
This isomorphism is the stable map of \cite{MO2012} (see \cite{Smirnov:2014npa} for a detailed discussion of the analogous statement in cohomology).

The case $r=1$ has been extensively studied by Haiman and collaborators \cite{Haiman1999,zbMATH02068005} to give an explicit construction of (modified) Macdonald functions as the Frobenius series of the fibers of the Garsia--Procesi bundle over fixed points in the Hilbert scheme of $n$ points in $\BC^2$, i.e.\ $\mathrm{Hilb}^n(\BC^2)\cong\mathcal{M}(1,n)$.
As a result of this intensive study of the theory of symmetric functions and their connection to geometry, Garsia, Haiman and Tesler (GHT) \cite[Theorem I.2]{Garsia2001} derived a powerful new identity from which many known properties of Macdonald polynomials follow as simple corollaries (such as bispectral duality or Macdonald–-Koornwinder reciprocity). This identity was later rederived by Carlsson, Nekrasov and Okounkov \cite{Carlsson:2013jka} using the the Macdonald--Mehta--Cherednik identity \cite{Cherednik:1997,Etingof:1997}.
Recent work on the wreath Macdonald case, corresponding to $r=1$ instanton moduli spaces on resolutions of $A_n$ singularities, has led to an extension of the result of GHT to higher rank analogues of the quantum toroidal algebra \cite{Romero:2025fiv,Romero:2025tes}.
However, no generalizations to higher levels $r>1$ are currently known.

In this paper, we make use of the representation theory of the quantum toroidal $\gl(1)$ algebra to extend several known results in the theory of Macdonald functions to the case of generalized Macdonald functions, including the GHT identity.

\paragraph{Main results.} Several studies on generalized Macdonald functions can be found in the literature, from both physics and mathematics perspectives \cite{Awata2011,Ohkubo:2014nqa,Zenkevich:2014lca,Awata:2015hax,Kononov:2016ycq,Zenkevich:2016xqu,Ohkubo:2017tvz,Zenkevich:2017tnb,Koroteev:2018isw,Fukuda:2019ywe,Mironov:2019uoy,Ohkubo:2019mtr}. The aim of this paper is to review some of the main results, and present several new identities.

Our first main result is a derivation of $e_1$-Pieri rules for GMP. Explicitly, we find at level two
\begin{equation}
(e_1(\bsx)+e_1(\bsy))P_{\l,\mu}(\bsx,\bsy|Q)=\sum_{\sAbox\in A(\l)}\Psi_\mu(Q\chi_\sAbox)\psi_\l(\Abox) P_{\l+\sAbox,\mu}(\bsx,\bsy|Q)+\sum_{\sAbox\in A(\mu)}\psi_\mu(\Abox) P_{\l,\mu+\sAbox}(\bsx,\bsy|Q)
\end{equation}
where $\psi_\l(\Abox)$ are the usual Pieri coefficients (at level one), and the function $\Psi_\mu(z)$ is defined in \eqref{expr_CYY} below. This formula is obtained from the observation  that, under specialization, the expansion of the Mukad\'e operator introduced in \cite{Fukuda:2019ywe} gives at first order the algebra generator $a_{-1}$. As a result, the matrix elements of this operator, which encode the Pieri rule, can be obtained from the expansion of the matrix elements of the Mukad\'e operator. The details of the derivation are provided in Appendix~\ref{app:Mukade}.

Our second main result is the proof of a factorization formula for the reproducing kernel of GMP, originally conjectured by Zenkevich \cite{Zenkevich:2014lca}. This kernel is a generalization of the usual Macdonald kernel. For instance, at level two, we prove that
\begin{equation}
\sum_{\l,\mu}b_{\l}b_{\mu}P_{\l,\mu}(\bsx,\bsy|Q)P_{\mu,\l}(\bsa,\bsb|Q^{-1}) = \mathe^{\sum_{k>0}\frac1k\frac{1-t^k}{1-q^k}
 \left(p_k(\bsx)p_k(\bsb)+p_k(\bsy)p_k(\bsa)+(1-t^kq^{-k})p_k(\bsx)p_k(\bsa)\right)}\,,
\end{equation}
where $b_\l$ and $b_\mu$ are standard combinatorial factors given in \eqref{eq:def-Macdonald-inner}.

The third main result is the extension of an identity obtained by Garsia, Haiman and Tesler \cite[Theorem~I.2]{Garsia2001} to GMP. For this purpose, we reinterpret this identity as a relation between Macdonald functions and Whittaker vectors,
\begin{equation}
 W_\l(\bsx) = \GHT\tP_\l(\bsx),
\end{equation}
where the Whittaker vector $W_\l(\bsx)$ coincides with the specialization of the Macdonald kernel $\Pi(\sp_\l|\bsx)$. In this identity, $\GHT$ is an operator constructed out of the framing operator (coinciding with the nabla operator in this representation) and the vacuum component of a vertex operator. With this interpretation, we extend the identity to higher levels. In general, Whittaker vectors no longer coincide with a specialization of the kernel, and do not have a factorized expression. However, they can still be defined algebraically, and written explicitly using higher vertex operators of the quantum toroidal $\gl(1)$ algebra.  

Unfortunately, we have only been able to provide a partial proof of the GHT identity at higher level, which relies on the Conjecture~\ref{conj:BH}. This conjecture is equivalent to the identity specialized to an empty $r$-tuple partition, and it has been extensively checked by computer. Upon assuming this conjecture, we show using algebraic methods that the identity holds for arbitrary $r$-tuple of partitions $\bl$.

Our paper also introduces several new objects, such as the extended quantum toroidal algebra by an outer automorphism (Section~\ref{sec_framing}), or the construction of higher vertex operator using $T_N$ diagrams (Appendix~\ref{sec_higher_VO}). We also propose an extension of the five-term relation obtained in \cite{Garsia:2018fiv}, and of the Fourier/Hopf pairing.

\paragraph{Organization.} The paper is organized as follows. The second section contains a description of the algebraic tools needed in the rest of the paper. It includes the definition of the algebra, the presentation of representations used here, the definition of higher vertex operators, and the extension of the algebra by \emph{framing operators} realizing certain outer automorphism. The third section introduces generalized Macdonald polynomials, and present their main properties, namely Pieri rule, inner products and the associated kernels. The fourth section is devoted to the GHT identity, the five-term relation and the definition of the Fourier/Hopf pairing. We first revisit the identity for ordinary Macdonald polynomials, reinterpreting it as a map between the polynomials and Whittaker vectors, and then extend it to GMP using the higher vertex operators defined previously.

The appendices present examples, technical results, and proofs. Specifically, Appendix~\ref{app_A1} recalls technical results on the Nekrasov factor, \ref{AppA2} and \ref{app:pieri} present the action of quantum toroidal $\gl(1)$ generators on Macdonald functions, \ref{app_D_H} discusses a twist of the coproduct and \ref{app:framing-op} the framing operator. Appendix~\ref{sec_higher_VO} contains the proof of the intertwining relation for higher vertex operators. Appendix~\ref{app_finite} presents an example of generalized Macdonald polynomials with finite (two) variables. Finally, Appendix~\ref{app:Mukade} gives the derivation of Pieri rules for GMP from matrix elements of the Mukad\'e operator.

\paragraph{Partitions.} The results presented in this paper rely heavily on combinatorial properties of partitions, and we need to prepare some notations. We denote partitions with Greek letters $\l,\mu,\nu,\dots$. A partition $\l=[\l_1,\l_2,\dots]$ is identified with the set of boxes $\Abox$ of coordinates $(i,j)$ with $i\in[\![1,\ell(\l)]\!]$ and $j\in[\![1,\l_i]\!]$ where $\ell(\l)$ is the length of the partitions. We denote $|\l|$ the total number of boxes, and $\l^T$ the transposed partition. We also define the standard combinatorial factors $n(\l)=\sum_{i=1}^{\ell(\l)}(i-1)\l_i$ and $z_\l=\prod_{k\geq1} k^{m_k} m_k!$, where $m_k$ is the multiplicity of the part $k$ in $\l$.
We say that two partitions satisfy $\mu\subset_k\l$ if $\mu$ is contained in $\l$ and $|\l|=|\mu|+k$. To any partition $\l$ we associate the set of addable boxes $A(\l)$ (resp. removable boxes $R(\l)$), i.e. the set of boxes $\Abox$ such that $\l+\Abox$ (resp. $\l-\Abox$) is also a partition.

To each box $\Abox\in\l$ we associate the \emph{box content} $\chi_\sAbox=q_1^{i-1}q_2^{j-1}$ expressed in terms of the quantum group parameters $q_1,q_2\in\mC^\times$. For convergence issues, we assume $|q_1|,|q_2|<1$. For any partition $\l$, we define the quantities $\chi_\l$, $g_\l$, $\sp_\l$ and $\me_\l$ as follows,
\be\label{expr_gl}
 \chi_\l := \sum_{\sAbox\in\l} \chi_\sAbox,
 \hspace{30pt}
 g_\l := \prod_{\sAbox\in\l} \chi_\sAbox
 = q_1^{n(\l)} q_2^{n(\l^T)},
\ee
\be\label{expr_sp}
 \sp_\l := \sum_{i=1}^\infty q_1^i q_2^{\l_i},
\ee
and
\be\label{expr_el}
 \me_\l := 1-(1-q_1)(1-q_2)\chi_\l
 = \sum_{\sAbox\in A(\l)} \chi_\sAbox - q_3^{-1}\sum_{\sAbox\in R(\l)}\chi_\sAbox
 = \frac{\sp_\l}{\sp_\emptyset},
\ee
where we assume $\l_i=0$ for $i>\ell(\l)$.

\paragraph{Plethystic notation.} We use the ``\emph{plethystic}'' notation for symmetric functions.
Consider the algebra $\Lambda$ of symmetric functions in a formal
infinite alphabet $(x_1,x_2,x_3,\dots)$, with coefficients in the field of rational
functions $\BQ(q_1,q_2)$.
Plethystic notation is a notational device which simplifies manipulation of symmetric function identities, by regarding the elements of $\Lambda$ as formal power series in the powersum symmetric functions $p_k$.
To begin with, if $E = E(t_1,t_2,t_3,\dots)$ is a formal Laurent series in the variables $t_1,t_2,t_3,\dots$ (which may include the parameters $q_1,q_2$) with coefficients in $\BZ$, we set
\be
 p_k(E) = E(t_1^k,t^k_2,t^k_3\dots)\,.
\ee
In the language of $\l$-rings, this corresponds to the $k$-th Adams operator on $\Lambda$.
More generally, if a certain symmetric function $f\in\Lambda$ is expressed as the formal
power series
\be
 f = f(p_1,p_2,p_3,\dots)
\ee
then we let
\be
 f(E) = f(p_1,p_2,p_3,\dots)\Big|_{p_k\mapsto E(t_1^k,t^k_2,t^k_3\dots)}
\ee
and refer to this operation as ``\emph{plethystic substitution}'' of $E$ into the symmetric function $f$.
Set $\bsx:=x_1+x_2+x_3+\dots$, then we will write
\be
 f(\bsx) = f(p_1,p_2,p_3,\dots)\Big|_{p_k\mapsto x_1^k+x_2^k+x_3^k+\dots}
\ee
for any $f\in\Lambda$, and regard $f(\bsx)$ as a symmetric function in the alphabet $\bsx$.
As an example, we will often consider Laurent series $E=E(q_1,q_2)$ such as $\sp_\lam$ and $\me_\lam$ defined in \eqref{expr_sp} and \eqref{expr_el}, and use them as arguments for symmetric functions, for instance
\be
 p_k(\sp_\lam) = \sum_{i=1}^\infty p_k(q_1^iq_2^{\lam_i})
 = \sum_{i=1}^\infty q_1^{ki}q_2^{k\lam_i}\ ,
 \hspace{30pt}
 p_k(\me_\lam) = 1-(1-q_1^k)(1-q_2^k)\sum_{\sAbox\in\lam} \chi_\sAbox^k\ .
\ee
We denote by $E^\vee(t_1,t_2,t_3,\dots)$ the series $E(t_1^{-1},t_2^{-1},t_3^{-1},\dots)$ and similarly $\bsx^\vee = \sum_{i\geq1} x_i^{-1}$.

We introduce the ``\emph{plethystic exponential}'' function as
\be
\label{eq:PE}
 \exp\left(\sum_{k>0}\frac{p_k(\bsx)p_k(\bsy)}{k}\right)
 = \sum_\l s_\l(\bsx)\,s_\l(\bsy)
 = \prod_{i,j}\frac1{1-x_iy_j}
\ee
where $\bsy=y_1+y_2+y_3+\dots$, is a second formal infinite alphabet and $s_\l$ are Schur symmetric functions.
At this point it is also convenient to introduce the ``\emph{translation}'' operator which acts on a symmetric function $f(\bsx)$ according to the plethystic formula
\be
 \exp\left(\sum_{k>0}p_k(\bsy)\frac{\p}{\p p_k(\bsx)}\right) f(\bsx)
 = f(\bsx+\bsy)
\ee
Formally, we may regard the symbols $k\frac{\p}{\p p_k(\bsx)}$ as the adjoint operators of $p_k(\bsx)$ w.r.t.\ the Hall inner product $\langle s_\l,s_\mu\rangle = z_\l\delta_{\l,\mu}$. We then write
\be
\label{eq:adjoint-p_k}
 \left\langle p_k(\bsx)f(\bsx),g(\bsx)\right\rangle
 = \left\langle f(\bsx),k\frac{\p}{\p p_k(\bsx)} g(\bsx)\right\rangle
\ee
for arbitrary $f,g\in\Lambda$.
A very important remark is that the adjoint operator of $p_k(\bsx)$ can only be defined when the symbol $\bsx$ is truly a formal infinite alphabet, i.e.\ when the formal variables $x_i$ are all independent of each other so that it makes sense to take derivatives w.r.t.\ the powersums $p_k(\bsx)$ in the r.h.s.\ of \eqref{eq:adjoint-p_k}.

\paragraph{Useful functions.} We define the following functions,
\begin{equation}\label{def_CYY}
 \CY_\l(z)=\exp\left(-\sum_{k>0}\frac{z^{-k}p_k(\me_\l)}{k}\right),
 \hspace{30pt}
 \Psi_\l(z)=\frac{\CY_\lam(q_3^{-1}z)}{\CY_\lam(z)}
 = \exp\left(\sum_{k>0}(1-q_3^k)\frac{z^{-k}p_k(\me_\l)}{k}\right),
\end{equation} 
and
\begin{equation}\label{def_S}
 S(z)
 = \dfrac{(1-q_1z)(1-q_2z)}{(1-z)(1-q_1q_2z)},
 \hspace{30pt}
 g(z) = \prod_{\a=1,2,3}\dfrac{1-q_\a z}{1-q_\a^{-1}z}
\end{equation} 
where $q_3=(q_1q_2)^{-1}$. The various expressions \eqref{expr_el} for $\me_\l$ give rise to different expressions for these functions,
\be
\begin{aligned}\label{expr_CYY}
 \CY_\l(z)&=\left(1-z^{-1}\right)\prod_{\sAbox\in\l}S(\chi_\sAbox/z)=\frac{\prod_{
 \sAbox\in A(\l)}1-\chi_\sAbox/z}{\prod_{\sAbox\in R(\l)}1-\chi_\sAbox/(q_3z)},\\
 \Psi_\l(z)&=\dfrac{1-q_3/z}{1-1/z}\prod_{\sAbox\in\l}g(z/\chi_\sAbox)=\prod_{
 \sAbox\in A(\l)}\frac{1-q_3\chi_\sAbox/z}{1-\chi_\sAbox/z}\times\prod_{\sAbox\in R(\l)}
 \frac{1-\chi_\sAbox/(q_3z)}{1-\chi_\sAbox/z},
\end{aligned}
\ee
We also introduce the (K-theoretic) Nekrasov factor $N_{\l,\mu}(z)$,
\begin{equation}
 N_{\l,\mu}(z)=\prod_{(i,j)\in\l}(1-z q_1^{\l^T_j-i+1}q_2^{j-\mu_i})
 \prod_{(i,j)\in\mu}(1-z q_1^{i-\mu^T_j}q_2^{\l_i-j+1})
\end{equation}
and the double $q$-Pochhammer
\begin{equation}
\label{eq:double-qPoc}
 \CG(z)=(z;q_1,q_2)_\infty = \prod_{i,j=1}^{\infty}(1-zq_1^{i-1}q_2^{j-1})
 =\exp\left(-\sum_{k>0}\frac{z^k}{k(1-q_1^k)(1-q_2^k)}\right).
\end{equation}

\paragraph{Multi-partitions.} We denote the $r$-tuple partitions with bold symbols, e.g.\ $\bl=(\l^{(1)},\l^{(2)},\dots,\l^{(r)})$. Such partitions are usually associated with a $r$-dimensional vector of weights $\bsv=(v_1,\dots,v_r)$. To a box $\Abox\in\bl$ corresponding to a box of coordinates $(i,j)$ in some $\l^{(\a)}$ we associate the box content $\chi_\sAbox=q_1^{i-1}q_2^{j-1}$ as before, and the weight $v_\sAbox=v_\a$. We extend the definitions \eqref{def_CYY} as follows,
\be
\label{eq:coloredCY-Psi}
 \CY_\bl(z|\bsv) = \prod_{\a=1}^r \CY_{\l^{(\a)}}(z/v_\a),
 \hspace{30pt}
 \Psi_\bl(z|\bsv) = \prod_{\a=1}^r \Psi_{\l^{(\a)}}(z/v_\a),
\ee
and we further introduce the functions
\be
\label{def_rl}
 r_\bl(\Abox|\bsv)=\res_{w=v_\sAbox\chi_\sAbox}\dfrac1{w\CY_{\bl}(w|\bsv)},
 \hspace{30pt}
 r_\bl^\ast(\Abox|\bsv)=\res_{w=v_\sAbox\chi_\sAbox}w^{-1}\CY_{\bl}(q_3^{-1}w|\bsv).
\ee

Multi-partitions are the natural labels for the basis of the space of multi-symmetric functions, i.e.\ functions separately symmetric in many alphabets. This space is isomorphic to the $r$-fold tensor product
\be
 \Lambda[\bsx^{(1)},\dots,\bsx^{(r)}]
 := \Lambda[\bsx^{(1)}]\otimes\dots\otimes\Lambda[\bsx^{(r)}]
 \equiv\Lambda^{\otimes r}
\ee
whose elements will be denoted as
\be
 f(\bsx^\bullet) \equiv f(\bsx^{(1)},\dots,\bsx^{(r)})\,.
\ee
One natural linear basis of such functions is that of the tensor product of power-sums in each alphabet. In the rest of the paper, we will employ the equivalent notations
\be
 p_k^{(\a)} := p_k(\bsx^{(\a)}) \equiv \underbrace{1\otimes\dots\otimes1}_{\a-1}
 \otimes\,p_k\otimes
 \underbrace{1\otimes\dots\otimes1}_{r-\a}
\ee
for $\a=1,\dots,r$. We observe also that there exists a natural ``diagonal'' embedding of the ring of symmetric functions in one alphabet into the ring of multi-symmetric functions in $r$ alphabets, namely $\Lambda\xhookrightarrow{}\Lambda^{\otimes r}$ such that the image of $f\in\Lambda$ is the multi-symmetric function $f(\sum_{\a=1}^r\bsx^{(\a)})\in\Lambda^{\otimes r}$.

\section{Quantum toroidal \texorpdfstring{$\gl(1)$}{gl(1)} algebra}

\subsection{Definition}\label{sec_Miki}
The quantum toroidal $\gl(1)$ algebra, here denoted $\CE$, is defined by the relations obeyed by the modes of the Drinfeld currents
\begin{equation}
x^\pm(z)=\sum_{k\in\mathbb{Z}}z^{-k}x^\pm_k,\quad \psi^\pm(z)=\sum_{k\geq0}z^{\mp k}\psi_{\pm k}^\pm.
\end{equation} 
These relations can be written in terms of exchange relations for the currents,
\be
\label{def_DIM}
\begin{aligned}
&[\psi^\pm(z),\psi^\pm(w)]=0,\quad \psi^+(z)\psi^-(w)=\dfrac{g(\hg w/z)}{g(\hg^{-1}w/z)}\psi^-(w)\psi^+(z),\\
&\psi^+(z)x^\pm(w)=g(\hg^{\pm1/2}z/w)^{\pm1}x^\pm(w)\psi^+(z),\quad \psi^-(z)x^\pm(w)=g(\hg^{\mp1/2}z/w)^{\pm1}x^\pm(w)\psi^-(z)\\
&\prod_{i=1,2,3}(z-q_i^{\pm1}w)\ x^\pm(z)x^\pm(w)=\prod_{i=1,2,3}(z-q_i^{\mp1}w)\ x^\pm(w)x^\pm(z),\\
&[x^+(z),x^-(w)]=\k\left(\d(\hg^{-1}z/w)\psi^+(\hg^{1/2}w)-\d(\hg z/w)\psi^-(\hg^{-1/2}w)\right),
\end{aligned}
\ee
where $\hg$ is a central element. In this formulation, $g(z)$ is the structure function defined in \eqref{def_S} and $\k$ is a $\mathbb{C}$-number that can be expressed in terms of the quantum group parameters $q_1,q_2$,
\begin{equation}
\k=\dfrac{(1-q_1)(1-q_2)}{(1-q_1q_2)}.
\end{equation}
Introducing the notation $\g=q_3^{1/2}$, the central element $\hg$ is sometimes denoted $\g^c$ with $c$ central. It is noted that the zero modes $\psi_0^+$ and $\psi_0^-$ are also central, and we introduce the central element $\bar c$ in order to denote $\psi_0^\pm=\g^{\mp\bar c}$.

The subalgebra generated by the elements $\psi_{\pm k}^\pm$ and the central element $\hg$ can be seen as a deformation of the Cartan subalgebra of simple Lie algebras but it is commutative only when $c=0$. It is sometimes useful to reformulate the Cartan generators $\psi_{\pm k}^\pm$ in terms of the modes $a_k$ of a $q$-deformed Heisenberg algebra,
\begin{equation}\label{def_ak}
\psi^\pm(z)=\psi_0^\pm\exp\left(\pm\sum_{k>0}z^{\mp k}a_{\pm k}\right).
\end{equation} 
These modes obey the following commutation relations, that can be derived from the original definition \eqref{def_DIM} of the algebra:
\be\label{com_ak}
\begin{aligned}
 & [a_k,a_l] = (\g^{ck}-\g^{-ck})c_k\d_{k+l}\,,\hspace{30pt}
 [a_k,x_l^\pm] = \pm\g^{\mp c|k|/2}c_k x_{l+k}^\pm\,,\\
 & [x_k^+,x_l^-] =
 \begin{cases}
 \k\g^{c(k-l)/2}\psi_{k+l}^+, & k+l>0,\\
 \k(\g^{c(k-l)/2-\bc}-\g^{-c(k-l)/2+\bc}), & k+l=0,\\
 -\k\g^{-c(k-l)/2}\psi_{k+l}^-, & k+l<0,\\
 \end{cases}
\end{aligned}
\ee
where we have introduced the coefficients $c_k$ of the expansion of the function $g(z)$:
\begin{equation}\label{exp_g}
[g(z)]_{\pm}=\exp\left(\pm\sum_{k>0}z^{\mp k}c_k\right),\quad c_k=-\dfrac1{k}\prod_{\a=1,2,3}(1-q_\a^k),
\end{equation} 
Here $[f(z)]_+$ refers to an expansion around $z=\infty$ of the function $f(z)$, while $[f(z)]_-$ denotes the expansion around $z=0$. The algebra $\CE$ is generated multiplicatively by the two central elements $(c,\bc)$ and the four generators $x_0^\pm$ and $a_{\pm1}$.

\paragraph{Gradings.} It is possible to supplement the algebra $\CE$ with two commuting grading operator $d$ and $\bd$ acting on Drinfeld currents as follows,
\begin{align}
\begin{split}
&[d,x_k^\pm]=-kx_k^\pm,\quad [d,\psi_{\pm k}^\pm]=\mp k\psi_{\pm k}^\pm,\\
&[\bd,x_k^\pm]=\pm x_k^\pm,\quad [\bd,\psi_{\pm k}^\pm]=0,\quad [\bd,d]=0.
\end{split}
\end{align}
We denote $\CE'$ the algebra extended by these two grading operators.

\paragraph{Hopf algebra structure.} The algebra $\CE$ has the structure of a Hopf algebra with the Drinfeld coproduct defined as
\begin{align}
\begin{split}\label{Drinfeld_coproduct}
&\D(x^+(z))=x^+(z)\otimes 1+\psi^-(\g^{c_{(1)}/2}z)\otimes x^+(\g^{c_{(1)}}z),\\
&\D(x^-(z))=x^-(\g^{c_{(2)}} z)\otimes \psi^+(\g^{c_{(2)}/2}z)+1\otimes x^-(z),\\
&\D(\psi^\pm(z))=\psi^\pm(\g^{\pm c_{(2)}/2}z)\otimes\psi^\pm(\g^{\mp c_{(1)}/2}z),\quad \D(a_k)=a_k\otimes \g^{-c|k|/2}+\g^{c|k|/2}\otimes a_k\\
\end{split}
\end{align}
with $c_{(1)}=c\otimes 1$, $c_{(2)}=1\otimes c$ and $\D(c)=c_{(1)}+c_{(2)}$, $\D(\bc)=\bc_{(1)}+\bc_{(2)}$. The expressions of antipode and counit will not be needed in this paper, we omit them here. This Hopf algebra stucture extends to $\CE'$ by imposing that grading elements are cocommutative.

%\begin{figure}
%\begin{center}
%\includegraphics[width=4cm]{d2.pdf}
%\end{center}
%\caption{DIM generators represented according to their degree $(x,y)=(\bd,-d)\in\mathbb{Z}\otimes\mathbb{Z}$.}
%\label{fig_DIM}
%\end{figure}

\paragraph{Automorphisms.} The quantum toroidal $\gl(1)$ algebra possess the group $\text{SL}(2,\mZ)$ of automorphisms, generated by the elements $\CS$ and $\CT$.\footnote{It should be noted that the transformation $\CT$ defined here, following an earlier convention, corresponds in fact to the element $T^{-1}$ of $\text{SL}(2,\mZ)$. Note also that the group of automorphisms is actually $\text{SL}(2,\mZ)$ only up to inner automorphisms. Indeed, with our convention we have $\CS^4=1$ but $(\CS\CT^{-1})^3=\iota^{-1}\CS^2$ with $\iota$ defined in \eqref{def_iota}.} The automorphism $\CS$ is known as Miki's automorphism \cite{Miki2007}, it is defined by its action on the generators $a_{\pm1}$, $x_0^\pm$ and $(c,\bc)$ of the algebra,
\begin{equation}\label{Miki_init}
a_1\to(\g-\g^{-1})x_0^+\to -a_{-1}\to -(\g-\g^{-1})x_0^-\to a_1,\quad (c,\bc)\to (-\bc,c)
\end{equation}
It is readily seen that this automorphism is of degree four. The transformations of other generators involve nested commutation relations, with the exception of the four modes $x_1^\pm$ and $x_{-1}^\pm$ which are transformed into each other,
\begin{equation}
x_1^+\to\g^{-(c+\bc)/2}x_{-1}^+\to -x_{-1}^-\to-\g^{(c+\bc)/2}x_1^-\to x_1^+.
\end{equation}
We will denote $b_k=\CS(a_k)$ and $y_k^\pm = \CS(x_k^\pm)$, and refer to \cite[Appendix~A]{Bourgine2018a} for an inductive construction of these generators. Miki's automorphism also acts on the grading operators, sending $(d,\bd)$ to $(-\bd,d)$. As a consequence, an element $e\in\CE$ of degrees $(d_e,\bd_e)$ is mapped to another element with degrees $(\bd_e,-d_e)$.
% As a result, the automorphism $\CS$ acts as a clockwise rotation of angle $90^\circ$ on the representation of the generators (c.f.\ Figure~\ref{fig_DIM}).

The action of the automorphism $\CT$ is relatively straightforward,
\begin{equation}
\CT(x_k^\pm)=x_{k\mp1}^\pm,\quad \CT(a_k)=a_k,\quad \CT(\psi_{\pm k}^\pm)=\g^{\mp c}\psi_{\pm k}^\pm,\quad \CT(c,\bc)=(c,\bc+c).
\end{equation} 
When discussing framing operators, it is also useful to introduce the automorphism $\CT^\perp=\CS\CT\CS^{-1}$. With our definitions, we have the properties $\CS\CT^\perp\CS^{-1}=\CS^2\CT\CS^{-2}=\CT$. We note that $\CT^\perp$ leaves $x_0^\pm$ invariant, $\CT^\perp(c,\bc)=(c-\bc,\bc)$ and 
\begin{equation}\label{CTCS}
\CT^\perp(a_{\pm1})=(\g-\g^{-1})\g^{\mp(c-\bc)/2}x_{\pm1}^\pm,\quad(\CT^\perp)^{-1}(a_{\pm1})=-(\g-\g^{-1})\g^{\pm(c+\bc)/2}x_{\pm1}^\mp,
\end{equation}

Furthermore, $\CT$ is compatible with the co-algebraic structure defined in \eqref{Drinfeld_coproduct}, i.e. $\D\circ\CT=(\CT\otimes\CT)\circ\D$, but neither $\CS$ nor $\CT^\perp$ are. In fact, $\CT^\perp$ is compatible with another co-algebraic structure based on the coproduct $\D^\perp=(\CS\otimes\CS)\circ\D\circ\CS^{-1}$. The automorphisms $\CT$ and $\CT^\perp$ can be extended to $\CE'$ by defining $\CT(d,\bd)=(d-\bd,d)$ and $\CT^\perp(d,\bd)=(d,\bd+d)$.

In this paper, we will use the following identity.
\begin{lemma}
\begin{equation}\label{id_TpTTp}
\CT^\perp\CT\CT^\perp=\iota\CS,
\end{equation}
where $\iota$ is the inner automorphism
\begin{equation}\label{def_iota}
\iota:e\in\CE\to \g^{(dc+\bd\bc)/2}e\g^{-(dc+\bd\bc)/2},
\end{equation}
i.e. $\iota(x_k^\pm)=\g^{(-kc\pm\bc)/2}x_k^\pm$ and $\iota(a_k)=\g^{-kc/2}a_k$.
\end{lemma}

\begin{proof}
Recall the action of $\CS^2$ as given in \cite{Bourgine2018a} (we use the same conventions here),
\begin{equation}
\CS^2(x_k^\pm)=-x_{-k}^\mp,\quad \CS^2(a_k)=-a_{-k},\quad \CS^2(c,\bc)=(-c,-\bc).
\end{equation} 
Obviously, we have $\CS^4=1$. Then, a direct computation gives
\begin{equation}
\CS^2\CT\CS^{-2}(x_k^\pm)=x_{k\mp1}^\pm,\quad \CS^2\CT\CS^{-2}(a_k)=a_{k},\quad \CS^2\CT\CS^{-2}(c,\bc)=(c,c+\bc),
\end{equation} 
and thus $[\CT,\CS^2]=[\CT^\perp,\CS^2]=0$. The action of $\CT^\perp$ on the generators is the result of direct calculation. The compatibility of $\CT$ with the co-algebraic structure follows from a direct comparison,
\begin{align}
\begin{split}
&\D(\CT(x^\pm(z))=z^{\mp1}\D(x^\pm(z)),\quad \D(\CT(\psi^\pm(z))=\g^{\mp c_{(1)}\mp c_{(2)}}\D(\psi^\pm(z)),\\
&(\CT\otimes\CT)\circ\D(x^+(z))=z^{-1}x^+(z)\otimes 1+ (\g^{c_{(1)}}z)^{-1}\g^{c_{(1)}}\psi^-(\g^{c_{(1)}/2}z)\otimes x^+(\g^{c_{(1)}}z),\\
&(\CT\otimes\CT)\circ\D(x^-(z))=(\g^{c_{(2)}} z)\g^{-c_{(2)}}x^-(\g^{c_{(2)}} z)\otimes \psi^+(\g^{c_{(2)}/2}z)+z1\otimes x^-(z),\\
&(\CT\otimes\CT)\circ\D(\psi^\pm(z))=\g^{\mp c_{(1)}\mp c_{(2)}}\psi^\pm(\g^{\pm c_{(2)}/2}z)\otimes\psi^\pm(\g^{\mp c_{(1)}/2}z)
\end{split}
\end{align}
The compatibility of $\CT^\perp$ with $\D^\perp$ is obtained by inserting the relation $\CT=\CS^{-1}\CT^\perp \CS$ in the relation $\D(\CT(e))=(\CT\otimes\CT)\circ\D(e)$.

To prove the identity \eqref{id_TpTTp}, it is sufficient to examine the transformation of the generators of $\CE$. For the central elements, we have
\begin{equation}
\CT^\perp\CT\CT^\perp(c,\bc)=\CT^\perp\CT(c-\bc,\bc)=\CT^\perp(c-\bc,c)=(-\bc,c)=\iota\CS(c,\bc).
\end{equation}
In a similar way,
\begin{align}
\begin{split}
\CT^\perp\CT\CT^\perp(a_{\pm1})&=(\g-\g^{-1})\CT^\perp\CT(\g^{\mp(c-\bc)/2}x^\pm_{\pm1})=(\g-\g^{-1})\CT^\perp(\g^{\pm\bc/2}x^\pm_{0})\\
&=(\g-\g^{-1})\g^{\pm\bc/2}x^\pm_{0}=\iota\CS(a_{\pm1}),\\
\CT^\perp\CT\CT^\perp(x_{0}^\pm)&=\CT^\perp\CT(x_{0}^\pm)=\CT^\perp(x_{\mp1}^\pm)=-\dfrac{\g^{\pm c/2}}{\g-\g^{-1}}a_{\mp1}=\iota\CS(x_0^\pm).
\end{split}
\end{align}
Finally, we note that
\begin{equation}
(\CS\CT^{-1})^3=\CS\CT^{-1}\CS\CT^{-1}\CS\CT^{-1}=\CT^{\perp-1}\CS^2\CT^{-1}\CT^{\perp-1}\CS=\CT^{\perp-1}\CT^{-1}\CT^{\perp-1}\CS^3=\CS^{-1}\iota^{-1}\CS^3=\iota^{-1}\CS^2.
\end{equation}
\end{proof}

\paragraph{Anti-automorphisms.} The algebra $\CE$ possess two involutive anti-automorphisms $\s_H$ and $\s_V$ defined as follows,
\begin{align}\label{def_sV_sH}
\begin{split}
&\s_H(x_k^\pm)=x_{-k}^\pm,\quad \s_H(\psi_{\pm k}^\pm)=\psi_{\mp k}^\mp,\quad \s_H(a_k)=-a_{-k},\quad \s_H(c,\bc)=(c,-\bc),\\
&\s_V(x_k^\pm)=-x_k^\mp,\quad \s_V(\psi_{\pm k}^\pm)=\psi_{\pm k}^\pm,\quad \s_V(a_k)=a_k,\quad \s_V(c,\bc)=(-c,\bc).
\end{split}
\end{align}
The anti-automorphisms extend to $\CE'$ defining $\s_H(d,\bd)=(d,-\bd)$ and $\s_V(d,\bd)=(-d,\bd)$.

\subsection{Representations}
In this paper, we only consider representations for which the central elements $(c,\bc)$ are integer. We denote these representations $\rho_\bsv^{(\ell_1,\ell_2)}$ where $\bsv$ is a complex vector of weights, and $\rho_\bsv^{(\ell_1,\ell_2)}(c)=\ell_1$, $\rho^{(\ell_1,\ell_2)}_\bsv(\bar c)=\ell_2$. The corresponding modules will be denoted $\CF_\bsv^{(\ell_1,\ell_2)}$. Strictly speaking, these modules do not depend on the weights $\bsv$, but it is convenient to keep track of them with this notation.

Outer automorphisms can be used to define new representations by composition. This property is particularly useful in the case of Miki's automorphism for which it can be interpreted as bispectral duality. The new representation $\rho_\bsv^{(\ell,\bell)}\circ\CS$ has levels $(-\bell,\ell)$, which suggests that it is isomorphic to the representation $\rho_{\tilde{\bsv}}^{(-\bell,\ell)}(e)$, i.e.
\begin{equation}\label{prop_S}
\rho^{(\ell,\bell)}_\bsv\circ\CS(e)=\CM_{\CS}^{(\ell,\bell)-1}\rho_{\tilde{\bsv}}^{(-\bell,\ell)}(e) \CM_{\CS}^{(\ell,\bell)},
\end{equation}
with the transformation matrix $\CM_\CS^{(\ell,\bell)}:\CF_{\bsv}^{(\ell,\bell)}\to \CF_{\tilde{\bsv}}^{(-\bell,\ell)}$. The precise relation between the weights $\bsv$ and $\tilde{\bsv}$ depends on the representation. To show that such an isomorphism exists, we only need to study the four generators $a_{\pm1}$ and $x_0^\pm$. This isomorphism is known for the representations of levels $(1,0)$ and $(0,1)$ \cite{Feigin2009a,Bourgine2018a}, and we will use the $e_1$-Pieri rule in the next section to show that it extends to representations $(m,0)$ and $(0,m)$. In this subsection, we recall the definition of these representations.

\subsubsection{Vertical representations}
The vertical representations are highest $\ell$-weight representations of levels $(0,m)$ introduced in \cite{Feigin:2011qua}. Their module $\CF_\bsv^{(0,m)}$ is spanned by a basis of states $\dket{\bl}$ labelled by $m$-tuple Young diagram $\bl=(\l^{(1)},\dots,\l^{(m)}$), with the higest weight $\dket{\vac}$ corresponding to the $m$-tuple of empty partition. The four Drinfeld currents act on these states as follows:
\begin{align}\label{def_vert_rep}
\begin{split}
&\rho_{\bsv}^{(0,m)}(x^+(z))\dket{\bl}
 = \sum_{\sAbox\in A(\bl)}\delta(v_\sAbox\chi_\sAbox/z)\,
 r_\bl(\Abox|\bsv)\dket{\bl+\Abox},\\
&\rho_{\bsv}^{(0,m)}(x^-(z))\dket{\bl}
 = \g^{-m}\sum_{\sAbox\in R(\bl)}\delta(v_\sAbox\chi_\sAbox/z)\,
 r_\bl^\ast(\Abox|\bsv)\dket{\bl-\Abox},\\
&\rho_{\bsv}^{(0,m)}(\psi^\pm(z))\dket{\bl}
 = \g^{-m}\left[\Psi_{\bl}(z|\bsv)\right]_\pm\dket{\bl}.
\end{split}
\end{align}
where the matrix elements $r_\bl$, $r^\ast_\bl$ and $\Psi_{\bl}$ are defined as in \eqref{def_rl} and \eqref{eq:coloredCY-Psi}, respectively.
In this representation, the Cartan current $\psi^\pm(z)$ are diagonal, and their eigenvalues are expressed in terms of the functions $\Psi_\bl(z|\bsv)$ defined in \eqref{eq:coloredCY-Psi} and depending on the weights $\bsv=(v_1,\dots,v_m)$. The current $x^+(z)$ (resp.\ $x^-(z)$) acts as creation (resp.\ annihilation) operators, and the summation is over the set of addable (resp.\ removable) boxes defined earlier. Their matrix elements involve residues of the rational functions $\CY_{\bl}(z|\bsv)$. In a generic situation, i.e. $v_\a/v_\b \neq q_1^{n_1}q_2^{n_2}$ for all $\a,\b$ and $n_1,n_2\in\BZ$, these functions have only single poles and single zeros. We note that when $m=1$, the matrix elements $r_\l(\Abox)$ and $r_\l^\ast(\Abox)$ are, in fact, independent of the weight $v$ of the representation. 

By expanding the functions $\Psi_\bl(z)$ both around $z=\infty$ and $z=0$, it is possible to deduce the representation of the modes $a_k$ describing the Cartan sector:
\begin{equation}
\begin{aligned}
 \rho^{(0,m)}_{\bsv}(a_k)\dket{\bl}
 &= \dfrac1k(1-q_3^k)\sum_{\a=1}^m v_\a^k \, p_k(\me_{\l^{(\a)}})\dket{\bl} \\
 &= c_k\left(\sum_{\sAbox\in\bl}v_\sAbox^k\chi_\sAbox^k-\dfrac1{(1-q_1^k)(1-q_2^k)}
 \sum_{\a=1}^mv_\a^k\right)\dket{\bl},\quad k\in\mathbb{Z}\setminus\{0\},
\end{aligned}
\end{equation}
where $c_k$ is the coefficient defined in \eqref{exp_g} and for $k<0$ we used $p_k(\me_\l)=p_{-k}(\me_\l^\vee)$.

\paragraph{Contragredient representation.} The contragredient representation follows from the scalar product
\begin{equation}\label{scalar_vert}
 (\dket{\bl},\dket{\bmu})_V = \dbra{\bl}\!\!\dket{\bmu} = \d_{\bl,\bmu}\ a_{\bl}(\bsv)^{-1},
 \quad a_{\bl}(\bsv) = \g^{m|\bl|}\prod_{\a,\b=1}^{m}
 \tilde{N}_{\l^{(\a)},\l^{(\b)}}(v_\a/v_\b)^{-1}.
\end{equation}
We observe that the basis vectors are orthogonal, but not orthonormal. Note that in the case of $m=1$, the (inverse) norms $a_\l$ are in fact independent of the weight $v$ of the representation. The norm has been chosen such that the representation satisfies $\rho_{\bsv}^{(0,m)\dagger}=\rho_{\bsv}^{(0,m)}\circ\s_V$ where $\s_V$ is the anti-automorphism defined in \eqref{def_sV_sH} and ``${}^\dagger$'' is the adjoint w.r.t.\ the inner product defined in \eqref{scalar_vert}.

\paragraph{Coproduct construction.} Vertical representations of higher level $(0,m)$ can be constructed as a tensor product of representations $(0,1)$ by repeated application of the coproduct, namely $\D^{(m-1)}=(\D(\otimes\Id)^{m-2})\circ(\D(\otimes\Id)^{m-3})\circ\cdots \circ\D$, followed by the tensor product $\rho_{v_1}^{(0,1)}\otimes\cdots\otimes\rho_{v_m}^{(0,1)}$. The natural basis for such construction is $\dket{\bl}^\otimes=\dket{\l^{(1)}}\otimes\cdots\otimes\dket{\l^{(m)}}$. Note however that this basis does not coincide with the basis $\dket{\bl}$, as can be easily realized by observing the action of $x^\pm(z)$:
\begin{align}
\begin{split}
&\left(\rho_{v_1}^{(0,1)}\otimes\cdots\otimes\rho_{v_m}^{(0,m)}\right)\D^{(m-1)}(x^+(z))\dket{\bl}^\otimes\\
&\hspace{30pt}=\sum_{\a=1}^m\sum_{\sAbox\in A(\l^{(\a)})}\delta(v_\sAbox\chi_\sAbox/z)\prod_{\b<\a}(\g^{-1}\Psi_{\l^{(\b)}}(z/v_\b))\ r_{\l^{(\a)}}(\Abox)\dket{\bl+\Abox}^\otimes,\\
\end{split}
\end{align}
and
\begin{align}
\begin{split}
&\left(\rho_{v_1}^{(0,1)}\otimes\cdots\otimes\rho_{v_m}^{(0,m)}\right)\D^{(m-1)}(x^-(z))\dket{\bl}^\otimes\\
&\hspace{30pt}=\gamma^{-1}\sum_{\a=1}^m\sum_{\sAbox\in R(\lambda^{(\a)})}\delta(v_\sAbox\chi_\sAbox/z)\prod_{\b>\a}(\g^{-1}\Psi_{\l^{(\b)}}(z/v_\b))\ r^\ast_{\l^{(\a)}}(\Abox)\dket{\bl-\Abox}^\otimes.
\end{split}
\end{align}
States of the two basis are simply related by a normalization factor, i.e. $\dket{\bl}^\otimes=G_{\bl}(\bsv)\dket{\bl}$. The expresssion of the normalization factor $G_{\bl}(\bsv)$ is given in Appendix~\ref{app_A1}, formula \eqref{def_gl}. This factor is such that all Young diagrams that make up $\bl$ in $\dket{\bl}$ play a symmetric role in the representation $(0,m)$. This choice is very natural from the point of view of the ADHM construction of the instanton moduli space, where the $r$-tuples of Young diagrams label the fixed points under a torus action. Note also that the action of the Cartan currents $\psi^\pm(z)$ takes the same form in both basis since it is diagonal, and so does not depend on the normalization of the states.

From the decomposition
\begin{equation}
a_{\bl}(\bsv)=\left(\prod_{\a=1}^m a_{\l^{(\a)}}\right) G_\bl(\bsv)G_\bl^\ast(\bsv),
\end{equation}
we deduce that the dual basis are related by the normalization factor $^\otimes\dbra{\bl}=G^\ast_{\bl}(\bsv)\dbra{\bl}$. The following relations will be used to express the vertical components of the generalized vertex operators defined in the next subsection,
\begin{equation}\label{eq:rule_G}
\prod_{\a=1}^m a_{\l^{(\a)}}\ \dket{\bl}^\otimes=G_\bl^\ast(\bsv)^{-1}a_\bl(\bsv)\dket{\bl},\quad \prod_{\a=1}^m a_{\l^{(\a)}}\ ^\otimes\dbra{\bl}=G_\bl(\bsv)^{-1}a_\bl(\bsv)\dbra{\bl}.
\end{equation}

\subsubsection{Horizontal representations}
The horizontal Fock representation $\rho_u^{(1,n)}$ has levels $(1,n)$ and depends on a weight $u\in\mC^\times$. In this representation, the generators act on a Fock space built upon the modes $J_k$ of a Heisenberg algebra, here normalised such that $[J_k,J_l]=k\d_{k+l}$,\footnote{Since there are no zero modes in this representation, we can identify $\ket{n}=\ket{\vac}$ and $\CF^{(1,n)}\simeq\CF^{(1,m)}$.}
\begin{equation}\label{def_CF_1n}
\CF_u^{(1,n)}=\mC[J_{-1},J_{-2},\dots]\ket{n},\quad\text{with}\quad J_{k>0}\ket{n}=0.
\end{equation}
The normal-ordering of vertex operators $:AB:$ is defined by placing the positive modes $J_{k>0}$ to the right. The Cartan modes $a_k$ are represented as the modes $J_k$ up to a factor depending on the quantum group parameters $q_1,q_2$,
\begin{equation}\label{q-osc-j}
\rho_u^{(1,n)}(a_k)=-\dfrac{\g^{-k/2}}{k}(1-q_2^k)(1-q_3^k)J_k,\quad \rho_u^{(1,n)}(a_{-k})=-\dfrac{\g^{-k/2}}{k}(1-q_1^k)(1-q_3^k)J_{-k},\quad (k>0).
\end{equation}
On the other hand, the currents $x^\pm(z)$ are expressed as vertex operators $\rho_u^{(1,n)}(x^\pm(z))=u^{\pm1}z^{\mp n}\eta^\pm(z)$ with
\begin{align}
\begin{split}\label{def_eta}
&\eta^+(z)=\exp\left(\sum_{k>0}\dfrac{z^{k}}{k}(1-q_1^k)J_{-k}\right)\exp\left(-\sum_{k>0}\dfrac{z^{-k}}{k}(1-q_2^k)J_k\right),\\
&\eta^-(z)=\exp\left(-\sum_{k>0}\dfrac{z^{k}}{k}\g^{k}(1-q_1^k)J_{-k}\right)\exp\left(\sum_{k>0}\dfrac{z^{-k}}{k}\g^{k}(1-q_2^k)J_k\right).
\end{split}
\end{align}
We note that $\rho^{(1,n)}_u=\rho_u^{(1,0)}\circ \CT^n$, and will focus primarily on $\rho_u^{(1,0)}$ in what follows. When $n=0$ and $u=1$, we will find it convenient to omit the corresponding representation $\rho_1^{(1,0)}$.

\paragraph{Symmetric functions.} The Fock modules $\CF^{(1,n)}$ are isomorphic as a ring to the space of symmetric functions $\L\equiv\L[\bsx]$ in one alphabet $\bsx$. This isomorphism is realized by sending the vacuum $\ket{n}$ to the constant function $1$, the creation operators $J_{-k}$ for $k>0$ to the power sums $p_k=p_k(\bsx)$ and the annihilation operators $J_k$ to the formal derivatives $k\frac{\p}{\p p_k(\bsx)}$.
As a result, the Cartan generators take the form
\begin{equation}\label{q-osc-p}
 \rho_u^{(1,0)}(a_k) = -\g^{-k/2}(1-q_2^k)(1-q_3^k)\dfrac{\p}{\p p_k(\bsx)},\quad
 \rho_u^{(1,0)}(a_{-k}) = -\dfrac{\g^{-k/2}}{k}(1-q_1^k)(1-q_3^k)p_k(\bsx),\quad (k>0).
\end{equation}
When no confusion ensues, we will omit to indicate the variables and simply denote $p_k=p_k(\bsx)$. On the other hand, the vertex operators $\eta^\pm(z)$ are expressed as follows,
\begin{align}
\begin{split}\label{eq:rho_x}
 & \eta^+(z) = \mathe^{\sum_{k>0}\frac{z^k}{k}(1-q_1^k)p_k}\,
 \mathe^{-\sum_{k>0}z^{-k}(1-q_2^k)\frac{\p}{\p p_k}}\,,\\
 & \eta^-(z) = \mathe^{-\sum_{k>0}\frac{z^{k}}{k}\g^{k}(1-q_1^k)p_k}\,
 \mathe^{\sum_{k>0}z^{-k}\g^{k}(1-q_2^k)\frac{\p}{\p p_k}}\,.
\end{split}
\end{align}
We also observe that in this representation $\psi^\pm(\g^{\mp\frac12}z)=\,:x^+(z)x^-(\g^{\mp1}z):$. Finally, we note that the grading operator $d$ is associated to the degree of homogeneity of symmetric functions, and so it can be represented as
\begin{equation}\label{def_L0}
\rho_u^{(1,0)}(d)=L_0=\sum_{k>0}kp_k\dfrac{\p}{\p p_k}.
\end{equation}

\paragraph{Macdonald functions and Miki's automorphism.} In the horizontal representation $\rho_u^{(1,0)}$, the elements $b_k$ form a commutative subalgebra. This algebra includes the zero modes $x_0^\pm\propto b_{\pm1}$ of the Drinfeld current. Under the identification of the Fock space with the ring of symmetric functions, the element $x_0^+$ coincides with the (bosonized) Macdonald operator with parameters $(q_1,q_2)=(t^{-1},q)$. This operator is known to be diagonalized by the Macdonald symmetric functions $P_\l(\bsx)$, with eigenvalues $u\,\me_\l$, where $\me_\l$ is given in \eqref{expr_el}. On this basis, the operator $x_0^-$ is also diagonal, with eigenvalues $u^{-1}\me_\l^\vee$,
\be
 \rho^{(1,0)}_u(x^+_0)\,P_\l(\bsx) = u\,\me_\l\,P_\l(\bsx)\,,
 \hspace{30pt}
 \rho^{(1,0)}_u(x^-_0)\,P_\l(\bsx) = u^{-1}\me^\vee_\l\,P_\l(\bsx)\,.
\ee

The action of the algebra $\CE$ on Macdonald functions can be reconstructed from the action of the generators $x_0^\pm$ and $a_{\pm1}$. In particular, the action of the operator $a_{-1}\propto e_1(\bsx)$ follows from the Pieri rules. With Macdonald functions normalized as $P_\l(\bsx) = m_\l(\bsx)+\dots$, where $m_\l(\bsx)$ are monomial symmetric functions, the Pieri rule for $e_1(\bsx)\equiv p_1(\bsx)$ takes the form (see \cite[Ch.VI, \textsection6]{Macdonald})
\be\label{PieriPsi}
 p_1\, P_\l = \sum_{\sAbox\in A(\l)} \psi_\l(\Abox)P_{\l+\sAbox}\,,
 \hspace{30pt}
 \psi_\l(i,j) = \prod_{i'=1}^{i-1}\dfrac{(1-q^{\l_{i'}-j+1}t^{i-i'-1})(1-q^{\l_{i'}-j}t^{i-i'+1})}{(1-q^{\l_{i'}-j+1}t^{i-i'})(1-q^{\l_{i'}-j}t^{i-i'})}\,.
\ee 
In addition, we have (see Appendix~\ref{app_A})
\be\label{dual_Pieri}
 \dfrac{\p}{\p p_1}P_\l = \sum_{\sAbox\in R(\l)}\psi_\l^\ast(\Abox)P_{\l-\sAbox}\,,
 \hspace{30pt}
 \psi_\l^\ast(i,j) = \prod_{j'=1}^{j-1}\dfrac{(1-q^{j-j'-1}t^{\l^T_{j'}-i+1})(1-q^{j-j'+1}t^{\l^T_{j'}-i})}{(1-q^{j-j'}t^{\l^T_{j'}-i+1})(1-q^{j-j'}t^{\l^T_{j'}-i})}\,.
\ee

This formula determines the action of the quantum toroidal $\gl(1)$ algebra on the Macdonald basis,
\be\label{act_apm1}
\begin{aligned}
 \rho_u^{(1,0)}(a_{-1})P_\l &= \g^{1/2}(\g-\g^{-1})(1-q_1)\sum_{\sAbox\in A(\l)}\psi_\l(\Abox)\ P_{\l+\sAbox},\\
 \rho_u^{(1,0)}(a_1)P_\l &= \g^{1/2}(\g-\g^{-1})(1-q_2)\sum_{\sAbox\in R(\l)}\psi_\l^\ast(\Abox)\ P_{\l-\sAbox}.
\end{aligned}
\ee
We also recall the definition of the Macdonald inner product
\be
\label{eq:def-Macdonald-inner}
 \langle P_\mu,P_\l\rangle_{q,t} = \delta_{\mu,\l} b_\lam^{-1}\,,
 \hspace{30pt}
 b_\l = \prod_{(i,j)\in\l}\dfrac{1-q^{\l_i-j}t^{\l_j^T-i+1}}{1-q^{\l_i-j+1}t^{\l_j^T-i}},
\ee
and the associated reproducing kernel,
\be
\label{Mac_kernel}
 \Pi(\bsx|\bsy)
 = \sum_\l b_\l P_\l(\bsx)P_\l(\bsy)
 = \exp\left(\sum_{k>0}\frac1k\frac{1-t^k}{1-q^k}p_k(\bsx)p_k(\bsy)\right)\,.
\ee 
The former can be defined in the power sum basis as
\begin{equation}
 \la p_\l,p_\mu\ra_{q,t}
 = \d_{\l,\mu} z_\l \prod_{k\in\lam} \dfrac{1-q^k}{1-t^k},
\end{equation}
where the combinatorial coefficient $z_\l$ is given in the paragraph \textbf{Partitions} of the Introduction.

\paragraph{Normalization.} To study the representations of the quantum toroidal $\gl(1)$ algebra, it is convenient to introduce a different normalization for the Macdonald symmetric functions, and introduce the \emph{spherical Macdonald functions} $\tP_\mu$ as the ordinary Macdonald functions normalized in the following way
\be
\label{eq:spherical-macdonalds}
 \tP_\l(\bsx) := \frac{P_\l(\bsx)}{P_\l(\sp_\emptyset)},
 \hspace{30pt}
 P_\l(\sp_\emptyset) = (-1)^{|\l|}\prod_{(i,j)\in\l}\left[\dfrac{q_1^{i-1}q_2^{j-1}}
 {1-q_1^{-1}q_2^{j-1}} \prod_{i'=1}^{i-1} \dfrac{1-q_1^{1+i-i'}q_2^{j-\l_{i'}}}
 {1-q_1^{i-i'}q_2^{j-\l_{i'}}}\right].
\ee
With this normalization, the $e_1$-Pieri rule and its dual take the form
\be
\begin{aligned}
 p_1\,\tP_\l &= \sum_{\sAbox\in A(\l)} \tilde{\psi}_\l(\Abox)\ \tP_{\l+\sAbox},
 \hspace{30pt} \tilde{\psi}_\l(\Abox) = \frac{P_{\l+\sAbox}(\sp_\vac)}{P_{\l}(\sp_\vac)}
 \psi_\l(\Abox) = -\dfrac1{1-t}r_\l(\Abox),\\
 \dfrac{\p}{\p p_1}\tP_\l &= \sum_{\sAbox\in R(\l)} \tilde{\psi}^\ast_\l(\Abox)\
 \tP_{\l-\sAbox}, \hspace{30pt}
 \tilde{\psi}^\ast_\l(\Abox) = \frac{P_{\l-\sAbox}(\sp_\vac)}{P_{\l}(\sp_\vac)}
 \psi^\ast_\l(\Abox) = \dfrac{q}{1-q}r^\ast_\l(\Abox),
\end{aligned}
\ee
with the coefficient $r_\l(\Abox)$ defined in \eqref{def_rl} (See Appendix~\ref{app:pieri} for the proof). In fact, with this choice of normalization, the action of the generators $a_{\pm1}$ takes the form
\be
\begin{aligned}
 \rho^{(1,0)}_u(a_{-1})\tP_\l &= q_1\g^{1/2}(\g-\g^{-1})\sum_{\sAbox\in A(\l)}r_\l(\Abox)\, \tP_{\l+\sAbox},\\
 \rho^{(1,0)}_u(a_1)\tP_\l &= q_2\g^{1/2}(\g-\g^{-1})\sum_{\sAbox\in R(\l)}r^\ast_\l(\Abox)\, \tP_{\l-\sAbox}.
\end{aligned}
\ee
This shows that the representations $\rho_v^{(1,0)}\circ\CS$ and $\rho_u^{(0,1)}$ are indeed isomorphic with the relation $u=-\g v$ between the weights, and the transformation $\CM_\CS$ sending the vertical basis vectors $\dket{\l}$ to the Macdonald basis elements $\g^{|\l|/2}(-q_1)^{|\l|}\tP_\l(\bsx)$. With this normalization, the Macdonald kernel takes the form
\be
 \Pi(\bsx|\bsy) = \sum_\l \tilde{b}_\l\tP_\l(\bsx)\tP_\l(\bsy)\,,
 \hspace{30pt}
 \text{with}
 \hspace{30pt}
 \tilde{b}_\l = b_\l P_\l(\sp_\emptyset)^2 = \dfrac{q^{-|\l|}}{\tilde{N}_{\l,\l}(1)},
\ee 
where $\tilde{N}_{\l,\mu}(Q)$ is the modified Nekrasov factor defined in Appendix~\ref{app_A1}. From this isomorphism, we deduce the eigenvalue of the higher Macdonald operators,
\be\label{eigen_bk}
 \rho_u^{(1,0)}(b_k) \tP_\l
 = \dfrac{(-1)^{k-1}}{k}(\g^k-\g^{-k}) u^k p_k(\me_\lam)\, \tP_\l\,.
\ee

\paragraph{Higher horizontal representations.} Representations of levels $(r,\bsn)$ can be introduced by evaluating the $(r-1)$-th coproduct $\D^{(r-1)}=(\D(\otimes\Id)^{r-2})\circ\cdots\circ(\D\otimes\Id)\circ\D$ of the algebra elements in the tensor product $\rho_{u_1}^{(1,n_1)}\otimes\cdots\otimes\rho_{u_{r}}^{(1,n_{r})}$ of horizontal representations. We denote these representations $\rho^{(r,\bsn)}_\bsu$ with $\bsn=(n_1,\dots,n_{r})$ and $\bsu=(u_1,\dots,u_{r})$, they act on the Fock space $\CF_{\bsu}^{(r,\bsn)}=\CF_{u_1}^{(1,n_1)}\otimes\cdots\otimes\CF_{u_{r}}^{(1,n_{r})}$ and we will use the notation $J_k^{(\a)}=(1\otimes)^{\a-1}J_k(\otimes1)^{r-\a}$ for the Heisenberg modes acting on the $\a$-th Fock space.
In particular, we have
\be
\begin{aligned}
 \rho_\bsu^{(r,\bsn)}(a_k) &= -\dfrac{\g^{-rk/2}}{k}(1-q_2^k)(1-q_3^k)
 \sum_{\a=1}^r \g^{(\a-1)k}J_k^{(\a)},\\
 \rho_\bsu^{(r,\bsn)}(a_{-k}) &= -\dfrac{\g^{-rk/2}}{k}(1-q_1^k)(1-q_3^k)
 \sum_{\a=1}^r \g^{(\a-1)k}J_{-k}^{(\a)}.
\end{aligned}
\ee

By construction, the Fock space $\CF_{\bsu}^{(r,\bsn)}$ is isomorphic to the ring of symmetric functions in $r$ alphabets $\L^{\otimes r}=\L[\bsx^{(1)}]\otimes\cdots\otimes \L[\bsx^{(r)}]=\L[\bsx^{(1)},\dots,\bsx^{(r)}]$. For convenience, we will introduce factors of $\g$ in this isomorphism and define
\be
\label{eq:iso-J-to-p}
 J_{-k}^{(\a)} = \g^{-(\a-1)k} p_k^{(\a)},
 \hspace{30pt}
 J_k^{(\a)} = k\g^{(\a-1)k}\dfrac{\p}{\p p_k^{(\a)}}
\ee
so that
\be
\begin{aligned}
 \rho_\bsu^{(r,\bsn)}(a_k) &= -\g^{-rk/2}(1-q_2^k)(1-q_3^k)
 \sum_{\a=1}^r q_3^{(\a-1)k}\dfrac{\p}{\p p_k^{(\a)}}, \\
 \rho_\bsu^{(r,\bsn)}(a_{-k}) &= -\dfrac{\g^{-rk/2}}{k}(1-q_1^k)(1-q_3^k)
 \sum_{\a=1}^r p_k^{(\a)}.    
\end{aligned}
\ee
When $|\bsn|=\sum_\a n_\a=0$, the operators $\rho_\bsu^{(r,\bsn)}(b_k)$ form again a commutative subalgebra containing the elements $\rho_\bsu^{(r,\bsn)}(x_0^\pm)$. We can write down explicitly the corresponding currents as a sum of vertex operators,
\begin{multline}
\label{eq:x(z)+}
 \rho^{(r,\bsn)}_\bsu(x^+(z)) = \sum_{\a=1}^r u_\a\g^{\sum_{\b=1}^{\a-1}n_\b}z^{-n_\a}
 \exp\left(\sum_{k>0}\frac{z^k}{k}(1-q_1^k)\Big((1-q_3^k)\sum_{\b=1}^{\a-1}p_k^{(\b)}+p_k^{(\a)}\Big)\right)\\
 \times \exp\left(-\sum_{k>0}z^{-k}(1-q_2^k)\frac{\p}{\p p_k^{(\a)}}\right),
\end{multline}
and
\begin{multline}
\label{eq:x(z)-}
 \rho^{(r,\bsn)}_\bsu(x^-(z)) = \sum_{\a=1}^r u_\a^{-1} \g^{-\sum_{\b=\a+1}^{r}n_\b}z^{n_\a}
 \exp\left(-\sum_{k>0}\frac{z^k}{k}\g^{rk}(1-q_1^k)q_3^{-(\a-1)k}p_k^{(\a)}\right)\\
 \times\exp\left(\sum_{k>0}z^{-k}q_3^{k}\g^{-rk}(1-q_2^k)\Big(q_3^{(\a-1)k}\frac{\p}{\p p_k^{(\a)}}+(1-q_3^{-k})\sum_{\b=\a+1}^{r}q_3^{(\b-1)k}\frac{\p}{\p p_k^{(\b)}}\Big)\right).
\end{multline}
Generalized Macdonald symmetric functions will be defined in Section~\ref{sec:gmp} as the common eigenfunctions of these operators inside of $\L^{\otimes r}$.

\subsection{Vertex operators}
Our construction is based on the reinterpretation of the formulas involving Macdonald symmetric functions using the vertex operators of the quantum toroidal $\gl(1)$ algebra. These vertex operators can be defined as intertwiners between the representations $(\rho_\bsv^{(0,m)}\otimes\rho_\bsu^{(r,\bsn)})\circ\D$ and $\rho_{\bsu'}^{(r,\bsn')}$. Here, we focus on the case $m=r$, and consider the vertex operators $\Phi^{(r,\bsn)}$ and $\Phi^{(r,\bsn)\ast}$ satisfying the following intertwining equations for any $e\in\CE$
\begin{align}
\begin{split}\label{prop_intw_r}
&\Phi^{(r,\bsn)}\left(\rho_{\bsv}^{(0,r)}\otimes\rho_\bsu^{(r,\bsn)}\right)\circ \D(e)=\rho_{\bsu'}^{(r,\bsn')}(e)\Phi^{(r,\bsn)},\\
&\Phi^{(r,\bsn)\ast}\rho_{\bsu'}^{(r,\bsn')}(e)=\left(\rho_{\bsv}^{(0,r)}\otimes\rho_\bsu^{(r,\bsn)}\right)\circ \D'(e)\Phi^{(r,\bsn)\ast},
\end{split}
\end{align}
where $\D'$ is the opposite coproduct. When $r=1$, this relation defines the vertex operators uniquely up to an overall normalization factor, they were obtained by Awata, Feigin and Shiraishi in \cite{AFS} and will be called here AFS vertex operators. For $r>1$, we will construct a set of solutions depending not only on the external weights $\bsu,\bsv$, but also on a set of internal parameters $\bsw$, by a gluing procedure on the AFS vertex operators. Vertex operators will be normalized such that
\begin{equation}\label{norm_VO}
\bra{\vac}\Phi^{(r,\bsn)}\left(\dket{\vac}\otimes\ket{\vac}\right)=\left(\dbra{\vac}\otimes\bra{\vac}\right)\Phi^{(r,\bsn)\ast}\ket{\vac}=1.
\end{equation}

\subsubsection{AFS vertex operators}
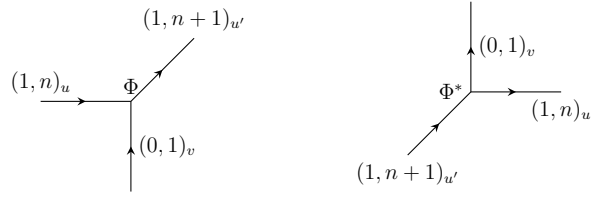
\begin{figure}
\begin{center}
\raisebox{-0.5\height}{\begin{tikzpicture}[scale=1.2]
\draw[postaction={on each segment={mid arrow=black}}] (0,0) -- (1,0) -- (1.7,0.7);
\draw[postaction={on each segment={mid arrow=black}}] (1,-1) -- (1,0);
\node[above,scale=0.7] at (1,0) {$\Phi$};
\node[above,scale=0.7] at (0,0) {$(1,n)_{u}$};
\node[right,scale=0.7] at (1,-0.5) {$(0,1)_{v}$};
\node[above,scale=0.7] at (1.7,0.7) {$(1,n+1)_{u'}$};
\end{tikzpicture}}
\hspace{10mm}
\raisebox{-0.5\height}{\begin{tikzpicture}[scale=1.2]
\draw[postaction={on each segment={mid arrow=black}}] (-0.7,-0.7) -- (0,0) -- (0,1);
\draw[postaction={on each segment={mid arrow=black}}] (0,0) -- (1,0);
\node[left,scale=0.7] at (0,0) {$\Phi^{\ast}$};
\node[below,scale=0.7] at (1,0) {$(1,n)_{u}$};
\node[right,scale=0.7] at (0,0.5) {$(0,1)_{v}$};
\node[below,scale=0.7] at (-0.7,-0.7) {$(1,n+1)_{u'}$};
\end{tikzpicture}}
\end{center}
\caption{AFS vertex operators}
\label{fig_AFS}
\end{figure}

When $r=1$, the intertwining relations \eqref{prop_intw_r} reduce to 
\be
\label{prop_intw}
\begin{aligned}
 &\Phi^{(1,n)}[u,v]\left(\rho_{v}^{(0,1)}\otimes\rho_u^{(1,n)}\ \D(e)\right)
 =\rho_{u'}^{(1,n+1)}(e)\,\Phi^{(1,n)}[u,v],\\
 &\Phi^{(1,n)\ast}[u,v]\,\rho_{u'}^{(1,n+1)}(e)
 =\left(\rho_{v}^{(0,1)}\otimes\rho_u^{(1,n)}\ \D'(e)\right)\Phi^{(1,n)\ast}[u,v].
\end{aligned}
\ee
These relations can be solved under the constraint $u'=-\g uv$. The solutions can be expanded in the vertical basis as follows
\be
\label{def_Phi}
 \Phi^{(1,n)}[u,v]=\sum_{\l} a_\l \dbra{\l}\otimes \Phi_{\l}^{(1,n)}[u,v],
 \hspace{30pt}
 \Phi^{(1,n)\ast}[u,v]=\sum_{\l}a_\l\dket{\l}\otimes\Phi_{\l}^{(1,n)\ast}[u,v],
\ee
where the coefficients $a_\l$, defined in \eqref{scalar_vert}, take care of the fact that the vertical basis is not orthonormal. The vertical components are vertex operators built as normal-ordered products of the vertex operators $\eta^\pm(z)$ defined in \eqref{def_eta},
\be
\label{def_AFS}
\begin{aligned}
 \Phi_\l^{(1,n)}[u,v] &= t_\l^{(1,n)}[u,v]\,
 :\Phi_\vac(v)\prod_{\sAbox\in\l}\eta^+(v\chi_\sAbox ):\,
 =t_\l^{(1,n)}[u,v]\,\mathe^{\sum_{k>0}\frac{v^k}{k}\frac{1-t^k}{1-q^k}p_k(\sp_\l)J_{-k}}
 \mathe^{-\sum_{k>0}\frac{v^{-k}q_3^{-k}}{k} p_k(\sp_\l^\vee)J_k},\\
 \Phi_{\l}^{(1,n)\ast}[u,v] &= t_{\l}^{(1,n)\ast}[u,v]\,
 :\Phi_\vac^\ast(v)\prod_{\sAbox\in\l}\eta^-(v\chi_\sAbox ):\,
 =t_\l^{(1,n)\ast}[u,v]\,\mathe^{-\sum_{k>0}\frac{v^k\g^k}{k}\frac{1-t^k}{1-q^k}p_k(\sp_\l)J_{-k}}
 \mathe^{\sum_{k>0}\frac{v^{-k}\g^{-k}}{k} p_k(\sp_\l^\vee)J_k},
\end{aligned}
\ee
with
\be
 \Phi_\vac(v) = \mathe^{-\sum_{k>0}\frac{v^k}{k(1-q_2^k)}J_{-k}}
 \mathe^{\sum_{k>0}\frac{v^{-k}q_3^{-k}}{k(1-q_1^{k})}J_k},\quad
 \Phi_\vac^\ast(v) = \mathe^{\sum_{k>0}\frac{\g^k v^k}{k(1-q_2^k)}J_{-k}}
 \mathe^{-\sum_{k>0}\frac{v^{-k}\g^{-k}}{k(1-q_1^{k})}J_k},
\ee
and
\be
 t_\l^{(1,n)}[u,v] = (-\g u v)^{|\l|} \prod_{\sAbox\in\l}(v\chi_\sAbox)^{-n-1},
 \hspace{30pt}
 t_{\l}^{(1,n)\ast}[u,v] = (\g u)^{-|\l|} \prod_{\sAbox\in\l}(v\chi_\sAbox)^n.
\ee
These operators have been represented in Figure~\ref{fig_AFS}. They can be glued together in two different ways. If the common representation is a vertical one, the gluing corresponds to the scalar product \eqref{scalar_vert}, and leads to a summation over partitions $\l$. On the other hand, if the common representation is a horizontal one, the gluing is realized by a simple product of vertex operators in the common Fock space.

\paragraph{Matrix elements.} The matrix elements of the AFS vertex operators in the Macdonald basis have been computed in \cite{AFS}. To state their result, we recall the definition of matrix elements of an operator $\CO\in\mathrm{End}(\L)$ in the Macdonald basis,
\be
 \bra{P_\l}\CO\ket{P_\mu} = \la P_\l,\CO\cdot P_\mu\ra_{q,t},
\ee
where $\la\cdot,\cdot\ra_{q,t}$ is the Macdonald scalar product defined in \eqref{eq:def-Macdonald-inner}. With this definition, the matrix elements read
\be
\label{AFS_matrix}
\begin{aligned}
 & \bra{P_\nu}\Phi_\l^{(1,n)}[u,v]\ket{P_\mu}
 = t_\l^{(1,n)}[u,v] v^{|\nu|-|\mu|}q_3^{-|\mu|}
 \sum_{\s} b_\s q_3^{|\s|}P_{\nu/\s}(\sp_\lam) P_{\mu/\s}(-\sp_\lam^\vee),\\
 & \bra{P_\nu}\Phi_\l^{(1,n)\ast}[u,v]\ket{P_\mu}
 = t_\l^{(1,n)\ast}[u,v] (\g v)^{|\nu|-|\mu|}
 \sum_\s b_\s P_{\nu/\s}(-\sp_\lam) P_{\mu/\s}(\sp_\lam^\vee),
\end{aligned}
\ee
where $\sp_\lam$ is defined in \eqref{expr_sp}.

\paragraph{Normal-ordering.} We denote $A(z)B(w)::f(z,w)$ the factor resulting from the normal-ordering of two vertex operators, i.e. $A(z)B(w)=f(z,w):A(z)B(w):$. Then, we have
\begin{align}\label{NO_Phi}
\begin{split}
&\Phi_\mu^{(1,n+1)}[u',v']\Phi_\l^{(1,n)}[u,v]::\dfrac{\CG(v/(q_3v'))}{N_{\l,\mu}(v/v')},\\
&\Phi_\mu^{(1,n)}[u',v']\Phi_\l^{(1,n)\ast}[u,v]::\dfrac{N_{\l,\mu}(\g v/v')}{\CG(v/(\g v'))},\\
&\Phi_\mu^{(1,n)\ast}[u',v']\Phi_\l^{(1,n)}[u,v]::\dfrac{N_{\l,\mu}(\g v/v')}{\CG(v/(\g v'))},\\
&\Phi_\mu^{(1,n-1)\ast}[u',v']\Phi_\l^{(1,n)\ast}[u,v]::\dfrac{\CG(v/v')}{N_{\l,\mu}(q_3v/v')}.
\end{split}
\end{align}

\subsubsection{Higher vertex operators}
\begin{figure}
\begin{center}
\begin{tikzpicture}[scale=0.7]
\draw[postaction={on each segment={mid arrow=black}}] (-1,0) -- (0,0) -- (0.7,0.7) -- (1.7,0.7) -- (2.4,1.4);
\draw[postaction={on each segment={mid arrow=black}}] (0,-1) -- (0,0) -- (0.7,0.7) -- (0.7,1.7) -- (2.1,3.1);
\draw[postaction={on each segment={mid arrow=black}}] (-1,1.7) -- (0.7,1.7);
\draw[postaction={on each segment={mid arrow=black}}] (1.7,-1) -- (1.7,0.7);
\node[below,scale=0.7] at (0,-1) {$(0,1)_{v_2}$};
\node[below,scale=0.7] at (1.7,-1) {$(0,1)_{v_1}$};
\node[left,scale=0.7] at (-1,0) {$(1,n_1)_{u_1}$};
\node[left,scale=0.7] at (-1,1.7) {$(1,n_2)_{u_2}$};
\node[right,scale=0.7] at (2.4,1.4) {$(1,n_1+1)_{u_1'}$};
\node[right,scale=0.7] at (2.1,3.1) {$(1,n_2+1)_{u_2'}$};
\end{tikzpicture}
\hspace{10mm}
\begin{tikzpicture}[scale=0.7]
\draw[postaction={on each segment={mid arrow=black}}] (-1.7,-.7) -- (-1,0) -- (0,0) -- (0.7,0.7) -- (1.7,.7);
\draw[postaction={on each segment={mid arrow=black}}] (-1.4,-2.4) -- (0,-1) -- (1.7,-1);
\draw[postaction={on each segment={mid arrow=black}}] (-1,0) -- (-1,1.7);
\draw[postaction={on each segment={mid arrow=black}}] (.7,.7) -- (.7,1.7);
\draw[postaction={on each segment={mid arrow=black}}] (0,-1) -- (0,0);
\node[above,scale=0.7] at (-1,1.7) {$(0,1)_{v_2}$};
\node[above,scale=0.7] at (.7,1.7) {$(0,1)_{v_1}$};
\node[right,scale=0.7] at (1.7,-1) {$(1,n_1)_{u_1}$};
\node[right,scale=0.7] at (1.7,.7) {$(1,n_2)_{u_2}$};
\node[below left,scale=0.7] at (-1.4,-2.4) {$(1,n_1+1)_{u_1'}$};
\node[below left,scale=0.7] at (-1.7,-.7) {$(1,n_2+1)_{u_2'}$};
\end{tikzpicture}
\end{center}
\caption{Diagram defining the vertex operators $\Phi^{(2,\bsn)}$ and $\Phi^{(2,\bsn)\ast}$.}
\label{fig_intw_m2}
\end{figure}

\begin{figure}
\begin{center}
\begin{tikzpicture}[scale=0.7]
\draw[postaction={on each segment={mid arrow=black}}] (-1,0) -- (0,0) -- (0.7,0.7) -- (1.7,0.7) -- (2.4,1.4) -- (3.4,1.4) -- (4.1,2.1);
\draw[postaction={on each segment={mid arrow=black}}] (-1,1.7) -- (0.7,1.7) -- (1.4,2.4) -- (2.4,2.4) -- (3.8,3.8);
\draw[postaction={on each segment={mid arrow=black}}] (-1,3.4) -- (1.4,3.4) -- (3.5,5.5);
\draw[postaction={on each segment={mid arrow=black}}] (0,-1) -- (0,0) -- (0.7,0.7) -- (0.7,1.7) -- (1.4,2.4) -- (1.4,3.4);
\draw[postaction={on each segment={mid arrow=black}}] (1.7,-1) -- (1.7,0.7) -- (2.4,1.4) -- (2.4,2.4);
\draw[postaction={on each segment={mid arrow=black}}] (3.4,-1) -- (3.4,1.4);
\node[below,scale=0.7] at (0,-1) {$(0,1)_{v_3}$};
\node[below,scale=0.7] at (1.7,-1) {$(0,1)_{v_2}$};
\node[below,scale=0.7] at (3.4,-1) {$(0,1)_{v_1}$};
\node[left,scale=0.7] at (-1,0) {$(1,n_1)_{u_1}$};
\node[left,scale=0.7] at (-1,1.7) {$(1,n_2)_{u_2}$};
\node[left,scale=0.7] at (-1,3.4) {$(1,n_3)_{u_3}$};
\node[right,scale=0.7] at (4.1,2.1) {$(1,n_1+1)_{u_1'}$};
\node[right,scale=0.7] at (3.8,3.8) {$(1,n_2+1)_{u_2'}$};
\node[right,scale=0.7] at (3.5,5.5) {$(1,n_3+1)_{u_3'}$};
\end{tikzpicture}
\hspace{10mm}
\begin{tikzpicture}[scale=0.7]
\draw[postaction={on each segment={mid arrow=black}}] (-2.7,.3) -- (-2,1) -- (-1,1) -- (-.3,1.7) -- (.7,1.7) -- (1.4,2.4) -- (2.4,2.4);
\draw[postaction={on each segment={mid arrow=black}}] (-2.05,-1.05) -- (-1,0) -- (0,0) -- (0.7,.7) -- (2.4,.7);
\draw[postaction={on each segment={mid arrow=black}}] (-1.4,-2.4) -- (0,-1) -- (2.4,-1);
\draw[postaction={on each segment={mid arrow=black}}] (-1,0) -- (-1,1);
\draw[postaction={on each segment={mid arrow=black}}] (.7,.7) -- (.7,1.7);
\draw[postaction={on each segment={mid arrow=black}}] (0,-1) -- (0,0);
\draw[postaction={on each segment={mid arrow=black}}] (-2,1) -- (-2,3.4);
\draw[postaction={on each segment={mid arrow=black}}] (-.3,1.7) -- (-.3,3.4);
\draw[postaction={on each segment={mid arrow=black}}] (1.4,2.4) -- (1.4,3.4);
\node[above,scale=0.7] at (-2,3.4) {$(0,1)_{v_3}$};
\node[above,scale=0.7] at (-.3,3.4) {$(0,1)_{v_2}$};
\node[above,scale=0.7] at (1.4,3.4) {$(0,1)_{v_1}$};
\node[right,scale=0.7] at (2.4,-1) {$(1,n_1)_{u_1}$};
\node[right,scale=0.7] at (2.4,.7) {$(1,n_2)_{u_2}$};
\node[right,scale=0.7] at (2.4,2.4) {$(1,n_3)_{u_3}$};
\node[below left,scale=0.7] at (-1.4,-2.4) {$(1,n_1+1)_{u_1'}$};
\node[below left,scale=0.7] at (-2.05,-1.05) {$(1,n_2+1)_{u_2'}$};
\node[below left,scale=0.7] at (-2.7,.3) {$(1,n_3+1)_{u_3'}$};
\end{tikzpicture}
\end{center}
\caption{Diagram defining the vertex operators $\Phi^{(3,\bsn)}$ and $\Phi^{(3,\bsn)\ast}$.}
\label{fig_intw_m3}
\end{figure}

Solutions of the intertwining equations \eqref{prop_intw} generalizing the AFS vertex operators to higher levels $r$ can be introduced by considering Figures~\ref{fig_intw_m2} and \ref{fig_intw_m3} and their obvious generalization. These figures represent the construction of the operators
\begin{equation}\label{def_Phi_rn}
 \Phi^{(r,\bsn)}:\CF_\bsv^{(0,r)}\otimes \CF_\bsu^{(r,\bsn)}\to \CF_{\bsu'}^{(r,\bsn')},
 \hspace{30pt}
 \Phi^{(r,\bsn)\ast}:\CF_{\bsu'}^{(r,\bsn')}\to \CF_\bsv^{(0,r)}\otimes \CF_\bsu^{(\r,\bsn)},
\end{equation}
from the gluing of elementary AFS vertex operators. This gluing is realized as a product of operators in the horizontal Fock spaces, and as a scalar product in the vertical modules. This construction provides solutions of the intertwining equation for $m=r$, and under the constraints $n'_\a=n_\a+1$ between levels of representations associated to external legs. These vertex operators depend on $2r$ external weights denoted collectively $\bsu=(u_1,\dots,u_r)$ and $\bsv=(v_1,\dots,v_r)$. In addition, they also depend on $r(r-1)/2$ internal vertical weights, denoted $\bsw=(w_{\a,\b})$ with $1\leq \b\leq \a\leq r-1$ which corresponds to an upper triangular matrix of size $(r-1)\times(r-1)$ (we also denote $w_{r,\b}=v_\b$ for $\b=1,\dots,r$). External weights are required to satisfy the relations
\begin{equation}\label{rel_weights}
\prod_{\a=1}^r u_\a'=\prod_{\a=1}^r (-\g u_\a v_\a),\quad u'_\a=
\begin{cases}
-\g u_\a w_{r-\a+1,1}\prod_{\b=1}^{r-\a}\dfrac{w_{r-\a+1,\b+1}}{w_{r-\a,\b}}, & \text{for }\Phi^{(r,\bsn)}\\
-\g u_\a w_{\a,\a}\prod_{\b=1}^{\a-1}\dfrac{w_{\a,\b}}{w_{\a-1,\b}}, & \text{for }\Phi^{(r,\bsn)\ast},
\end{cases}
\end{equation}
It is interesting to note that diagrams of the form \ref{fig_intw_m2}, \ref{fig_intw_m3} can be interpreted as $(p,q)$-brane webs defining 5d $\CN=1$ supersymmetric gauge theories \cite{Aharony1997,Aharony1997a}, or toric diagram defining topological strings amplitudes on toric Calabi--Yau threefolds \cite{Leung1998}. More specifically, these particular diagrams are associated to the trinion gauge theories, they have been studied extensively in this context (see e.g. \cite{Coman2019}).\footnote{J.E.B.\ would like to thank Elli Pomoni for drawing his attention to this type of brane webs many years ago.}

The operator $\Phi^{(r,\bsn)}$ is built by gluing $r(r+1)/2$ AFS vertex operators of type $\Phi$ and $r(r-1)/2$ of type $\Phi^\ast$, using $r(r-1)$ horizontal couplings and $r(r-1)/2$ vertical couplings, and the construction of $\Phi^{(r,\bsn)\ast}$ is parallel. The definition is independent of the choice of the vertical basis. It is convenient to use the symmetric basis to define the vertical components,
\begin{align}
\begin{split}
&\Phi^{(r,\bsn)}[\bsu,\bsv,\bsw]=\sum_{\bl} a_{\bl}(\bsv)\ \Phi_{\bl}^{(r,\bsn)}[\bsu,\bsv,\bsw]\  \dbra{\bl}\\
&\Phi^{(r,\bsn)\ast}[\bsu,\bsv,\bsw]=\sum_{\bl} a_\bl(\bsv)\ \Phi_{\bl}^{(r,\bsn)\ast}[\bsu,\bsv,\bsw]\  \dket{\bl}
\end{split}
\end{align}
These operators can be constructed recursively in three different ways, and this is used in Appendix~\ref{sec_higher_VO} to show that they indeed solve the intertwining relation \eqref{prop_intw_r}.

\paragraph{Critical value.} For certain critical values of the internal weights $\bsw$, the expression of the vacuum components of the vertex operators greatly simplifies. The reason for this simplification is the presence of zeros coming from Nekrasov factors (and more specifically the relations \eqref{N_vac}) that forces the summations over internal Young diagrams to reduce to a single term corresponding to the empty partition. For each vertex operators, there are two critical values that correspond to restrict internal partitions summations either from the left or the right.

In the case of the vertex operator $\Phi^{(r,\bsn)}[\bsu,\bsv,\bsw]$, we find the critical values $w_{\a,\b}^{(I)}=\g^{r-\a}v_\b$ and $w_{\a,\b}^{(II)}=\g^{r-\a}v_{r-\a+\b}$, corresponding respectively to the formulas
\be
\begin{aligned}
\Phi_{\vac}^{(r,\bsn)}[\bsu,\bsv,\bsw^{(I)}]&=\mathe^{-\sum_{k>0}\frac1{k(1-q_2^k)}\sum_{\a=1}^r\left(v_{r-\a+1}^k+(1-q_3^k)\sum_{\b=\a+1}^{r}v_{r-\b+1}^k\right)p_k^{(\a)}}\mathe^{\sum_{k>0}\frac{q_3^{-k}}{1-q_1^k}\sum_{\a=1}^r v_{r-\a+1}^{-k}\frac{\p}{\p p_k^{(\a)}}},\\
\Phi_{\vac}^{(r,\bsn)}[\bsu,\bsv,\bsw^{(II)}]&=\mathe^{-\sum_{k>0}\frac1{k(1-q_2^k)}\sum_{\a=1}^r\left(v_\a^k+(1-q_3^k)\sum_{\b=\a+1}^r v_\b^k\right)p_k^{(\a)}}\mathe^{\sum_{k>0}\frac{q_3^{-k}}{1-q_1^k}\sum_{\a=1}^r v_\a^{-k}\frac{\p}{\p p_k^{(\a)}}}.
\end{aligned}
\ee
We note that the two expressions only differ by the substitution $v_\a\to v_{r-\a+1}$. With these specializations, we have $u'_\a=-\g^{2\a-r}u_\a v_{r-\a+1}$ (I) and $u'_\a=-\g^{2\a-r}u_\a v_\a$ (II), respectively.

Similarly, in the case of the vertex operator $\Phi^{(r,\bsn)\ast}[\bsu,\bsv,\bsw]$, we find the two critical values $w^{(I)}_{\a,\b}=\g^{r-\a}v_\b$ and $w^{(II)}_{\a,\b}=\g^{r-\a}v_{r+\b-\a}$, corresponding to the expressions
\be
\begin{aligned}\label{Phi_vac_st}
\Phi_{\vac}^{(r,\bsn)\ast}[\bsv,\bsw^{(I)}]&=\mathe^{\sum_{k>0}\frac{\g^{rk}}{k(1-q_2^k)}\sum_{\a=1}^rq_3^{-(\a-1)k}v_\a^k p_k^{(\a)}}\mathe^{-\sum_{k>0}\frac{\g^{-rk}}{1-q_1^k}\sum_{\a=1}^rq_3^{(\a-1)k}\left(v_\a^{-k}+(1-q_3^{-k})\sum_{\b=1}^{\a-1}v_\b^{-k}\right)\frac{\p}{\p p_k^{(\a)}}},\\
\Phi_{\vac}^{(r,\bsn)\ast}[\bsv,\bsw^{(II)}]&=\mathe^{\sum_{k>0}\frac{\g^{rk}}{k(1-q_2^k)}\sum_{\a=1}^rq_3^{-(\a-1)k}v_{r-\a+1}^k p_k^{(\a)}}\mathe^{-\sum_{k>0}\frac{\g^{-rk}}{1-q_1^k}\sum_{\a=1}^rq_3^{(\a-1)k}\left(v_{r-\a+1}^{-k}+(1-q_3^{-k})\sum_{\b=1}^{\a-1}v_{r-\b+1}^{-k}\right)\frac{\p}{\p p_k^{(\a)}}}.
\end{aligned}
\ee
With these specializations, we have $u'_\a=-\g^{r+2-2\a}u_\a v_\a$ (I) and $u'_\a=-\g^{r+2-2\a}u_\a v_{r-\a+1}$ (II), respectively.

\subsubsection{Vacuum components}
It has been observed \cite{Bourgine2021b} that the vacuum component of the AFS vertex operators obey certain algebraic relations that can be obtained as a projection of the intertwining relations \eqref{prop_intw} on the vacuum state $\dket{\vac}$. These relations follow from the highest weight property of the vertical vacuum state, and so they naturally extend to vertex operators with higher levels. Let $\CN^\pm\subset\CE$ denote the nilpotent subalgebras generated by products of operators $x_k^\pm$ for $k\in\mZ$ and $\g^{a c}$ for $a\in\mC^\times$. We do not include powers of $\g^{a\bc}$ here since $\rho_{\bsv}^{(0,r)}(\g^{a\bc})=\g^{a r}$, and so it must be treated separately.  We have the following relations for the vacuum components of the vertex operators,
\begin{align}
\begin{split}\label{rel_intw_vac}
&\Phi_\vac^{(r,\bsn)}[\bsu,\bsv,\bsw]\rho^{(r,\bsn)}_{\bsu}(\g^{a\bc}e_-)=\g^{-a r}\rho^{(r,\bsn')}_{\bsu'}(\g^{a\bc}e_-)\Phi_\vac^{(r,\bsn)}[\bsu,\bsv,\bsw],\quad e_-\in\CN^-,\\
&\Phi_\vac^{(r,\bsn)\ast}[\bsu,\bsv,\bsw]\rho^{(r,\bsn')}_{\bsu'}(\g^{a\bc}e_+)=\g^{a r}\rho^{(r,\bsn)}_{\bsu}(\g^{a\bc}e_+)\Phi_\vac^{(r,\bsn)\ast}[\bsu,\bsv,\bsw],\quad e_+\in\CN^+.
\end{split}
\end{align}
In addition, we also have the following relations for the action of the modes $a_k$,
\begin{align}
\begin{split}\label{ak_intw_vac}
&\Phi_\vac^{(r,\bsn)}[\bsu,\bsv,\bsw]\rho^{(r,\bsn)}_{\bsu}(a_k)-\rho^{(r,\bsn')}_{\bsu'}(a_k)\Phi_\vac^{(r,\bsn)}[\bsu,\bsv,\bsw]=-\dfrac{\g^{-r|k|/2}}{k}(1-q_3^k)\left(\sum_{\a=1}^r v_\a^k\right) \Phi_\vac^{(r,\bsn)}[\bsu,\bsv,\bsw],\\
&\Phi_\vac^{(r,\bsn)\ast}[\bsu,\bsv,\bsw]\rho^{(r,\bsn')}_{\bsu'}(a_k)-\rho^{(r,\bsn)}_{\bsu}(a_k)\Phi_\vac^{(r,\bsn)\ast}[\bsu,\bsv,\bsw]=\dfrac{\g^{r|k|/2}}{k}(1-q_3^k)\left(\sum_{\a=1}^r v_\a^k\right) \Phi_\vac^{(r,\bsn)\ast}[\bsu,\bsv,\bsw].
\end{split}
\end{align}
In this paper, we will need the following properties of $\CN^\pm$.
\begin{lemma}
We have:
\begin{enumerate}
    \item[(i)] $\CT(\CN^\pm)\subset\CN^\pm$,
    \item[(ii)] $[a_k,\CN^\pm]\subset\CN^\pm$,
    \item[(iii)] Let $b_k=\CS(a_k)$ and $y_k^\pm=\CS(x_k^\pm)$, then we have for $k>0$, 
    \begin{equation}
    b_{\pm k}\in\CN^\pm,\quad \g^{\pm k\bc/2}y_k^\pm\in\CN^+,\quad \g^{\pm k\bc/2}y_{-k}^\pm\in\CN^-,
    \end{equation}
    \item[(iv)] $\g^{\mp k\bc/2}\CT^\perp(a_{\pm k})\in\CN^\pm$ for $k>0$.
\end{enumerate}
\end{lemma}

\begin{proof}
The first two properties follow, respectively, from the action of the automorphism $\CT$ on $x_k^\pm$, and the algebraic relation \eqref{com_ak}, i.e. $[a_k,x_l^\pm]\propto x^\pm_{k+l}$. The third property can be obtained from the explicit expression of the modes $b_k$ and $y_k^\pm$ found by Miki \cite{Miki2007}. To encode the modes $b_k$, it is convenient to introduce 
\be
 \xi^\pm(z) = \CS(\psi^\pm(z))=\xi_0^\pm \mathe^{\pm\sum_{k>0}z^{\mp k}b_{\pm k}}.
\ee
We have \cite[Appendix A]{Bourgine2021b},
\begin{align}
\begin{split}
&y^\pm_k=(\pm)^k\g^{-(c\pm k\bc)/2}\s_1^{-(k-1)}\left(\text{ad}_{x_0^+}\right)^{k-1} x_{\mp1}^+,\quad y^\pm_{-k}=-(\pm)^k \g^{(c\mp k\bc)/2} \s_1^{-(k-1)}\left(\text{ad}_{x_0^-}\right)^{k-1} x^-_{\mp1},\\
&\xi_{\pm k}^\pm=-(\mp)^k(\g-\g^{-1})\s_1^{-(k-1)}\g^{\mp c}\text{ad}_{x_{\mp1}^\pm}\left(\text{ad}_{x_0^\pm}\right)^{k-2}x_{\pm1}^{\pm},\quad \xi_{\pm1}^\pm=\pm\g^{\mp c}(\g-\g^{-1})x_0^\pm,\quad \xi_0^\pm=\g^{\mp c},
\end{split}
\end{align}
with $\s_1=(q_1^{1/2}-q_1^{-1/2})(q_2^{1/2}-q_2^{-1/2})$, which clearly shows that positive modes belong to $\CN^+$ and negative ones to $\CN^-$.

For $k=1$, the last property follows from the explicit expression \eqref{CTCS}. Let us consider $k>2$. By definition, $\CT^\perp(a_k)=-\CS\CT(b_{-k})$, so we need to show that $\g^{\mp k\bc/2}\CS\CT(b_{\mp k})\in\CN^\pm$, or, equivalently, $\g^{\mp\bc\pm k\bc/2}\CS\CT(\xi^\pm_{\pm k})\in \CN^\mp$ since $\g^{\mp\bc}\CS\CT(\xi_0^\pm)=1$. For $k>2$,
\begin{align}
\begin{split}
&\CT(\xi^\pm_{\pm k})=-(\mp)^k(\g-\g^{-1})\s_1^{-(k-1)}\g^{\mp c}\text{ad}_{x_{\mp2}^\pm}\left(\text{ad}_{x_{\mp1}^\pm}\right)^{k-2}x_{0}^{\pm},\\
\implies &\CS\CT(\xi^\pm_{\pm k})=(\mp)^k\s_1^{-(k-1)}\g^{\pm\bc}\text{ad}_{y_{\mp2}^\pm}\left(\text{ad}_{y_{\mp1}^\pm}\right)^{k-2}a_{\mp1}.
\end{split}
\end{align}
According to the second and third properties, $\g^{\mp\bc\pm k\bc/2}\CS\CT(\xi^\pm_{\pm k})\in[\CN^\mp,a_{\mp1}]\subset\CN^\mp$.
\end{proof}

\subsection{Framing operator}\label{sec_framing}
The GHT identity involves the operator $\nabla$ acting on symmetric functions and defined as the operator diagonal on the Macdonald basis with eigenvalues 
\be\label{eigen_nabla}
 \nabla\,P_\l(\bsx) = g_\l \,P_\l(\bsx)
\ee
with $g_\l$ defined as in \eqref{expr_gl}.
In order to derive an algebraic interpretation of the GHT identity, we need to introduce a version of this operator at the universal level. To do so, we use the automorphisms $\CT$ and $\CT^\perp$ of the quantum toroidal $\gl(1)$ algebra $\CE$ introduced previously to define extensions of this algebra, in a similar way as e.g. extended Weyl groups are defined. Hence, we define the extended algebra $\CE_\text{ext.}$ (resp. $\CE_\text{ext.}^\perp$) as the algebra generated by elements of $\CE$ and the invertible element $F$ (resp. $F^\perp$) satisfying the exchange relation
\begin{equation}\label{def_F}
eF=F\CT(e),\quad (\text{resp.}\quad e F^\perp=F^\perp \CT^\perp(e)),\quad \forall e\in\CE.
\end{equation} 
We note that the distributivity and associativity of the product is guaranteed by the fact that $\CT$ and $\CT^\perp$ are automorphisms. Invertibility follows from the fact that both $\CT$ and $\CT^\perp$ are invertible.

From the explicit action of the automorphisms on generators, we deduce that
\begin{equation}
[F,c]=[F,a_{\pm k}]=0,\quad [F,\bc]=-c F,\quad x^\pm_k F=F x^\pm_{k\mp1},
\end{equation}
and similarly
\begin{align}
\begin{split}
&[F^\perp,\bar c]=[F^\perp,x_0^\pm]=[F^\perp,b_k]=0,\quad [F^\perp,c]=\bc F^\perp,\\
&a_{\pm1}F^\perp=(\g-\g^{-1})\g^{\mp c/2}F^\perp x_{\pm1}^\pm,\quad x_{\mp1}^\pm F^\perp=-(\g-\g^{-1})^{-1}\g^{\pm c/2}F^\perp a_{\mp1}.
\end{split}
\end{align}
We note that the element $c$ (resp. $\bc$) is no longer central in the extended algebra. In particular, it implies that $F^\perp$ cannot be represented on an irreducible $\CE$-module, unless it has level $\bc=0$, which is indeed the case of horizontal modules $(r,0)$ considered in this paper.

The automorphisms of $\CE$ can be extended as follows. The automorphism $\CT$ (resp. $\CT^\perp$) is extended as an automorphism of $\CE_\text{ext.}$ (resp. $\CE_\text{ext.}^\perp$) by defining $\CT(F)=F$ (resp. $\CT^\perp(F^\perp)=F^\perp$).\footnote{The extension of the automorphism $\CT$ follows from $\CT(e)\CT(F)=\CT(F)\CT^2(e)$ which becomes $e\CT(F)=\CT(F)\CT(e)$ after substitution $e\to\CT^{-1}(e)$. A similar argument holds for $F^\perp$.} In addition, Miki's automorphism $\CS$ can be extended as a homomorphisms $\CS:\CE_\text{ext.}\to\CE_\text{ext.}^\perp$ and $\CS:\CE_\text{ext.}^\perp\to\CE_\text{ext.}$ by defining $F=\CS(F^\perp)$ and $\CS(F)=F^\perp$.\footnote{To define the extension of Miki's automorphism, we look for the image $\CS(F)$ such that we have $\CS(e F)=\CS(e)\CS(F)=\CS(F)\CS\CT(e)$. The relation $eF^\perp=F^\perp \CS\CT\CS^{-1}(e)$ becomes $\CS(e)F^\perp=F^\perp\CS\CT(e)$ if we substitute $e\to \CS(e)$, and so we can take $\CS(F)=F^\perp$. In the same way, we have $\CS(eF^\perp)=\CS(e)\CS(F^\perp)=\CS(F^\perp) \CS^2\CT\CS^{-1}(e)$ and so substituting $e\to \CS^{-1}(e)$, we have $e\CS(F^\perp)=\CS(F^\perp) \CS^2\CT\CS^{-2}(e)=\CS(F^\perp)\CT(e)$ and we can identify $\CS(F^\perp)=F$.} Note that, to have a proper automorphism instead of homomorphisms between different algebras, we would have to introduce an extension of $\CE$ by both $F$ and $F^\perp$. To define the product of these elements, we need to define $\CT(F^\perp)$ and $\CT^\perp(F)$, which then requires us to introduce other elements of the subgroup generated by $\CT$ and $\CT^\perp$. Although this seems possible, it is not required here and we will only consider the minimal extension.\footnote{We can say a little more about this. For example, we should have $FF^\perp=F^\perp\CT^\perp(F)$. Looking for $\CT(F^\perp)$, we apply $\CT$ to the relation $eF^\perp=F^\perp\CT^\perp(e)$ which gives $\CT(e)\CT(F^\perp)=\CT(F^\perp)\CT\CT^\perp(e)$ and so $e\CT(F^\perp)=\CT(F^\perp)\CT\CT^\perp\CT^{-1}(e)$ after the substitution $e\to \CT^{-1}(e)$. So we need to introduce the automorphism $\CT\CT^\perp\CT^{-1}$, and it is not hard to check that
\begin{equation}
\CT\CT^\perp\CT^{-1}(a_{\pm1})=(\g-\g^{-1})\g^{\pm\bc/2}x_0^\pm,\quad \CT\CT^\perp\CT^{-1}(x_0^\pm)=\pm\dfrac{\g-\g^{-1}}{c_1}[x_0^\pm,x_{\mp1}^\pm]=\g^{\pm(c+2\bc)/2}y^\pm_{\pm2},\quad \CT\CT^\perp\CT^{-1}(x_{\mp1}^\pm)=x_{\mp1}^\pm.
\end{equation} 
So this is a new automorphism, and we need to introduce the extension by another new element. We should then repeat the process with this new element, and so on until all elements of the automorphism group generated by $\CT$ and $\CT^\perp$ are considered.} The extended algebra $\CE_\text{ext.}$ (resp. $\CE_\text{ext.}^\perp$) equipped with the coproduct $\D$ (resp. $\D^\perp$) such that $\D(F)=F\otimes F$ (resp. $\D^\perp(F^\perp)=F^\perp\otimes F^\perp$) is a Hopf algebra.\footnote{The relation $\D\circ\CT=(\CT\otimes\CT)\circ\D$ with the Drinfeld coproduct implies that $\D$ is an homomorphism of the extended algebra if $\D(e)\D(F)=\D(F)(\CT\otimes \CT)\D(e)$. Denoting
\begin{equation}
\D(e)=\sum_i e_i^{(1)}\otimes e_i^{(2)}\implies \left(\sum_i e_i^{(1)}\otimes e_i^{(2)}\right)\D(F)=\D(F)\left(\sum_i \CT(e_i^{(1)})\otimes \CT(e_i^{(2)})\right)
\end{equation} 
which is indeed solved by $\D(F)=F\otimes F$.} Finally, it is useful to mention the possibility to extend the anti-automorphism $\s_H$ to $\CE_\text{ext.}^\perp$ by $\s_H(F^\perp)=F^\perp$, which is a consequence of the relation $\s_H\circ(\CT^\perp)^{-1}\circ\s_H=\CT^\perp$.\footnote{Since $\s_H$ is an anti-automorphism, the image of the defining relation for $F^\perp$ is $\s_H(F^\perp)\s_H(e)=\s_H\circ\CT^\perp(e)\s_H(F^\perp)$. Replacing $e\to (\CT^\perp)^{-1}\circ\s_H$ and using the fact that $\s_H$ is an involution, we deduce that $e\s_H(F^\perp)=\s_H(F^\perp)\s_H\circ(\CT^\perp)^{-1}\circ\s_H(e)$.}

\paragraph{Representations.} We now examine different representations for the extended algebra, starting from the simpler case of $\CE_\text{ext.}$. The vertical Fock modules have level $c=0$, and so they can sustain a representation of $F$. Since $F$ commutes with the Cartan algebra generated by $\psi^\pm(z)$, it can be represented as a diagonal operator on the $\ell$-weight basis $\dket{\bl}$. It it easy to check that the expression
\begin{equation}
\rho^{(0,m)}_{\bsv}(F)\dket{\bl}=\prod_{\sAbox\in\bl} v_\sAbox \chi_\sAbox\dket{\bl}.
\end{equation} 
satisfies all the required algebraic properties.

It is also possible to define a horizontal Fock representation for the extended algebra $\CE_\text{ext.}$. Although we will not be using this construction in this paper, we mention it here for the sake of completeness. Since the representation $\rho^{(1,0)}$ has the level $c=1$, it is not possible to represent $F$ on the module $\CF^{(1,0)}$ defined in \eqref{def_CF_1n}. Instead, it is necessary to consider the direct sum
\begin{equation}
\CF_\text{ext.}=\oplus_{n\in\mZ}\CF^{(1,n)},
\end{equation} 
and introduce the zero modes $[J_0,\hat{\mathsf{Q}}]=1$, such that the vacuum states of the Fock space $\CF^{(1,n)}$ is obtained as $\ket{n}=\mathe^{n\hat{\mathsf{Q}}}\ket{\vac}$ with $J_0\ket{\vac}=0$. On this module, we define the representation $\rho^{(H)}_u$ of $\CE_\text{ext.}$ as
\begin{align}
\begin{split}
&\rho_u^{(H)}(a_k)=-\dfrac{\g^{-k/2}}{k}(1-q_2^k)(1-q_3^k)J_k,\quad \rho^{(H)}_u(a_{-k})=-\dfrac{\g^{-k/2}}{k}(1-q_1^k)(1-q_3^k)J_{-k},\quad (k>0),\\
&\rho_u^{(H)}(x^+(z))=uz^{-J_0}\eta^+(z),\quad \rho_u^{(H)}(x^-(z))=u^{-1}z^{J_0}\eta^-(z),\\
&\rho_u^{(H)}(\bc)=J_0,\quad \rho_u^{(H)}(F)=\mathe^{\hat{\mathsf{Q}}}.
\end{split}
\end{align}
The restriction of $\rho_u^{(H)}$ to $\CE$ acts on $\CF^{(1,n)}$ as $\rho_u^{(1,n)}$.

We now turn to the algebra $\CE_\text{ext.}^\perp$ and consider horizontal Fock representations of levels $(r,0)$. In these representations, the operator $F^\perp$ commutes with the generators $x_0^\pm$, and so they can be diagonalized simultaneously. Let us denote as $\ket{\bl}$ the eigenbasis\footnote{In Section~\ref{sec:gmp} we will explicitly identify this basis with the generalized Macdonald basis.} of the level $r$ horizontal module. Then one can check that the expression
\begin{equation}
\label{eq:framing-horizontal}
 \rho_{\bsu}^{(r,0)}(F^\perp)\ket{\bl}
 = \prod_{\sAbox\in\bl}(-\g^{-1}u_\sAbox\chi_\sAbox)\, \ket{\bl},
\end{equation}
satisfies all the required algebraic properties (see Appendix~\ref{app:framing-op} for the proof). Moreover, this expression is compatible with the isomorphism between the representations $\rho^{(0,m)}$ and $\rho^{(m,0)}\circ\CS$, and so it allows us to extend this isomorphism to the representations of the extended algebra.

In the case $r=1$, this representation reduces to $\rho_u^{(1,0)}(F^\perp)=(-\g^{-1}u)^{L_0}\nabla$, and we recover the usual nabla operator, and $L_0$ is the grading operator defined in \eqref{def_L0}. For this representation, the operator $F^\perp$ was introduced in \cite{Bourgine2021b} in the context of topological strings and integrable hierarchies. It was called \emph{framing operator} as it results in a modification of the framing factors of topological strings amplitudes.

\begin{remark}
It is also possible to consider the vector representation with levels $(0,0)$ of the quantum toroidal $\gl(1)$ algebra. It is possible to represent both $F$ and $F^\perp$ on this representation, and $\rho^{(0,0)}(F^\perp)$ coincides with the $q$-Borel transform. Taking the $m$-th coproduct evaluated in a tensor product of vector representations defines a representation acting on symmetric Laurent series of $m$ variables. For this representation, the operator $F^\perp$ coincides with the operator $\CT_m$ defined in \cite{Langmann:2020utd}, up to a minor modification.
\end{remark}

\section{Generalized Macdonald functions}
\label{sec:gmp}
In this section, we consider the representations $\rho_\bsu^{(r,\bsn)}$ introduced in the previous section with the specialization $\bsn=\bsn_0=(0,\dots,0)$. These representations act on the Fock module $\CF_\bsu^{(r,\bsn_0)}$ isomorphic to the ring of symmetric functions in $r$ alphabets $\L[\bsx^{(1)},\dots,\bsx^{(r)}]$. By construction, the tensor products of Macdonald symmetric functions $P_{\l^{(1)}}(\bsx^{(1)})\cdots P_{\l^{(r)}}(\bsx^{(r)})$ form a basis of this module indexed by $r$-tuple partitions $\bl=(\l^{(1)},\dots,\l^{(r)})$. As we will see shortly, the operator $\rho_\bsu^{(r,\bsn_0)}(x_0^+)$ has a triangular action on this basis with respect to a certain grading. Generalized Macdonald symmetric functions of level $r$ are defined as the eigenbasis of this operator.

In addition to the parameters $(q,t)$, they depend on the weights $\bsu=(u_1,\dots,u_r)$.
Because the eigenvalue equation is homogeneous w.r.t.\ rescalings of all the weights, it follows that the eigenfunctions depend on the weights only up to an overall rescaling. Alternatively, one could express the eigenfunctions as functions of the ratios $u_\a/u_{\a+1}$ only. In the case of $r=1$, this homogeneity allows to get rid of the one weight altogether, and for $r=2$, one can reduce the dependence to one parameter only, which we will indicate as $Q=u_1/u_2$ for convenience. 

Since eigenvalues are non-degenerate, the eigenfunctions are uniquely determined upon fixing their normalization to
\begin{equation}
 P_{\bl}(\bsx^{(1)},\dots,\bsx^{(r)}|u_1,\dots,u_r)
 = P_{\l^{(1)}}(\bsx^{(1)})\cdots P_{\l^{(r)}}(\bsx^{(r)})+\text{lower degrees}.
\end{equation}
With this normalization, Generalized Macdonald symmetric functions are uniquely defined as the solutions to the eigenvalue equations
\be\label{eq:gmp-eigenvectors-eq}
\begin{aligned}
 \rho_\bsu^{(r,\bsn_0)}(x^+_0) \cdot P_{\bl}(\bsx^\bullet|\bsu)
 &= \Big(\sum_{\a=1}^r u_\a\,\me_{\lam^{(\a)}}\Big) P_{\bl}(\bsx^\bullet|\bsu),
\end{aligned}
\ee
with $\rho_\bsu^{(r,\bsn_0)}(x^+_0)$ the zero-mode of the current in \eqref{eq:x(z)+}.
Since the generators $x_0^+$ and $x_0^-$ commute in this representation, $x^-_0$ is also diagonal in this basis, and we have
\begin{equation}
 \rho_\bsu^{(r,\bsn_0)}(x^-_0) \cdot P_{\bl}(\bsx^\bullet|\bsu)
 = \Big(\sum_{\a=1}^r u_\a^{-1} \me_{\lam^{(\a)}}^\vee\Big) P_{\bl}(\bsx^\bullet|\bsu).
\end{equation}

The change of basis from tensor products of level-1 Macdonald functions to GMP is encoded by a triangular matrix $A_{\bl,\bmu}(\bsu)$, i.e.
\begin{equation}\label{def_P_g}
P_{\bl}(\bsx^\bullet|\bsu)=\sum_{\bmu\preceq\bl} A_{\bl,\bmu}(\bsu) P_{\mu^{(1)}}(\bsx^{(1)})\cdots P_{\mu^{(r)}}(\bsx^{(r)}),
\end{equation}
where the ordering will be specified later. Thus, the study of generalized Macdonald symmetric functions is in essence a study of this matrix. It should be emphasized that the formulas derived in this section do not depend on the choice of normalization of the ordinary Macdonald functions $P_\l(\bsx)$. For instance, it is possible to work with the alternative normalization
\be
\label{def_modif_GMP}
 \tP_{\bl}(\bsx^\bullet|\bsu)
 := \dfrac{P_{\bl}(\bsx^\bullet|\bsu)}{\prod_{\a=1}^r P_{\l^{(\a)}}(\sp_\vac)}
 = \tP_{\l^{(1)}}(\bsx^{(1)})\cdots \tP_{\l^{(r)}}(\bsx^{(r)})+\text{lower degrees}.
\ee
The formulas for $\tP_\bl$ are the same as those for $P_\bl$ when the coefficients $b_\l,\psi_\l(\Abox),\psi_\l^\ast(\Abox)$ are replaced by their ``tilde'' counterpart. In particular, the matrix $A_{\bl,\bmu}(\bsu)$ is replaced with $\tilde{A}_{\bl,\bmu}(\bsu)=A_{\bl,\bmu}(\bsu)n_{\bmu}/n_{\bl}$, with $n_\bl=\prod_\a P_{\l^{(\a)}}(\sp_\vac)$.

This section is organized as follows. We first analyze the case of level $r=2$, namely we discuss the triangular nature of the operators, present Pieri-like formulas and introduce the corresponding kernels. Then, these results are generalized to an arbitrary level $r$. 

\subsection{Generalized Macdonald functions at level \texorpdfstring{$r=2$}{r=2}}
We start by analyzing the triangular structure of the operator $\rho_{u_1,u_2}^{(2,\bsn_0)}(x_0^+)$. Using the expression of the coproduct, this operator can be written in the form $u_2^{-1}\rho^{(2,\bsn_0)}_{u_1,u_2}(x_0^+)=D+X$ with
\begin{equation}\label{expr_x0p}
 D = Q x_0^+\otimes 1+1\otimes x_0^+,
 \hspace{30pt}
 X=\sum_{k>0}\g^{-\frac{k}{2}}\psi_{-k}^-\otimes x_k^+,
 \hspace{30pt}
 Q=\frac{u_1}{u_2},
\end{equation} 
where we omitted in the r.h.s. the representations $\rho_1^{(1,0)}$, and denoted $Q=u_1/u_2$ the ratio of weights. As the notation suggests, the operator $D$ is diagonal on the basis $P_\l(\bsx)P_\mu(\bsy)$, it has eigenvalues $Q\me_\l+\me_\mu$. To analyze the action of the operator $X$, we use the algebraic relations to rewrite
\begin{equation}
\g^{-kc/2}\psi_{-k}^-\otimes x_k^+=-\k^{-1}[x_0^+\otimes1,x_{-k}^-\otimes x_k^+].
\end{equation}
The action of $x_k^+$ (resp. $x_{-k}^-$) with $k>0$ on Macdonald functions $P_\l(\bsx)$ is derived in Appendix~\ref{AppA2}, it produces a sum of Macdonald functions $P_\mu(\bsx)$ where $\mu$ is obtained from $\l$ by adding (resp. removing) $k$ boxes. Thus, the action of $x_{-k}^-\otimes x_k^+$ can be described by introducing the partial ordering $\bmu\preceq_k\bl$ on couple of partitions, which indicates that $\mu^{(1)}$ is obtained from $\l^{(1)}$ by removing $k$ boxes and $\mu^{(2)}$ from $\l^{(2)}$ by adding the same number $k$ of boxes. We further denote $\bmu\preceq\bl$ (resp. $\bmu\prec\bl$) if there is a $k\geq0$ (resp. $k>0$) such that $\bmu\preceq_k\bl$. With this definition, we have
\begin{equation}
XP_{\l^{(1)}}(\bsx^{(1)})P_{\l^{(2)}}(\bsx^{(2)})=\sum_{\bmu\prec\bl}X_{\bl,\bmu} P_{\mu^{(1)}}(\bsx^{(1)})P_{\mu^{(2)}}(\bsx^{(2)}).
\end{equation}
Note that the summation is actually finite since the process of removing boxes to $\l^{(1)}$ terminates after $|\l^{(1)}|$ steps. This ordering defines a triangular structure for the operator $\rho_{u_1,u_2}^{(2,\bsn_0)}(x_0^+)$. Thus, this operator can be diagonalized as 
\begin{equation}
u_2^{-1}\rho_{u_1,u_2}^{(2,\bsn_0)}(x_0^+)=A(Q)D(Q)A(Q)^{-1},    
\end{equation}
and the change of basis \eqref{def_P_g} with the triangular matrix $A(Q)$ defines the generalized Macdonald functions.

The matrix $A(Q)$ is chosen with one on the diagonal, i.e. $A_{\bl,\bl}(Q)=1$, which fixes the normalization of the generalized Macdonald. Decomposing $A(Q)=1+U(Q)$ with $U(Q)$ strictly triangular, we deduce the relation
\begin{equation}
[D(Q),U(Q)]=-X(1+U(Q)).
\end{equation}
Projecting this relation on the basis $P_{\l^{(1)}}(\bsx^{(1)})P_{\l^{(2)}}(\bsx^{(2)})$, we can determine inductively the matrix elements of $U(Q)$ in terms of the matrix elements of $X$,\footnote{The matrix elements of the operator $X$ are determined explicitly in Appendix~\ref{AppA2}, they are independent of the variable $Q$.}
\begin{equation}
U_{\bl,\bmu}(Q)=-\dfrac{X_{\bl,\bmu}+\sum_{\bmu\prec\bnu\prec\bl}X_{\bl,\bnu}U_{\bnu,\bmu}(Q)}{Q(\me_{\l^{(1)}}-\me_{\mu^{(1)}})+\me_{\l^{(2)}}-\me_{\mu^{(2)}}},
\end{equation}
since the r.h.s. involves matrix elements with strictly smaller indices. This is an induction on $k=|\l^{(1)}|-|\mu^{(1)}|=|\mu^{(2)}|-|\l^{(2)}|$ starting from the order $k=1$ for which $\mu^{(1)}=\l^{(1)}-\Abox$ and $\mu^{(2)}=\l^{(2)}+\AboxB$,
\begin{equation}
U_{\bl,\bmu}(Q)=\dfrac1{(1-q_1)(1-q_2)}\dfrac{X_{\bl,\bmu}}{Q\chi_\sAbox-\chi_\sAboxB}.
\end{equation}

The algebraic relations imply that the generator $x_0^-$ commutes with $x_0^+$ in this representation. Hence, $\rho_{u_1,u_2}^{(2,\bsn_0)}(x_0^-)$ is also diagonal on the basis of generalized Macdonald functions. It has a similar triangular structure, namely 
\begin{equation}
u_2\rho^{(2,\bsn_0)}_{u_1,u_2}(x_0^-)=Q^{-1}x_0^-\otimes 1+1\otimes x_0^-+Q^{-1}\sum_{k>0}\g^{\frac{k}{2}}x_{-k}^-\otimes\psi_k^+
\end{equation} 
with the operator $x_{-k}^-\otimes\psi_k^+$ also acting on the basis $P_\l(\bsx)P_\mu(\bsy)$ by removing $k$ boxes to $P_\mu(\bsy)$ and simultaneously adding $k$ boxes to $P_\l(\bsx)$. From the same argument as in the case of $\rho^{(2,\bsn_0)}_{u_1,u_2}(x_0^+)$, it follows that its eigenvalues are given by the linear combination $Q^{-1}\me_\l^\vee+\me_\mu^\vee$.

In this construction of generalized Macdonald symmetric function, the coefficients $X_{\bl,\bmu}$ and $U_{\bl,\bmu}$ only depend on the partitions, and not explicitly on the number of variables. It is thus possible to define generalized Macdonald functions with a finite number of variables. GMP indexed by the $r$-tuple partition $\bl$ with $\ell(\l^{(\a)})\leq N_\a$ form a basis of the space of polynomials in $\sum_\a N_\a$ variables $x_i^{(\a)}$, $i=1,\dots,N_\a$, invariant under $S_{N_1}\times\cdots\times S_{N_r}$. In Appendix~\ref{app_finite}, we explore the case of level $r=2$, with $N_1=N_2=1$.

\paragraph{Example.} Using the triangular structure, it is possible to compute inductively the generalized Macdonald functions. Every operator is graded by the total number of boxes $|\bl|=|\l^{(1)}|+|\l^{(2)}|$, and so the diagonalization can be done independently in each subspace. At degree zero, the subspace is of dimension one, and we have trivially $P_{\vac,\vac}(\bsx,\bsy|Q)=1$. At degree one, the subspace is two-dimensional. Here, it is easier to work in normalization $\tP_{\bl}$, and compute
\begin{equation}
r_{[1]}^\ast(1,1)=q_3(1-q_1)(1-q_2),\quad r_\vac(1,1)=1.
\end{equation}
Due to the triangular structure, $X_{\bl,\bmu}=0$ if $\l^{(1)}=\vac$. The only non-zero coefficient is\footnote{Note that the change of basis is trivial here since $n_\bl=n_\bmu$ and so $\tilde{U}_{(\vac,[1])([1],\vac)}=U_{(\vac,[1])([1],\vac)}$.}
\be
 \tilde{X}_{(\vac,[1])([1],\vac)}
 = c_1\implies U_{(\vac,[1])([1],\vac)}
 = \dfrac{1-q_3}{1-Q}
\ee
We deduce that
\be
 P_{[1],\vac}(\bsx,\bsy|Q) = P_{[1]}(\bsx),\quad P_{\vac,[1]}(\bsx,\bsy|Q)
 = P_{[1]}(\bsy)+\dfrac{1-t/q}{1-Q}P_{[1]}(\bsx)\,.
\ee

\subsubsection{Pieri rules}
In Appendix~\ref{app:Mukade}, we derive the action of the generators $a_{\pm1}$ on generalized Macdonald functions for arbitrary level $r$. Specializing to $r=2$, we find
\be
\label{Pieri_g}
\begin{aligned}
(\g-\g^{-1})^{-1}(1-q_1)^{-1}&\rho^{(2,\bsn_0)}(a_{-1})P_{\l,\mu}(\bsx,\bsy|Q)\\
=&\sum_{\sAbox\in A(\l)}\Psi_\mu(Q\chi_\sAbox)\psi_\l(\Abox) P_{\l+\sAbox,\mu}(\bsx,\bsy|Q)+\sum_{\sAbox\in A(\mu)}\psi_\mu(\Abox) P_{\l,\mu+\sAbox}(\bsx,\bsy|Q),\\
(\g-\g^{-1})^{-1}(1-q_2)^{-1}&\rho^{(2,\bsn_0)}(a_1)P_{\l,\mu}(\bsx,\bsy|Q)\\
=&\sum_{\sAbox\in R(\l)}\psi_\l^\ast(\Abox)\ P_{\l-\sAbox,\mu}(\bsx,\bsy|Q)+\sum_{\sAbox\in R(\mu)}\Psi_\l(\chi_\sAbox/Q)\psi_\mu^\ast(\Abox)\ P_{\l,\mu-\sAbox}(\bsx,\bsy|Q),
\end{aligned}
\ee
with the matrix elements $\psi_\l(\Abox)$ and $\psi_\l^\ast(\Abox)$ defined in \eqref{PieriPsi} and \ref{dual_Pieri} respectively. These relations are equivalent to a generalized $e_1$-Pieri rule and its dual,
\be
\begin{aligned}
&(p_1(\bsx)+p_1(\bsy))P_{\l,\mu}(\bsx,\bsy|Q)=\sum_{\sAbox\in A(\l)}\Psi_\mu(Q\chi_\sAbox)\psi_\l(\Abox) P_{\l+\sAbox,\mu}(\bsx,\bsy|Q)+\sum_{\sAbox\in A(\mu)}\psi_\mu(\Abox) P_{\l,\mu+\sAbox}(\bsx,\bsy|Q),\\
&\left(\dfrac{\p}{\p p_1(\bsx)}+q_3\dfrac{\p}{\p p_1(\bsy)}\right)P_{\l,\mu}(\bsx,\bsy|Q)=\sum_{\sAbox\in R(\l)}\psi_\l^\ast(\Abox)\ P_{\l-\sAbox,\mu}(\bsx,\bsy|Q)+\sum_{\sAbox\in R(\mu)}\Psi_\l(\chi_\sAbox/Q)\psi_\mu^\ast(\Abox)\ P_{\l,\mu-\sAbox}(\bsx,\bsy|Q).
\end{aligned}
\ee
A similar relation holds for $\tilde{P}_{\l,\mu}(\bsx,\bsy|Q)$ with $\psi_\l(\Abox)$, $\psi_\l^\ast(\Abox)$ replaced by their tilde counterpart.

\subsubsection{Reproducing kernels}
The Macdonald reproducing kernel $\Pi(\bsx|\bsa)$ defined in \eqref{Mac_kernel} is characterized by the following properties,
\be
\begin{aligned}
 \left.M_k\right|_\bsx \Pi(\bsx|\bsa) &= \left.M_k\right|_\bsa \Pi(\bsx|\bsa),\\
 p_k(\bsx)\Pi(\bsx|\bsa) &= k\dfrac{1-q^k}{1-t^k}\dfrac{\p}{\p p_k(\bsa)}\Pi(\bsx|\bsa),\\
 k\dfrac{\p}{\p p_k(\bsx)}\Pi(\bsx|\bsa) &= \dfrac{1-t^k}{1-q^k}p_k(\bsa)\Pi(\bsx|\bsa),
\end{aligned}
\ee
where $M_k$ are the Macdonald operators \cite{Shiraishi:2006afa,Awata:2009sz}, i.e.
\be
 M_k =
 \prod_{i=1}^k \oint\frac{\mathd z_i}{2\pi\mathi z_i} \prod_{1\leq i<j\leq k}\frac{1-z_i z_j^{-1}}
 {1-t^{-1}z_i z_j^{-1}} :\eta^+(z_1)\cdots\eta^+(z_k):\,,
 \hspace{30pt}
 M_k P_\l = \frac{e_k(\sp_\l)}{e_k(\sp_\vac)}\, P_\l\,.
\ee
Notice that the first property is a corollary of the identity $\eta^+(z)|_\bsx\Pi(\bsx|\bsa)=\eta^+(tz^{-1})|_\bsa\Pi(\bsx|\bsa)$.
Using the isomorphism with the horizontal Fock module of the quantum toroidal algebra $\CE$, these relations become equivalent to the following algebraic characterization,
\be
\label{eq:adj-r=1}
 \left.\rho^{(1,0)}_u(e)\right|_\bsx \Pi(\bsx|\bsa)
 = \left.\rho^{(1,0)}_u(\tilde{\s}_H(e))\right|_\bsa \Pi(\bsx|\bsa),\quad\forall e\in\CE,
\ee
where $\tilde{\s}_H$ is the anti-homomorphism defined in \eqref{def_sV_sH} twisted by an inner automorphism such that $\tilde{\s}_H(e)=q_1^{-d}\s_H(e)q_1^{d}$. Indeed, this is an involutive anti-automorphism such that $\tilde{\s}_H(x_0^\pm)=x_0^\pm$ and $\tilde{\s}_H(a_1)=-q_1^{-1}a_{-1}$.\footnote{Involution is easy to show, since $\tilde{\s}_H\circ\tilde{\s}_H(e)=q_1^{-d}\s_H(q_1^{-d}\s_H(e)q_1^d)q_1^d= q_1^{-d}\s_H(q_1^{d})\s_H(\s_H(e))\s_H(q_1^{-d})q_1^d=e$.}

There are two natural ways to extend the property \eqref{eq:adj-r=1} to representations of level two, namely
\be
\label{prop_kernels}
\begin{aligned}
 \left.\rho^{(2,\bsn_0)}_{u_1,u_2}(e)\right|_{\bsx,\bsy}\Pi(\bsx,\bsy|\bsa,\bsb|Q)
 &= \left.\rho^{(2,\bsn_0)}_{u_1,u_2}(\tilde{\s}_H(e))\right|_{\bsa,\bsb}
 \Pi(\bsx,\bsy|\bsa,\bsb|Q),\\
 \left.\rho^{(2,\bsn_0)}_{u_1,u_2}(e)\right|_{\bsx,\bsy}\Pi_{\mathrm{Z}}(\bsx,\bsy|\bsa,\bsb|Q)
 &= \left.\rho^{(2,\bsn_0)}_{u_2,u_1}(\tilde{\s}_H(e))\right|_{\bsa,\bsb}
 \Pi_{\mathrm{Z}}(\bsx,\bsy|\bsa,\bsb|Q),
\end{aligned}
\ee
together with the normalization
\be
\label{prop_kernels_norm}
 \Pi(0,0|\bsa,\bsb|Q) = \Pi(\bsx,\bsy|0,0|Q) = 1 = \Pi_{\mathrm{Z}}(\bsx,\bsy|0,0|Q) = \Pi_{\mathrm{Z}}(0,0|\bsa,\bsb|Q)\,.
\ee
Each of these equations will provide the definition of a reproducing kernel and an associated inner product, as we will show below.

\begin{proposition}
\label{prop:kernel-r=2}
The equations in \eqref{prop_kernels} together with the normalization condition in \eqref{prop_kernels_norm} define uniquely the kernels $\Pi$, $\Pi_{\mathrm{Z}}$, and the solutions can be expanded on the basis of generalized Macdonald symmetric functions as
\begin{align}\label{def_kernels}
\begin{split}
&\Pi(\bsx,\bsy|\bsa,\bsb|Q)=\sum_{\l,\mu}b_{\l,\mu}(Q)P_{\l,\mu}(\bsx,\bsy|Q)P_{\l,\mu}(\bsa,\bsb|Q),\\
&\Pi_{\mathrm{Z}}(\bsx,\bsy|\bsa,\bsb|Q)=\sum_{\l,\mu}b_{\l}b_{\mu}P_{\l,\mu}(\bsx,\bsy|Q)P_{\mu,\l}(\bsa,\bsb|Q^{-1}),
\end{split}
\end{align}
with the coefficient\footnote{It is important to emphasize that the coefficients of the expansion of the kernel $\Pi_{\mathrm{Z}}$ are independent of $Q$.} $b_\l$ given in \eqref{Mac_kernel} and
\begin{equation}\label{expr_blm}
b_{\l,\mu}(Q):=b_\l b_\mu\dfrac{\tilde{N}_{\l,\mu}(Q)}{\tilde{N}_{\mu,\l}(Q^{-1})}.
\end{equation}
\end{proposition}
\begin{proof}
Due to the structure of the algebra $\CE$, it is enough to prove that the kernels defined in \eqref{def_kernels} satisfy the algebraic relations \eqref{prop_kernels} for the four generators $x_0^\pm$ and $a_{\pm1}$. The case $e=x_0^\pm$ follows immediately from the fact that generalized Macdonald functions are eigenvectors of $\rho_\bsu^{(r,\bsn_0)}(x_0^\pm)$ as in  \eqref{eq:gmp-eigenvectors-eq}. For the second kernel, we need to use the fact that the eigenvalues $u_1\me_\l+u_2\me_\mu$ are invariant under the simultaneous exchange of the weights $u_1,u_2$ and the partitions $\l,\mu$.

The non-trivial part of the proof lies in the action of the generators $a_{\pm1}$, for which we use our previous result on the Pieri rule \eqref{Pieri_g}. Starting from the kernel $\Pi_{\mathrm{Z}}(\bsx,\bsy|\bsa,\bsb|Q)$, we find
\be
\begin{aligned}
&(\g-\g^{-1})^{-1}(1-q_2)^{-1}\left.\rho^{(2,\bsn_0)}_{u_1,u_2}(a_1)\right|_{\bsx,\bsy}\Pi_{\mathrm{Z}}(\bsx,\bsy|\bsa,\bsb|Q)\\
&=\sum_{\l,\mu}b_\l b_\mu\left[\sum_{\sAbox\in R(\l)}\psi_\l^\ast(\Abox)\ P_{\l-\sAbox,\mu}(\bsx,\bsy|Q)+\sum_{\sAbox\in R(\mu)}\Psi_\l(\chi_\sAbox/Q)\psi_\mu^\ast(\Abox)\ P_{\l,\mu-\sAbox}(\bsx,\bsy|Q)\right]P_{\mu,\l}(\bsa,\bsb|Q^{-1}),\\
&=\sum_{\l,\mu}\sum_{\sAbox\in A(\l)}b_{\l+\sAbox}b_\mu \psi_{\l+\sAbox}^\ast(\Abox)\ P_{\l,\mu}(\bsx,\bsy|Q)P_{\mu,\l+\sAbox}(\bsa,\bsb|Q^{-1})\\
&+\sum_{\l,\mu}\sum_{\sAbox\in A(\mu)}b_\l b_{\mu+\sAbox}\Psi_\l(\chi_\sAbox/Q)\psi_{\mu+\sAbox}^\ast(\Abox)\ P_{\l,\mu}(\bsx,\bsy|Q)P_{\mu+\sAbox,\l}(\bsa,\bsb|Q^{-1})\,,
\end{aligned}
\ee
after a change of summation labels $\l-\Abox\mapsto \l$ in the first term and $\mu-\Abox\mapsto\mu$ in the second one. This is to be compared with
\be
\begin{aligned}
(\g-\g^{-1})^{-1}&(1-q_2)^{-1}\left.\rho^{(2,\bsn_0)}_{u_2,u_1}(-t a_{-1})\right|_{\bsa,\bsb}\Pi_{\mathrm{Z}}(\bsx,\bsy|\bsa,\bsb|Q)\\
&=\frac{1-t}{1-q}\sum_{\l,\mu}\sum_{\sAbox\in A(\mu)}b_\l b_\mu \Psi_\l(\chi_\sAbox/Q)\psi_{\mu}(\Abox)\ P_{\l,\mu}(\bsx,\bsy|Q)P_{\mu+\sAbox,\l}(\bsa,\bsb|Q^{-1})\\
&+\frac{1-t}{1-q}\sum_{\l,\mu}\sum_{\sAbox\in A(\l)}b_\l b_\mu \psi_{\l}(\Abox)\ P_{\l,\mu}(\bsx,\bsy|Q)P_{\mu,\l+\sAbox}(\bsa,\bsb|Q^{-1})\,.
\end{aligned}
\ee
The two expressions are indeed identical as a consequence of the variation formulas \eqref{var_b_l}.

Now, turning to the other kernel $\Pi(\bsx,\bsy|\bsa,\bsb|Q)$, we find
\be
\begin{aligned}
&(\g-\g^{-1})^{-1}(1-q_2)^{-1}\left.\rho^{(2,\bsn_0)}_{u_1,u_2}(a_1)\right|_{\bsx,\bsy}\Pi(\bsx,\bsy|\bsa,\bsb|Q)\\
=&\sum_{\l,\mu}b_{\l,\mu}(Q)\left[\sum_{\sAbox\in R(\l)}\psi_\l^\ast(\Abox)\ P_{\l-\sAbox,\mu}(\bsx,\bsy|Q)+\sum_{\sAbox\in R(\mu)}\Psi_\l(\chi_\sAbox/Q)\psi_\mu^\ast(\Abox)\ P_{\l,\mu-\sAbox}(\bsx,\bsy|Q)\right]P_{\l,\mu}(\bsa,\bsb|Q),\\
=&\sum_{\l,\mu}\sum_{\sAbox\in A(\l)}b_{\l+\sAbox,\mu}(Q)\psi_{\l+\sAbox}^\ast(\Abox)\ P_{\l,\mu}(\bsx,\bsy|Q)P_{\l+\sAbox,\mu}(\bsa,\bsb|Q)\\
+&\sum_{\l,\mu}\sum_{\sAbox\in A(\mu)}b_{\l,\mu+\sAbox}(Q)\Psi_\l(\chi_\sAbox/Q)\psi_{\mu+\sAbox}^\ast(\Abox)\ P_{\l,\mu}(\bsx,\bsy|Q)P_{\l,\mu+\sAbox}(\bsa,\bsb|Q)\,,
\end{aligned}
\ee
where we use the same change of variables. Comparing with 
\be
\begin{aligned}
(\g-\g^{-1})^{-1}&(1-q_2)^{-1}\left.\rho^{(2,\bsn_0)}_{u_1,u_2}(-t a_{-1})\right|_{\bsa,\bsb}\Pi(\bsx,\bsy|\bsa,\bsb|Q)\\
&=\frac{1-t}{1-q}\sum_{\l,\mu}\sum_{\sAbox\in A(\l)}b_{\l,\mu}(Q)\Psi_\mu(Q\chi_\sAbox)\psi_{\l}(\Abox)\ P_{\l,\mu}(\bsx,\bsy|Q)P_{\l+\sAbox,\mu}(\bsa,\bsb|Q)\\
&+\frac{1-t}{1-q}\sum_{\l,\mu}\sum_{\sAbox\in A(\mu)}b_{\l,\mu}(Q)\psi_{\mu}(\Abox)\ P_{\l,\mu}(\bsx,\bsy|Q)P_{\l,\mu+\sAbox}(\bsa,\bsb|Q)\,,
\end{aligned}
\ee
we observe that the two expressions match provided that the coefficients $b_{\l,\mu}(Q)$ obey the recursion relations
\be
\label{cond_blm}
 \dfrac{b_{\l+\sAbox,\mu}(Q)}{b_{\l,\mu}(Q)}
 = \frac{1-t}{1-q}\dfrac{\psi_\l(\Abox)}{\psi_{\l+\sAbox}^\ast(\Abox)}
 \Psi_\mu(Q\chi_\sAbox),
 \hspace{30pt}
 \dfrac{b_{\l,\mu+\sAbox}(Q)}{b_{\l,\mu}(Q)}
 =\frac{1-t}{1-q}\dfrac{\psi_\mu(\Abox)}{\psi_{\mu+\sAbox}^\ast(\Abox)}
 \Psi_\l(\chi_\sAbox/Q)^{-1}
\ee 
This is solved by the expression \eqref{expr_blm}. A similar argument can be developed for the action of $a_{-1}$, it leads to conditions on the coefficients $b_{\l,\mu}(Q)$ that are equivalent to \eqref{cond_blm}.
\end{proof}

It is important to mention that the kernel $\Pi_{\mathrm{Z}}$ was previously considered by Zenkevich in \cite{Zenkevich:2014lca} where a factorized expression for it has been conjectured. Here we give a proof of the factorization formula.
\begin{proposition}
The reproducing kernel $\Pi_{\mathrm{Z}}$ can be explicitly resummed in the form of a plethystic exponential as
\be\label{Zenkevich}
 \Pi_{\mathrm{Z}}(\bsx,\bsy|\bsa,\bsb|Q) = \exp\left(\sum_{k>0}\frac1k\frac{1-t^k}{1-q^k}
 \left(p_k(\bsx)p_k(\bsb)+p_k(\bsy)p_k(\bsa)+(1-t^kq^{-k})p_k(\bsx)p_k(\bsa)\right)
 \right)\,,
\ee
which is independent of the weights.
\end{proposition}
\begin{proof}
For convenience, we denote as $E(\bsx,\bsy|\bsa,\bsb)$ the plethystic
exponential
\be
 E(\bsx,\bsy|\bsa,\bsb) := \mathe^{\sum_{k>0}\frac1k\frac{1-t^k}{1-q^k}
 \left(p_k(\bsx)p_k(\bsb)+p_k(\bsy)p_k(\bsa)+(1-t^kq^{-k})p_k(\bsx)p_k(\bsa)\right)}\,.
\ee
The normalization condition \eqref{prop_kernels_norm} is then automatically satisfied by $E(\bsx,\bsy|\bsa,\bsb)$. We need to show that the function $E(\bsx,\bsy|\bsa,\bsb)$ does indeed satisfy the characterizing equations for the kernel $\Pi_{\mathrm{Z}}$ in \eqref{prop_kernels}.

For an element $e\in\CE$, consider
\begin{align}
\begin{split}
\rho_{u_2,u_1}^{(2,\bsn_0)}(\tilde{\s}_H(e))&=q_1^{-L_0\otimes1-1\otimes L_0}\rho_{u_2,u_1}^{(2,\bsn_0)}(\s_H(e))q_1^{L_0\otimes1+1\otimes L_0}\\
&=q_1^{-L_0\otimes1-1\otimes L_0}\left(\left(\rho_{u_2}^{(1,0)}\circ\s_H\otimes \rho_{u_1}^{(1,0)}\circ\s_H\right)\circ \D_H(e)\right)q_1^{L_0\otimes1+1\otimes L_0}\\
&=\left(\rho_{u_2}^{(1,0)}\circ\tilde{\s}_H\otimes \rho_{u_1}^{(1,0)}\circ\tilde{\s}_H\right)\circ \D_H(e),
\end{split}
\end{align}
where we introduced the coproduct twisted by $\s_H$, $\D_H=(\s_H\otimes\s_H)\circ\D\circ\s_H$. The coproduct $\D_H$ can equivalently be written as a twist of the opposite coproduct $\Delta'$ by the 2-tensor corresponding to the Cartan factor of the universal $R$-matrix of $\CE$, i.e.
\be
\label{eq:DHKD'K}
 \D_H = \CK^{-1}\D'\CK\,,
 \quad\text{with}\quad
 \CK = \g^{\d+\bar\d} \mathe^{-\sum_{k>0}\frac1{c_k}(a_k\otimes a_{-k})} \g^{\d+\bar\d},
\ee
and $\d=(c\otimes d+d\otimes c)/2$, $\bar\d=(\bc\otimes\bd+\bd\otimes\bc)/2$ (see Appendix~\ref{app_D_H} for a proof). As a result, we can write
\be
\begin{aligned}
 \rho_{u_2,u_1}^{(2,\bsn_0)}(\tilde{\s}_H(e))
 &= \left(\rho_{u_2}^{(1,0)}\otimes\rho_{u_1}^{(1,0)}\right)\circ
 \left(\tilde{\s}_H\otimes\tilde{\s}_H\right)
 \left(\CK^{-1}\D'(e)\CK\right) \\
 &= \left(\rho_{u_2}^{(1,0)}\otimes\rho_{u_1}^{(1,0)}\right)
 \left(\tilde{\s}_H\otimes\tilde{\s}_H(\CK)\cdot
 \tilde{\s}_H\otimes\tilde{\s}_H(\D'(e))\cdot
 \tilde{\s}_H\otimes\tilde{\s}_H(\CK^{-1})\right) \\
 &= K\cdot\left(\left(\rho_{u_2}^{(1,0)}\circ\tilde{\s}_H\otimes \rho_{u_1}^{(1,0)}\circ\tilde{\s}_H\right)\circ \D'(e)\right)\cdot K^{-1} \\
 &= K\cdot\left(\rho_{u_2}^{(1,0)}(\tilde{\s}_H(e^{(2)}))\otimes\rho_{u_1}^{(1,0)}(\tilde{\s}_H(e^{(1)}))\right)\cdot K^{-1}
\end{aligned}
\ee
where we used that $\tilde{\s}_H$ is an anti-involution, together with Sweedler's notation $\D(e)=e^{(1)}\otimes e^{(2)}$ for the coproduct and
\be
\begin{aligned}
 K :=& \left(\rho_{u_2}^{(1,0)}\otimes\rho_{u_1}^{(1,0)}\right)
 \left(\tilde{\s}_H\otimes\tilde{\s}_H(\CK)\right)\\
 =& \,\g^{\frac12(L_0\otimes1+1\otimes L_0)}
 \left(\rho_{u_2}^{(1,0)}\otimes \rho_{u_1}^{(1,0)}\right)
 \left(\mathe^{-\sum_{k>0}\frac1{c_k}a_{-k}\otimes a_{k}}\right)
 \g^{\frac12(L_0\otimes1+1\otimes L_0)} \\
 =& \,\exp\left(\sum_{k>0}(1-q_3^k)p_k\otimes\frac{\p}{\p p_k}\right)\,.
\end{aligned}
\ee
Using the previous result, we can write
\be
\begin{aligned}
 \left.\rho_{u_2,u_1}^{(2,\bsn_0)}(\tilde{\s}_H(e))\right|_{\bsa,\bsb}
 E(\bsx,\bsy|\bsa,\bsb)
 &= \left.K\right|_{\bsa,\bsb} \cdot
 \left.\rho_{u_2}^{(1,0)}(\tilde{\s}_H(e^{(2)}))\right|_{\bsa} \cdot
 \left.\rho_{u_1}^{(1,0)}(\tilde{\s}_H(e^{(1)}))\right|_{\bsb}
 \cdot \left.K^{-1}\right|_{\bsa,\bsb}
 E(\bsx,\bsy|\bsa,\bsb) \\
 &= \left.K\right|_{\bsa,\bsb} \cdot
 \left.\rho_{u_2}^{(1,0)}(\tilde{\s}_H(e^{(2)}))\right|_{\bsa} \cdot
 \left.\rho_{u_1}^{(1,0)}(\tilde{\s}_H(e^{(1)}))\right|_{\bsb}
 \Pi(\bsx|\bsb)\Pi(\bsy|\bsa) \\
 &= \left.K\right|_{\bsa,\bsb} \cdot
 \left.\rho_{u_2}^{(1,0)}(e^{(2)})\right|_{\bsy} \cdot
 \left.\rho_{u_1}^{(1,0)}(e^{(1)})\right|_{\bsx}
 \Pi(\bsx|\bsb)\Pi(\bsy|\bsa) \\
 &= \left.\rho_{u_1}^{(1,0)}(e^{(1)})\right|_{\bsx} \cdot
 \left.\rho_{u_2}^{(1,0)}(e^{(2)})\right|_{\bsy} \cdot
 \left.K\right|_{\bsa,\bsb}
 \Pi(\bsx|\bsb)\Pi(\bsy|\bsa) \\
 &= \left.\rho_{u_1,u_2}^{(2,\bsn_0)}(e)\right|_{\bsx,\bsy}
 E(\bsx,\bsy|\bsa,\bsb)
\end{aligned}
\ee
where we used that
\be
 \left.K\right|_{\bsa,\bsb} \Pi(\bsx|\bsb)\Pi(\bsy|\bsa)
 = \Pi(\bsx|\bsb+(1-t/q)\bsa)\Pi(\bsy|\bsa)
 = E(\bsx,\bsy|\bsa,\bsb)
\ee
as well as the identity \eqref{eq:adj-r=1}. We therefore have shown that $E(\bsx,\bsy|\bsa,\bsb)$ satisfies the equations \eqref{prop_kernels} and we can conclude that $\Pi_{\mathrm{Z}}(\bsx,\bsy|\bsa,\bsb|Q)=E(\bsx,\bsy|\bsa,\bsb)$.
\end{proof}

Despite the fact that the original definition of $\Pi_{\mathrm{Z}}$ depends explicitly on the ratio of weights $Q$ through the generalized Macdonald functions in the sum over pairs of partition, this proposition shows that $\Pi_{\mathrm{Z}}$ is actually independent of the weights. In the following, we will often make use of this fact and write $\Pi_{\mathrm{Z}}(\bsx,\bsy|\bsa,\bsb)$ for the reproducing kernel.

\subsubsection{Inner products}
We can define the inner products associated to the two reproducing kernels by their matrix elements in the GMP basis,
\begin{align}
\begin{split}
 & \la P_{\l,\mu}(\bsx,\bsy|Q),P_{\rho,\s}(\bsx,\bsy|Q)\ra
 = \d_{\l,\rho}\d_{\mu,\s}b_{\l,\mu}(Q)^{-1},\\
 & \la P_{\mu,\l}(\bsx,\bsy|Q^{-1}),P_{\rho,\s}(\bsx,\bsy|Q)\ra_{\mathrm{Z}}
 = \d_{\l,\rho}\d_{\mu,\s}b_\l^{-1} b_\mu^{-1}.
\end{split}
\end{align}
From this definition and the previous expansion formulas, we have 
\begin{align}
\begin{split}
&\la P_{\l,\mu}(\bsx,\bsy|Q),\Pi(\bsx,\bsy|\bsa,\bsb|Q)\ra= P_{\l,\mu}(\bsa,\bsb|Q),\\
&\la P_{\l,\mu}(\bsx,\bsy|Q),\Pi_{\mathrm{Z}}(\bsx,\bsy|\bsa,\bsb|Q)\ra_{\mathrm{Z}}= P_{\l,\mu}(\bsa,\bsb|Q).
\end{split}
\end{align}

Now, suppose we denote as $\langle-,-\rangle^\otimes_{q,t}$ the inner product associated to the product of kernels $\Pi(\bsx|\bsb)\Pi(\bsy|\bsa)$, so that
\be
 \langle P_\mu(\bsx)P_\lam(\bsy),P_\rho(\bsx)P_\sigma(\bsy)\rangle^\otimes_{q,t}
 = \langle P_\lam,P_\rho\rangle_{q,t}\,\langle P_\mu,P_\sigma\rangle_{q,t}
 = \d_{\l,\rho}\d_{\mu,\s}b_\l^{-1} b_\mu^{-1}.
\ee
Then, for any pair of functions $f,g\in\Lambda[\bsx,\bsy]$, we can write the identity
\be
 \langle f(\bsx,\bsy),g(\bsx,\bsy)\rangle_{\mathrm{Z}}
 = \langle f(\bsx,\bsy),\left.K^{-1}\right|_{\bsx,\bsy}g(\bsx,\bsy)\rangle^\otimes_{q,t}
 = \langle f(\bsx,\bsy),g(\bsx,\bsy-(1-t/q)\bsx)\rangle^\otimes_{q,t}
\ee
Since the operator $K$ is self-adjoint w.r.t.\ the product $\langle-,-\rangle^\otimes_{q,t}$, it follows that the new product $\langle-,-\rangle_{\mathrm{Z}}$ is symmetric in its arguments.
It also follows that
\be
\label{eq:adjoint-r=2}
 \left\langle\rho^{(2,\bsn_0)}_{u_2,u_1}(\tilde{\s}_H(e))\,f,g\right\rangle_{\mathrm{Z}} = 
 \left\langle f,\rho^{(2,\bsn_0)}_{u_1,u_2}(e)\,g\right\rangle_{\mathrm{Z}}\,,
 \hspace{30pt}
 \forall\,e\in\CE,
\ee
which is equivalent to \eqref{prop_kernels}. Using \eqref{eq:adjoint-r=2} we are now able to define adjunction w.r.t.\ the inner product $\langle-,-\rangle_{\mathrm{Z}}$ through the use of the involution $\tilde{\s}_H$ in the algebra and, for example, we can write
\be
 \left\langle(p_k(\bsx)+p_k(\bsy))\,f,g\right\rangle_{\mathrm{Z}} = 
 k\frac{1-q^k}{1-t^k} \left\langle f,
 \left(\frac{\p}{\p p_k(\bsx)}+t^kq^{-k}\frac{\p}{\p p_k(\bsy)}\right)g\right\rangle_{\mathrm{Z}}\,.
\ee

\subsection{Generalized Macdonald functions at higher levels}
\subsubsection{Pieri rules}
In Appendix~\ref{app:Mukade}, Proposition~\ref{prop:pieriApp}, we derive the following expression for the action of the generators $a_{\pm1}$ on generalized Macdonald functions,
\be
\begin{aligned}
\label{apm_GMP}
 & \rho_{\bsu}^{(r,\bsn_0)}(a_{-1})P_{\bl}(\bsx^\bullet|\bsu)
 = \g^{1-\frac{r}{2}}(\g-\g^{-1})(1-q_1)\sum_{\a=1}^r\sum_{\sAbox\in A(\l^{(\a)})}
 \psi_{\l^{(\a)}}(\Abox) \prod_{\b=\a+1}^{r}\Psi_{\l^{(\b)}}(u_\a\chi_\sAbox/u_\b)\
 P_{\bl+\sAbox}(\bsx^\bullet|\bsu),\\
 & \rho_{\bsu}^{(r,\bsn_0)}(a_1)P_{\bl}(\bsx^\bullet|\bsu)
 = \g^{1-\frac{r}{2}}(\g-\g^{-1})(1-q_2)\sum_{\a=1}^r\sum_{\sAbox\in R(\l^{(\a)})}
 \psi_{\l^{(\a)}}^\ast(\Abox) \prod_{\b=1}^{\a-1}\Psi_{\l^{(\b)}}(u_\a\chi_\sAbox/u_\b)\
 P_{\bl-\sAbox}(\bsx^\bullet|\bsu).
\end{aligned}
\ee
Our derivation takes advantage of the results obtained in \cite{Fukuda:2019ywe} on the Mukad\'e operator, which reduces to a generating function of Cartan generators after specialization of the weights as in Lemma~\ref{lem:mukade}. In terms of the action of power sums, this is as follows,
\begin{align}
\begin{split}
&\left(\sum_{\a=1}^rq_3^{\a-1}\dfrac{\p}{\p p_1(\bsx^{(\a)})}\right)P_{\bl}(\bsx^\bullet|\bsu)=\sum_{\a=1}^r\sum_{\sAbox\in R(\l^{(\a)})}\psi^\ast_{\l^{(\a)}}(\Abox)\prod_{\b=1}^{\a-1}\Psi_{\l^{(\b)}}(u_\a\chi_\sAbox/u_\b)\ P_{\bl-\sAbox}(\bsx^\bullet|\bsu),\\
&\left(\sum_{\a=1}^r p_1(\bsx^{(\a)})\right)P_{\bl}(\bsx^\bullet|\bsu)=\sum_{\a=1}^r\sum_{\sAbox\in A(\l^{(\a)})}\psi_{\l^{(\a)}}(\Abox)\prod_{\b=\a+1}^{r}\Psi_{\l^{(\b)}}(u_\a\chi_\sAbox/u_\b)\ P_{\bl+\sAbox}(\bsx^\bullet|\bsu).
\end{split}
\end{align}

\subsubsection{Reproducing kernel}
The construction of the reproducing kernel follows very closely that of the case $r=2$. Indeed, we define the level $r$ reproducing kernel as the unique solution to the equations
\be
\label{eq:reproducing-kernel-r-def}
 \left.\rho_{u_r,\dots,u_1}^{(r,\bsn_0)}(\tilde{\s}_H(e))\right|_{\bsa^\bullet}
 \Pi_{\mathrm{Z}}(\bsx^\bullet|\bsa^\bullet)
 = \left.\rho_{u_1,\dots,u_r}^{(r,\bsn_0)}(e)\right|_{\bsx^\bullet}
 \Pi_{\mathrm{Z}}(\bsx^\bullet|\bsa^\bullet)
\ee
such that $\Pi_{\mathrm{Z}}(0,\dots,0|\bsa^\bullet)=1=\Pi_{\mathrm{Z}}(\bsx^\bullet|0,\dots,0)$.
The solution can be expanded over the level $r$ GMP basis, leading to the identity
\be
\label{eq:repr-kernel-rank-r}
 \Pi_{\mathrm{Z}}(\bsx^\bullet|\bsa^\bullet)
 = \sum_\blam \prod_{\a=1}^r b_{\lam^{(\a)}}
 P_{\lam^{(1)},\dots,\lam^{(r)}}(\bsx^{(1)},\dots,\bsx^{(r)}|u_1,\dots,u_r)
 P_{\lam^{(r)},\dots,\lam^{(1)}}(\bsa^{(1)},\dots,\bsa^{(r)}|u_r,\dots,u_1)
\ee
which can be proven in a similar way as in Proposition~\ref{prop:kernel-r=2}.

In order to show that the reproducing kernel factorizes as in the case of level $r=2$, \eqref{Zenkevich}, we follow the same logic as in the previous section and we first define a $r$-tensor $\CK^{(r)}$ which connects the iterated coproducts $\D_H$ and $\D'$.

We define the iterated twisted coproduct inductively as
\be
 \Delta_H^{(r-1)} := (\Delta_H(\otimes\Id)^{r-2})\circ\Delta^{(r-2)}_H
 = (\underbrace{\s_H\otimes\dots\otimes\s_H}_r)\circ\Delta^{(r-1)}\circ\s_H \,,
\ee
and similarly,
\be
 \Delta^{'(r-1)} := (\Delta'(\otimes\Id)^{r-2})\circ\Delta^{'(r-2)} \,,
\ee
so that
\be
 \Delta_H^{(r-1)} = (\CK^{(r-1)})^{-1}\Delta^{'(r-1)}\CK^{(r-1)} \,,
\ee
for some $r$-tensor $\CK^{(r-1)}$, where, by induction, we have
\be
 \CK^{(r-1)} = ((\D'(\otimes\Id)^{r-3})(\CK^{(r-2)}))(\CK(\otimes\Id)^{r-3}),
\ee
with $\CK^{(1)}$ equal to the 2-tensor $\CK$ in \eqref{eq:DHKD'K}.
Using recursively $(\D'\otimes\Id)(\CK)=\CK_{1,2}\CK_{1,3}\CK_{2,3}(\CK_{1,2})^{-1}$,\footnote{This follows from $(\D_H\otimes\Id)(\CK)=\CK_{1,3}\CK_{2,3}$ together with \eqref{eq:DHKD'K}.} we obtain
\be
\label{eq:recursionCK}
 \CK^{(r-1)} = (\CK^{(r-2)}\otimes\Id) \prod_{\a=1}^{r-1}\CK_{\a,r}
\ee
where $\CK_{\a,\b}$ is the $r$-tensor whose only non-trivial legs are in the $\a$-th and $\b$-th tensor factors. Observe here, that the $\CK_{\a,\b}$ do not commute in general, however we have $\CK_{\a,\b}\CK_{i,j}=\CK_{i,j}\CK_{\a,\b}$ if $\a\neq j$ and $\b\neq i$. This implies that the last product in the r.h.s.\ of \eqref{eq:recursionCK} does not depend on the ordering.

We want to show that the tensor $\CK^{(r-1)}$ can be used to define the notion of adjoint of an operator in the level $r$ representation in the sense of \eqref{eq:reproducing-kernel-r-def}.
Let us first consider the operator acting in the l.h.s.\ of \eqref{eq:reproducing-kernel-r-def},
\be
\label{eq:rankrKtensorId}
\begin{aligned}
 \rho_{u_r,\dots,u_1}^{(r,\bsn_0)}(\tilde{\s}_H(e))
 &= \left(\rho_{u_r}^{(1,0)}\otimes\dots\otimes
 \rho_{u_1}^{(1,0)}\right) \left(\tilde{\s}_H^{\otimes r}
 \left(\D^{(r-1)}_H(e)\right)\right) \\
 &= \left(\rho_{u_r}^{(1,0)}\otimes\dots\otimes
 \rho_{u_1}^{(1,0)}\right) \left(\tilde{\s}_H^{\otimes r}
 \left((\CK^{(r-1)})^{-1}\D^{'(r-1)}(e)\CK^{(r-1)}\right)\right) \\
 &= \left(\rho_{u_r}^{(1,0)}\otimes\dots\otimes
 \rho_{u_1}^{(1,0)}\right)
 \left(\tilde{\s}_H^{\otimes r}(\CK^{(r-1)})\,
 \tilde{\s}_H^{\otimes r}(\D^{'(r-1)}(e))\,
 \tilde{\s}_H^{\otimes r}(\CK^{(r-1)})^{-1}\right) \\
 &= K^{(r-1)}\left(\rho_{u_r}^{(1,0)}(\tilde{\s}_H(e^{(r)}))
 \otimes\dots\otimes
 \rho_{u_1}^{(1,0)}(\tilde{\s}_H(e^{(1)}))\right)(K^{(r-1)})^{-1},
\end{aligned}
\ee
where we have defined
\be
 K^{(r-1)} := \left(\rho^{(1,0)}_{u_r}\otimes\dots\otimes\rho^{(1,0)}_{u_1}\right)
 \left(\tilde{\s}_H^{\otimes r}(\CK^{(r-1)})\right).
\ee
From \eqref{eq:recursionCK}, we have the recursion
\be
 K^{(r-1)} = \Big(\prod_{\a=1}^{r-1} K_{\a,r} \Big) (K^{(r-2)}\otimes\Id),
 \hspace{30pt}
 \left.K_{\a,\b}\right|_{\bsx^\bullet}
 = \exp\left(\sum_{k>0}(1-q_3^k)p_k(\bsx^{(\a)})\frac{\p}{\p p_k(\bsx^{(\b)})}\right)
\ee
whose solution can be written as
\be
 K^{(r-1)} = \prod_{\a_1=1}^{r-1} K_{\a_1,r} \prod_{\a_2=1}^{r-2} K_{\a_2,r-1}
 \cdots \prod_{\a_{r-2}=1}^{2} K_{\a_{r-2},3} \prod_{\a_{r-1}=1}^{1} K_{\a_{r-1},2} \,.
\ee
If we define the function
\be
\begin{aligned}
 E(\bsx^\bullet|\bsa^\bullet)
 :=& \left.K^{(r-1)}\right|_{\bsa^\bullet}
 \prod_{\a=1}^r \Pi(\bsx^{(r-\a+1)}|\bsa^{(\a)}) \\
 =&\, \prod_{\a=1}^r \Pi(\bsx^{(r-\a+1)}|(1-t/q)
 \sum_{\b<\a}\bsa^{(\b)}+\bsa^{(\a)}) \\
 =&\, \exp\left(\sum_{k>0}\frac1{k}\frac{1-t^k}{1-q^k}
 \sum_{\a=1}^r p_k(\bsx^{(r-\a+1)})\left((1-t^kq^{-k})\sum_{\b<\a}p_k(\bsa^{(\b)})+p_k(\bsa^{(\a)})\right)\right)
\end{aligned}
\ee
then, using \eqref{eq:rankrKtensorId}, we can show that
\be
\label{eq:adj-rankr-kernel}
\begin{aligned}
 \left.\rho_{u_r,\dots,u_1}^{(r,\bsn_0)}(\tilde{\s}_H(e))\right|_{\bsa^\bullet}
 E(\bsx^\bullet|\bsa^\bullet)
 &= \left.K^{(r-1)}\right|_{\bsa^\bullet}
 \prod_{\a=1}^r\left.\rho_{u_{r-\a+1}}^{(1,0)}(\tilde{\s}_H(e^{(r-\a+1)}))\right|_{\bsa^{(\a)}}
 \left.(K^{(r-1)})^{-1}\right|_{\bsa^\bullet}
 E(\bsx^\bullet|\bsa^\bullet) \\
 &= \left.K^{(r-1)}\right|_{\bsa^\bullet}
 \prod_{\a=1}^r\left.\rho_{u_{r-\a+1}}^{(1,0)}(\tilde{\s}_H(e^{(r-\a+1)}))\right|_{\bsa^{(\a)}}
 \prod_{\a=1}^r\Pi(\bsx^{(r-\a+1)}|\bsa^{(\a)}) \\
 &= \left.K^{(r-1)}\right|_{\bsa^\bullet}
 \prod_{\a=1}^r\left.\rho_{u_{r-\a+1}}^{(1,0)}(e^{(r-\a+1)})\right|_{\bsx^{(r-\a+1)}}
 \prod_{\a=1}^r\Pi(\bsx^{(r-\a+1)}|\bsa^{(\a)}) \\
 &= \left.\rho_{u_1,\dots,u_r}^{(r,\bsn_0)}(e)\right|_{\bsx^\bullet}
 E(\bsx^\bullet|\bsa^\bullet)
\end{aligned}
\ee
which is the defining equation \eqref{eq:reproducing-kernel-r-def} of the reproducing kernel.
Together with the normalization property, $E(0,\dots,0|\bsa^\bullet)=1=E(\bsx^\bullet|0,\dots,0)$, we conclude that $E(\bsx^\bullet|\bsa^\bullet)$ coincides with the reproducing kernel $\Pi_{\mathrm{Z}}(\bsx^\bullet|\bsa^\bullet)$, and thus we obtain the Cauchy identity
\begin{multline}
\label{eq:Cauchy-r}
 \exp\left(\sum_{k>0}\frac1{k}\frac{1-t^k}{1-q^k}
 \sum_{\a=1}^r p_k(\bsx^{(r-\a+1)})\left((1-t^kq^{-k})\sum_{\b<\a}p_k(\bsa^{(\b)})+p_k(\bsa^{(\a)})\right)\right)\\
 = \sum_\blam \prod_{\a=1}^r b_{\lam^{(\a)}}
 P_{\lam^{(1)},\dots,\lam^{(r)}}(\bsx^{(1)},\dots,\bsx^{(r)}|u_1,\dots,u_r)
 P_{\lam^{(r)},\dots,\lam^{(1)}}(\bsa^{(1)},\dots,\bsa^{(r)}|u_r,\dots,u_1).
\end{multline}

\subsubsection{Inner product}
The associated inner product is then defined accordingly as
\be
\label{eq:twisted-scalar-prod}
 \langle f(\bsx^\bullet),g(\bsx^\bullet)\rangle_{\mathrm{Z}}
 := \langle f(\bsx^\bullet),\left.(K^{(r-1)})^{-1}\right|_{\bsx^\bullet}
 g(\bsx^\bullet)\rangle^\otimes_{q,t}
\ee
for any $f,g\in\Lambda^{\otimes r}$, where the inner product in the right is the one associated to the reproducing kernel $\prod_{\a=1}^r\Pi(\bsx^{(r-\a+1)}|\bsa^{(\a)})$.
We observe that the operator $K^{(r-1)}$ is self-adjoint w.r.t.\ the inner product $\langle\,,\,\rangle^\otimes_{q,t}$, which implies that the twisted inner product $\langle\,,\,\rangle_{\mathrm{Z}}$ is symmetric in its arguments, which is also equivalent to the statement that the kernel $\Pi_{\mathrm{Z}}(\bsx^\bullet|\bsa^\bullet)$ is symmetric under the simultaneous exchange of $\bsx^{(\a)}$ and $\bsa^{(\a)}$, for all $\a=1,\dots,r$.
In terms of the product $\langle\,,\,\rangle_{\mathrm{Z}}$, the property \eqref{eq:reproducing-kernel-r-def} can be rewritten as
\be
\label{eq:adjoint-r}
 \left\langle\rho^{(r,\bsn_0)}_{u_r,\dots,u_1}(\tilde{\s}_H(e))\,f,g\right\rangle_{\mathrm{Z}} = 
 \left\langle f,\rho^{(r,\bsn_0)}_{u_1,\dots,u_r}(e)\,g\right\rangle_{\mathrm{Z}}\,,
 \hspace{30pt}
 \forall\,e\in\CE,\quad
 \forall\,f,g\in\Lambda^{\otimes r}.
\ee
In conclusion, we find that the generalized Macdonald functions in level $r$ satisfy the orthogonality relation
\be
 \langle P_{\l^{(r)},\dots,\l^{(1)}}(\bsx^{(1)},\dots,\bsx^{(r)}|u_r,\dots,u_1),
 P_{\mu^{(1)},\dots,\mu^{(r)}}(\bsx^{(1)},\dots,\bsx^{(r)}|u_1,\dots,u_r) \rangle_{\mathrm{Z}}
 = \prod_{\a=1}^r b^{-1}_{\l^{(\a)}} \delta_{\l^{(\a)},\mu^{(\a)}},
\ee
which follows from \eqref{eq:repr-kernel-rank-r}, and that the reproducing kernel behaves as a delta-function for the inner product,
\be
 \langle \Pi_{\mathrm{Z}}(\bsx^\bullet|\bsa^\bullet),
 f(\bsx^{(1)},\dots,\bsx^{(r)}) \rangle_{\mathrm{Z}}
 = f(\bsa^{(1)},\dots,\bsa^{(r)}),
 \hspace{30pt}
 \forall f\in\Lambda^{\otimes r},
\ee
as expected. It is worth noting that neither the kernel nor the inner product in this equation depend on the weights $\bsu$.
\begin{remark}
Besides \eqref{eq:adjoint-r}, we can define the adjoint of any multi-symmetric function in $\Lambda^{\otimes r}$ w.r.t.\ the inner product $\langle\,,\,\rangle_{\mathrm{Z}}$ by first decomposing onto the power-sum basis $p_k^{(\a)}$ and then using
\be
\label{adj_Z}
\begin{aligned}
&(p_k^{(\a)})^\dagger = k\frac{1-q^k}{1-t^k}\left(\frac{\p}{\p p_k^{(r-\a+1)}}
 -(1-t^kq^{-k})\sum_{\b=1}^{\a-1} (tq^{-1})^{(\a-\b-1)k}
 \frac{\p}{\p p_k^{(r-\b+1)}}\right),\\
&\left(\frac{\p}{\p p_k^{(\a)}}\right)^\dagger = \dfrac1k\dfrac{1-t^k}{1-q^k}\left(p_k^{(r-\a+1)}+(1-t^kq^{-k})\sum_{\b=\a+1}^r p_k^{(r-\b+1)}\right)
\end{aligned}
\ee
where ``${}^\dagger$'' denotes the adjoint operator. This follows from the definition in \eqref{eq:twisted-scalar-prod}, so that we have
\begin{equation}
p_k(\bsx^{(\a)})^\dagger\Pi_{\mathrm{Z}}(\bsx^\bullet|\bsa^\bullet)=p_k(\bsa^{(\a)})\Pi_{\mathrm{Z}}(\bsx^\bullet|\bsa^\bullet),\quad \left(\dfrac{\p}{\p p_k(\bsx^{(\a)})}\right)^\dagger\Pi_{\mathrm{Z}}(\bsx^\bullet|\bsa^\bullet)=\left(\dfrac{\p}{\p p_k(\bsa^{(\a)})}\right)\Pi_{\mathrm{Z}}(\bsx^\bullet|\bsa^\bullet).
\end{equation}
This operation is involutive.
\end{remark}

\subsection{Isomorphism between horizontal and vertical representations}
\label{sec:iso-h-v}
The $e_1$-Pieri rule can also be used to prove the isomorphism \eqref{prop_S} between the representations $\rho_\bsv^{(0,r)}$ and $\rho_\bsu^{(r,\bsn_0)}\circ\CS$, and find the explicit relation between the different basis. Since the algebra $\CE$ is generated multiplicatively by $a_{\pm1}$, $x_0^\pm$ (and the central elements), it is sufficient to study the action of these elements. The diagonal action of $\rho_\bsv^{(0,r)}(a_{\pm1})$ on the basis $\dket{\bl}$ matches the action of $\rho_\bsu^{(r,\bsn_0)}(\CS(a_{\pm1}))\propto\rho_\bsu^{(r,\bsn_0)}(x_0^\pm)$, provided that we identify the weights as $u_\a=-\g v_\a$. Since the spectrum of these operators is non-degenerate, the vectors $\dket{\bl}$ are mapped to the GMPs, up to a normalization factor. This factor will be determined by comparing the action of $\rho_\bsv^{(0,r)}(x_0^+)$ with $\rho_\bsu^{(r,\bsn_0)}(\CS(x_0^+))$. 

From the expression \eqref{apm_GMP} of the action of $\rho^{(0,r)}_\bsu(a_{-1})$ on the GMPs, and taking into account the extra factor in the definition of Miki's automorphism \eqref{Miki_init}, we deduce
\begin{equation}
\rho_\bsu^{(r,\bsn_0)}(\CS(x_0^+))\tP_{\bl}(\bsx^\bullet|\bsu)=-q_1\g^{1-r/2}\sum_{\a=1}^r\sum_{\sAbox\in A(\l^{(\a)})}r_{\l^{(\a)}}(\Abox)\prod_{\b=\a+1}^{r}\Psi_{\l^{(\b)}}(u_\a\chi_\sAbox/u_\b)\ \tP_{\bl+\sAbox}(\bsx^\bullet|\bsu)
\end{equation}
where we used the normalized GMPs defined in \eqref{def_modif_GMP} to facilitate the comparison. Inserting the coefficient $G_\bl^\ast(\bsu)$ defined in Appendix~\ref{app_A1} (formula \eqref{def_gl}), and using the property \eqref{prop_G_r}, the formula becomes
\begin{equation}
\rho_\bsu^{(r,\bsn_0)}(\CS(x_0^+))G_\bl^\ast(\bsu)^{-1}\tP_{\bl}(\bsx^\bullet|\bsu)=-q_1\g^{r/2+1-\a}\sum_{\sAbox\in A(\bl)}r_{\bl}(\Abox|\bsu)\ G_{\bl+\sAbox}^\ast(\bsu)^{-1}\tP_{\bl+\sAbox}(\bsx^\bullet|\bsu).
\end{equation}
Comparison with the formula \eqref{def_vert_rep} for the vertical action of $\rho^{(0,r)}_\bsv(x_0^+)$ leads to the identification of the image of the vector $\dket{\bl}$ with $(-q_1\g^{1/2})^{|\bl|}\g^{-(r-1)|\bl|/2+\sum_\a (r-\a)|\l^{(\a)}|}G_\bl^\ast(\bsu)^{-1}\tP_\bl(\bsx^\bullet|\bsu)$. A similar comparison can be made for the actions of $\rho_\bsv^{(0,r)}(x_0^-)$ and $\rho_\bsu^{(r,\bsn_0)}(\CS(x_0^-))$, which leads to the same relation between the basis elements.

\subsection{Framing operator on generalized Macdonald functions}
Let us define the framing operator $\nabla(\bsu)$ to be diagonal on the basis of generalized Macdonald functions with weights $\bsu$.
\begin{definition}
Let $\nabla(\bsu):\CF_\bsu^{(r,0)}\to\CF_\bsu^{(r,0)}$ be the operator defined by
\be
\label{eq:def-nabla}
 \nabla(\bsu) P_\blam(\bsx^\bullet|\bsu)
 := \Big(\prod_{\sAbox\in\blam} u_\sAbox\chi_\sAbox\Big)
 P_\blam(\bsx^\bullet|\bsu)\,.
\ee
\end{definition}
We then have the identification
\be\label{rep_framing}
 \rho_\bsu^{(r,\bsn_0)}(F^\perp) = (-\g^{-1})^{L_0} \nabla(\bsu).
\ee
Note that the operator representing $F^\perp$ must be self-adjoint since $\tilde{\s}_H(F^\perp)=F^\perp$,\footnote{It was shown previously that $\s_H(F^\perp)=F^\perp$, and $[F^\perp,d]=0$.} which is in fact true of both operators $L_0$ and $\nabla(\bsu)$, separately.
It therefore follows that the properties in \eqref{prop_kernels} can be extended to all elements $e\in\CE_\text{ext.}^\perp$.

In what follows we will need the following identity for the action of $F^\perp$ on certain plethystic exponentials.
\begin{conjecture}
\label{conj:BH}
\begin{multline}
\label{nabla_r}
 \rho_{\bsu}^{(r,\bsn_0)}(F^\perp)
 \exp\left(\sum_{k>0}\frac{(-1)^k}{k(1-q_2^k)}\sum_{\a=1}^r v_\a^k p_k(\bsx^{(\a)})\right)\\
 = \exp\left(-\sum_{k>0}\frac{1}{k(1-q_2^k)}\sum_{\a=1}^r p_k(\bsx^{(\a)})
 \Big((u'_\a)^k+(1-q_3^k)\sum_{\b=\a+1}^r(u'_\b)^k\Big)\right)
\end{multline}
with $u'_\a=-\g^{-1} u_\a v_\a$.
\end{conjecture}
We have checked by computer that the identity holds up to degree 7 in $\bsx^{(\a)}$ for $r=2$, degree 5 for $r=3$, degree 3 for $r=4$ and degree 2 for $r=5,6$. Based on this experimental evidence, we expect the conjecture to hold at any degree and any level. For $r=1$, the conjecture reduces to a known identity of the $\nabla$ operator, which can be proven using the techniques developed in \cite{Garsia2001}. For the sake of completeness, we recall the proof of this formula, using our notations, in Appendix \ref{App:conjecture}.

\begin{remark}
Using the Cauchy identity \eqref{eq:Cauchy-r} to expand both sides of \eqref{nabla_r}, we obtain the evaluation identity
\be
 \frac{P_\bl(\bsx^{(\a)}=u_\a v_\a \sp_\vac|\bsu)}
 {P_\bl(\bsx^{(\a)}=-(v_\a-(1-q_3)\sum_{\b<\a}q_3^{\a-\b-1}v_\b)\sp_\vac|\bsu)}
 = \prod_{\sAbox\in\blam} (-u_\sAbox\chi_\sAbox),
\ee
which, for $v_\a=1$, simplifies to
\be
\label{eq:spec-id2}
 \frac{P_\bl(u_1\sp_\vac,u_2\sp_\vac,\dots,u_r\sp_\vac|\bsu)}
 {P_\bl(-\sp_\vac,-q_3\sp_\vac,\dots,-q_3^{r-1}\sp_\vac|\bsu)}
 = \prod_{\sAbox\in\blam} (-u_\sAbox\chi_\sAbox).
\ee
\end{remark}

\section{GHT identity}
In this section, we revisit an identity for Macdonald functions obtained by Garsia, Haiman and Tesler in \cite{Garsia2001}, in the language of the quantum toroidal $\gl(1)$ algebra. It leads us to propose a generalization to Macdonald functions of higher levels. Then, using this results, we also propose a generalization of the five-term relation obtained in \cite{Garsia:2018fiv}, and of the Fourier/Hopf pairing defined in \cite{Cherednik:1995mac,Okounkov:2001are,Beliakova:2021cyc}.

\subsection{Revisiting the level 1 GHT identity}
In \cite{Garsia2001}, GHT derived an identity that relates the Macdonald functions $P_\l(\bsx)$ to the Macdonald kernel with one of the alphabets specialized to $\sp_\l=(t^{-1}q^{\l_1},t^{-2}q^{\l_2},\dots)$. More specifically, we define the operator
\be
\label{eq:def-GHT-oper}
 \GHT := \nabla \mathe^{\sum_{k>0}\frac{(-1)^{k}}{k(1-q^k)}p_k(\bsx)}t^{-L_0}
 \mathe^{\sum_{k>0}\frac{(-1)^{k}}{1-t^k}\frac{\p}{\p p_k(\bsx)}}\nabla\,,
\ee
where $\nabla$ is the operator introduced in Section~\ref{sec_framing} that is diagonal on the Macdonald basis with the eigenvalues given by \eqref{eigen_nabla}. The operator $L_0$ defined in \eqref{def_L0}, is the representation of the element $d\in\CE$.
Then \cite[Theorem I.2]{Garsia2001} can be stated as
\be
\label{eq:GHT-thm}
 \GHT\tP_\l(\bsx) = \Pi(\sp_\l|\bsx)
 \,,
\ee
where $\tP_\l$ are the spherical Macdonald functions \eqref{eq:spherical-macdonalds}.
The r.h.s.\ of this equation is an infinite series of symmetric functions in $\bsx$, that we will interpret shortly as a type of Whittaker vector in the Fock module. Moreover, we will show that all such vectors can be obtained from the action of the AFS vertex operators on the vacuum. This gives an algebraic interpretation of the GHT identity which can be used to generalize \eqref{eq:GHT-thm} to Fock modules of higher levels, i.e.\ to the case of generalized Macdonald functions.

\subsubsection{Whittaker vector}
\begin{definition}
An element $W_\lam(\bsx|u)$ in the Fock module $\CF^{(1,0)}_u\cong\L[\bsx]$ is called a (level-$1$) \emph{Whittaker vector} if it transforms under the action of the toroidal algebra $\CE$ as
\be
\label{prop_Whit}
\begin{aligned}
 \rho_u^{(1,0)}(a_k) W_\l(\bsx|u)
 &= \dfrac{\g^{-k/2}}{k}(1-q_3^k) u^k p_k(\me_\l)W_\l(\bsx|u),\quad (k>0),\\
 \rho_u^{(1,0)}(x_k^+)W_{\l}(\bsx|u)
 &= \sum_{\sAbox\in A(\l)}(u\chi_{\sAbox})^k\,r_\l(\Abox)W_{\l+\sAbox}(\bsx|u),\quad (k\geq 0).
\end{aligned}
\ee
and it is normalized such that $W_\l(0|u) = u^{|\l|}$.
\end{definition}
We observe that the first property implies that $W_\l(\bsx|u)$ is an eigenvector of the annihilation operators, namely
\be
 k\dfrac{\p}{\p p_k(\bsx)} W_\l(\bsx|u)
 = \frac{1-t^k}{1-q^k} u^k p_k(\sp_\l) W_\l(\bsx|u)
 \,.
\ee
Thus, to be specific, $W_\l(\bsx|u)$ is a Whittaker vector for the $\widehat{\gl(1)}$ subalgebra of the quantum toroidal $\gl(1)$ algebra in the horizontal Fock module.
This property, together with the normalization condition, defines the Whittaker vector uniquely at level $r=1$. The definition of the Whittaker vector therefore, assumes a specific form of the eigenvalue. Then, the second equation in \eqref{prop_Whit} follows from the first.
On the other hand, at level $r>1$, the first relation is no longer sufficient to define the Whittaker vector uniquely and the second relation will be used to pinpoint this notion.

\begin{proposition}
The vectors $W_\l(\bsx|u)$ can be obtained from the action of the AFS vertex operator $\Phi_\l$ on the vacuum of the Fock space, as
\be
\label{eq:Whittaker-from-intertwiner}
 W_\l(\bsx|u) = \Phi_\l^{(1,-1)}[-\g^{-1},u]\ket{\vac}\,,
\ee
where the vacuum state $\ket{\vac}$ is identified with the constant symmetric function $1\in\Lambda[\bsx]$.
\end{proposition}
\begin{proof}
First, we observe that $\bra{\vac}\Phi_\l^{(1,-1)}[-\g^{-1},u]\ket{\vac}=t^{(1,-1)}_\l[-\g^{-1},u]=u^{|\l|}$, which is the normalization property.
Next, we can check that the properties \eqref{prop_Whit} of the Whittaker vector simply follow from the intertwining properties \eqref{prop_intw} of the vertex operator $\Phi_\l^{(1,-1)}$.
Recall that $\rho_u^{(1,-1)}(x_k^+)=\rho_u^{(1,0)}(\CT^{-1}(x_k^+))=\rho_u^{(1,0)}(x_{k+1}^+)$ and $\rho_u^{(1,0)}(x_k^+)\ket{\vac}=0$ for $k>0$.
It then follows that, for $k>0$,
\be
\begin{aligned}
 \rho_{u}^{(1,0)}(x_k^+)\,&\Phi_\l^{(1,-1)}[-\g^{-1},u]\ket{\vac}
 = \Phi_\l^{(1,-1)}[-\g^{-1},u] \left(\rho^{(0,1)}_u\otimes
 \rho^{(1,-1)}_{-\g^{-1}}\circ\Delta(x_k^+)\right) \ket{\vac} \\
 &= \Phi_\l^{(1,-1)}[-\g^{-1},u] \left(\rho^{(0,1)}_u\otimes
 \rho^{(1,-1)}_{-\g^{-1}}\left(x_k^+\otimes 1
 +\sum_{l\geq0}\g^{-(k+l/2)c}\psi_{-l}^-\otimes x_{k+l}^+\right)\right) \ket{\vac} \\
 &= \Phi_\l^{(1,-1)}[-\g^{-1},u] \left(\rho^{(0,1)}_u(x_k^+)\otimes1\right) \ket{\vac} \\
 &= \sum_{\sAbox\in A(\l)}(u\chi_{\sAbox})^k\, r_\l(\Abox)\,
 \Phi_{\l+\sAbox}^{(1,-1)}[-\g^{-1},u]\ket{\vac}\,,
\end{aligned}
\ee
where the first factor in the tensor product acts on the vertical leg while the second on the horizontal one, i.e.\ on the vacuum vector.

The other defining property can be shown by exploiting the fact that $\rho^{(1,-1)}_u(a_{k+1})\ket{\vac}=\rho_u^{(1,-1)}(x_{k}^+)\ket{\vac}=0$ for $k\geq0$. We deduce that
\be
\begin{aligned}
 \rho_{u}^{(1,0)}(a_k)\,&\Phi_\l^{(1,-1)}[-\g^{-1},u]\ket{\vac}
 = \Phi_\l^{(1,-1)}[-\g^{-1},u] \left(\rho^{(0,1)}_u\otimes
 \rho^{(1,-1)}_{-\g^{-1}}\circ\Delta(a_k)\right) \ket{\vac} \\
 &= \Phi_\l^{(1,-1)}[-\g^{-1},u] \left(\rho^{(0,1)}_u\otimes
 \rho^{(1,-1)}_{-\g^{-1}}\left(a_k\otimes\g^{-ck/2}
 +\g^{ck/2}\otimes a_k\right)\right) \ket{\vac} \\
 &= \Phi_\l^{(1,-1)}[-\g^{-1},u] \left(\rho^{(0,1)}_u(a_k)\otimes\g^{-k/2}\right) \ket{\vac} \\
 &= \dfrac{\g^{-k/2}}{k}(1-q_3^k) u^k p_k(\me_\l) \Phi_\l^{(1,-1)}[-\g^{-1},u]\ket{\vac}\,,\\
\end{aligned}
\ee
for $k>0$, which concludes the proof.
\end{proof}

\begin{proposition}
The vectors $W_\l(\bsx|u)$ can be explicitly described as specializations of Macdonald kernels as
\be
 W_\l(\bsx|u) = u^{|\l|} \Pi(u\,\sp_\l|\bsx)\,,
\ee
with $\sp_\l$ as in \eqref{expr_sp}.
\end{proposition}

\begin{proof}
We simply compute
\be
\begin{aligned}
 \Phi_\l^{(1,n)}[-\g^{-1},u]\ket{\vac}
 &= t_\l^{(1,n)}[-\g^{-1},u]\, \mathe^{\sum_{k>0}\frac{u^k}{k}\frac{1-t^k}{1-q^k}p_k(\sp_\l)J_{-k}} \ket{\vac}\\
 &= u^{|\l|} \,\mathe^{\sum_{k>0}\frac{u^k}{k}\frac{1-t^k}{1-q^k}p_k(\sp_\l)p_k(\bsx)}\\
 &= u^{|\l|} \Pi(u\,\sp_\l|\bsx),
\end{aligned}
\ee
which completes the proof.
\end{proof}

\subsubsection{The operator \texorpdfstring{$\mathring{\GHT}$}{V}}

It is useful at this point to rewrite the operator $\GHT$ in \eqref{eq:def-GHT-oper} in terms of intertwiners as we did for the Whittaker vectors $W_\l$. In order to do so, we first note that
\be
 \Phi_\vac^\ast(-\g^{-1})
 = \mathe^{\sum_{k>0}\frac{(-1)^k}{k(1-q^k)}p_k(\bsx)}
 \mathe^{-\sum_{k>0}\frac{(-1)^{-k}}{1-t^{-k}}\frac{\p}{\p p_k(\bsx)}}
 = \mathe^{\sum_{k>0}\frac{(-1)^k}{k(1-q^k)}p_k(\bsx)}t^{-L_0}
 \mathe^{\sum_{k>0}\frac{(-1)^{-k}}{1-t^{k}}\frac{\p}{\p p_k(\bsx)}}t^{L_0},
\ee
and, since $L_0$ commutes with $\nabla$ (they are both diagonal in the same basis), we obtain
\be\label{def_Xi}
 \GHT = \nabla \Phi_\vac^{\ast}(-\g^{-1})\nabla t^{-L_0}
 = t^{-L_0}\nabla \Phi_\vac^{\ast}(-t\g^{-1})\nabla\,.
\ee
In fact, it is instructive to consider a more general operator, namely
\be\label{def_tXi}
 \mathring{\GHT}[u,v] := \rho_{u}^{(1,0)}(F^\perp)\Phi_\vac^{\ast}(v)
 \rho_{u'}^{(1,0)}(F^\perp):\CF_{u'}^{(1,0)}\to\CF_u^{(1,0)}
\ee
with $u'=-\g uv$.

\begin{proposition}
The operator $\mathring{\GHT}[u,v]$ obeys the exchange relations
\be\label{rel_Xi}
\begin{aligned}
 \rho_u^{(1,0)}(a_k)\mathring{\GHT}[u,v]
 &= \g^{-k/2}\mathring{\GHT}[u,v]\rho_{u'}^{(1,0)}(b_k),\quad (k>0) \\
 \rho_u^{(1,0)}(x_0^+)\mathring{\GHT}[u,v]
 &= -\g^{1/2}(\g-\g^{-1})^{-1}\mathring{\GHT}[u,v]\rho_{u'}^{(1,0)}(a_{-1})\,,
\end{aligned}
\ee
in particular, the first relation with $k=1$ implies
\be
 \rho_u^{(1,0)}(a_1)\mathring{\GHT}[u,v]
 = \g^{-1/2}(\g-\g^{-1})\mathring{\GHT}[u,v]\rho_{u'}^{(1,0)}(x_0^+)\,,
\ee
with $u'=-\g u v$.
\end{proposition}

\begin{proof}
To derive the algebraic relations \eqref{rel_Xi}, we use the intertwining properties \eqref{rel_intw_vac} and \eqref{ak_intw_vac} of $\Phi_\vac^\ast(v)$, and \eqref{def_F} for $F^\perp$. We find, for $k>0$,
\be
\begin{aligned}
 \rho_u^{(1,0)}(a_k)\mathring{\GHT}[u,v]
 &= \rho_u^{(1,0)}(a_k)\rho_{u}^{(1,0)}(F^\perp)\Phi_\vac^{\ast}(v)\rho_{u'}^{(1,0)}(F^\perp) \\
 &= \rho_{u}^{(1,0)}(F^\perp)\rho_u^{(1,0)}(\g^{k\bc/2}\g^{-k\bc/2}\CT^\perp(a_k))\Phi_\vac^{\ast}(v)\rho_{u'}^{(1,0)}(F^\perp) \\
 &= \g^{-k/2}\rho_{u}^{(1,0)}(F^\perp)\Phi_\vac^{\ast}(v)\rho_{u'}^{(1,1)}(\CT^\perp(a_k))\rho_{u'}^{(1,0)}(F^\perp) \\
 &= \g^{-k/2}\rho_{u}^{(1,0)}(F^\perp)\Phi_\vac^{\ast}(v)\rho_{u'}^{(1,0)}(\CT\CT^\perp(a_k))\rho_{u'}^{(1,0)}(F^\perp) \\
 &= \g^{-k/2}\rho_{u}^{(1,0)}(F^\perp)\Phi_\vac^{\ast}(v)\rho_{u'}^{(1,0)}(F^\perp)\rho_{u'}^{(1,0)}(\CT^\perp\CT\CT^\perp(a_k)) \\
 &= \g^{-k/2}\mathring{\GHT}[u,v]\rho_{u'}^{(1,0)}(b_k)\,,
\end{aligned}
\ee
where we used the fact that $\g^{-k\bc/2}\CT^\perp(a_k)\in\CN^+$ for $k>0$ to obtain the second line, $\rho_u^{(1,1)}=\rho_u^{(1,0)}\circ\CT$ for the third line, and finally the identity \eqref{id_TpTTp} in the last line. The proof for the second relation follows the same steps. Indeed, since $\CT^\perp(x_0^+)=x_0^+\in\CN^+$, we have
\begin{equation}
\rho_u^{(1,0)}(x_0^+)\mathring{\GHT}[u,v]=\mathring{\GHT}[u,v]\rho_{u'}^{(1,0)}(\iota\CS(x_0^+))=-\g^{1/2}(\g-\g^{-1})^{-1}\mathring{\GHT}[u,v]\rho_{u'}^{(1,0)}(a_{-1}).
\end{equation}
\end{proof}

Observe that, by its definition, the morphism $\mathring{\GHT}[u,v]$ is not self-adjoint as it does not map a module to itself. We remark however that we can get a self-adjoint operator by specializing the weights appropriately. In fact, we define
\be
 \GHT(v):=\mathring{\GHT}[-\g v,-\g^{-1}]t^{-L_0}
 = \nabla(v) \,\mathe^{\sum_{k>0}\frac{(-1)^{k}}{k(1-q^k)}p_k}t^{-L_0}
 \mathe^{\sum_{k>0}\frac{(-1)^{k}}{1-t^k}\frac{\p}{\p p_k}}\nabla(v),
\ee
with $\nabla(v)$ defined as in \eqref{eq:def-nabla}. For $v=1$, this clearly reduces to the operator in \eqref{eq:def-GHT-oper}.

\subsubsection{Algebraic proof of the GHT identity}
Combining the exchange relations satisfied by the operator $\GHT(u)$, and the algebraic characterization of the Macdonald functions, we show that $\GHT(u)\tP_\l(\bsx)$ satisfies the same algebraic relations \eqref{prop_Whit} as the Whittaker vector $W_\l(\bsx|u)$.
\begin{proposition}
We have the algebraic relations
\be
\begin{aligned}
 \rho_u^{(1,0)}(a_k)\GHT(u)\tP_\l(\bsx) &= \dfrac{\g^{-k/2}}{k}(1-q_3^k)u^k p_k(\me_\l)\GHT(u)\tP_\l(\bsx),\quad (k>0),\\
 \rho_u^{(1,0)}(x_k^+)\GHT(u)\tP_\l(\bsx) &= \sum_{\sAbox\in A(\l)}
 (u\chi_\sAbox)^k\,r_\l(\Abox)\GHT(u)\tilde{P}_{\l+\sAbox}(\bsx),\quad (k\geq0).
\end{aligned}
\ee
\end{proposition}
\begin{proof}
Using the first relation in \eqref{rel_Xi}, together with \eqref{eigen_bk} for the eigenvalue of the operators $b_k$ in the horizontal representation, we obtain
\be
 \rho_u^{(1,0)}(a_k)\GHT(u)\tP_\l(\bsx)
 = \g^{-k/2}\GHT(u)\rho_{-\g u}^{(1,0)}(b_k)\tP_\l(\bsx)
 = \dfrac{\g^{-k/2}}{k}(1-q_3^k)u^k p_k(\me_\l)\GHT(u)\tP_\l(\bsx).
\ee
Similarly, from the second relation in \eqref{rel_Xi}, we have
\be
\begin{aligned}
 \rho_u^{(1,0)}(x_0^+)\GHT(u)\tP_\l(\bsx)
 &= (-\g^{-1})\rho_{-\g u}^{(1,0)}(x_0^+)\GHT(u)\tP_\l(\bsx)
 = t\g^{-1/2}(\g-\g^{-1})^{-1}\GHT(u)\rho_{-\g u}^{(1,0)}(a_{-1})\tP_\l(\bsx),\\
 &= \sum_{\sAbox\in A(\l)}r_\l(\Abox)\GHT(u)\tP_{\l+\sAbox}(\bsx).
\end{aligned}
\ee
Combining the two, using $x_k^+=c_k^{-1}\g^{kc/2}[a_k,x_0^+]$ and $c_k$ as in \eqref{exp_g}, we deduce
\be
\begin{aligned}
 \rho_u^{(1,0)}(x_k^+)\GHT(u)\tP_\l(\bsx)
 &= \sum_{\sAbox\in A(\l)}r_\l(\Abox)\dfrac{1-q_3^k}{kc_k}
 u^kp_k(\me_{\l+\sAbox}-\me_\l)\, \GHT(u)\tP_{\l+\sAbox}(\bsx),\\
 &= \sum_{\sAbox\in A(\l)}(u\chi_\sAbox)^k\,r_\l(\Abox)\GHT(u)\tP_{\l+\sAbox}(\bsx).
\end{aligned}
\ee
\end{proof}
This implies that $\GHT(u)\tP_\l(\bsx)$ and $W_\l(\bsx|u)$ transform in the same way under the action of the algebra.
In order to be able to identify them as the same vector in the module, we observe that the relations in \eqref{prop_Whit} imply that any Whittaker vector $W_\l(\bsx|u)$ can be obtained as a linear combination of repeated actions of the modes $x_k^+$ acting as creation operators on the vacuum $1$. This is the analogous statement as saying that $\tP_\l$ satisfy a recursion on the number of boxes in $\l$ induced by the Pieri rules. Since we have proven that the two recursions are the same, we just need to show that the two vectors are the same for the base case of the empty partition, $\l=\vac$.
In this particular case, the identity $\GHT(u)\tP_\vac(\bsx)=W_\vac(\bsx|u)$ becomes equivalent to the following property of the operator $\nabla$,
\be\label{rel_nabla}
 \nabla\, \mathe^{\sum_{k>0}\frac{z^k}{k(1-q^k)}p_k(\bsx)}
 = \mathe^{-\sum_{k>0}\frac{(-z)^k}{k(1-q^k)}p_k(\bsx)},
\ee
which is in fact a special case of the Conjecture~\ref{conj:BH}, and it is proven for $r=1$.

Equivalently, one could show that $\GHT(u)\tP_\l(\bsx)$ and $W_\l(\bsx|u)$ satisfy the same normalization condition, then by uniqueness of the definition of Whittaker vector, equality would follow. To show that the normalizations are the same, we would have to prove
\be
 (-u)^{|\l|} g_\l \tP_\l(-\sp_\vac) = u^{|\l|},
\ee
which is in fact a special case of the identity \eqref{eq:spec-id2}.

We finally obtain the following.
\begin{proposition}
We have the identity
\be
 \GHT(u)\tP_\l(\bsx)=W_\l(\bsx|u).
\ee
\end{proposition}
For $u=1$, this is equivalent to the GHT identity in \eqref{eq:GHT-thm}, originally proven in \cite{Garsia2001}.

\subsection{Extension of the GHT identity to higher levels}
In order to generalize the GHT identity \eqref{eq:GHT-thm} to higher levels, we need to introduce the corresponding definitions of the operator $\GHT$ and the Whittaker vectors $W_\l(\bsx)$. Then, we show that the latter are related to the generalized Macdonald symmetric functions studied in the previous section. In this subsection, we need to introduce the notation $\bsn_k:=(k,\dots,k)$ for the $r$-dimensional vectors with components $k\in\mZ$, and notice that $\rho_\bsu^{(r,\bsn_{k+1})}=\rho_\bsu^{(r,\bsn_k)}\circ\CT$.\footnote{This property follows from the compatibility relation between the automorphism $\CT$ and the Drinfeld coproduct. E.g. for $r=2$, $\rho^{(2,\bsn_k)}\circ\CT=(\rho^{(1,k)}\otimes\rho^{(1,k)})\circ\D\circ\CT=(\rho^{(1,k)}\circ\CT\otimes\rho^{(1,k)}\circ\CT)\circ\D=(\rho^{(1,k+1)}\otimes\rho^{(1,k+1)})\circ\D=\rho^{(2,\bsn_{k+1})}$.}

We note that a similar construction of Whittaker vectors for the quantum toroidal $\gl(1)$ algebra using vertex operators was used in \cite{BershteinGonin2020} (see also \cite{Negut2014,Tsymbaliuk2014} for a more geometric approach). However, the vertex operators and the corresponding vectors constructed here are different from those considered in \cite{BershteinGonin2020}.

\subsubsection{Generalized Whittaker vectors}
\begin{definition}\label{def:GeneralizedWhittaker}
We define the generalized Whittaker vector at level $r$ as
\be
 W_\bl(\bsx^\bullet|\bsu,\bsv,\bsw):=C_\bl(\bsu,\bsv,\bsw)\Phi_{\bl}^{(r,\bsn_{-1})}[\bsu,\bsv,\bsw]\ket{\vac},
\ee
with the normalization coefficients\footnote{This is precisely the coefficient introduced in Section~\ref{sec:iso-h-v} to relate the vertical basis vectors $\dket{\bl}$ with the $P_\bl(\bsx^\bullet|\bsu)$.}
\begin{equation}
C_\bl(\bsu,\bsv,\bsw)=\g^{-\frac12|\bl|-\sum_{\a=1}^r (r-\a)|\l^{(\a)}|}G_\bl^\ast(\bsv)(-q_1\g^{1/2})^{-|\bl|}\prod_{\a=1}^rP_{\l^{(\a)}}(\sp_\vac),
\end{equation}
and $G_\bl^\ast(\bsv)$ are defined in \eqref{def_gl}.
\end{definition}

\begin{proposition}
\label{prop:id_Whitt_r}
The generalized Whittaker vector transforms as:
\begin{align}
\begin{split}\label{id_Whitt_r}
(i)\quad &\rho_{\bsu'}^{(r,\bsn_0)}(a_k)W_\bl(\bsx^\bullet|\bsu,\bsv,\bsw)=-\dfrac{\g^{-rk/2}}{k}(\g^k-\g^{-k})\left(\sum_{\a=1}^r (\g v_\a)^kp_k(\me_{\l^{(\a)}}) \right)W_{\bl}(\bsx^\bullet|\bsu,\bsv,\bsw), \quad (k>0),\\
(ii)\quad &\rho_{\bsu'}^{(r,\bsn_0)}(x_k^+)W_\bl(\bsx^\bullet|\bsu,\bsv,\bsw)\\
&=-\g(1-q_1)\sum_{\a=1}^r\sum_{\sAbox\in A(\l^{(\a)})}(v_\a\chi_\sAbox)^k \psi_{\l^{(\a)}}(\Abox)\prod_{\b=\a+1}^r\Psi_{\l^{(\b)}}\left(\frac{u'_\a}{u'_\b}\chi_\sAbox\right) W_{\bl+\sAbox}(\bsx^\bullet|\bsu,\bsv,\bsw),\quad (k\geq0)\\
(iii)\quad &\rho_{\bsu'}^{(r,\bsn_0)}(x_k^-)W_\bl(\bsx^\bullet|\bsu,\bsv,\bsw)= -\g(1-q_2)\sum_{\sAbox\in R(\bl)}(\g^{-r}v_{\sAbox}\chi_{\sAbox})^k \psi_\l^\ast(\Abox) \prod_{\b=1}^{\a-1}\Psi_{\l^{(\b)}}\left(\frac{u'_\a}{u'_\b}\chi_\sAbox\right) W_{\bl-\sAbox}(\bsx^\bullet|\bsu,\bsv,\bsw)\,\\
 &+\d_{k,1}\left(\sum_{\a=1}^r u_\a^{-1}\right) W_{\bl}(\bsx^\bullet|\bsu,\bsv,\bsw)\,,\quad (k> 0)\,.
 \end{split}
\end{align}
\end{proposition}

The first relation expresses the fact that the formal series of vectors $W_\bl(\bsx^\bullet|\bsu,\bsv,\bsw)$ are eigenvectors of the annihilation operators $\rho_u^{(r,\bsn_0)}(a_k)$. As such, they, can be interpreted as a generalization of the Whittaker vectors $W_\l(\bsx)$ introduced in the previous section, i.e. as Whittaker vectors of the $\widehat{\gl(1)}$ subalgebra of $\CE$ generated by $a_k$ for $k\in\mZ$. However, this Whittaker condition no longer uniquely defines the vectors for modules $(r,\bsn_0)$ with $r>1$. Indeed, this notion is defined for Verma modules, and the Fock module $\CF^{(1,0)}$ is a Verma module for the algebra $\widehat{\gl(1)}$, but it is no longer the case at higher level. This translates into the fact that $\rho_\bsu^{(r,\bsn_0)}(a_k)$ is proportional to a specific linear combination of $\p/\p p_k(\bsx^{(\a)})$, and leaves the freedom of multiplying by any orthogonal linear combination of the power sums.

In this section, we focus on a specific set of solutions of the Whittaker conditions that are realized by the vertical component of a vertex operators, as defined in Definition~\ref{def:GeneralizedWhittaker}. Alternatively, this set of solutions can be constructed inductively from $W_{\vac,\dots,\vac}(\bsx^\bullet|\bsu,\bsv,\bsw)$ using the relation $(ii)$. The relation $(iii)$ further pinpoints this specific set of solutions.

\begin{proof}
All three properties follow from the intertwining relations of the vertex operator, and the fact that $\rho^{(r,\bsn_{-1})}_\bsu(x_k^+)\ket{\vac}=\rho^{(r,\bsn_{-1})}_\bsu(a_{k+1})\ket{\vac}=0$ for $k\geq0$. These lead to the following algebraic properties,
\be
\label{eq:comm-whitt-r}
\begin{aligned}
 \rho_{\bsu'}^{(r,\bsn_0)}(a_k)\Phi_{\bl}^{(r,\bsn_{-1})}[\bsu,\bsv,\bsw]\ket{\vac}
 &= \dfrac{\g^{-rk/2}}{k}(1-q_3^k)\left(\sum_{\a=1}^r v_\a^kp_k(\me_{\l^{(\a)}}) \right)\Phi_{\bl}^{(r,\bsn_{-1})}[\bsu,\bsv,\bsw]\ket{\vac}\,,
 \quad (k>0)\,,\\
 \rho_{\bsu'}^{(r,\bsn_0)}(x_k^+) \Phi_{\bl}^{(r,\bsn_{-1})}[\bsu,\bsv,\bsw]\ket{\vac}
 &= \sum_{\sAbox\in A(\bl)}(v_{\sAbox}\chi_{\sAbox})^k r_\bl(\Abox|\bsv)
 \Phi_{\bl+\sAbox}^{(r,\bsn_{-1})}[\bsu,\bsv,\bsw]\ket{\vac}\,,\quad (k\geq 0)\,,\\
  \rho_{\bsu'}^{(r,\bsn_0)}(x_k^-) \Phi_{\bl}^{(r,\bsn_{-1})}[\bsu,\bsv,\bsw]\ket{\vac}
 &= \sum_{\sAbox\in R(\bl)}(\g^{-r}v_{\sAbox}\chi_{\sAbox})^k r_\bl^\ast(\Abox|\bsv)
 \Phi_{\bl-\sAbox}^{(r,\bsn_{-1})}[\bsu,\bsv,\bsw]\ket{\vac}\,\\
 &+\d_{k,1}\left(\sum_{\a=1}^r u_\a^{-1}\right) \Phi_{\bl}^{(r,\bsn_{-1})}[\bsu,\bsv,\bsw]\ket{\vac}\,,\quad (k> 0)\,,
\end{aligned}
\ee
which prove the Proposition, after the appropriate change of normalization.
\end{proof}

\begin{corollary}
\label{corollary_I}
Let $\bl^R=(\l^{(r)},\dots,\l^{(1)})$ be the multi-partition with reversed order w.r.t.\ $\bl$ and similarly let $\bsv^R=(v^{(r)},\dots,v^{(1)})$ and $w^R_{\a,\b}=w_{\a,\a+1-\b}$. Under this permutation, generalized Whittaker vectors transform as
\begin{equation}
W_{\bl^R}(\bsx^\bullet|\bsu,\bsv^R,\bsw^R)=\g^{-\sum_\a (2\a-r-1)|\l^{(\a)}|}\dfrac{G_\bl(\bsv)}{G_\bl^\ast(\bsv)}W_{\bl}(\bsx^\bullet|\bsu,\bsv,\bsw)
\end{equation}
Note that the reversal symmetry exchanges the specializations $\bsw=\bsw^{(I)}$ and $\bsw=\bsw^{(II)}$.
\end{corollary}

\begin{proof}
This can be shown inductively from the algebraic relations of Proposition~\ref{prop:id_Whitt_r}, and the expression of $W_{\vac,\dots,\vac}(\bsx^\bullet|\bsu,\bsv,\bsw)$.
\end{proof}

In general, the relation between Whittaker vectors and the reproducing kernel seems lost in this context, and the vectors do not seem to possess a factorized expression. However, in the particular case $\bl=(\l^{(1)},\vac,\dots,\vac)$, and upon specialization of the internal weights to $\bsw=\bsw^{(II)}$, we do find a factorized expression
\begin{align}
\begin{split}
 & W_{\l^{(1)},\vac,\dots,\vac}(\bsx^\bullet|\bsu,\bsv,\bsw^{(II)})\\
 = &(-\g u_1 v_1)^{|\l^{(1)}|}\prod_{\sAbox\in\l^{(1)}}\prod_{\b=2}^r\left(1-q_3v_\b/(v_1\chi_\sAbox)\right) \Pi_{\mathrm{Z}}(v_r\sp_{\vac},v_{r-1}\sp_\vac,\dots,v_2\sp_\vac, v_1\sp_{\l^{(1)}}|\bsx^{(1)},\dots,\bsx^{(r)}).
\end{split}
\end{align}

In the particular case $\l^{(1)}=\vac$, this expression reduces to
\begin{equation}
 W_{\vac,\dots,\vac}(\bsx^\bullet|\bsu,\bsv,\bsw^{(II)})
 = \mathe^{-\sum_{k>0}\frac1{k(1-q_2^k)}\sum_{\a=1}^r p_k(\bsx^{(\a)})\left(v_\a^k+(1-q_3^k)
 \sum_{\b=\a+1}^r v_\b^k\right)}.
\end{equation}

\begin{remark}
Let us define an analogue of the Whittaker module by
\begin{equation}
\CW=\Span\{ W_\bl(\bsx^\bullet|\bsu,\bsv,\bsw),\quad \bl\in\CP^{\otimes r}\}
\end{equation}
where $\CP^{\otimes r}$ is the set of $r$-tuple partitions.
Let $\CB^{\perp,+}$ be the Borel subalgebra associated to $\D^\perp=(\CS^{-1}\otimes\CS^{-1})\circ\D\circ\CS$, and generated by $\{c,\bc,b_k,y_l^-\}$ for $k>0$ and $l\in\mZ$ (recall that $b_k=\CS(a_k)$ and $y_l^-=\CS(x_l^-)$). Let $\CB^{\perp,+}_{-1}$ be the subalgebra generated by $\{c,\bc,a_k, x_0^+,x_1^-\}$ for $k>0$. Note that this subalgebra contains the generators $y_l^-$ for $l\geq -1$ only. The formulas \eqref{id_Whitt_r} define an action $\tilde{\rho}$ of $\CB^{\perp,+}_{-1}$ on $\CW$, with the generators acting on the $r$-tuple partitions $\bl$,
\begin{equation}
\left.\rho_{\bsu'}^{(r,\bsn_0)}(b)\right|_{\bsx^\bullet} W_\bl(\bsx^\bullet|\bsu,\bsv,\bsw) = \left.\tilde{\rho}(b)\right|_\bl W_\bl(\bsx^\bullet|\bsu,\bsv,\bsw) ,\quad b\in \CB^{\perp,+}_{-1}.
\end{equation}
Apart from an extra shift for $x_1^-$, this action coincides with the vertical representation $\rho_\bsv^{(0,r)}(b)$, which seems to be a manifestation of bispectral duality. It would be interesting to understand what is the maximal subalgebra for which this identity holds, and if it contains the Borel subalgebra $\CB^{\perp,+}$.
\end{remark}

\subsubsection{The operator \texorpdfstring{$\mathring{\GHT}$}{V}} The algebraic definition \eqref{def_tXi} of the operator $\mathring{\GHT}$ can be extended to level $r>1$ using the vacuum component of the higher vertex operators defined in \eqref{def_Phi_rn},
\begin{equation}
\mathring{\GHT}[\bsu,\bsv,\bsw]:=\rho_{\bsu}^{(r,\bsn_0)}(F^\perp)\Phi_\vac^{(r,\bsn_0)\ast}[\bsu,\bsv,\bsw]\rho_{\bsu'}^{(r,\bsn_0)}(F^\perp):\CF_{\bsu'}^{(r,\bsn_0)}\to\CF_{\bsu}^{(r,\bsn_0)}
\end{equation}
where the weights $\bsu'$ are fixed through the relations \eqref{rel_weights}.
Similarly to its level one counterpart, this operator obeys the following algebraic relations,
\be
\label{rel_Xi_r}
\begin{aligned}
&\rho_{\bsu}^{(r,\bsn_0)}(a_k){\mathring{\GHT}}[\bsu,\bsv,\bsw]=\g^{-rk/2}{\mathring{\GHT}}[\bsu,\bsv,\bsw]\rho_{\bsu'}^{(r,\bsn_0)}(b_k),\quad (k>0)\\
&\rho_{\bsu}^{(r,\bsn_0)}(x_0^+){\mathring{\GHT}}[\bsu,\bsv,\bsw]=-\g^{r/2}(\g-\g^{-1})^{-1}{\mathring{\GHT}}[\bsu,\bsv,\bsw]\rho_{\bsu'}^{(r,\bsn_0)}(a_{-1}),\\
&\rho_{\bsu}^{(r,\bsn_0)}(x_1^-){\mathring{\GHT}}[\bsu,\bsv,\bsw]=-\sum_{\a=1}^r(\g v_\a){\mathring{\GHT}}[\bsu,\bsv,\bsw] -\g^{-r}{\mathring{\GHT}}[\bsu,\bsv,\bsw]\rho_{\bsu'}^{(r,\bsn_0)}(x_{1}^+).
\end{aligned}
\ee
The proof of the first two relations is a simple extension of the one given for $r=1$, and does not present any difficulty. The third relation is new, it follows from the transformation property of $x_1^-$ under $\CT^\perp$, and the relation \eqref{ak_intw_vac}.

\subsubsection{GHT identity}
Let us introduce the vector
\begin{equation}
 \Whit_\bl(\bsx^\bullet|\bsu,\bsv,\bsw)
 :={\mathring{\GHT}}[\bsu,\bsv,\bsw]P_\bl(\bsx^\bullet|\bsu')\in\CF_{\bsu}^{(r,\bsn_0)},
\end{equation}
with $\bsu'$ determined by $\bsu,\bsv,\bsw$ to be the expression \eqref{rel_weights}.
The goal is to show that, for a specific choice of weights, the vector $\Whit_\bl$ is equal to the generalized Whittaker vector $W_\bl$. First, we observe the following.
\begin{proposition}
We have the algebraic relations
\be
\begin{aligned}
(i)\quad &\rho_{\bsu}^{(r,\bsn_0)}(a_k)\Whit_\bl(\bsx^\bullet|\bsu,\bsv,\bsw)=\frac{\g^{-rk/2}}{k}(\g^k-\g^{-k})\left(\sum_{\a=1}^r p_k(u'_\a\me_{\l^{(\a)}})\right) \Whit_{\bl}(\bsx^\bullet|\bsu,\bsv,\bsw)\,,\quad(k>0)\,,\\
(ii)\quad &\rho_{\bsu}^{(r,\bsn_0)}(x_k^+)\Whit_\bl(\bsx^\bullet|\bsu,\bsv,\bsw)\\
&=-\g(1-q_1)\sum_{\a=1}^r\sum_{\sAbox\in A(\l^{(\a)})}(-\g^{-1}u'_\a\chi_\sAbox)^k \psi_{\l^{(\a)}}(\Abox)\prod_{\b=\a+1}^r\Psi_{\l^{(\b)}}\left(\frac{u'_\a}{u'_\b}\chi_\sAbox\right)\ \Whit_{\bl+\sAbox}(\bsx^\bullet|\bsu,\bsv,\bsw)\,,\quad (k\geq0)\,,\\
(iii)\quad &\rho_{\bsu'}^{(r,\bsn_0)}(x_k^-)\Whit_\bl(\bsx^\bullet|\bsu,\bsv,\bsw)\\
&=-\g(1-q_2)\sum_{\a=1}^r\sum_{\sAbox\in R(\l^{(\a)})}(-\g^{-(r+1)}u'_\a\chi_{\sAbox})^k \psi_{\l^{(\a)}}^\ast(\Abox) \prod_{\b=1}^{\a-1}\Psi_{\l^{(\b)}}\left(\frac{u'_\a}{u'_\b}\chi_\sAbox\right)\Whit_{\bl-\sAbox}(\bsx^\bullet|\bsu,\bsv,\bsw)\,\\
 &-\d_{k,1}\left(\sum_{\a=1}^r (\g v_\a)\right) \Whit_{\bl}(\bsx^\bullet|\bsu,\bsv,\bsw)\,,\quad (k> 0)\,.
\end{aligned}
\ee
\end{proposition}
\begin{proof}
The proof for $x_0^+$ and $a_k$ simply follows from \eqref{rel_Xi_r}, combined with the properties of the generalized Macdonald symmetric functions. For $k>0$ in ($ii$), we use $x_k^+=c_1^{-k}\g^{k c/2}(\ad_{a_1})^k x_0^+$. For the third relation, we need
\begin{equation}
\rho_\bsu^{(r,\bsn_0)}(x_1^+)P_\bl(\bsx^\bullet|\bsu)=-(1-q_2)\sum_{\a=1}^r\sum_{\sAbox\in R(\l^{(\a)})}(u_\a\chi_{\sAbox}) \psi_{\l^{(\a)}}^\ast(\Abox) \prod_{\b=1}^{\a-1}\Psi_{\l^{(\b)}}\left(\frac{u_\a}{u_\b}\chi_\sAbox\right)P_{\bl-\sAbox}(\bsx^\bullet|\bsu).
\end{equation}
The proof for modes $x_k^-$ goes along the same lines.
\end{proof}
At this stage, we can already compare the algebraic relations satisfied by $\Whit_\bl(\bsx^\bullet|\mathring{\bsu},\mathring{\bsv},\mathring{\bsw})$ and $W_\bl(\bsx^\bullet|\bsu,\bsv,\bsw)$. They coincide if we identify the weights as follows,
\begin{equation}\label{cond_weights}
v_\a=-\g^{-1}\mathring{u}_\a',\quad u_\a=-\g^{-1}\mathring{v}_\a^{-1},\quad u'_\a=A \mathring{u}'_\a,
\end{equation}
for some $A\in\mC$. In general, the relation \eqref{rel_weights} that determines $\mathring{\bsu}'$ from the other weights is relatively complicated, but it reduces to $\mathring{u}'_\a=-\g^{r+2-2\a}\mathring{u}_\a \mathring{v}_\a$ if we specialize the internal weights to $\mathring{\bsw}=\bsw^{(I)}$.\footnote{The other specialization $\mathring{\bsw}=\bsw^{(II)}$ corresponds to the reflection $\mathring{v}_\a\to \mathring{v}_{r+1-\a}$.} With this specialization, we find
\be
\begin{aligned}
 \Whit_{\vac,\dots,\vac}(\bsx^\bullet|\mathring{\bsu},\mathring{\bsv},\bsw^{(I)})
 &= \rho_{\mathring{\bsu}}^{(r,0)}(F^\perp)
 \exp\left(\sum_{k>0}\frac{\g^{rk}}{k(1-q_2^k)}\sum_{\a=1}^r
 q_3^{-(\a-1)k}\mathring{v}_\a^k p_k(\bsx^{(\a)})\right) \\
 &=\exp\left(-\sum_{k>0}\frac{1}{k(1-q_2^k)}\sum_{\a=1}^r p_k(\bsx^{(\a)})
 \Big((\mathring{u}''_\a)^k+(1-q_3^k)\sum_{\b=\a+1}^r(\mathring{u}''_\b)^k\Big)\right)
\end{aligned}
\ee
using the conjectured formula \eqref{nabla_r} for the action of $F^\perp$, and denoting $\mathring{u}''_\a=\g^{r-2\a+1}\mathring{u}_\a\mathring{v}_\a$. The comparison with $W_\vac(\bsx^\bullet|\bsu,\bsv,\bsw^{(I)})$ leads to the identification $v_\a=\mathring{u}''_\a=\g^{r-2\a+1}\mathring{u}_\a\mathring{v}_\a$, which is indeed consistent with the previous one. After specialization of the internal weights, we have $u'_\a=-\g^{2\a-r}u_\a v_\a$ and $\mathring{u}_\a'=-\g^{r+2-2\a}\mathring{u}_\a\mathring{v}_\a$. The conditions on the weights \eqref{cond_weights} fixes the vertical weights $\mathring{v}_\a=-A^{-1}\g^{-r-2+2\a}$, and we will choose $A=1$ so that $\mathring{u}_\a'=\mathring{u}_\a$. We conclude that
\begin{equation}\label{eq:GHT_r}
W_\bl(\bsx^\bullet|\bsu,\bsv,\bsw^{(II)})=\mathring{W}(\bsx^\bullet|\mathring{\bsu},\mathring{\bsv},\bsw^{(I)}),\quad \mathring{v}_\a=-\g^{-r-2+2\a},\quad u_\a=\g^{r+1-2\a},\quad v_\a=-\g^{-1}\mathring{u}_\a.
\end{equation}
We end up with only one set of free parameters corresponding to $\mathring{u}_\a=-\g v_\a$.

\paragraph{Summary.} The equation \eqref{eq:GHT_r} can already be interpreted as a higher level generalization of the GHT identity. It is possible to rewrite it in a closer form to the original formula by introducing the operator $\GHT(\bsv)$ defined as a specialization of the operator ${\mathring{\GHT}}[\mathring{\bsu},\mathring{\bsv},\mathring{\bsw}]t^{-L_0}$ to the weights $\mathring{\bsw}=\bsw^{(I)}$, $\mathring{v}_\a=-\g^{-r-2+2\a}$ and $\mathring{u}_\a=-\g v_\a$. An explicit expression for this operator is obtained by combining the expression \eqref{Phi_vac_st} of the vertex operator, and the representation \eqref{rep_framing} of the framing operator,
\be
\label{eq:GHTV(v)}
 \GHT(\bsv) := \nabla(\bsv)\,
 \exp\left(\sum_{k>0}\frac{(-1)^k}{k(1-q^k)}\sum_{\a=1}^r p_k^{(\a)}\right)
 t^{-L_0}
 \exp\left(\sum_{k>0}\frac{(-1)^k}{1-t^k}\sum_{\a=1}^r q_3^{(\a-1)k}\frac{\p}{\p p_k^{(\a)}}\right)
 \nabla(\bsv)
\ee
Using that $\nabla(\bsv)$ is self-adjoint w.r.t.\ the inner product $\langle\,,\,\rangle_\mathrm{Z}$, together with \eqref{adj_Z}, we immediately obtain that the operator $\GHT(\bsv)$ is also self-adjoint, i.e.\ $\GHT(\bsv)^\dagger=\GHT(\bsv)$.

Similarly, we modify the normalization of the generalized Whittaker vector, and consider its specialization to the weights $u_\a=\g^{r+1-2\a}$, in order to define
\begin{equation}\label{def_Wbl}
W_\bl(\bsx^\bullet|\bsv):=(-1)^{|\bl|}\g^{-|\bl|-\sum_\a(r-\a)|\l^{(\a)}|}G_\bl^\ast(\bsv)\left.\Phi_\bl^{(r,\bsn_{-1})}[\bsu,\bsv,\bsw^{(II)}]\ket{\vac}\right|_{u_\a\to\g^{r+1-2\a}}.
\end{equation}
This new normalization is adapted to the spherical generalized Macdonald basis. Using these two objects, we can formulate the main result of this section.

\begin{proposition}
\label{prop:GHTr}
The level-$r$ generalization of the GHT identity reads
\be
\label{eq:GHTr}
 \GHT(\bsv)\tP_\bl(\bsx^\bullet|\bsv) = W_\bl(\bsx^\bullet|\bsv),
\ee
and the r.h.s. satisfies the relations:\footnote{To simplify this expression, we used the properties $\rho_{a\bsu}^{(r,\bsn_0)}(a_k)=\rho_{\bsu}^{(r,\bsn_0)}(a_k)$ and $\rho_{a\bsu}^{(r,\bsn_0)}(x_k^\pm)=a^{\pm1}\rho_{\bsu}^{(r,\bsn_0)}(x_k^\pm)$ of the horizontal Fock representation.}
\begin{align}
\begin{split}
(i)\quad &\rho_{\bsv}^{(r,\bsn_0)}(a_k)W_\bl(\bsx^\bullet|\bsv)= -\dfrac{\g^{-rk/2}}{k}(\g^k-\g^{-k})\left(\sum_{\a=1}^r (\g v_\a)^kp_k(\me_{\l^{(\a)}}) \right)W_\bl(\bsx^\bullet|\bsv), \quad (k>0),\\
(ii)\quad &\rho_{\bsv}^{(r,\bsn_0)}(x_k^+)W_\bl(\bsx^\bullet|\bsv)\\
&=\sum_{\a=1}^r\sum_{\sAbox\in A(\l^{(\a)})}(v_\a\chi_\sAbox)^k r_{\l^{(\a)}}(\Abox)\prod_{\b=\a+1}^r\Psi_{\l^{(\b)}}\left(\frac{v_\a}{v_\b}\chi_\sAbox\right) W_{\bl+\sAbox}(\bsx^\bullet|\bsv),\quad (k\geq0)\\
(iii)\quad &\rho_{\bsv}^{(r,\bsn_0)}(x_k^-)W_\bl(\bsx^\bullet|\bsv)=\sum_{\sAbox\in R(\bl)}(\g^{-r}v_{\sAbox}\chi_{\sAbox})^k r_{\l^{(\a)}}^\ast(\Abox) \prod_{\b=1}^{\a-1}\Psi_{\l^{(\b)}}\left(\frac{v_\a}{v_\b}\chi_\sAbox\right) W_{\bl-\sAbox}(\bsx^\bullet|\bsu,\bsv,\bsw)\,\\
 &-\d_{k,1}\g\dfrac{\g^r-\g^{-r}}{\g-\g^{-1}}W_\bl(\bsx^\bullet|\bsv)\,,\quad (k> 0)\,.
 \end{split}
\end{align}
\end{proposition}

\paragraph{Example.} In this example, we give an explicit expression for the Whittaker vectors at level $r=2$. The starting point is the construction of the vertex operators presented in Appendix~\ref{sec_higher_VO},
\begin{align}
\begin{split}
&\Phi^{(2,n_{-1})}[\bsu,\bsv,\bsw^{(II)}]\\
=&\sum_{\l,\mu}a_\l a_\mu \dbra{\l}\otimes\dbra{\mu}\otimes\sum_{\nu\subseteq\mu} a_\nu \Phi_\l^{(1,-1)}[\g^{-1}u_1,v_1]\Phi_\nu^{(1,-1)\ast}[\g^{-1}u_1,\g v_2]\Phi_\mu^{(1,-1)}[u_1,v_2]\otimes\Phi_\nu^{(1,-1)}[u_2,\g v_2].
\end{split}
\end{align}
We note that the sum over partitions is restricted to $\nu\subseteq\mu$ due to the specialization $\bsw=\bsw^{(II)}$. Using the relation $a_\l a_\mu\dbra{\l}\otimes\dbra{\mu}=G_{\l,\mu}(\bsv)^{-1}a_{\l,\mu}(\bsv)\dbra{\l,\mu}$, we extract the vertical component,
\begin{equation}
\Phi_{\l,\mu}^{(2,n_{-1})}[\bsu,\bsv,\bsw^{(II)}]=G_{\l,\mu}(\bsv)^{-1}\sum_{\nu\subseteq\mu} a_\nu \Phi_\l^{(1,-1)}[\g^{-1}u_1,v_1]\Phi_\nu^{(1,-1)\ast}[\g^{-1}u_1,\g v_2]\Phi_\mu^{(1,-1)}[u_1,v_2]\otimes\Phi_\nu^{(1,-1)}[u_2,\g v_2]
\end{equation}
After normal ordering, and extracting the normalization factors of the vertex operators, we find\footnote{In this derivation, we ignore the normalization factor of the type $\CG$-functions that will ultimately be canceled.}
\begin{equation}
\Phi_{\l,\mu}^{(2,n_{-1})}[\bsu,\bsv,\bsw^{(II)}]=(-u_1v_1)^{|\l|}(-u_1 v_2)^{|\mu|}\dfrac{\tN_{\mu,\l}(v_2/v_1)}{N_{\mu,\l}(v_2/v_1)}\sum_{\nu\subseteq\mu} \dfrac{a_\nu}{g_\nu}\left(-\g\dfrac{u_2}{u_1}\right)^{|\nu|}N_{\nu,\l}(q_3 v_2/v_1)N_{\mu,\nu}(1)\CO_{\l,\mu,\nu}(\bsv)
\end{equation}
where $\CO_{\l,\mu,\nu}(\bsv)$ is an operator normalized to one. Acting on the vacuum, it produces
\begin{equation}
\CO_{\l,\mu,\nu}(\bsv)\ket{\vac}=\exp\left(\sum_{k>0}\frac1k\frac{1-t^k}{1-q^k}\Big(
 \left(p_k(v_1\sp_\l)+p_k(v_2\sp_\mu)-q_3^kp_k(v_2\sp_\nu)\right)p_k(\bsx)
 +p_k(v_2\sp_\nu)p_k(\bsy)
 \Big)\right)
\end{equation}
Taking into account the rescaling by a factor $(-\g)^{-|\l|-|\mu|}\g^{-|\l|}G_{\l,\mu}^\ast(\bsv)$ in \eqref{def_Wbl}, and the specialization $u_1=\g$, $u_2=\g^{-1}$, we find\footnote{Simplification of the prefactor follows from
\begin{equation}
\dfrac{\tN_{\mu,\l}(v_2/v_1)}{\tN_{\l,\mu}(v_1/v_2)N_{\mu,\l}(v_2/v_1)}=\dfrac{q_3^{|\mu|}}{N_{\mu,\l}(q_3v_2/v_1)}.
\end{equation}}
\begin{multline}
 W_{\l,\mu}(\bsx,\bsy|v_1,v_2)
 = v_1^{|\l|}(q_3v_2)^{|\mu|}
 \sum_{\nu\subseteq\mu} q_3^{-{|\nu|}}
 \frac{N_{\mu,\nu}(1)}{N_{\nu,\nu}(1)}
 \frac{N_{\nu,\l}(q_3v_2/v_1)}{N_{\mu,\l}(q_3v_2/v_1)} \\
 \times\exp\left(\sum_{k>0}\frac1k\frac{1-t^k}{1-q^k}\Big(
 \left(p_k(v_1\sp_\l)+p_k(v_2\sp_\mu)-q_3^kp_k(v_2\sp_\nu)\right)p_k(\bsx)
 +p_k(v_2\sp_\nu)p_k(\bsy)
 \Big)\right).
\end{multline}
The GHT identity takes the form
\be
 W_{\l,\mu}(\bsx,\bsy|v_1,v_2)=\GHT(\bsv)\tilde{P}_{\l,\mu}(\bsx,\bsy|v_1,v_2),
\ee
with
\be
 \GHT(\bsv) = \nabla(\bsv)\, \mathe^{\sum_{k>0}\frac{(-1)^k}{k(1-q^k)}(p_k(\bsx)+p_k(\bsy))}
 t^{-L_0}\mathe^{\sum_{k>0}\frac{(-1)^k}{1-t^k}\left(\frac{\p}{\p p_k(\bsx)}+q_3^k\frac{\p}{\p p_k(\bsy)}\right)}\nabla(\bsv).
\ee

\paragraph{Specialization to $r=1$.} For $r=1$, the spherical Macdonald function no longer depends on the weight $v$, i.e. $\tP_\l(\bsx|v)=\tP_\l(\bsx)$. However, the nabla operator still depends on $v$, as $\nabla(v)=v^{L_0}\nabla$, and so $\GHT(v)=v^{L_0}\GHT v^{L_0}$. Thus, the GHT formula \eqref{eq:GHTr} becomes
\begin{equation}
v^{-|\l|} v^{-L_0} W_\l(\bsx|v)=\GHT \tP_\l(\bsx).
\end{equation}
To relate this expression to the one derived in the previous subsection, we examine the Whittaker vector obtained here from \eqref{def_Wbl} as
\begin{equation}
W_\l(\bsx|v)=(-\g^{-1})^{|\l|}\Phi_\l^{(1,-1)}[1,v]\ket{\vac}= v^{|\l|}\Pi(\bsx|v\e_\l).
\end{equation}
Thus, by defining $W_\l(\bsx)=v^{-|\l|} v^{-L_0} W_\l(\bsx|v)=\Pi(\bsx|\e_\l)$ we recover indeed the original GHT formula. 

\subsubsection{Inductive construction of the generalized Whittaker vector}

The generalized Whittaker vector defined in \eqref{def_Wbl} satisfies certain recursion relations corresponding to the decomposition of the vertex operator $\Phi^{(r,\bsn_{-1})}_\bl$ discussed in Appendix~\ref{app:B1}. In this subsection, we use the other specialization of the internal weights, which allows us to reverse the order of the partitions $\l^{(\a)}$ in $\bl$ (see Corollary \ref{corollary_I}).

The first type of inductive construction of the Whittaker vector is a consequence of the decomposition in \eqref{rec_I}, where the vertex operator is cut along the vertical legs. This allows to reconstruct the Whittaker vector by gluing horizontal strips of the form
\begin{center}
 \begin{tikzpicture}[
  scale=0.75,
  baseline=(current bounding box.center),
  arrow inside/.style = {
    postaction={decorate},
    decoration={markings, mark=at position 0.5 with {\arrow{stealth}}}
  }
 ]
  \draw[arrow inside] (-3,2) node [left] {$\emptyset$} -- (-2,1);
  \draw[arrow inside] (-2,1) -- (-1,1);
  \draw[arrow inside] (-1,1) -- (0,0);
  \draw[arrow inside] (0,0) -- (1,0);
  \draw[arrow inside] (1,0) -- (2,-1);
  \draw[arrow inside] (2,-1) -- (3,-1) node [right] {};
  \draw[arrow inside] (-1,1) -- (-1,2) node [above] {$\mu^{(1)}$};
  \draw[arrow inside] (1,0) -- (1,1) node [above] {$\mu^{(2)}$};
  \draw[arrow inside] (-2,0) node [below] {$\l^{(1)}$} -- (-2,1);
  \draw[arrow inside] (0,-1) node [below] {$\l^{(2)}$} -- (0,0);
  \draw[arrow inside] (2,-2) node [below] {$\l^{(3)}$} -- (2,-1);
 \end{tikzpicture}
\end{center}
where each strip has $r-1$ vertical legs on the top and $r$ vertical legs at the bottom. The left $(1,-1)$-leg is contracted with the vacuum and the right horizontal leg is kept ``open'' corresponding to the tensor factor $\L[\bsx^{(1)}]$. These are indeed the vertical matrix elements of the operators defined in \eqref{eq:opAuvv}, after a specialization of the weights.

Gluing strips of increasing length from the top produces generalized Whittaker vectors of higher level. In particular, we can write the first recursion formula
\begin{multline}
 W_{\l^{(1)},\dots,\l^{(r)}}(\bsx^{(1)},\dots,\bsx^{(r)}|u_1,\dots,u_r) \\
 = \sum_{\mu^{(1)},\dots,\mu^{(r-1)}}
 M_{\bl,\bmu}(\bsu)
 \exp\left(\sum_{k>0}\frac1k\frac{1-t^k}{1-q^k}p_k(\bsx^{(1)})
 \Big(\sum_{\a=1}^r p_k(u_\a\sp_{\l^{(\a)}})
 -q_3^k\sum_{\a=2}^r p_k(u_\a\sp_{\mu^{(\a-1)}}) \Big)\right)\\
 \times W_{\mu^{(1)},\dots,\mu^{(r-1)}}(\bsx^{(2)},\dots,\bsx^{(r)}|u_2,\dots,u_r),
\end{multline}
with coefficients
\be
 M_{\bl,\bmu}(\bsu) := u_1^{|\bl|-|\bmu|}
 \left[\prod_{\a=1}^r\prod_{\b=1}^{\a-1}
 \frac{N_{\mu^{(\b)},\l^{(\a)}}(q_3u_{\b+1}/u_\a)}
 {N_{\l^{(\b)},\l^{(\a)}}(u_{\b}/u_\a)}\right]
 \times
 \left[\prod_{\a=1}^{r-1}\prod_{\b=1}^{\a}
 \frac{N_{\l^{(\b)},\mu^{(\a)}}(u_\b/u_{\a+1})}
 {N_{\mu^{(\b)},\mu^{(\a)}}(q_3u_{\b+1}/u_{\a+1})}\right],
\ee
which vanish identically unless $\mu^{(\a)}\subseteq\l^{(\a+1)}$ for every $\a$.

The second type of inductive construction is a consequence of the decomposition in \eqref{rec_II}, where the vertex operator is cut along the horizontal legs. This allows to reconstruct the Whittaker vector by gluing vertical strips as in Figure~\ref{fig:Whittaker-recursion2}, each strip having $r-1$ horizontal legs on the left and $r$ on the right. The top $(1,-1)$-leg contracts with the vacuum, while the bottom vertical leg is labeled by one of the partitions in $\bl$.
\begin{figure}[ht!]
 \centering
 \begin{tblr}{
  colspec = {Q[c,m]Q[c,m]Q[c,m]Q[c,m]Q[c,m]Q[c,m]Q[c,m]Q[c,m]},
 }
 $W_{\bl}(\bsx^\bullet|\bsu)$ & $=$ &
 \begin{tikzpicture}[
  scale=0.75,
  baseline=(current bounding box.center),
  arrow inside/.style = {
    postaction={decorate},
    decoration={markings, mark=at position 0.5 with {\arrow{stealth}}}
  }
 ]
  \draw[arrow inside] (0,0) node [left] {$\vac$} -- (1,-1);
  \draw[arrow inside] (1,-1) -- (2,-1) node [right] {};
  \draw[arrow inside] (1,-2) node [below] {$\l^{(1)}$} -- (1,-1);
 \end{tikzpicture}
 & $\times$ &
 \begin{tikzpicture}[
  scale=0.75,
  baseline=(current bounding box.center),
  arrow inside/.style = {
    postaction={decorate},
    decoration={markings, mark=at position 0.5 with {\arrow{stealth}}}
  }
 ]
  \draw[arrow inside] (0,0) -- (1,-1);
  \draw[arrow inside] (1,-1) -- (2,-1) node [right] {};
  \draw[arrow inside] (1,-2) node [below] {$\l^{(2)}$} -- (1,-1);
  \draw[arrow inside] (0,0) -- (0,1);
  \draw[arrow inside] (0,1) -- (1,1) node [right] {};
  \draw[arrow inside] (-1,2) node [left] {$\vac$} -- (0,1);
  \draw[arrow inside] (-1,0) node [left] {} -- (0,0);
 \end{tikzpicture}
 & $\times$ &
 \begin{tikzpicture}[
  scale=0.75,
  baseline=(current bounding box.center),
  arrow inside/.style = {
    postaction={decorate},
    decoration={markings, mark=at position 0.5 with {\arrow{stealth}}}
  }
 ]
  \draw[arrow inside] (0,0) -- (1,-1);
  \draw[arrow inside] (1,-1) -- (2,-1) node [right] {};
  \draw[arrow inside] (1,-2) -- (1,-1);
  \draw[arrow inside] (0,0) -- (0,1);
  \draw[arrow inside] (0,1) -- (1,1) node [right] {};
  \draw[arrow inside] (-1,2) node [left] {$\vac$} -- (0,1);
  \draw[arrow inside] (-1,0) node [left] {} -- (0,0);
  \draw[arrow inside] (0,-2) node [left] {} -- (1,-2);
  \draw[arrow inside] (1,-2) -- (2,-3);
  \draw[arrow inside] (2,-3) -- (3,-3) node [right] {};
  \draw[arrow inside] (2,-4) node [below] {$\l^{(3)}$} -- (2,-3);
 \end{tikzpicture}
 & $\cdots$
 \end{tblr}
 \captionof{figure}{Inductive (horizontal) construction of the generalized Whittaker vector.}
 \label{fig:Whittaker-recursion2}
\end{figure}
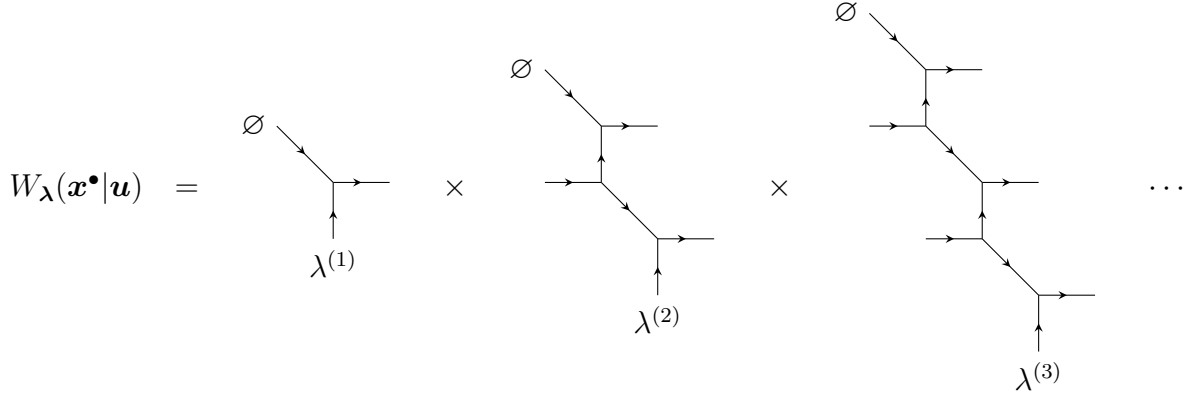

Each such strip can be thought of as an operator $\mathcal{A}^{(k)}_{\l}(u_k)$ of the form
\be
 \mathcal{A}^{(r)}_{\l}(u_r) : \CF_{\g u_1,\dots,\g u_{r-1}}^{(r-1,\bsn_0)}
 \to \CF_{u_1,\dots,u_{r-1},u_r}^{(r,\bsn_0)}.
\ee
defined as
\be
\label{eq:Aoperators}
\begin{aligned}
 \mathcal{A}^{(r)}_{\nu^{(r)}}(u_r)
 :=&\, u_1^{|\nu^{(r)}|}
 \sum_{\nu^{(1)}\subseteq\nu^{(2)}\subseteq\dots\subseteq\nu^{(r-1)}\subseteq\nu^{(r)}}
 \left[\prod_{\a=1}^{r-1} \left(\frac{u_{r-\a+1}}{u_{r-\a}}\right)^{|\nu^{(\a)}|}
 \frac{N_{\nu^{(\a)},\nu^{(\a+1)}}(q_3)}{N_{\nu^{(\a)},\nu^{(\a)}}(q_3)}\right] \\
 &\times \exp\left(\sum_{k>0}\frac{u_r^k}{k}\frac{1-t^k}{1-q^k}\Big(\sum_{\a=1}^{r-1}
 \left(p_k(\sp_{\nu^{(r-\a+1)}})-q_3^kp_k(\sp_{\nu^{(r-\a)}})\right)p_k^{(\a)}
 +p_k(\sp_{\nu^{(1)}})p_k^{(r)}\Big)\right)\\
 &\times \exp\left(-\sum_{k>0}q_3^{-k}u_r^{-k}\sum_{\a=1}^{r-1}
 \left(p_k(\sp^\vee_{\nu^{(r-\a+1)}})-p_k(\sp^\vee_{\nu^{(r-\a)}})\right)
 \frac{\partial}{\partial p_k^{(\a)}}
 \right).
\end{aligned}
\ee
We can then write the second recursion relation
\be
 W_{\lam^{(1)},\dots,\l^{(r)}}(\bsx^{(1)},\dots,\bsx^{(r)}|u_1,\dots,u_r)
 = \mathcal{A}^{(r)}_{\l^{(r)}}(u_r)\cdot
 W_{\lam^{(1)},\dots,\l^{(r-1)}}(\bsx^{(1)},\dots,\bsx^{(r-1)}|u_1,\dots,u_{r-1}).
\ee
This allows us to decompose the generalized Whittaker vectors completely through the repeated action of the $\mathcal{A}$-operators as
\be
\label{eq:WhittakerToA}
 W_\bl(\bsx^\bullet|\bsu)
 = \mathcal{A}^{(r)}_{\l^{(r)}}(u_r)\cdots \mathcal{A}^{(1)}_{\l^{(1)}}(u_1)\cdot 1.
\ee

Using the GHT identity in Proposition~\ref{prop:GHTr}, the operator $\GHT(\bsu)$ maps generalized Whittaker vectors to generalized (spherical) Macdonald functions. Consequently, after conjugation by $\GHT(\bsu)$, the generalized Macdonald functions admit a recursive construction analogous to \eqref{eq:WhittakerToA}. In fact, if we define the formal creation operators
\be
\label{eq:Bcreation}
 \mathcal{B}^{(r)}_\l(u_r)
 := \GHT(u_1,\dots,u_{r-1},u_r)^{-1}
 \mathcal{A}^{(r)}_\l(u_r)
 \GHT(u_1,\dots,u_{r-1})\,
\ee
we can formally decompose the generalized Macdonald functions as
\be
 \tP_\bl(\bsx^\bullet|\bsu)
 = \mathcal{B}^{(r)}_{\l^{(r)}}(u_r)\cdots\mathcal{B}^{(1)}_{\l^{(1)}}(u_1)\cdot 1.
\ee
The definition in \eqref{eq:Bcreation}, however, appears to be of limited practical use. An open problem is to find an explicit formula for the $\mathcal{B}$-operators, similar in spirit to the construction of the creation operators given, for example, in \cite{Kirillov:1996qdi,Kirillov:1996aff,Lapointe:1998det,Dali:2024mac}.

\begin{remark}
For the $\mathcal{A}$-operator labeled by the empty partition, the sum in \eqref{eq:Aoperators} reduces to one term with all $\nu^{(\a)}=\vac$, and the derivative part of the operator becomes trivial. In this case, we obtain the simple expression
\be
 \mathcal{A}^{(r)}_\vac(u_r)
 = \exp\left(\sum_{k>0}\frac{u_r^k}{k}\frac{1-t^k}{1-q^k}p_k(\sp_\vac)
 \Big((1-q_3^k)\sum_{\a=1}^{r-1}p_k^{(\a)}+p_k^{(r)}\Big)\right).
\ee
\end{remark}

\subsection{Five-term relation}

While a five-term relation in higher level is expected on general grounds based on an analogous universal relation in the algebra \cite{Zenkevich:2021aus,Hu:2024pro}, it is worth giving an explicit proof of such an identity directly in the representation $\rho_\bsu^{(r,\bsn_0)}$ for the sake of completeness and to have explicit expressions in our conventions.
In this section, we follow the strategy of proof in \cite{Garsia:2018fiv,Romero:2025fiv}.

Define, for any $f\in\Lambda$, the diagonal operators $\D^\pm_\bsu[f]$ as\footnote{In the special case $f=p_1$, we have $\D^\pm_\bsu[p_1]=-\rho_\bsu^{(r,\bsn_0)}(x^\pm_0)/(1-t^{\pm1})$.}
\be
\begin{aligned}
 \D^+_\bsu[f] P_\blam(\bsx^\bullet|\bsu)
 &:= f\Big(\sum_{\a=1}^r u_\a\sp_{\lam^{(\a)}}\Big) P_\blam(\bsx^\bullet|\bsu) \\
 \D^-_\bsu[f] P_\blam(\bsx^\bullet|\bsu)
 &:= f\Big(\sum_{\a=1}^r u_\a^{-1}\sp_{\lam^{(\a)}}^{\vee}\Big) P_\blam(\bsx^\bullet|\bsu)
\end{aligned}
\ee
and the operator $\Delta(z|\bsu)$ as
\be
 \Delta(z|\bsu) P_\blam(\bsx^\bullet|\bsu)
 = \exp\left(-\sum_{k>0} \frac{z^k}{k} \sum_{\a=1}^r
 u_\a^k\, p_k(\chi_{\lam^{(\a)}})\right)
 P_\blam(\bsx^\bullet|\bsu)
\ee
where we used the plethysm $p_k(\chi_\lam)=\sum_{\sAbox\in\lam}\chi_\sAbox^k$.
We observe also that, using \eqref{expr_el} to write $\chi_\l=-\frac{t\sp_\l}{1-q}+\frac{1}{(1-q)(1-t^{-1})}$, we can equivalently write
\be
\label{eq:delta2epsilon}
 \Delta(z|\bsu)
 = \exp\left(\sum_{k>0}\frac{(zt)^k \D^+_\bsu[p_k]}{k(1-q^k)} \right)
 \prod_{\a=1}^r(u_\a z;q,t^{-1})_\infty
\ee
which shows that $\Delta(z|\bsu)$ is a special instance of the more general operator $\D^+_\bsu[f]$, for a specific choice of $f$.

\begin{proposition}
\label{prop:GHTcor}
For any $f\in\Lambda$, we have the relations
\be
 f\Big(\sum_\a \bsx^{(\a)}\Big)^\dagger\,\GHT(\bsu)
 = \GHT(\bsu)\,\D^+_\bsu[f]\,,
\ee
and
\be
 \GHT(\bsu)\,f\Big(\sum_\a\bsx^{(\a)}\Big)
 = \D^+_\bsu[f]\,\GHT(\bsu)\,,
\ee
where ${}^\dagger$ denotes the adjoint operator w.r.t.\ the inner product $\langle\,,\,\rangle_\mathrm{Z}$ in \eqref{eq:twisted-scalar-prod}, and $\GHT(\bsu)$ is the operator in \eqref{eq:GHTV(v)}.
\end{proposition}
\begin{proof}
The first relation follows from the first identity in \eqref{rel_Xi_r}, after the appropriate specialization of weights which reduces $\mathring{\GHT}[\mathring{\bsu},\mathring{\bsv},\mathring{\bsw}]t^{-L_0}$ to $\GHT(\bsu)$.
The second relation is obtained from the first by taking the adjoint and using that $\GHT(\bsu)^\dagger=\GHT(\bsu)$ as well as $\D^\pm_\bsu[f]^\dagger=\D^\pm_\bsu[f]$.
\end{proof}

The five-term relation (or pentagon identity), is an identity concerning the operator
\be
 \FTW(a,b|\bsu) := \sum_{i,j} a^i b^j \FTW_{i,j}(\bsu)
 =
 \Delta(b|\bsu)
 \exp\left(-\sum_{k>0}\frac{a^k\sum_\a p_k^{(\a)}}{k(1-q^k)}\right)
 \Delta(b|\bsu)^{-1}
 \exp\left(\sum_{k>0}\frac{a^k\sum_\a p_k^{(\a)}}{k(1-q^k)}\right),
\ee
and the main goal of this section is to show that the expansion of $\FTW(a,b|\bsu)$ in powers of $a,b$ is diagonal, i.e.\ that $\FTW(a,b|\bsu)=\sum_{i} (ab)^i \FTW_{i,i}(\bsu)$. Moreover, we will give an explicit formula for the diagonal elements $\FTW_{i,i}(\bsu)$ which can be resummed in closed form.

Following \cite{Garsia:2018fiv}, the strategy of proof is to separately show that the operators
\be
\label{eq:FTW1}
 \Delta(b|\bsu)
 \exp\left(-\sum_{k>0}\frac{a^k\sum_\a p_k^{(\a)}}{k(1-q^k)}\right)
 \Delta(b|\bsu)^{-1}
\ee
and
\be
\label{eq:FTW2}
 \exp\left(-\sum_{k>0}\frac{a^k\sum_\a p_k^{(\a)}}{k(1-q^k)}\right)
 \Delta(b|\bsu)^{-1}
 \exp\left(\sum_{k>0}\frac{a^k\sum_\a p_k^{(\a)}}{k(1-q^k)}\right)
\ee
are both triangular but in opposite directions. This in turn implies that $\FTW(a,b|\bsu)$ is diagonal. We start by analyzing the operator in \eqref{eq:FTW1} and expanding in powers of $a$ using the Cauchy identity \eqref{eq:PE},
\be
 \exp\left(-\sum_{k>0}\frac{a^k\sum_\a p_k^{(\a)}}{k(1-q^k)}\right)
 = \sum_{n=0}^\infty a^n s_{[n]}\left(-\frac{\sum_\a\bsx^{(\a)}}{1-q}\right),
\ee
where $s_{[n]}$ are the Schur symmetric functions, and in the r.h.s.\ we used the plethystic substitution
\be
 p_k\left(-\frac{\sum_\a\bsx^{(\a)}}{1-q}\right)
 =-\frac{\sum_\a p_k^{(\a)}}{1-q^k}.
\ee
\begin{proposition}
The operator
\be
\label{eq:operatorFT1}
 \Delta(b|\bsu)\,
 s_{[n]}\left(-\frac{\sum_\a \bsx^{(\a)}}{1-q}\right)
 \Delta(b|\bsu)^{-1}
\ee
is a polynomial in $b$ of degree $n$. Furthermore, the top degree in $b$ is given by
\be
 \oint\frac{\mathd b}{2\pi\mathi b} b^{-n}
 \Delta(b|\bsu)\,
 s_{[n]}\left(-\frac{\sum_\a \bsx^{(\a)}}{1-q}\right)
 \Delta(b|\bsu)^{-1}
 = \nabla(\bsu)\,
 (-1)^n s_{[n]}\left(-\frac{\sum_\a \bsx^{(\a)}}{1-q}\right)
 \nabla(\bsu)^{-1},
\ee
where the contour integral in $b$ is used to select the constant term in the Laurent expansion of the integrand. 
In particular, this means that $\FTW_{i,j}(\bsu) = 0$ for $i<j$, and for $i = j$, we have
\be
 \FTW_{i,i}(\bsu) = (-1)^i \nabla(\bsu)\,
 s_{[i]}\left(-\frac{\sum_\a \bsx^{(\a)}}{1-q}\right)
 \nabla(\bsu)^{-1}.
\ee
\end{proposition}
\begin{proof}
Let us introduce coefficients $d_{\bnu/\blam}$ defined as
\be
 s_{[n]}\left(-\frac{\sum_\a \bsx^{(\a)}}{1-q}\right)
 P_\blam(\bsx^\bullet|\bsu)
 = \sum_{\{\bnu:\blam\subset_n\bnu\}}
 d_{\bnu/\blam} P_\bnu(\bsx^\bullet|\bsu)
\ee
where $\blam\subset_n\bnu$ indicates that $\blam$ is contained in $\bnu$ and they differ by $n$ boxes. This is indeed the case since the operator in the l.h.s.\ can be realized as a non-commutative polynomial in the algebra generators $(\CT^\perp)^k(a_{-1})$. These elements act on the horizontal Fock module by raising the degree by one, and on the generalized Macdonald basis they all act via the generalized Pieri rules\footnote{The twist by $\CT^\perp$ changes the precise coefficient in the Pieri rule but it does not change the overall structure.}. This implies that the r.h.s.\ can be expanded as a linear combination of generalized Macdonald functions whose multi-partitions $\bnu$ can be obtained by adding $n$ boxes to $\bl$, one by one in all possible ways.

Then, we see that
\be
\begin{aligned}
 \Delta(b|\bsu)\,
 & s_{[n]}\left(-\frac{\sum_\a \bsx^{(\a)}}{1-q}\right)
 \Delta(b|\bsu)^{-1}
 P_\blam(\bsx^\bullet|\bsu) \\
 &= \sum_{\{\bnu:\blam\subset_n\bnu\}}
 d_{\bnu/\blam}
 \prod_{\sAbox\in\bnu/\blam}(1-b\, u_\sAbox \chi_\sAbox)
 P_\bnu(\bsx^\bullet|\bsu) \\
 &= \sum_{\{\bnu:\blam\subset_n\bnu\}}
 d_{\bnu/\blam}
 \left(1+\dots+(-b)^n \frac{\prod_{\sAbox\in\bnu}u_\sAbox\chi_\sAbox}{\prod_{\sAbox\in\blam}u_\sAbox\chi_\sAbox}\right)
 P_\bnu(\bsx^\bullet|\bsu) \\
 &=
 \left(s_{[n]}\left(-\frac{\sum_\a \bsx^{(\a)}}{1-q}\right)
 +\dots+
 (-b)^n \nabla(\bsu)\, s_{[n]}\left(-\frac{\sum_\a \bsx^{(\a)}}{1-q}\right) \nabla(\bsu)^{-1} \right)
 P_\blam(\bsx^\bullet|\bsu)
\end{aligned}
\ee
where we used that $\nabla(\bsu)$ is diagonal on the basis $P_\blam(\bsx^\bullet|\bsu)$ with eigenvalues as in \eqref{eq:def-nabla}.
This proves that the operator in \eqref{eq:operatorFT1} is polynomial of degree $n$ in $b$ and the last line gives an expression for the coefficient of $b^n$. Since the action from the right of the operator $\exp\left(\sum_{k>0}\frac{a^k\sum_\a p_k^{(\a)}}{k(1-q^k)}\right)$ can only increase the degree in $a$, we obtain that $\FTW_{i,j}(\bsu)=0$ if $i<j$. 
\end{proof}

We repeat the analysis for the operator in \eqref{eq:FTW2}.
\begin{proposition}
The operator
\be
\label{eq:operatorFT2}
 \oint\frac{\mathd b}{2\pi\mathi b} b^{-n}
 \exp\left(-\sum_{k>0}\frac{a^k\sum_\a p_k^{(\a)}}{k(1-q^k)}\right)
 \Delta(b|\bsu)^{-1}
 \exp\left(\sum_{k>0}\frac{a^k\sum_\a p_k^{(\a)}}{k(1-q^k)}\right)
\ee
is a polynomial in $a$ of degree $n$. In particular, this means that $\FTW_{i,j}(\bsu) = 0$ for $i>j$.
\end{proposition}
\begin{proof}
From Proposition~\ref{prop:GHTcor}, we obtain
\be
\begin{aligned}
 \GHT(\bsu)^{-1} & \GHT(\bsu)\,
 \exp\left(-\sum_{k>0}\frac{a^k\sum_\a p_k^{(\a)}}{k(1-q^k)}\right)
 \Delta(b|\bsu)^{-1}
 \exp\left(\sum_{k>0}\frac{a^k\sum_\a p_k^{(\a)}}{k(1-q^k)}\right) \\
 &= \GHT(\bsu)^{-1}
 \exp\left(-\sum_{k>0}\frac{a^k \D^+_\bsu[p_k]}{k(1-q^k)}\right)
 \exp\left(-\sum_{k>0}\frac{(bt)^k\sum_\a p_k^{(\a)}}{k(1-q^k)}\right)^\dagger
 \exp\left(\sum_{k>0}\frac{a^k \D^+_\bsu[p_k]}{k(1-q^k)}\right)
 \GHT(\bsu) \\
 &= \GHT(\bsu)^{-1}
 \Delta(at^{-1}|\bsu)^{-1}
 \exp\left(-\sum_{k>0}\frac{(bt)^k\sum_\a p_k^{(\a)}}{k(1-q^k)}\right)^\dagger
 \Delta(at^{-1}|\bsu)\,
 \GHT(\bsu)
\end{aligned}
\ee
where we used \eqref{eq:delta2epsilon}.
Let us introduce coefficients $c_{\bnu/\blam}$ as
\be
 s_{[n]}\left(-\frac{t\sum_\a \bsx^{(\a)}}{1-q}\right)^\dagger
 P_\bnu(\bsx^\bullet|\bsu)
 = \sum_{\{\blam:\blam\subset_n\bnu\}}
 c_{\bnu/\blam} P_\blam(\bsx^\bullet|\bsu)
\ee
from which, we similarly obtain that
\be
\begin{aligned}
 &\Delta(at^{-1}|\bsu)^{-1}
 s_{[n]}\left(-\frac{t\sum_\a \bsx^{(\a)}}{1-q}\right)^\dagger
 \Delta(at^{-1}|\bsu)
 P_\bnu(\bsx^\bullet|\bsu) \\
 &= \sum_{\{\blam:\blam\subset_n\bnu\}}
 c_{\bnu/\blam}
 \prod_{\{\sAbox\in\bnu/\blam\}}(1-at^{-1} u_\sAbox \chi_\sAbox)
 P_\blam(\bsx^\bullet|\bsu) \\
 &=
 \left(s_{[n]}\left(-\frac{t\sum_\a \bsx^{(\a)}}{1-q}\right)^\dagger
 +\dots+
 (-at^{-1})^n \nabla(\bsu)^{-1} s_{[n]}\left(-\frac{t\sum_\a \bsx^{(\a)}}{1-q}\right)^\dagger \nabla(\bsu) \right)
 P_\bnu(\bsx^\bullet|\bsu)
\end{aligned}
\ee
is a polynomial in $a$ of degree $n$. Since polynomiality in $a$ is preserved under conjugation by $\GHT(\bsu)$, we find that the operator in \eqref{eq:operatorFT2} is also polynomial in $a$ of degree $n$. Finally, this shows that $\FTW_{i,j}(\bsu)=0$ for $j<i$, since the left action of $\Delta(b|\bsu)$ only increases the $b$ degree.
\end{proof}
We are now ready to state the following.
\begin{theorem}
The higher level generalization of the five-term relation of \cite{Garsia:2018fiv} is
\begin{multline}\label{eq:FiveTerms}
 \nabla(\bsu)
 \exp\left(-\sum_{k>0}\frac{(-ab)^{k}}{k(1-q^k)}\sum_{\a=1}^r p_k^{(\a)}\right)
 \nabla(\bsu)^{-1} \\
 =
 \Delta(b|\bsu)
 \exp\left(-\sum_{k>0}\frac{a^k}{k(1-q^k)}\sum_{\a=1}^r p_k^{(\a)}\right)
 \Delta(b|\bsu)^{-1}
 \exp\left(\sum_{k>0}\frac{a^k}{k(1-q^k)}\sum_{\a=1}^r p_k^{(\a)}\right)
\end{multline}
as an identity of operators acting on $\CF^{(r,0)}_\bsu$.
\end{theorem}
\begin{proof}
Combining the previous two propositions we find
\be
 \FTW(a,b|\bsu) = \sum_i (ab)^i \FTW_{i,i}(\bsu)
 = \nabla(\bsu)
 \exp\left(-\sum_{k>0}\frac{(-ab)^{k}}{k(1-q^k)}\sum_{\a=1}^r p_k^{(\a)}\right)
 \nabla(\bsu)^{-1}
\ee
which proves the statement of the theorem.
\end{proof}
\begin{remark}
Taking the inverse of both sides of the five-term relation and acting on the vacuum, we obtain the identity
\be
\label{eq:FTinvVac}
 \Delta(z|\bsu)
 \exp\left(\sum_{k>0}\frac{\sum_{\a=1}^r p_k^{(\a)}}{k(1-q^k)}\right)
 =
 \exp\left(\sum_{k>0}\frac{\sum_{\a=1}^r p_k^{(\a)}}{k(1-q^k)}
 \left(1-z^k \Big(u_\a^k+(1-t^k q^{-k}) \sum_{\b=\a+1}^r u_\b^k\Big)\right)\right)
\ee
where in the r.h.s.\ we have used the conjectural identity \eqref{nabla_r} to rewrite explicitly the action of $\nabla(\bsu)$ on the plethystic exponential.
Using the Cauchy identity \eqref{eq:Cauchy-r} to expand both sides of \eqref{eq:FTinvVac}, we obtain the evaluation identity
\be
\label{eq:spec-id3}
 \frac{P_\bl(-(1-zu_1)\sp_\vac,-(q_3-zu_2)\sp_\vac,\dots,-(q_3^{r-1}-zu_r)\sp_\vac|\bsu)}
 {P_\bl(-\sp_\vac,-q_3\sp_\vac,\dots,-q_3^{r-1}\sp_\vac|\bsu)}
 = \prod_{\sAbox\in\blam} (1-zu_\sAbox\chi_\sAbox),
\ee
where we recall that $-\sp_\vac=\frac1{1-t}$. In the limit $z\to\infty$, we recover the evaluation \eqref{eq:spec-id2}.
\end{remark}

\subsection{Fourier/Hopf pairing}

As observed before, the operator $\GHT(\bsu)$ is self-adjoint w.r.t.\ the inner product $\langle\,,\,\rangle_\mathrm{Z}$. This fact implies that we can use $\GHT(\bsu)$ to ``\emph{twist}'' the inner product to give another, different inner product on the module $\CF^{(r,0)}_\bsu\cong\Lambda^{\otimes r}$.

\begin{definition}
For any $f,g\in\Lambda^{\otimes r}$, we define the bilinear pairing
\be
\label{eq:Fourier-r}
 (f,g)_\mathrm{F} := \langle f,\GHT(\bsu)\,g\rangle_\mathrm{Z}.
\ee
\end{definition}
The pairing in \eqref{eq:Fourier-r} is symmetric due to self-adjointedness of $\GHT(\bsu)$.
Moreover, due to Proposition~\ref{prop:GHTcor}, we have
\be
\begin{aligned}
 \left(f\Big(\sum_\a \bsx^{(\a)}\Big)\,g(\bsx^\bullet),h(\bsx^\bullet)\right)_\mathrm{F}
 &= \left\langle f\Big(\sum_\a \bsx^{(\a)}\Big)\,g(\bsx^\bullet),\GHT(\bsu)\,h(\bsx^\bullet)\right\rangle_\mathrm{Z} \\
 &= \left\langle g(\bsx^\bullet),f\Big(\sum_\a \bsx^{(\a)}\Big)^\dagger\,\GHT(\bsu)\,h(\bsx^\bullet)\right\rangle_\mathrm{Z} \\
 &= \left\langle g(\bsx^\bullet),\GHT(\bsu)\,\D^+_\bsu[f]\,h(\bsx^\bullet)\right\rangle_\mathrm{Z} \\
 &= \left( g(\bsx^\bullet),\D^+_\bsu[f]\,h(\bsx^\bullet)\right)_\mathrm{F},
\end{aligned}
\ee
for any $f\in\Lambda$ and $g,h\in\Lambda^{\otimes r}$.
Hence, we conclude that $\D^+_\bsu[f]$ is the adjoint operator to the function $f\Big(\sum_\a \bsx^{(\a)}\Big)$ w.r.t.\ the pairing $(\,,\,)_\mathrm{F}$.
This makes $(\,,\,)_\mathrm{F}$ into a type of Fourier/Hopf pairing as defined in \cite{Cherednik:1995mac,Okounkov:2001are,Beliakova:2021cyc}.

We notice however, that the previous identity only holds if $f$ is in the image of the ``diagonal'' embedding of $\Lambda$ into $\Lambda^{\otimes r}$ (i.e.\ $f\mapsto f(\sum_\a\bsx^{(\a)})\in\Lambda^{\otimes r}$), while it fails for more general elements of the module. This is related to the fact that, for $r>1$, the Whittaker vectors are no longer reproducing kernels for $\langle\,,\,\rangle_\mathrm{Z}$, so that the inner product with $W_\bl$ does not lead to an evaluation homomorphism, as in the case $r=1$.

We now construct an orthogonal basis for the Fourier/Hopf pairing.
Define the operator
\be
 \mathsf{G} := \nabla(\bsu)^{-1}\exp\left(-\sum_{k>0}\frac{(-1)^k}{k(1-t^k)}
 \sum_{\a=1}^r q_3^{(\a-1)k}\frac{\p}{\p p_k^{(\a)}} \right)\nabla(\bsu),
\ee
which can be regarded as an infinite series $\mathsf{G}=1+\dots$ of operators of increasingly negative degree.
Then, we find
\be
 \GHT(\bsu) = (\mathsf{G}^\dagger)^{-1}\nabla(\bsu)t^{-L_0}\nabla(\bsu)\mathsf{G}^{-1},
\ee
where $\nabla(\bsu)t^{-L_0}\nabla(\bsu)$ is diagonal on the GMP basis.
We define a new basis of inhomogeneous functions $P^\ast_\l(\bsx^\bullet|\bsu)$ as
\be
 P^\ast_\l(\bsx^\bullet|\bsu) := \mathsf{G}\,P_\l(\bsx^\bullet|\bsu).
\ee
From the definition \eqref{eq:Fourier-r} it then follows that
\be
\begin{aligned}
 \Big(P^\ast_{\l^{(r)},\dots,\l^{(1)}}(\bsx^\bullet|&u_r,\dots,u_1),
 P^\ast_{\mu^{(1)},\dots,\mu^{(r)}}(\bsx^\bullet|u_1,\dots,u_r)\Big)_\mathrm{F}\\
 &= \left\langle P_{\l^{(r)},\dots,\l^{(1)}}(\bsx^\bullet|u_r,\dots,u_1),
 \mathsf{G}^\dagger\,\GHT(\bsu)\, \mathsf{G}\,P_{\mu^{(1)},\dots,\mu^{(r)}}(\bsx^\bullet|u_1,\dots,u_r)\right\rangle_\mathrm{Z}\\
 &= \left\langle P_{\l^{(r)},\dots,\l^{(1)}}(\bsx^\bullet|u_r,\dots,u_1),
 \nabla(\bsu)t^{-L_0}\nabla(\bsu)\,P_{\mu^{(1)},\dots,\mu^{(r)}}(\bsx^\bullet|u_1,\dots,u_r)\right\rangle_\mathrm{Z}\\
 &= \prod_{\a=1}^r t^{-|\l^{(\a)}|} u_\a^{2|\l^{(\a)}|} g_{\l^{(\a)}}^2
 b_{\l^{(\a)}} \delta_{\l^{(\a)},\mu^{(\a)}},
\end{aligned}
\ee
which shows that $P^\ast_\blam$ are an orthogonal basis for the Fourier/Hopf pairing.

For $r=1$, the $P^\ast_\l$ coincide with interpolation Macdonald functions \cite{Knop:1996sym,Sahi:1996int,Sahi:1996dif,Okounkov:1997shi}, as it was shown in \cite{Garsia2001} using the GHT identity or in \cite[Proposition~2]{Carlsson:2013jka} using the Macdonald--Mehta--Cherednik identities.
For general $r$, we have the following vanishing property. First, let us define generalized skew-Macdonald functions $P_{\blam/\bnu}$ as
\be
 P_\blam(\bsx^\bullet+zq_3^{\bullet-1}|\bsu)
 = \sum_{\bnu\subseteq\blam} P_\bnu(\bsx^\bullet|\bsu)
 P_{\blam/\bnu}(zq_3^{\bullet-1}|\bsu),
\ee
where $\bsx^\bullet+zq_3^{\bullet-1}$ denotes the shift $\bsx^{(\a)}\mapsto\bsx^{(\a)}+zq_3^{\a-1}$ which can be obtained as the action of a plethystic exponential in the operators $\rho_\bsu^{(r,\bsn_0)}(a_k)$, $k>0$.
Then we can compute the pairing
\be
\begin{aligned}
 &\Big(P_{\l^{(1)},\dots,\l^{(r)}}(\bsx^\bullet|u_1,\dots,u_r),
 P^\ast_{\mu^{(r)},\dots,\mu^{(1)}}(\bsx^\bullet|u_r,\dots,u_1)\Big)_\mathrm{F}\\
 &=\Big\langle P_{\l^{(1)},\dots,\l^{(r)}}(\bsx^\bullet|u_1,\dots,u_r),\GHT(\bsu)\,
 P^\ast_{\mu^{(r)},\dots,\mu^{(1)}}(\bsx^\bullet|u_r,\dots,u_1)\Big\rangle_\mathrm{Z}\\
 &= \Big(\prod_{\a=1}^r u_\a^{|\mu^{(\a)}|}g_{\mu^{(\a)}}\Big)^2 t^{-|\bmu|}
 \Big\langle \mathsf{G}^{-1}P_{\l^{(1)},\dots,\l^{(r)}}(\bsx^\bullet|u_1,\dots,u_r),
 P_{\mu^{(r)},\dots,\mu^{(1)}}(\bsx^\bullet|u_r,\dots,u_1)\Big\rangle_\mathrm{Z}\\
 &= \Big(\prod_{\a=1}^r u_\a^{|\mu^{(\a)}|}g_{\mu^{(\a)}}\Big)^2 t^{-|\bmu|}
 \sum_{\bnu\subseteq\blam}\Big(\prod_{\a=1}^r\frac{(-u_\a)^{|\l^{(\a)}|}g_{\l^{(\a)}}}
 {(-u_\a)^{|\nu^{(\a)}|}g_{\nu^{(\a)}}}\Big)
 P_{\blam/\bnu}(-\sp_\vac q_3^{\bullet-1}|\bsu)\\
 &\hspace{150pt}\times\Big\langle P_{\nu^{(1)},\dots,\nu^{(r)}}(\bsx^\bullet|u_1,\dots,u_r),
 P^\ast_{\mu^{(r)},\dots,\mu^{(1)}}(\bsx^\bullet|u_r,\dots,u_1)\Big\rangle_\mathrm{Z}\\
 &= \Big(\prod_{\a=1}^r (-t^{-1}u_\a)^{|\mu^{(\a)}|}g_{\mu^{(\a)}}b_{\mu^{(\a)}}^{-1}
 (-u_\a)^{|\l^{(\a)}|}g_{\l^{(\a)}}\Big)
 P_{\blam/\bmu}(-\sp_\vac q_3^{\bullet-1}|\bsu),
\end{aligned}
\ee
where $P_{\blam/\bmu}$ is identically zero if $\bmu$ is not contained in $\blam$.
In the case $r=1$, the Fourier/Hopf pairing with $\tilde{P}_\l$ defines an evaluation at the locus $\sp_\l$, hence the previous formula says that $P^\ast_\mu(\sp_\l)=0$ if $\mu\not\subseteq\l$, which is the vanishing property of the interpolation Macdonald functions.
For $r>1$, however, the map $f\mapsto(\tilde{P}_\blam,f)_\mathrm{F}$ is not an evaluation in general\footnote{For generic $\bl$, it is a sum of different evaluations.}, hence the previous vanishing does not define an interpolation locus.

\section{Discussion}
In this paper, we proposed an extension to generalized Macdonald functions for several identities known in the case of usual Macdonald symmetric functions. It includes the GHT formula \ref{eq:GHTr}, the five-term relation \ref{eq:FiveTerms},  and the Fourier/Hopf pairing. This extension relies on new algebraic objects, like framing operators and higher vertex operators, that can be defined for representation of any levels $(r,0)$ with $r\in\mZ^{>0}$. The proof of these identities rely on the algebraic property of these objects, but also on the Conjecture ~\ref{conj:BH} in the case of the GHT identity. We have not been able to find an algebraic interpretation of the latter that could lead to a complete proof. Yet, it may be possible to prove this conjecture by studying the action of the operators $y_{-k}^-=\CS(x_{-k}^-)$ for $k>1$ on both sides of the equality. Combined with the already known action of $a_k$ for $k>1$, the corresponding algebraic transformations would uniquely characterize this quantity. We hope to be able to come back to this problem in a future work.

It is instructive to draw a parallel between our work and recent studies of wreath Macdonald polynomials \cite{Romero:2025tes,Romero:2025fiv} (see also the excellent survey \cite{Orr:2023wre}). These wreath Macdonald polynomials are associated to the higher rank deformation of the algebra, namely quantum toroidal $\sl(n)$, rather than higher level representations of quantum toroidal $\gl(1)$. Yet, similar results were found in this case, e.g.\ Pieri rules \cite{Romero:2025fiv}, a factorization of the kernel, and a GHT identity \cite{Romero:2025tes}. In fact, in this case the logic seems reversed since Pieri rules are obtained as a consequence of the GHT identity, while we used the Pieri rule to obtain this identity here. This suggests that it might be possible to derive the GHT identity from first principle in the case of GMP as well.
    
It would also be interesting to develop further the vertex operator approach to wreath Macdonald polynomials using the vertex operators associated to the quantum toroidal $\sl(n)$ algebra \cite{Awata:2017lqa}. This approach should also extend to the generalized framework developed in \cite{Bourgine:2019phm} based on $(\nu_1,\nu_2)$ cyclic deformations of the quantum toroidal $\sl(n)$ algebra. The case of the surface defect $(\nu_1,\nu_2)=(1,0)$ should be of particular interest thanks to its deep connection to the Bethe/gauge correspondence \cite{Jeong:2021rll}. Finally, the case of supersymmetric Macdonald polynomials \cite{Sergeev2009}, with their connection to corner VOA \cite{Cheewaphutthisakun:2025zoc}, would also be interesting to study.

\paragraph{Acknowledgments.} J.E.B.\ thanks Misha Bershtein, Fabrizio Nieri, Elli Pomoni, Junichi Shiraishi and Yegor Zenkevich for various discussions and many useful comments.
L.C.\ thanks Victor Mishnyakov and Yegor Zenkevich for illuminating discussions.
The work of L.C.\ was supported by the ARC Discovery Grant DP210103081.
J.E.B.\ and L.C.\ gratefully acknowledge support during the MATRIX Program``New Deformations of Quantum Field and Gravity Theories'', (Creswick AU, 22 Jan – 2 Feb 2024), where part of this work was conducted.
A.S.\ gratefully acknowledges the support provided by the Australian Government Research Training Program Scholarship and the Rowden White Scholarship.

\appendix
\section{Useful formulas and technical results}\label{app_A}
\subsection{Nekrasov factor}\label{app_A1}
The Nekrasov factor can be defined by the formula \ref{def_nek}, or, equivalently, assuming $z$ generic,
\be
\label{def_nek}
\begin{aligned}
 N_{\l,\mu}(z) &= \prod_{\superp{\sAbox\in\l}{\sAboxB\in\mu}}S(z\chi_\sAbox/\chi_\sAboxB)
 \times\prod_{\sAbox\in\l}(1-q_3^{-1}z\chi_\sAbox)\times
 \prod_{\sAboxB\in\mu}(1-z\chi_\sAboxB^{-1}).
\end{aligned}
\ee
It satisfies the ``reflection'' property
\be
\label{Nek_refl}
 N_{\l,\mu}(z) = N_{\mu,\l}(q_3z^{-1})\times
 \prod_{\sAbox\in\l}(-q_3^{-1}z\chi_\sAbox)\times
 \prod_{\sAboxB\in\mu}(-z\chi_\sAboxB^{-1}).
\ee
Alternatively, it can also be expressed using the functions $\CY_\l(z)$ defined in \eqref{def_CYY},
\begin{equation}\label{Nek_CYY}
N_{\l,\mu}(z)=\prod_{\sAbox\in\l}(1-zq_3^{-1}\chi_\sAbox)\times\prod_{\sAboxB\in\mu}\CY_\l(z^{-1}\chi_\sAboxB)=\prod_{\sAboxB\in\mu}(1-z\chi_\sAboxB^{-1})\times\prod_{\sAbox\in\l}(-zq_3^{-1}\chi_\sAbox)\CY_\mu(q_3^{-1}z\chi_\sAbox).
\end{equation} 
We deduce from this expression the variation formulas
\begin{align}\label{rec_nek}
\begin{split}
&\dfrac{N_{\l\pm \sAbox,\mu}(z)}{N_{\l,\mu}(z)}=(-q_3^{-1}z\chi_\sAbox)^{\pm1}\CY_\mu(q_3^{-1}z\chi_\sAbox)^{\pm1},
\hspace{30pt}
 \dfrac{N_{\l,\mu\pm \sAboxB}(z)}{N_{\l,\mu}(z)}=\CY_\l(z^{-1}\chi_\sAboxB)^{\pm1}.
\end{split}
\end{align}
The following specialization to $\l=\vac$ or $\mu=\vac$ are also useful,
\begin{equation}\label{N_vac}
 N_{\l,\vac}(z)=\prod_{\sAbox\in\l}(1-zq_3^{-1}\chi_\sAbox),
 \hspace{30pt}
 N_{\vac,\mu}(z)=\prod_{\sAboxB\in\mu}(1-z\chi_\sAboxB^{-1}).
\end{equation} 

It is convenient to absorb a monomial factor in the definition of the Nekrasov factor in order to make it more symmetric under the exchange of the two partitions,
\begin{equation}\label{def_tN}
\tilde{N}_{\l,\mu}(z)=N_{\l,\mu}(z)\prod_{\sAbox\in\l}(-q_3^{-1}z\chi_\sAbox)^{-1}\implies \tN_{\l,\mu}(z)=\tN_{\mu,\l}(q_3z^{-1}).
\end{equation} 
As a result, the variation formulas under addition of boxes simplify as well,
\begin{equation}\label{prop_tN}
 \dfrac{\tilde{N}_{\l\pm \sAbox,\mu}(z)}{\tilde{N}_{\l,\mu}(z)}
 =\CY_\mu(q_3^{-1}z\chi_\sAbox)^{\pm1},
 \hspace{30pt}
 \dfrac{\tilde{N}_{\l,\mu\pm \sAboxB}(z)}{\tilde{N}_{\l,\mu}(z)}
 =\CY_\l(z^{-1}\chi_\sAboxB)^{\pm1}.
\end{equation}

In this paper, we will need the following properties that apply to both $N_{\l,\mu}(z)$ and $\tN_{\l,\mu}(z)$.
\begin{lemma}
We have
\begin{align}
\begin{split}\label{prop_N_sp}
 N_{\l,\mu}(1)&=0,\quad\text{if}\quad\mu\not\subseteq\l,\\
 N_{\l,\mu}(q_3)&=0,\quad\text{if}\quad\l\not\subseteq\mu.
\end{split}
\end{align}
\end{lemma}
\begin{proof}
These properties can be obtained using the formula \eqref{Nek_CYY} expressing the Nekrasov factor in terms of $\CY$-functions, and the expression \eqref{expr_CYY} for these functions,
\be
\label{eq:neklemma}
\begin{aligned}
&N_{\l,\mu}(z)=\prod_{\sAbox\in\l}(1-zq_3^{-1}\chi_\sAbox)\times\prod_{\sAboxB\in\mu}\dfrac{\prod_{\sAbox\in A(\l)}(1-z\chi_\sAbox/\chi_\sAboxB)}{\prod_{\sAbox\in R(\l)}(1-zq_3^{-1}\chi_{\sAbox}/\chi_{\sAboxB})},\\
&N_{\l,\mu}(z)=\prod_{\sAboxB\in\mu}(1-z\chi_\sAboxB^{-1})\times\prod_{\sAbox\in\l}(-zq_3^{-1}\chi_\sAbox)\dfrac{\prod_{\sAboxB\in A(\mu)}(1-z^{-1}q_3\chi_\sAboxB/\chi_\sAbox)}{\prod_{\sAboxB\in R(\mu)}(1-z^{-1}\chi_\sAboxB/\chi_\sAbox)}.
\end{aligned}
\ee
Examining the first expression at $z=1$, and noticing that for any box of coordinate $q_3^{-1}\chi_\sAbox$ for $\Abox\in R(\l)$, there exists a box in $A(\l)$ in the same column with a lower row index, we deduce that the Nekrasov factor vanishes if the Young diagram $\mu$ is not included in $\l$. A similar argument applies to the second expression specialized to $z=q_3$.
\end{proof}

\paragraph{Normalization factors.}
From the expression of $b_\l$ in \eqref{eq:def-Macdonald-inner}, we derive the variation formula for a box of coordinates $(i,j)$,
\begin{equation}
\dfrac{b_{\l+\sAbox}}{b_\l}=\dfrac{1-t}{1-q}\prod_{j'=1}^{j-1}\dfrac{(1-q^{j-j'}t^{\l^T_{j'}-i+1})(1-q^{j-j'}t^{\l^T_{j'}-i})}{(1-q^{j-j'-1}t^{\l^T_{j'}-i+1})(1-q^{j-j'+1}t^{\l^T_{j'}-i})}\times\prod_{i'=1}^{i-1}\dfrac{(1-q^{\l_{i'}-j}t^{i-i'+1})(1-q^{\l_{i'}-j+1}t^{i-i'-1})}{(1-q^{\l_{i'}-j}t^{i-i'})(1-q^{\l_{i'}-j+1}t^{i-i'})}
\end{equation}
Note that $\psi_{\l+\sAbox}(\Abox)=\psi_\l(\Abox)$, and $\psi^\ast_{\l-\sAbox}(\Abox)=\psi_\l^\ast(\Abox)$, and
\begin{equation}\label{var_b_l}
\dfrac{b_{\l+\sAbox}}{b_\l}=\dfrac{1-t}{1-q}\dfrac{\psi_{\l}(\Abox)}{\psi_\l^\ast(\Abox)},
\hspace{30pt}
\dfrac{b_{\l-\sAbox}}{b_\l}=\dfrac{1-q}{1-t}\dfrac{\psi_{\l}^\ast(\Abox)}{\psi_\l(\Abox)}.
\end{equation}
Using the expression of the kernel \eqref{Mac_kernel} and the Pieri rule \eqref{PieriPsi}, we show that
\begin{align}
\begin{split}
\dfrac{\p}{\p p_1(\bsx)}\Pi(\bsx|\bsy)&=\dfrac{1-t}{1-q}p_1(\bsy)\Pi(\bsx|\bsy)=\dfrac{1-t}{1-q}\sum_\l b_\l P_\l(\bsx)p_1(\bsy)P_\l(\bsy)\\
=&\dfrac{1-t}{1-q}\sum_\l b_\l P_\l(\bsx)\sum_{\sAbox\in A(\l)}\psi_\l(\Abox)P_{\l+\sAbox}(\bsy)\\
=&\dfrac{1-t}{1-q}\sum_\l\sum_{\sAbox\in R(\l)} b_{\l-\sAbox}\psi_{\l}(\Abox) P_{\l-\sAbox}(\bsx)P_{\l}(\bsy)\\
=&\sum_\l\sum_{\sAbox\in R(\l)} b_{\l}\psi_{\l}^\ast(\Abox) P_{\l-\sAbox}(\bsx)P_{\l}(\bsy).
\end{split}
\end{align}
Projecting onto $P_\l(\bsy)$, we deduce the rule \eqref{dual_Pieri}.

\paragraph{Normalization of the vertical basis.} To compare vertical basis vectors in different normalizations, we introduce the following quantities,
\begin{equation}\label{def_gl}
G_{\bl}(\bsv):=\g^{\sum_{\a=1}^m (\a-1)|\l^{(\a)}|}\prod_{\superp{\a,\b=1}{\a>\b}}^m\tilde{N}_{\l^{(\a)},\l^{(\b)}}(v_\a/v_\b)^{-1},\quad G_{\bl}^\ast(\bsv):=\g^{\sum_{\a=1}^m(m-\a)|\l^{(\a)}|}\prod_{\superp{\a,\b=1}{\a<\b}}^m\tilde{N}_{\l^{(\a)},\l^{(\b)}}(v_\a/v_\b)^{-1}.
\end{equation}
They satisfy the variation formulas for $\Abox\in\l^{(\a)}$,
\begin{align}\label{var_gl}
\begin{split}
&\dfrac{G_{\bl\pm\sAbox}(\bsv)}{G_\bl(\bsv)}=\g^{\pm(\a-1)}\prod_{\b<\a}\CY_{\l^{(\b)}}(q_3^{-1}v_\a\chi_\sAbox/v_\b)^{\mp1}\times\prod_{\b>\a}\CY_{\l^{(\b)}}(v_\a\chi_\sAbox/v_\b)^{\mp1},\\
&\dfrac{G^\ast_{\bl\pm\sAbox}(\bsv)}{G^\ast_\bl(\bsv)}=\g^{\pm(m-\a)}\prod_{\b>\a}\CY_{\l^{(\b)}}(q_3^{-1}v_\a\chi_\sAbox/v_\b)^{\mp1}\times\prod_{\b<\a}\CY_{\l^{(\b)}}(v_\a\chi_\sAbox/v_\b)^{\mp1}.
\end{split}
\end{align}
In particular, we have
\begin{equation}
\left(\dfrac{G_{\bl+\sAbox}(\bsv)}{G_\bl(\bsv)}\right)\left(\dfrac{G^\ast_{\bl+\sAbox}(\bsv)}{G^\ast_\bl(\bsv)}\right)^{-1}=\g^{2\a-m-1}\prod_{\b<\a}\Psi_{\l^{(\b)}}(v_\a\chi_\sAbox/v_\b)^{-1}\times\prod_{\b>\a}\Psi_{\l^{(\b)}}(v_\a\chi_\sAbox/v_\b).
\end{equation}

Moreover, we note that the matrix elements \eqref{def_rl} of the vertical action on the basis $\dket{\bl}$ take the following form for a box $\Abox\in A(\l^{(\a)})$ (resp. $\Abox\in R(\l^{(\a)})$),
\begin{equation}
r_\bl(\Abox|\bsv)=r_{\l^{(\a)}}(\Abox)\prod_{\superp{\b=1}{\b\neq\a}}^m \CY_{\l^{(\b)}}(v_\a\chi_{\sAbox}/v_\b)^{-1},\quad r_\bl^\ast(\Abox|\bsv)=r^\ast_{\l^{(\a)}}(\Abox)\prod_{\superp{\b=1}{\b\neq\a}}^m \CY_{\l^{(\b)}}(q_3^{-1}v_\a\chi_{\sAbox}/v_\b).
\end{equation}
As a result, we find
\begin{align}
\begin{split}\label{prop_G_r}
\left(\dfrac{G^\ast_{\bl+\sAbox}(\bsv)}{G^\ast_\bl(\bsv)}\right)^{-1}r_\bl(\Abox|\bsv)&=\g^{\a-m}r_{\l^{(\a)}}(\Abox)\prod_{\b=\a+1}^m \Psi_{\l^{(\b)}}(v_\a\chi_{\sAbox}/v_\b),\\
\left(\dfrac{G^\ast_{\bl-\sAbox}(\bsv)}{G^\ast_\bl(\bsv)}\right)^{-1}r_\bl^\ast(\Abox|\bsv)&=\g^{m-\a}r_{\l^{(\a)}}^\ast(\Abox)\prod_{\b=1}^{\a-1} \Psi_{\l^{(\b)}}(v_\a\chi_{\sAbox}/v_\b).
\end{split}
\end{align}

\subsection{Action of \texorpdfstring{$x_k^\pm$}{xk+-} on Macdonald functions}\label{AppA2}
In this appendix, we derive explicit expressions for the action of the generators $x_k^\pm$ on Macdonald functions $P_\l(\bsx)$ in the representation $\rho_1^{(1,0)}$. This action involves the addition/removal of sets of $|k|$ boxes, and we first need to introduce some notation. We denote the corresponding nested summations,
\begin{align}
\begin{split}
&\sum_{\{x_1,\dots,x_k\}\in A(\l)}=\sum_{x_1\in A(\l)}\sum_{x_2\in A(\l+x_1)}\cdots \sum_{x_k\in A(\l+x_1+\dots+x_{k-1})},\\
&\sum_{\{x_1,\dots,x_k\}\in R(\l)}=\sum_{x_1\in R(\l)}\sum_{x_2\in R(\l-x_1)}\cdots \sum_{x_k\in R(\l-x_1-\dots-x_{k-1})},
\end{split}
\end{align}
and the product of matrix elements involved in the action of $a_{\pm1}$ as
\begin{align}
\begin{split}
&\psi_\l^\ast(x_1,\dots,x_k)=\psi_\l(x_1)\psi_{\l+x_1}(x_2)\cdots \psi_{\l+x_1+\dots+x_{k-1}}(x_k),\quad \{x_1,\dots,x_k\}\in A(\l),\\
&\psi_\l(x_1,\dots,x_k)=\psi_\l^\ast(x_1)\psi_{\l-x_1}^\ast(x_2)\cdots \psi_{\l-x_1-\dots-x_{k-1}}^\ast(x_k),\quad \{x_1,\dots,x_k\}\in R(\l)
\end{split}
\end{align}

The matrix elements can be determined inductively using the following relations for $k>0$,
\begin{equation}
x_k^\pm=(\pm1)^k \g^{\pm kc/2}c_1^{-k}\left(\ad_{a_1}\right)^k x_0^\pm,\quad x_{-k}^\pm=(\pm1)^k \g^{\pm kc/2}c_1^{-k}\left(\ad_{a_{-1}}\right)^k x_0^\pm.
\end{equation}
It is easy to show that the following expressions solve the induction,
\begin{align}
\begin{split}
x_k^\pm P_\l(\bsx)&=(\mp1)^k\dfrac{\g^{(-1\pm1/2)k+1\mp1}}{(1-q_1)^{k-1}(1-q_2)^{-1}}\sum_{\{x_1,\dots,x_k\}\in R(\l)}  E_k(\chi_{x_1}^{\pm1},\dots,\chi_{x_k}^{\pm1})\psi_\l^\ast(x_1,\dots,x_k)\ P_{\l-x_1-\dots-x_{k}}(\bsx),\\
x_{-k}^\pm P_\l(\bsx)&=(\mp1)^k\dfrac{\g^{(-1\pm1)\frac{k}{2}+1\mp1}}{(1-q_1)^{-1}(1-q_2)^{k-1}}\sum_{\{x_1,\dots,x_k\}\in A(\l)}E_k(\chi_{x_1}^{\pm1},\dots,\chi_{x_k}^{\pm1})\psi_\l(x_1,\dots,x_k)P_{\l+x_1+\dots+x_k}(\bsx),
\end{split}
\end{align}
provided that the family of polynomial $E_k(\chi_{x_1},\dots,\chi_{x_k})$ satisfy the recursion relations
\begin{equation}\label{rec_Ek}
E_{k+1}(\chi_1,\dots,\chi_{k+1})=E_{k}(\chi_1,\dots,\chi_k)-E_{k}(\chi_2,\dots,\chi_{k+1}),
\end{equation}
with $E_1(\chi_1)=\chi_1$. This is solved explicitly by
\begin{equation}
E_k(\chi_1,\dots,\chi_k)=\sum_{l=1}^k (-1)^{l-1}\left(\superp{k-1}{l-1}\right)\chi_l.
\end{equation} 

\paragraph{Matrix elements of $X$.} The previous results can be used to deduce the action of $\psi_{-k}^-=-\k^{-1}\g^{-k/2}[x_{-k}^+,x_0^-]$,
\begin{equation}
\psi_{-k}^-P_\l(\bsx)=(-)^{k-1}\dfrac{(1-q_3)\g^{-k/2}}{(1-q_1)^{-1}(1-q_2)^{k-1}}\equskip\sum_{\{x_1,\dots,x_k\}\in A(\l)}\equskip \psi_\l(x_1,\dots,x_k)E_k(\chi_{x_1},\dots,\chi_{x_k})\left(\sum_{i=1}^k\chi_{x_i}^{-1}\right)P_{\l+x_1+\dots+x_k}(\bsx).
\end{equation} 
Combining these results, we find that the matrix element $X_{\bl,\bmu}$ is nonzero only if there exists two sequences of $k$ boxes $\{x_1,\dots,x_k\}\in A(\l^{(1)})$ and $\{y_1,\dots,y_k\}\in R(\l^{(2)})$ such that $\mu^{(1)}=\l^{(1)}+x_1+\dots+x_k$ and $\mu^{(2)}=\l^{(2)}-y_1+\dots-y_k$. The corresponding matrix element is\footnote{An extra factor $\g^k$ needs to be included here since we twisted the isomorphism between the free boson Fock space and the ring of symmetric functions as in \eqref{eq:iso-J-to-p}.}
\be
\label{expr_X}
\begin{aligned}
X_{\bl,\bmu}&=-\dfrac{(1-q_3)}{(1-q_1)^{k-2}(1-q_2)^{k-2}}\psi_{\l^{(1)}}(x_1,\dots,x_k)\psi_{\l^{(2)}}^\ast(y_1,\dots,y_k)\\
&\times E_k(\chi_{x_1},\dots,\chi_{x_k})E_k(\chi_{y_1},\dots,\chi_{y_k})\left(\sum_{i=1}^k\chi_{x_i}^{-k}\right).
\end{aligned}
\ee

\subsection{Twisted coproduct \texorpdfstring{$\D_H$}{DeltaH}}\label{app_D_H}
The twisted coproduct acts on the Drinfeld currents $x^\pm(z)$ as
\begin{align}
\begin{split}
&\D_H(x^+(z))=x^+(z)\otimes 1+\psi^+(\hg_{(1)}^{-1/2}z)\otimes x^+(\hg_{(1)}^{-1}z),\\
&\D_H(x^-(z))=x^-(\hg_{(2)}^{-1} z)\otimes \psi^-(\hg_{(2)}^{-1/2}z)+1\otimes x^-(z),
\end{split}
\end{align}
while the coproduct of central, grading, and Cartan elements coincide with the action of $\D$. We recall the expression of the Cartan factor $\CK$ of the universal R-matrix,
\begin{equation}
\CK=\g^{\d}\bCK\g^{\d},\quad \bCK=\g^{2\bar\d}\exp\left(-\sum_{k>0}\dfrac1{c_k} a_k\otimes a_{-k}\right),
\end{equation} 
with
\begin{equation}
\d=\hf(c\otimes d+d\otimes c)\quad \bar\d=\hf(\bc\otimes\bd+\bd\otimes\bc).
\end{equation} 

First, we focus on the adjoint action of $\d$. For any two Drinfeld currents $a(z)$ and $b(z)$, we have
\begin{equation}\label{act_d}
\g^\d a(z)\otimes b(z)\g^{-\d}=a(\hg_{(2)}^{1/2}z)\otimes b(\hg_{(1)}^{1/2}z).
\end{equation} 
Next, we turn to $\bar\d$. The role of this operator is to generate the zero modes of the Cartan currents $\psi^\pm(z)$ when commuted with the operators $x^\pm(z)$, for instance:
\begin{equation}
\g^{\mp2\bar\d}(x^+(z)\otimes1)\g^{\pm2\bar\d}=x^+(z)\otimes(\psi_0^-)^{\mp1},\quad \g^{\mp2\bar\d}(1\otimes x^+(z))\g^{\pm2\bar\d}=(\psi_0^+)^{\pm1}\otimes x^+(z).
\end{equation} 

It remains to treat the operator $\bCK$, after some tedious but straightforward computation, we find
\begin{align}
\begin{split}
&\bCK^{\mp1}(\psi^+(z)\otimes1)\bCK^{\pm1}=\psi^+(z)\otimes 1,\quad \bCK^{\mp1}(1\otimes \psi^+(z))\bCK^{\pm1}=\psi^+(\hg_{(2)}z)^{\pm1}\psi^+(\hg_{(2)}^{-1}z)^{\mp1}\otimes \psi^+(z)\\
&\bCK^{\mp1}(\psi^-(z)\otimes1)\bCK^{\pm1}=\psi^-(z)\otimes\psi^-(\hg_{(1)}z)^{\pm1}\psi^-(\hg_{(1)}^{-1}z)^{\mp1},\quad \bCK^{\mp1}(1\otimes \psi^-(z))\bCK^{\pm1}=1\otimes\psi^-(z)\\
&\bCK^{\mp1}(x^+(z)\otimes1)\bCK^{\pm1}=x^+(z)\otimes\psi^-(\hg_{(1)}^{-1/2}z)^{\mp1},\quad \bCK^{\mp1}(1\otimes x^+(z))\bCK^{\pm1}=\psi^+(\hg_{(2)}^{1/2}z)^{\pm1}\otimes x^+(z)\\
&\bCK^{\mp1}(x^-(z)\otimes1)\bCK^{\pm1}=x^-(z)\otimes\psi^-(\hg_{(1)}^{1/2}z)^{\pm1},\quad \bCK^{\mp1}(1\otimes x^-(z))\bCK^{\pm1}=\psi^+(\hg_{(2)}^{-1/2}z)^{\mp1}\otimes x^-(z)\\
\end{split}
\end{align}
Combining these expressions with the action \eqref{act_d} of $\g^\d$, we have
\begin{align}\label{act_K_inv}
\begin{split}
&\CK(\psi^+(z)\otimes1)\CK^{-1}=\psi^+(\g_{(2)}z)\otimes 1,\quad \CK(1\otimes \psi^+(z))\CK^{-1}=\psi^+(\g_{(1)}^{1/2}\hg_{(2)}^{3/2}z)^{-1}\psi^+(\g_{(1)}^{1/2}\hg_{(2)}^{-1/2}z)\otimes \psi^+(\hg_{(1)}z)\\
&\CK(\psi^-(z)\otimes1)\CK^{-1}=\psi^-(\hg_{(2)}z)\otimes\psi^-(\hg_{(1)}^{-1/2}\g_{(2)}^{1/2}z)\psi^-(\hg_{(1)}^{3/2}\g_{(2)}^{1/2}z)^{-1},\quad \CK(1\otimes \psi^-(z))\CK^{-1}=1\otimes\psi^-(\g_{(1)}z)\\
&\CK(x^+(z)\otimes1)\CK^{-1}=x^+(\g_{(2)}z)\otimes\psi^-(\hg_{(2)}^{1/2}z),\quad \CK(1\otimes x^+(z))\CK^{-1}=\psi^+(\hg_{(1)}^{1/2}\hg_{(2)}z)\otimes x^+(\g_{(1)}z)\\
&\CK(x^-(z)\otimes1)\CK^{-1}=x^-(\g_{(2)}z)\otimes\psi^-(\hg_{(1)}\hg_{(2)}^{1/2}z)^{-1},\quad \CK(1\otimes x^-(z))\CK^{-1}=\psi^+(\hg_{(1)}^{1/2}z)\otimes x^-(\hg_{(1)}z)\\
\end{split}
\end{align}
In the same way, it is possible to obtain
\begin{align}\label{act_K}
\begin{split}
&\CK^{-1}(\psi^+(z)\otimes1)\CK=\psi^+(\g_{(2)}^{-1}z)\otimes 1,\quad \CK^{-1}(1\otimes \psi^+(z))\CK=\psi^+(\g_{(1)}^{-1/2}\hg_{(2)}^{1/2}z)\psi^+(\g_{(1)}^{-1/2}\hg_{(2)}^{-3/2}z)^{-1}\otimes \psi^+(\hg_{(1)}^{-1}z)\\
&\CK^{-1}(\psi^-(z)\otimes1)\CK=\psi^-(\hg_{(2)}^{-1}z)\otimes\psi^-(\hg_{(1)}^{1/2}\g_{(2)}^{-1/2}z)\psi^-(\hg_{(1)}^{-3/2}\g_{(2)}^{-1/2}z)^{-1},\quad \CK^{-1}(1\otimes \psi^-(z))\CK=1\otimes\psi^-(\g_{(1)}^{-1}z)\\
&\CK^{-1}(x^+(z)\otimes1)\CK=x^+(\g_{(2)}^{-1}z)\otimes\psi^-(\hg_{(1)}^{-1}\hg_{(2)}^{-1/2}z)^{-1},\quad \CK^{-1}(1\otimes x^+(z))\CK=\psi^+(\hg_{(1)}^{-1/2}z)\otimes x^+(\g_{(1)}^{-1}z)\\
&\CK^{-1}(x^-(z)\otimes1)\CK=x^-(\g_{(2)}^{-1}z)\otimes\psi^-(\hg_{(2)}^{-1/2}z),\quad \CK^{-1}(1\otimes x^-(z))\CK=\psi^+(\hg_{(1)}^{-1/2}\hg_{(2)}^{-1}z)^{-1}\otimes x^-(\hg_{(1)}^{-1}z)\\
\end{split}
\end{align}
The property $\CK\D_H\CK^{-1}=\D'$ follows from these expressions. Note that a similar property has been obtained for $\Uqsl$ in \cite{Khoroshkin1994}.

\subsection{Pieri rules for ordinary Macdonald functions}\label{app:pieri}

In this section, we review the proof of Pieri rules for ordinary Macdonald functions. We reformulate the well-known proof by Macdonald in \cite[Ch.VI, \textsection6]{Macdonald} in the language of symmetric functions in infinitely many variables, making use of the Fourier/Hopf pairing \cite{Cherednik:1995mac,Okounkov:2001are,Beliakova:2021cyc}.

With the normalization in \eqref{eq:spherical-macdonalds}, the spherical Macdonald functions $\tP_\l$ satisfy the \emph{bispectral duality} \cite[Ch.VI, \textsection6, (6.6)]{Macdonald}
\be
\label{eq:bispectral-duality}
 \tP_\lam(\sp_\mu) = \tP_\mu(\sp_\lam)
\ee
also known as Macdonald--Koornwinder reciprocity. We can now define the Fourier/Hopf pairing $(\,,\,)_\mathrm{F}$ as the inner product
\be
 (\tP_\lam,\tP_\mu)_\mathrm{F} = \tP_\lam(\sp_\mu)
\ee
which extends by linearity to the entire ring of symmetric functions $\Lambda$, and it is symmetric by virtue of \eqref{eq:bispectral-duality}. Notice that this is the level $r=1$ version of the pairing defined in \eqref{eq:Fourier-r}, after setting the weight $u$ to 1.

With this definition, we have for any $f,g\in\Lambda$ the identity
\be
 \left(\eta_0^+\cdot f,g\right)_\mathrm{F} = \left(f,\tilde{e}_1\cdot g\right)_\mathrm{F}
\ee
where $\tilde{e}_1 := \tP_{[1]}$.
This allows us to compute explicitly the matrix element of the operator of multiplication by $\tilde{e}_1$ in terms of the action of the operator $\eta^+_0$ in \eqref{eq:rho_x}. The computation goes as follows.

Let $\lam,\mu$ be two arbitrary partitions. Then we consider the pairing
\be
\begin{aligned}
 (\tilde{e}_1\cdot \tP_\lam,\tP_\mu)_\mathrm{F}
 &= (\tP_\lam,\eta_0^+\cdot\tP_\mu)_\mathrm{F} \\
 &= \oint\frac{\mathd z}{2\pi\mathi z}
 \left(\tP_\lam,\mathe^{\sum_{k>0}\frac{z^k}{k}(1-q_1^k)p_k(\bsx)}
 \mathe^{-\sum_{k>0}z^{-k}(1-q_2^k)\frac{\p}{\p p_k(\bsx)}}\tP_\mu\right)_\mathrm{F} \\
 &= \oint\frac{\mathd z}{2\pi\mathi z}
 \left(\tP_\lam(\bsx),\mathe^{\sum_{k>0}\frac{z^k}{k}(1-q_1^k)p_k(\bsx)}\tP_\mu(\bsx-(1-q_2)z^{-1})\right)_\mathrm{F} \\
 &= \oint\frac{\mathd z}{2\pi\mathi z}
 \mathe^{\sum_{k>0}\frac{z^k}{k}(1-q_1^k)p_k(\sp_\lam)}
 \tP_\mu(\sp_\lam-(1-q_2)z^{-1}) \\
 &= \oint\frac{\mathd z}{2\pi\mathi z}
 \mathe^{\sum_{k>0}\frac{z^k}{k} q_1^k p_k(\me_\lam)}
 \tP_\mu(\sp_\lam-(1-q_2)z^{-1}) \\
 &= \oint\frac{\mathd z}{2\pi\mathi z}
 \frac{1}{\CY_\lam(q_1^{-1}z^{-1})}
 \tP_\mu(\sp_\lam-(1-q_2)z^{-1}) \\
 &= -\oint\frac{\mathd w}{2\pi\mathi w}
 \frac{1}{\CY_\lam(w)}
 \tP_\mu(\sp_\lam-(1-q_2)q_1w) \\
\end{aligned}
\ee
where we used the relation $\me_\lam=-(1-q_1^{-1})\sp_\lam$.
From the definition of $\me_\lam$ in \eqref{expr_el} it follows that the contour integral can now be evaluated by taking (minus) the residues at the zeroes of the function $\CY_\lam(w)$, i.e.\ $w=\chi_\sAbox$ for all the boxes that can be added to the partition $\lam$.
We then obtain
\be
 (\tilde{e}_1\cdot \tP_\lam,\tP_\mu)_\mathrm{F}
 = \sum_{\sAbox\in A(\lam)} \res_{w=\chi_\sAbox}
 \frac{1}{w\CY_\lam(w)}
 \tP_\mu(\sp_\lam-(1-q_2)q_1w)
 = \sum_{\sAbox\in A(\lam)}
 r_\lam(\Abox)\,
 \tP_\mu(\sp_\lam-(1-q_2)q_1\chi_\sAbox)
\ee
Using the definition of $\sp_\lam$ in \eqref{expr_sp}, we can write 
\be
 \tP_\mu(\sp_\lam-(1-q_2)q_1\chi_\sAbox)
 = \tP_\mu(\sp_{\lam+\sAbox})
 = (\tP_{\lam+\sAbox},\tP_\mu)_\mathrm{F}
\ee
and, by the non-degeneracy of the pairing, we finally obtain the identity
\be
\label{eq:f_tilde_PE}
 \tilde{e}_1\cdot \tP_\lam
 = \sum_{\sAbox\in A(\lam)}
 r_\lam(\Abox)\,
 \tP_{\lam+\sAbox}
\ee
which is the Pieri rule for the spherical Macdonald functions.

\subsection{Framing operator}\label{app:framing-op}

We show now that $\nabla(\bsu)$ in \eqref{eq:def-nabla} is the operator that represents the action of $F^\perp$ in the horizontal representation as in \eqref{eq:framing-horizontal}.
In order to do so, it is enough to check the commutation relations with the generators of the algebra $\CE$.
\begin{proposition}
Let $\tilde{a}_{-1} := -\frac{\g^{c/2}q_1^{-1}}{(1-q_3)} a_{-1}$.
Then we have
\be
 \rho^{(r,\bsn_0)}_{\bsu}(x^\pm_0)
 = \nabla(\bsu)^{-1} \rho^{(r,\bsn_0)}_{\bsu}(x^\pm_0) \nabla(\bsu)\,,
\ee
\be
\label{eq:framing-id-2}
 \rho^{(r,\bsn_0)}_{\bsu}(\tilde{a}_{-1})
 = \nabla(\bsu)^{-1} \rho^{(r,\bsn_0)}_{\bsu}(q_1^{-1}x^+_{-1}) \nabla(\bsu)\,,
\ee
\be
\label{eq:framing-id-3}
 \rho^{(r,\bsn_0)}_{\bsu}(\tilde{a}_{-1})
 = \nabla(\bsu) \rho^{(r,\bsn_0)}_{\bsu}(-\g^cq_2x^-_{-1}) \nabla(\bsu)^{-1}\,.
\ee
\end{proposition}

\begin{proof}
Commutation between $\nabla(\bsu)$ and $\rho^{(r,\bsn_0)}_{\bsu}(x^\pm_0)$ is obvious as they are both defined as operators that are diagonal on the generalized Macdonald basis (for the same values of weights $\bsu$).
The second identity can be proven by using the commutation relation
\be
 q_1^{-k-1} x_{-k-1}^+
 = \frac1{(1-q_1)(1-q_2)} \left[\tilde{a}_{-1},q_1^{-k} x_{-k}^+\right]\,.
\ee
Then, for any vector of partitions $\blam$, we can write
\be
\begin{aligned}
 & \nabla(\bsu)^{-1} \rho^{(r,\bsn_0)}_{\bsu}(q_1^{-1}x^+_{-1}) \nabla(\bsu)
 P_\blam(\bsx^\bullet|\bsu) \\
 &= \frac1{(1-q_1)(1-q_2)}
 \prod_{\sAbox\in\blam} u_\sAbox\chi_\sAbox \nabla(\bsu)^{-1}
 \left(
 \rho^{(r,\bsn_0)}_{\bsu}(\tilde{a}_{-1}) \rho^{(r,\bsn_0)}_{\bsu}(x^+_0)
 - \rho^{(r,\bsn_0)}_{\bsu}(x^+_0) \rho^{(r,\bsn_0)}_{\bsu}(\tilde{a}_{-1})
 \right) P_\blam(\bsx^\bullet|\bsu) \\
 &= \frac1{(1-q_1)(1-q_2)}
 \sum_{\{\bnu:\bnu=\blam+\sAbox\}}
 \frac{\prod_{\sAbox\in\blam} u_\sAbox\chi_\sAbox}
 {\prod_{\sAbox\in\bnu} u_\sAbox\chi_\sAbox}
 \sum_{\a=1}^r u_\a(\me_{\lam^{(\a)}}-\me_{\nu^{(\a)}})
 C_{\bnu/\blam}
 P_\bnu(\bsx^\bullet|\bsu) \\
 &= \sum_{\{\bnu:\bnu=\blam+\sAbox\}}
 \frac{\prod_{\sAbox\in\blam} u_\sAbox\chi_\sAbox}
 {\prod_{\sAbox\in\bnu} u_\sAbox\chi_\sAbox}
 \left(\sum_{\sAbox\in\bnu} u_\sAbox\chi_\sAbox
 -\sum_{\sAbox\in\blam} u_\sAbox\chi_\sAbox\right)
 C_{\bnu/\blam}
 P_\bnu(\bsx^\bullet|\bsu) \\
 &= \sum_{\{\bnu:\bnu=\blam+\sAbox\}}
 C_{\bnu/\blam}
 P_\bnu(\bsx^\bullet|\bsu) \\
 &= \rho_\bsu^{(r,\bsn_0)}(\tilde{a}_{-1}) P_\blam(\bsx^\bullet|\bsu)
\end{aligned}
\ee
where $C_{\bnu/\blam}$ is the coefficient that appears in the generalized Pieri rule, i.e.\
\be
 \rho_\bsu^{(r,\bsn_0)}(\tilde{a}_{-1}) P_\bnu(\bsx^\bullet|\bsu)
 =\sum_{\{\bnu:\bnu=\blam+\sAbox\}} C_{\bnu/\blam}
 P_\bnu(\bsx^\bullet|\bsu)\,,
\ee
and, for $\bnu=\blam+\Abox$, we have used the identity
\be
 \frac{\prod_{\sAbox\in\blam} u_\sAbox\chi_\sAbox}
 {\prod_{\sAbox\in\bnu} u_\sAbox\chi_\sAbox}
 \left(\sum_{\sAbox\in\bnu} u_\sAbox\chi_\sAbox
 -\sum_{\sAbox\in\blam} u_\sAbox\chi_\sAbox\right) = 1\,.
\ee
therefore we obtain
\be
 \nabla(\bsu)^{-1} \rho^{(r,\bsn_0)}_{\bsu}(q_1^{-1}x^+_{-1}) \nabla(\bsu)
 = \rho^{(r,\bsn_0)}_{\bsu}(\tilde{a}_{-1})
\ee
which completes the proof of \eqref{eq:framing-id-2}.

The third identity can be proven using the commutation relation
\be
 (-\g^cq_2)^{k+1}x_{-k-1}^-
 = \frac1{(1-q_1^{-1})(1-q_2^{-1})} \left[\tilde{a}_{-1},(-\g^cq_2)^{k}x_{-k}^-\right]
\ee
Then, similarly to the previous case, we find
\be
\begin{aligned}
 &
 \nabla(\bsu) \rho^{(r,\bsn_0)}_{\bsu}(-\g^cq_2x^-_{-1}) \nabla(\bsu)^{-1}
 P_\blam(\bsx^\bullet|\bsu) \\
 &= \frac1{(1-q_1^{-1})(1-q_2^{-1})}
 \frac{1}{\prod_{\sAbox\in\blam} u_\sAbox\chi_\sAbox}
 \nabla(\bsu) \left(\rho^{(r,\bsn_0)}_{\bsu}(\tilde{a}_{-1}) \rho^{(r,\bsn_0)}_{\bsu}(x^-_0)
 -\rho^{(r,\bsn_0)}_{\bsu}(x^-_0) \rho^{(r,\bsn_0)}_{\bsu}(\tilde{a}_{-1}) \right)
 P_\blam(\bsx^\bullet|\bsu) \\
 &= \frac1{(1-q_1^{-1})(1-q_2^{-1})}
 \sum_{\{\bnu:\bnu=\blam+\sAbox\}}
 \frac{\prod_{\sAbox\in\bnu} u_\sAbox\chi_\sAbox}
 {\prod_{\sAbox\in\blam} u_\sAbox\chi_\sAbox}
 \sum_{\a=1}^r u_\a^{-1}(\me^\vee_{\lam^{(\a)}}-\me^\vee_{\nu^{(\a)}})
 C_{\bnu/\blam}
 P_\bnu(\bsx^\bullet|\bsu) \\
 &= \sum_{\{\bnu:\bnu=\blam+\sAbox\}}
 \frac{\prod_{\sAbox\in\bnu} u_\sAbox\chi_\sAbox}
 {\prod_{\sAbox\in\blam} u_\sAbox\chi_\sAbox}
 \left(\sum_{\sAbox\in\bnu} (u_\sAbox\chi_\sAbox)^{-1}
 -\sum_{\sAbox\in\blam} (u_\sAbox\chi_\sAbox)^{-1}\right)
 C_{\bnu/\blam}
 P_\bnu(\bsx^\bullet|\bsu) \\
 &= \sum_{\{\bnu:\bnu=\blam+\sAbox\}}
 C_{\bnu/\blam}
 P_\bnu(\bsx^\bullet|\bsu) \\
 &= \rho^{(r,\bsn_0)}_{\bsu}(\tilde{a}_{-1}) P_\blam(\bsx^\bullet|\bsu) \\
\end{aligned}
\ee
where we used that, for $\bnu=\blam+\Abox$, we have
\be
 \frac{\prod_{\sAbox\in\bnu} u_\sAbox\chi_\sAbox}
 {\prod_{\sAbox\in\blam} u_\sAbox\chi_\sAbox}
 \left(\sum_{\sAbox\in\bnu} (u_\sAbox\chi_\sAbox)^{-1}
 -\sum_{\sAbox\in\blam} (u_\sAbox\chi_\sAbox)^{-1}\right) = 1\,.
\ee
It then follows that
\be
 \nabla(\bsu) \rho^{(r,\bsn_0)}_{\bsu}(-\g^cq_2x^-_{-1}) \nabla(\bsu)^{-1}
 = \rho^{(r,\bsn_0)}_{\bsu}(\tilde{a}_{-1})
\ee
which completes the proof of \eqref{eq:framing-id-3}.
\end{proof}
Observe that by taking the adjoint of equations \eqref{eq:framing-id-2}, \eqref{eq:framing-id-3} w.r.t.\ the inner product $\langle-,-\rangle_{\mathrm{Z}}$, one obtains analogous commutation relations between the operator $\nabla(\bsu)$ and the $\rho_\bsu^{(r,\bsn_0)}(x^\pm_1)$.

\subsection{Proof of the Conjecture \ref{conj:BH} at level one}\label{App:conjecture}
In order to prove the identity, we need to show that the inner products
\be
\label{eq:BH-induction}
 \Big\langle f(\bsx),\rho^{(1,0)}_u(F^\perp)
 \mathe^{\sum_{k>0}\frac{(-1)^k p_k(\bsx)}{k(1-q_2^k)}}\Big\rangle_{q,t}
 = \Big\langle f(\bsx),
 \mathe^{-\sum_{k>0}\frac{(-\g^{-1}u)^k p_k(\bsx)}{k(1-q_2^k)}}\Big\rangle_{q,t}
\ee
are equal for any symmetric function $f\in\L$.
We observe that any function $f$ of degree $n$ can be written as a degree-$n$ non-commutative polynomial in the algebra generators $x^-_{-1}$ and $a_{-1}$, acting on the vacuum of the Fock module. This allows us to prove the equality by induction on the degree of $f$, provided the identity holds for degree zero. The degree zero case is obviously true given that
\be
 \Big\langle 1,\rho^{(1,0)}_u(F^\perp)
 \mathe^{\sum_{k>0}\frac{(-1)^k p_k}{k(1-q_2^k)}}\Big\rangle_{q,t}
 = \Big\langle \rho^{(1,0)}_u(F^\perp)1,
 \mathe^{\sum_{k>0}\frac{(-1)^k p_k}{k(1-q_2^k)}}\Big\rangle_{q,t}
 = 1 = \Big\langle 1,
 \mathe^{-\sum_{k>0}\frac{u^k p_k}{k(1-q_2^k)}}\Big\rangle_{q,t}
\ee
where we used that $\rho^{(1,0)}_u(F^\perp)$ is self-adjoint and it has eigenvalue 1 on the vacuum.
The induction argument goes as follows. Suppose \eqref{eq:BH-induction} holds
for some $f$ of degree $n$. Then we consider the scalar product of the l.h.s.\ with the function $\rho^{(1,0)}_u(a_{-1})f$ of degree $n+1$,
\be
\begin{aligned}
 \Big\langle \rho^{(1,0)}_u(a_{-1})\,f,&\rho^{(1,0)}_u(F^\perp)
 \mathe^{\sum_{k>0}\frac{(-1)^k p_k}{k(1-q_2^k)}}\Big\rangle_{q,t}
 = \Big\langle f,\rho^{(1,0)}_u(\tilde{\s}_H(a_{-1}))\rho^{(1,0)}_u(F^\perp)
 \mathe^{\sum_{k>0}\frac{(-1)^k p_k}{k(1-q_2^k)}}\Big\rangle_{q,t}\\
 &= \Big\langle f,\rho^{(1,0)}_u(-q_1a_1)\rho^{(1,0)}_u(F^\perp)
 \mathe^{\sum_{k>0}\frac{(-1)^k p_k}{k(1-q_2^k)}}\Big\rangle_{q,t}\\
 &= \Big\langle f,-q_1(\g-\g^{-1})\g^{-1/2}\rho^{(1,0)}_u(F^\perp)\rho^{(1,0)}_u(x^+_1)
 \mathe^{\sum_{k>0}\frac{(-1)^k p_k}{k(1-q_2^k)}}\Big\rangle_{q,t}.
\end{aligned}
\ee
Using the representation of $x^+(z)$ in the level-one horizontal module, we find
\be
\label{eq:BH-ind-id1}
\begin{aligned}
 \rho^{(1,0)}_u(x^+_1)
 \mathe^{\sum_{k>0}\frac{(-1)^k p_k}{k(1-q_2^k)}}
 &= u\oint\frac{\mathd z}{2\pi\mathi z} z
 \mathe^{\sum_{k>0}\frac{z^k}{k}(1-q_1^k) p_k}
 \mathe^{-\sum_{k>0}z^{-k}(1-q_2^k) \frac{\p}{\p p_k}}
 \mathe^{\sum_{k>0}\frac{(-1)^k p_k}{k(1-q_2^k)}} \\
 &= u\,\mathe^{\sum_{k>0}\frac{(-1)^k p_k}{k(1-q_2^k)}}
 \oint\frac{\mathd z}{2\pi\mathi z} z
 \mathe^{-\sum_{k>0}\frac{(-z)^{-k}}{k}}
 \mathe^{\sum_{k>0}\frac{z^k}{k}(1-q_1^k) p_k} \\
 &= u\,\mathe^{\sum_{k>0}\frac{(-1)^k p_k}{k(1-q_2^k)}}
 \oint\frac{\mathd z}{2\pi\mathi z} z
 (1+z^{-1}) \mathe^{\sum_{k>0}\frac{z^k}{k}(1-q_1^k) p_k} \\
 &= u\,\mathe^{\sum_{k>0}\frac{(-1)^k p_k}{k(1-q_2^k)}}
\end{aligned}
\ee
Hence,
\be
\begin{aligned}
 \Big\langle \rho^{(1,0)}_u(a_{-1})\,f,\rho^{(1,0)}_u(F^\perp)
 \mathe^{\sum_{k>0}\frac{(-1)^k p_k}{k(1-q_2^k)}}\Big\rangle_{q,t}
 &= -uq_1(\g-\g^{-1})\g^{-1/2}\Big\langle f,\rho^{(1,0)}_u(F^\perp)\,
 \mathe^{\sum_{k>0}\frac{(-1)^k p_k}{k(1-q_2^k)}}\Big\rangle_{q,t}\\
 &= -uq_1(\g-\g^{-1})\g^{-1/2}\Big\langle f,
 \mathe^{-\sum_{k>0}\frac{(-\g^{-1}u)^k p_k}{k(1-q_2^k)}}\Big\rangle_{q,t}\\
 &= q_1\g^{-1/2}(1-q_2)(1-q_3)\Big\langle f, \frac{\p}{\p p_1}
 \mathe^{-\sum_{k>0}\frac{(-\g^{-1}u)^k p_k}{k(1-q_2^k)}}\Big\rangle_{q,t}\\
 &= \Big\langle f, -q_1\rho^{(1,0)}_u(a_1)
 \mathe^{-\sum_{k>0}\frac{(-\g^{-1}u)^k p_k}{k(1-q_2^k)}}\Big\rangle_{q,t}\\
 &= \Big\langle \rho^{(1,0)}_u(a_{-1}) f,
 \mathe^{-\sum_{k>0}\frac{(-\g^{-1}u)^k p_k}{k(1-q_2^k)}}\Big\rangle_{q,t},
\end{aligned}
\ee
where in the second line we used the induction hypothesis \eqref{eq:BH-induction}.
We now need to repeat the argument for the action of $x^-_{-1}$. Starting with the l.h.s., we use
\be
 \tilde{\s}_H(x^-_{-1})F^\perp
 = -q_1x^-_1 F^\perp
 = q_1(\g-\g^{-1})^{-1}\g^{-c/2}F^\perp a_1,
\ee
together with
\be
 \rho^{(1,0)}_u(a_1)\mathe^{\sum_{k>0}\frac{(-1)^k p_k}{k(1-q_2^k)}}
 = -\g^{-1/2}(1-q_2)(1-q_3)\frac{\p}{\p p_1} \mathe^{\sum_{k>0}\frac{(-1)^k p_k}{k(1-q_2^k)}}
 = \g^{-1/2}(1-q_3) \mathe^{\sum_{k>0}\frac{(-1)^k p_k}{k(1-q_2^k)}},
\ee
to obtain
\be
\begin{aligned}
 \Big\langle \rho^{(1,0)}_u(x^-_{-1})\,f,\rho^{(1,0)}_u(F^\perp)
 \mathe^{\sum_{k>0}\frac{(-1)^k p_k}{k(1-q_2^k)}}\Big\rangle_{q,t}
 &= -q_1\Big\langle f,\rho^{(1,0)}_u(F^\perp)\,
 \mathe^{\sum_{k>0}\frac{(-1)^k p_k}{k(1-q_2^k)}}\Big\rangle_{q,t}\\
 &= -q_1\Big\langle f,
 \mathe^{-\sum_{k>0}\frac{(-\g^{-1}u)^k p_k}{k(1-q_2^k)}}\Big\rangle_{q,t},
\end{aligned}
\ee
where the second line is obtained using the induction assumption.
Moving to the r.h.s., we need to compute
\be
\label{eq:BH-ind-id2}
\begin{aligned}
 \rho^{(1,0)}_u(x^-_1)\mathe^{-\sum_{k>0}\frac{(-\g^{-1}u)^k p_k}{k(1-q_2^k)}}
 &= u^{-1}\oint\frac{\mathd z}{2\pi\mathi z} z
 \mathe^{-\sum_{k>0}\frac{z^k}{k}\g^k(1-q_1^k) p_k}
 \mathe^{\sum_{k>0}z^{-k}\g^k(1-q_2^k) \frac{\p}{\p p_k}}
 \mathe^{-\sum_{k>0}\frac{(-\g^{-1}u)^k p_k}{k(1-q_2^k)}}\\
 &= u^{-1}\mathe^{-\sum_{k>0}\frac{(-\g^{-1}u)^k p_k}{k(1-q_2^k)}}
 \oint\frac{\mathd z}{2\pi\mathi z} z
 \mathe^{-\sum_{k>0}\frac{(-z^{-1}u)^k}{k}}
 \mathe^{-\sum_{k>0}\frac{z^k}{k}\g^k(1-q_1^k) p_k} \\
 &= u^{-1}\mathe^{-\sum_{k>0}\frac{(-\g^{-1}u)^k p_k}{k(1-q_2^k)}}
 \oint\frac{\mathd z}{2\pi\mathi z} z(1+z^{-1}u)
 \mathe^{-\sum_{k>0}\frac{z^k}{k}\g^k(1-q_1^k) p_k} \\
 &= \mathe^{-\sum_{k>0}\frac{(-\g^{-1}u)^k p_k}{k(1-q_2^k)}},
\end{aligned}
\ee
which shows that,
\be
\begin{aligned}
 \Big\langle \rho^{(1,0)}_u(x^-_{-1})\,f,\rho^{(1,0)}_u(F^\perp)
 \mathe^{\sum_{k>0}\frac{(-1)^k p_k}{k(1-q_2^k)}}\Big\rangle_{q,t}
 &= \Big\langle f,\rho^{(1,0)}_u(-q_1x^-_1)
 \mathe^{-\sum_{k>0}\frac{(-\g^{-1}u)^k p_k}{k(1-q_2^k)}}\Big\rangle_{q,t}\\
 &= \Big\langle \rho^{(1,0)}_u(x^-_{-1})f,
 \mathe^{-\sum_{k>0}\frac{(-\g^{-1}u)^k p_k}{k(1-q_2^k)}}\Big\rangle_{q,t},
\end{aligned}
\ee
and it concludes the proof at level 1.

\begin{remark}
At level higher than 1, the same identities can be shown to hold for $x^-_{-1}$ and $a_{-1}$, since they come from universal relations of the algebra $\CE$.
In particular, one needs to show that
\be
 \rho_\bsu^{(r,\bsn_0)}(x^+_1)
 \exp\left(\sum_{k>0}\frac{(-1)^k}{k(1-q_2^k)}\sum_{\a=1}^r v_\a^k p_k^{(\a)}\right)
 = \Big(\sum_{\a=1}^r u_\a v_\a\Big)
 \exp\left(\sum_{k>0}\frac{(-1)^k}{k(1-q_2^k)}\sum_{\a=1}^r v_\a^k p_k^{(\a)}\right),
\ee
and
\begin{multline}
 \rho_\bsu^{(r,\bsn_0)}(x^-_1)
 \exp\left(-\sum_{k>0}\frac{(-\g^{-1})^k}{k(1-q_2^k)}\sum_{\a=1}^r p_k^{(\a)}
 \Big((u_\a v_\a)^k+(1-q_3^k)\sum_{\b=\a+1}^r(u_\b v_\b)^k\Big)\right)\\
 = \Big(\sum_{\a=1}^r \g^{-r+2\a-1} v_\a\Big)
 \exp\left(-\sum_{k>0}\frac{(-\g^{-1})^k}{k(1-q_2^k)}\sum_{\a=1}^r p_k^{(\a)}
 \Big((u_\a v_\a)^k+(1-q_3^k)\sum_{\b=\a+1}^r(u_\b v_\b)^k\Big)\right),
\end{multline}
whose proofs are straightforward generalizations of \eqref{eq:BH-ind-id1} and \eqref{eq:BH-ind-id2}, respectively.
The non-trivial part of the proof, however, is the fact that it is no longer true that any multi-symmetric function $f(\bsx^\bullet)$ can be expressed as a non-commutative polynomial in those two generators. Instead, one needs to use all generators of the form $(\CT^\perp)^k (a_{-1})$ (for $k\geq0$ at least), where $\CT^\perp(a_{-1}) \propto x^-_{-1}$.
\end{remark}

\section{Covariance of higher vertex operators}\label{sec_higher_VO}
In this appendix, we study the higher vertex operators represented in Figures~\ref{fig_intw_m2} and \ref{fig_intw_m3} and their generalizations to higher level $r$. In particular, we will show by induction that these objects solve the intertwining relations \eqref{prop_intw_r}. Note that in this section, we consider the decomposition of the vertex operator on the vertical states $\dket{\bl}^\otimes$, i.e.
\be
\begin{aligned}\label{def_VO}
 \Phi^{(r,\bsn)}[\bsu,\bsv,\bsw]
 &= \sum_{\bl} \prod_\a a_{\l^{(\a)}}\ \Phi_{\bl}^{(r,\bsn)}[\bsu,\bsv,\bsw]\  \dbra{\bl}^\otimes\\
 \Phi^{(r,\bsn)\ast}[\bsu,\bsv,\bsw]
 &= \sum_{\bl} \prod_\a a_{\l^{(\a)}}\ \Phi_{\bl}^{(r,\bsn)\ast}[\bsu,\bsv,\bsw]\  \dket{\bl}^\otimes
\end{aligned}
\ee
This is so that the vertical representation is given by the coproduct construction. In order to recover the vertical component used in the main text, it should be multiplied by an extra factor $G_\bl(\bsv)$ or $G^\ast_\bl(\bsv)$ according to the rules \eqref{eq:rule_G}. This factor does not depend on the internal vertical weight, and so the arguments about the critical values developed below remain unchanged. Note also that the formulas \eqref{def_VO} do not reflect the choice of normalization \eqref{norm_VO} that will be imposed at the end of the section.

\subsection{Vertex operator \texorpdfstring{$\Phi^{(r,\bsn)}$}{Phi(r,n)}}
\label{app:B1}
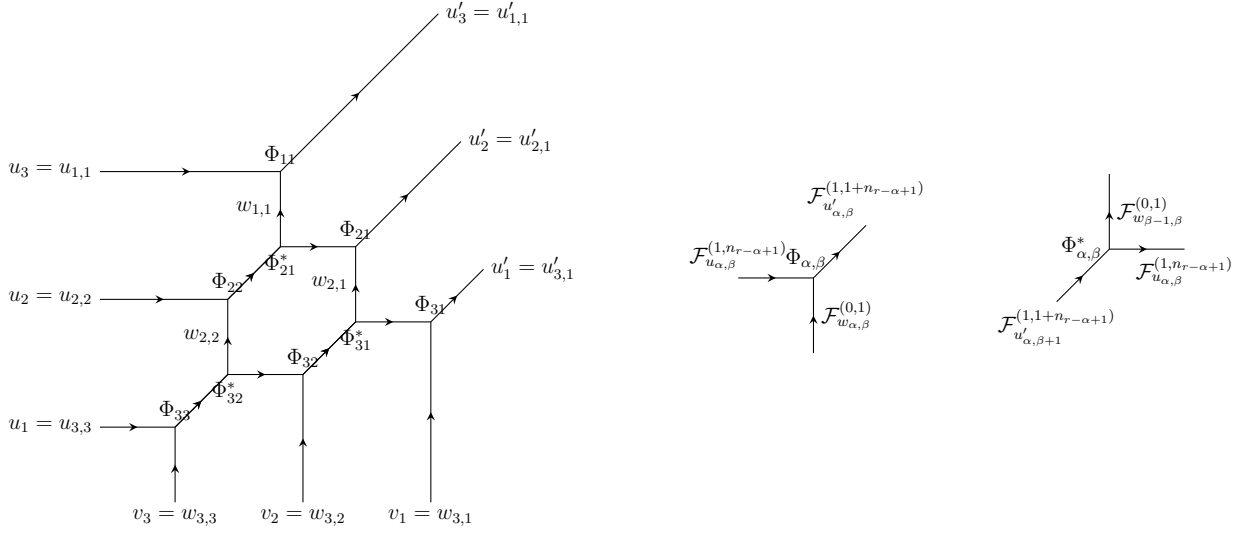
\begin{figure}
\begin{center}
\raisebox{-0.5\height}{\begin{tikzpicture}[scale=1]
\draw[postaction={on each segment={mid arrow=black}}] (-1,0) -- (0,0) -- (0.7,0.7) -- (1.7,0.7) -- (2.4,1.4) -- (3.4,1.4) -- (4.1,2.1);
\draw[postaction={on each segment={mid arrow=black}}] (-1,1.7) -- (0.7,1.7) -- (1.4,2.4) -- (2.4,2.4) -- (3.8,3.8);
\draw[postaction={on each segment={mid arrow=black}}] (-1,3.4) -- (1.4,3.4) -- (3.5,5.5);
\draw[postaction={on each segment={mid arrow=black}}] (0,-1) -- (0,0) -- (0.7,0.7) -- (0.7,1.7) -- (1.4,2.4) -- (1.4,3.4);
\draw[postaction={on each segment={mid arrow=black}}] (1.7,-1) -- (1.7,0.7) -- (2.4,1.4) -- (2.4,2.4);
\draw[postaction={on each segment={mid arrow=black}}] (3.4,-1) -- (3.4,1.4);
\node[above,scale=0.7] at (0,0) {$\Phi_{33}$};
\node[below,scale=0.7] at (0.7,0.7) {$\Phi^\ast_{32}$};
\node[above,scale=0.7] at (1.7,0.7) {$\Phi_{32}$};
\node[below,scale=0.7] at (2.4,1.4) {$\Phi^\ast_{31}$};
\node[above,scale=0.7] at (3.4,1.4) {$\Phi_{31}$};
\node[above,scale=0.7] at (0.7,1.7) {$\Phi_{22}$};
\node[below,scale=0.7] at (1.4,2.4) {$\Phi^\ast_{21}$};
\node[above,scale=0.7] at (2.4,2.4) {$\Phi_{21}$};
\node[above,scale=0.7] at (1.4,3.4) {$\Phi_{11}$};
\node[left,scale=0.7] at (.7,1.2) {$w_{2,2}$};
\node[left,scale=0.7] at (2.4,1.9) {$w_{2,1}$};
\node[left,scale=0.7] at (1.4,2.9) {$w_{1,1}$};
\node[below,scale=0.7] at (0,-1) {$v_3=w_{3,3}$};
\node[below,scale=0.7] at (1.7,-1) {$v_2=w_{3,2}$};
\node[below,scale=0.7] at (3.4,-1) {$v_1=w_{3,1}$};
\node[left,scale=0.7] at (-1,0) {$u_1=u_{3,3}$};
\node[left,scale=0.7] at (-1,1.7) {$u_2=u_{2,2}$};
\node[left,scale=0.7] at (-1,3.4) {$u_3=u_{1,1}$};
\node[right,scale=0.7] at (4.1,2.1) {$u_1'=u_{3,1}'$};
\node[right,scale=0.7] at (3.8,3.8) {$u_2'=u_{2,1}'$};
\node[right,scale=0.7] at (3.5,5.5) {$u_3'=u_{1,1}'$};
\end{tikzpicture}}
\hspace{30pt}
\raisebox{-0.5\height}{\begin{tikzpicture}[scale=1]
\draw[postaction={on each segment={mid arrow=black}}] (0,0) -- (1,0) -- (1.7,0.7);
\draw[postaction={on each segment={mid arrow=black}}] (1,-1) -- (1,0);
\node[above,scale=0.7] at (0.9,0) {$\Phi_{\a,\b}$};
\node[above,scale=0.7] at (0,0) {$\CF^{(1,n_{r-\a+1})}_{u_{\a,\b}}$};
\node[right,scale=0.7] at (1,-0.5) {$\CF^{(0,1)}_{w_{\a,\b}}$};
\node[above,scale=0.7] at (1.7,0.7) {$\CF^{(1,1+n_{r-\a+1})}_{u_{\a,\b}'}$};
\end{tikzpicture}}
\hspace{5mm}
\raisebox{-0.5\height}{\begin{tikzpicture}[scale=1]
\draw[postaction={on each segment={mid arrow=black}}] (-0.7,-0.7) -- (0,0) -- (0,1);
\draw[postaction={on each segment={mid arrow=black}}] (0,0) -- (1,0);
\node[left,scale=0.7] at (0,0) {$\Phi^{\ast}_{\a,\b}$};
\node[below,scale=0.7] at (1,0) {$\CF^{(1,n_{r-\a+1})}_{u_{\a,\b}}$};
\node[right,scale=0.7] at (0,0.5) {$\CF^{(0,1)}_{w_{\b-1,\b}}$};
\node[below,scale=0.7] at (-0.7,-0.7) {$\CF^{(1,1+n_{r-\a+1})}_{u_{\a,\b+1}'}$};
\end{tikzpicture}}
\end{center}
\caption{Diagram associated to $\Phi^{(r,\bsn)}$ showing the choice of labels for weights and intermediate vertex operators (right).}
\label{fig_labelling}
\end{figure}

The first step is to introduce a proper labeling of the intermediate weights and vertex operators, this is done in Figure~\ref{fig_labelling} for $r=3$, and can easily be generalized to higher levels. Hence, we introduced the notation AFS vertex operators $\Phi_{\a,\b}=\Phi^{(1,n_{r-\a+1})}[u_{\a,\b},w_{\a,\b}]$ and $\Phi_{\a,\b}^\ast=\Phi_{\a,\b}^{\ast(1,n_{r-\a+1})}[u_{\a,\b},w_{\a-1,\b}]$  with $1\leq \b\leq \a\leq r$ (resp. $1\leq \b<\a\leq r$). With this notation, we have
\be
 u_\a=u_{r-\a+1,r-\a+1},
 \hspace{30pt}
 v_\b=w_{r,\b},
 \hspace{30pt}
 u'_\a=u_{r-\a+1,1},
\ee
and the weights constraint of AFS vertex operators, namely $u'_{\a,\b}=-\g u_{\a,\b}w_{\a,\b}=-\g u_{\a,\b-1}w_{\a-1,\b-1}$ fixes all intermediate $u_{\a,\b}$ and $u'_{\a,\b}$, and leads to the relations \eqref{rel_weights}.

The vertical component of the vertex operator $\Phi^{(r,\bsn)}[\bsu,\bsv,\bsw]$ can be expressed as a sum over the realizations of $r(r-1)/2$ partitions $\l^{(\a,\b)}$ with $1\leq \b\leq \a\leq r-1$ of vertex operators acting on the tensor product of $r$ bosonic Fock spaces,
\be
\label{exp_Phi_rbsn}
\Phi_\bl^{(r,\bsn)}[\bsu,\bsv,\bsw]=\sum_{\{\l^{(\a,\b)}\}}\prod_{\a,\b}a_{\l^{(\a,\b)}}\ \Phi_{r,1}\Phi_{r,1}^\ast\Phi_{r,2}\cdots\Phi_{r,r-1}^\ast\Phi_{r,r}\otimes\Phi_{r-1,1}\Phi^\ast_{r-1,1}\cdots\Phi_{r-1,r-2}^\ast\Phi_{r-1,r-1}\otimes\cdots\otimes \Phi_{1,1},
\ee
where we used the following shortcut notations for the AFS vertex operators,
\be
 \Phi_{\a,\b}=\Phi_{\l^{(\a,\b)}}^{(1,n_{r-\a+1})}[u_{\a,\b},w_{\a,\b}],
 \hspace{30pt}
 \Phi_{\a,\b}^\ast=\Phi_{\l^{(\a-1,\b)}}^{(1,n_{r-\a+1})\ast}[u_{\a,\b},w_{\a-1,\b}],
\ee
and we denoted $\l^{(r,\a)}=\l^{(\a)}$ the fixed partitions of the external vertical modules, $\bl=(\l^{(1)},\dots,\l^{(r)})$.

\paragraph{Recursive construction.} There are three ways to build recursively the vertex operators, corresponding to the three sets of outer legs. First, considering the lower vertical legs, i.e.
\begin{center}
\begin{tikzpicture}[scale=1]
\draw[dotted,thick] (-1,1) -- (0,1) -- (3.4,2.4) -- (4.1,3.1);
\draw[postaction={on each segment={mid arrow=black}}] (-1,0) -- (0,0) -- (0.7,0.7) -- (1.7,0.7) -- (2.4,1.4) -- (3.4,1.4) -- (4.1,2.1);
\draw[postaction={on each segment={mid arrow=black}}] (-1,1.7) -- (0.7,1.7) -- (1.4,2.4) -- (2.4,2.4) -- (3.8,3.8);
\draw[postaction={on each segment={mid arrow=black}}] (-1,3.4) -- (1.4,3.4) -- (3.5,5.5);
\draw[postaction={on each segment={mid arrow=black}}] (0,-1) -- (0,0) -- (0.7,0.7) -- (0.7,1.7) -- (1.4,2.4) -- (1.4,3.4);
\draw[postaction={on each segment={mid arrow=black}}] (1.7,-1) -- (1.7,0.7) -- (2.4,1.4) -- (2.4,2.4);
\draw[postaction={on each segment={mid arrow=black}}] (3.4,-1) -- (3.4,1.4);
\node[above,scale=0.7] at (0,0) {$\Phi_{33}$};
\node[below,scale=0.7] at (0.7,0.7) {$\Phi^\ast_{32}$};
\node[above,scale=0.7] at (1.7,0.7) {$\Phi_{32}$};
\node[below,scale=0.7] at (2.4,1.4) {$\Phi^\ast_{31}$};
\node[above,scale=0.7] at (3.4,1.4) {$\Phi_{31}$};
\node[above,scale=0.7] at (0.7,1.7) {$\Phi_{22}$};
\node[below,scale=0.7] at (1.4,2.4) {$\Phi^\ast_{21}$};
\node[above,scale=0.7] at (2.4,2.4) {$\Phi_{21}$};
\node[above,scale=0.7] at (1.4,3.4) {$\Phi_{11}$};
\node[left,scale=0.7] at (.7,1.2) {$w_{2,2}$};
\node[left,scale=0.7] at (2.4,1.9) {$w_{2,1}$};
\node[left,scale=0.7] at (1.4,2.9) {$w_{1,1}$};
\node[below,scale=0.7] at (0,-1) {$v_3=w_{3,3}$};
\node[below,scale=0.7] at (1.7,-1) {$v_2=w_{3,2}$};
\node[below,scale=0.7] at (3.4,-1) {$v_1=w_{3,1}$};
\node[left,scale=0.7] at (-1,0) {$u_1=u_{3,3}$};
\node[left,scale=0.7] at (-1,1.7) {$u_2=u_{2,2}$};
\node[left,scale=0.7] at (-1,3.4) {$u_3=u_{1,1}$};
\node[right,scale=0.7] at (4.1,2.1) {$u_1'=u_{3,1}'$};
\node[right,scale=0.7] at (3.8,3.8) {$u_2'=u_{2,1}'$};
\node[right,scale=0.7] at (3.5,5.5) {$u_3'=u_{1,1}'$};
\end{tikzpicture}
\end{center}
Cutting along the dashed line, gives the following recursion,
\be
\label{rec_I}
 \Phi^{(r+1,\bsn)}_{\l^{(1)},\dots,\l^{(r+1)}}
 =\sum_{\mu^{(1)},\dots,\mu^{(r)}}\prod_{\a=1}^r a_{\mu^{(\a)}}\
 \Phi_{\l^{(1)}}\Phi_{\mu^{(1)}}^\ast\cdots\Phi_{\l^{(r)}}\Phi_{\mu^{(r)}}^\ast
 \Phi_{\l^{(r+1)}}\otimes \Phi^{(r,\bsn)}_{\mu^{(1)},\dots,\mu^{(r)}},
\ee
where we omitted the weights and used the shortcut notations $\Phi_{\l^{(\a)}}=\Phi_{\l^{(\a)}}^{(1,n_1)}[u_{r,\a},w_{r,\a}]$ and $\Phi_{\mu^{(\a)}}^\ast=\Phi_{\mu^{(\a)}}^{(1,n_1)\ast}[u_{r,\a},w_{r-1,\a}]$. Alternatively, we can add an additional `column' either on the left or on the right,
\begin{center}
\begin{tikzpicture}[scale=1]
\draw[dotted,thick] (.85,-1) -- (.85,0) -- (1.7,1.7) -- (1.7,2.7) -- (3.4,4.4);
\draw[postaction={on each segment={mid arrow=black}}] (-1,0) -- (0,0) -- (0.7,0.7) -- (1.7,0.7) -- (2.4,1.4) -- (3.4,1.4) -- (4.1,2.1);
\draw[postaction={on each segment={mid arrow=black}}] (-1,1.7) -- (0.7,1.7) -- (1.4,2.4) -- (2.4,2.4) -- (3.8,3.8);
\draw[postaction={on each segment={mid arrow=black}}] (-1,3.4) -- (1.4,3.4) -- (3.5,5.5);
\draw[postaction={on each segment={mid arrow=black}}] (0,-1) -- (0,0) -- (0.7,0.7) -- (0.7,1.7) -- (1.4,2.4) -- (1.4,3.4);
\draw[postaction={on each segment={mid arrow=black}}] (1.7,-1) -- (1.7,0.7) -- (2.4,1.4) -- (2.4,2.4);
\draw[postaction={on each segment={mid arrow=black}}] (3.4,-1) -- (3.4,1.4);
\node[above,scale=0.7] at (0,0) {$\Phi_{33}$};
\node[below,scale=0.7] at (0.7,0.7) {$\Phi^\ast_{32}$};
\node[above,scale=0.7] at (1.7,0.7) {$\Phi_{32}$};
\node[below,scale=0.7] at (2.4,1.4) {$\Phi^\ast_{31}$};
\node[above,scale=0.7] at (3.4,1.4) {$\Phi_{31}$};
\node[above,scale=0.7] at (0.7,1.7) {$\Phi_{22}$};
\node[below,scale=0.7] at (1.4,2.4) {$\Phi^\ast_{21}$};
\node[above,scale=0.7] at (2.4,2.4) {$\Phi_{21}$};
\node[above,scale=0.7] at (1.4,3.4) {$\Phi_{11}$};
\node[left,scale=0.7] at (.7,1.2) {$w_{2,2}$};
\node[left,scale=0.7] at (2.4,1.9) {$w_{2,1}$};
\node[left,scale=0.7] at (1.4,2.9) {$w_{1,1}$};
\node[below,scale=0.7] at (0,-1) {$v_3=w_{3,3}$};
\node[below,scale=0.7] at (1.7,-1) {$v_2=w_{3,2}$};
\node[below,scale=0.7] at (3.4,-1) {$v_1=w_{3,1}$};
\node[left,scale=0.7] at (-1,0) {$u_1=u_{3,3}$};
\node[left,scale=0.7] at (-1,1.7) {$u_2=u_{2,2}$};
\node[left,scale=0.7] at (-1,3.4) {$u_3=u_{1,1}$};
\node[right,scale=0.7] at (4.1,2.1) {$u_1'=u_{3,1}'$};
\node[right,scale=0.7] at (3.8,3.8) {$u_2'=u_{2,1}'$};
\node[right,scale=0.7] at (3.5,5.5) {$u_3'=u_{1,1}'$};
\end{tikzpicture}
\hspace{10mm}
\begin{tikzpicture}[scale=1]
\draw[dotted,thick] (1.7+.85,-1) -- (1.7+.85,.9) -- (.2,2.8) -- (-1,2.8);
\draw[postaction={on each segment={mid arrow=black}}] (-1,0) -- (0,0) -- (0.7,0.7) -- (1.7,0.7) -- (2.4,1.4) -- (3.4,1.4) -- (4.1,2.1);
\draw[postaction={on each segment={mid arrow=black}}] (-1,1.7) -- (0.7,1.7) -- (1.4,2.4) -- (2.4,2.4) -- (3.8,3.8);
\draw[postaction={on each segment={mid arrow=black}}] (-1,3.4) -- (1.4,3.4) -- (3.5,5.5);
\draw[postaction={on each segment={mid arrow=black}}] (0,-1) -- (0,0) -- (0.7,0.7) -- (0.7,1.7) -- (1.4,2.4) -- (1.4,3.4);
\draw[postaction={on each segment={mid arrow=black}}] (1.7,-1) -- (1.7,0.7) -- (2.4,1.4) -- (2.4,2.4);
\draw[postaction={on each segment={mid arrow=black}}] (3.4,-1) -- (3.4,1.4);
\node[above,scale=0.7] at (0,0) {$\Phi_{33}$};
\node[below,scale=0.7] at (0.7,0.7) {$\Phi^\ast_{32}$};
\node[above,scale=0.7] at (1.7,0.7) {$\Phi_{32}$};
\node[below,scale=0.7] at (2.4,1.4) {$\Phi^\ast_{31}$};
\node[above,scale=0.7] at (3.4,1.4) {$\Phi_{31}$};
\node[above,scale=0.7] at (0.7,1.7) {$\Phi_{22}$};
\node[below,scale=0.7] at (1.4,2.4) {$\Phi^\ast_{21}$};
\node[above,scale=0.7] at (2.4,2.4) {$\Phi_{21}$};
\node[above,scale=0.7] at (1.4,3.4) {$\Phi_{11}$};
\node[left,scale=0.7] at (.7,1.2) {$w_{2,2}$};
\node[left,scale=0.7] at (2.4,1.9) {$w_{2,1}$};
\node[left,scale=0.7] at (1.4,2.9) {$w_{1,1}$};
\node[below,scale=0.7] at (0,-1) {$v_3=w_{3,3}$};
\node[below,scale=0.7] at (1.7,-1) {$v_2=w_{3,2}$};
\node[below,scale=0.7] at (3.4,-1) {$v_1=w_{3,1}$};
\node[left,scale=0.7] at (-1,0) {$u_1=u_{3,3}$};
\node[left,scale=0.7] at (-1,1.7) {$u_2=u_{2,2}$};
\node[left,scale=0.7] at (-1,3.4) {$u_3=u_{1,1}$};
\node[right,scale=0.7] at (4.1,2.1) {$u_1'=u_{3,1}'$};
\node[right,scale=0.7] at (3.8,3.8) {$u_2'=u_{2,1}'$};
\node[right,scale=0.7] at (3.5,5.5) {$u_3'=u_{1,1}'$};
\end{tikzpicture}
\end{center}
Two more recursion formulas follow, corresponding to the left and right cuts, respectively.
The first way of splitting the diagram gives 
\begin{equation}\label{rec_II}
\Phi^{(r+1,\bsn)}_{\l^{(1)},\dots,\l^{(r+1)}}=\left(\Phi^{(r,\bsn)}_{\l^{(1)},\dots,\l^{(r)}}\otimes1\right)\left(\sum_{\mu^{(1)},\dots,\mu^{(r)}}\prod_{\a=1}^r a_{\mu^{(\a)}}\ \Phi_{\mu^{(1)}}^\ast\Phi_{\l^{(r+1)}}\otimes\Phi_{\mu^{(2)}}^\ast\Phi_{\mu^{(1)}}\otimes\cdots\otimes\Phi_{\mu^{(r)}}^\ast\Phi_{\mu^{(r-1)}}\otimes\Phi_{\mu^{(r)}}\right),
\end{equation} 
with $\Phi_{\l^{(r+1)}}=\Phi_{\l^{(r+1)}}^{(1,n_1)}[u_{1},v_{r+1}]$, $\Phi_{\mu^{(\a)}}=\Phi_{\mu^{(\a)}}^{(1,n_{\a+1})}[u_{r-\a,r-\a},w_{r-\a,r-\a}]$, $\Phi_{\mu^{(\a)}}^\ast=\Phi_{\mu^{(\a)}}^{(1,n_\a)\ast}[u_{r-\a+1,r-\a+1},w_{r-\a,r-\a+1}]$, while the second gives
\begin{equation}\label{rec_III}
\Phi^{(r+1,\bsn)}_{\l^{(1)},\dots,\l^{(r+1)}}=\left(\sum_{\mu^{(1)},\dots,\mu^{(r)}}\prod_{\a=1}^r a_{\mu^{(\a)}}\ \Phi_{\l^{(1)}}\Phi_{\mu^{(1)}}^\ast\otimes\Phi_{\mu^{(1)}}\Phi_{\mu^{(2)}}^\ast\otimes\cdots\otimes\Phi_{\mu^{(r-1)}}\Phi_{\mu^{(r)}}^\ast\otimes\Phi_{\mu^{(r)}}\right)\left(\Phi^{(r,\bsn)}_{\l^{(2)},\dots,\l^{(r+1)}}\otimes1\right),
\end{equation} 
with $\Phi_{\l^{(1)}}=\Phi_{\l^{(1)}}^{(1,n_1)}[u_{r,1},v_1]$, $\Phi_{\mu^{(\a)}}=\Phi_{\mu^{(\a)}}^{(1,n_{\a+1})}[u_{r-\a,1},w_{r-\a,1}]$ and $\Phi_{\mu^{(\a)}}^\ast=\Phi_{\mu^{(\a)}}^{(1,n_\a)\ast}[u_{r-\a+1,1},w_{r-\a,1}]$.

\paragraph{Intertwining relation.} To prove the intertwining relation, we will use the recursion relation \eqref{rec_I} obtained by cutting the bottom part of the vertex operator. It will be convenient to introduce the operator
\be
\label{eq:opAuvv}
 A^{(r,n)}[u,\bsv,\bsv']:\CF_{\bsv}^{(0,r)}\otimes \CF_u^{(1,n)}\to \CF_{\bsv'}^{(0,r-1)}\otimes \CF^{(1,n+1)}_{u'},
 \hspace{30pt}
 u'=-\g v_r\prod_{\a=1}^{r-1}\frac{v_\a}{v'_\a},
\ee
defined by
\begin{multline}
A^{(r,n)}[u,\bsv,\bsv']=\sum_{\l^{(1)},\dots,\l^{(r)}}\prod_{\a=1}^r a_{\l^{(\a)}}\equskip\sum_{\mu^{(1)},\dots,\mu^{(r-1)}}\prod_{\b=1}^{r-1} a_{\mu^{(\b)}}\  \left(\dket{\l^{(1)}}\otimes\cdots\otimes\dket{\l^{(r)}}\right)\left(\dbra{\mu^{(1)}}\otimes\cdots\otimes\dbra{\mu^{(r-1)}}\right)\\
\otimes \Phi_{\l^{(1)}}[u_1,v_1]\Phi_{\mu^{(1)}}^\ast[u_{1},v'_{1}]\cdots\Phi_{\l^{(r-1)}}^{\ast}[u_{r-1},v_{r-1}]\Phi_{\mu^{(r-1)}}^\ast[u_{r-1},v'_{r-1}]\Phi_{\l^{(r)}}[u_r,v_r]
\end{multline}
with $u_r=u$ and $u_{\a-1}=u_\a v_\a/v'_{\a-1}$. We will introduce a symbol for this operator $A^{(r,n)}$,
\begin{center}
\begin{tikzpicture}[scale=1]
\node[ellipse,draw,scale=.7] (M) at (0,0) {$A^{(r,n)}$};
\draw[postaction={on each segment={mid arrow=black}}] (-1,0) -- (M);
\draw[postaction={on each segment={mid arrow=black}}] (0,-1) -- (M);
\draw[postaction={on each segment={mid arrow=black}}] (M) -- (0,1);
\draw[postaction={on each segment={mid arrow=black}}] (M) -- (.7,.7);
\node[left,scale=.7] at (-1,0) {$\CF_u^{(1,n)}$};
\node[below,scale=.7] at (0,-1) {$\CF_\bsv^{(0,r)}$};
\node[above,scale=.7] at (0,1) {$\CF_{\bsv'}^{(0,r-1)}$};
\node[above right,scale=.7] at (.7,.7) {$\CF_{u'}^{(1,n+1)}$};
\end{tikzpicture}
\end{center}
so that the recursion relation \eqref{rec_I} can be described diagrammatically as the decomposition
\begin{center}
\begin{tikzpicture}[scale=1]
\draw[postaction={on each segment={mid arrow=black}}] (0,0) -- (1,0) -- (1.7,0.7);
\draw[postaction={on each segment={mid arrow=black}}] (1,-1) -- (1,0);
\node[below right,scale=0.7] at (1,0) {$\Phi^{(r+1,\bsn)}$};
\node[scale=1] at (2.5,0) {$=$};
\draw[postaction={on each segment={mid arrow=black}}] (3+0,1) -- (3.5+1,1) -- (3.5+1.7,1.7);
\node[below right,scale=0.7] at (3.5+1,1) {$\Phi^{(r,\tilde{\bsn})}$};
\node[ellipse,draw,scale=.6] (M) at (4.5,-1) {$A^{(r+1,n_1)}$};
\draw[postaction={on each segment={mid arrow=black}}] (M) -- (4.5,1);
\draw[postaction={on each segment={mid arrow=black}}] (3,-1) -- (M);
\draw[postaction={on each segment={mid arrow=black}}] (4.5,-2) -- (M);
\draw[postaction={on each segment={mid arrow=black}}] (M) -- (5.2,.7-1);
\end{tikzpicture}
\end{center}
which corresponds to $\Phi^{(r+1,\bsn)}[\bsu,\bsv,\bsw]=A^{(r+1,n_1)}[u_1,\bsv,\tilde{\bsv}]\cdot \Phi^{(r,\tilde{\bsn})}[\tilde{\bsu},\tilde{\bsv},\bsw]$
where the product is taken along the vertical modules and $\tilde{\bsn}=(n_2,\dots,n_{r+1})$, $\tilde{\bsu}=(u_2,\dots,u_{r+1})$, $\tilde{\bsv}=(w_{r-1,1},\dots,w_{r-1,r-1})$.

The proof relies on the decomposition of the horizontal representations $\rho_{\bsu}^{(r+1,\bsn)}$ into $(\rho_{u_1}^{(1,n_1)}\otimes \rho_{\tilde{\bsu}}^{(r,\tilde{\bsn})})\circ\D$. We will use Sweedler's notation for the coproduct, i.e. we omit the summation and denote
\begin{equation}
\D(e)=e_{(1)}\otimes e_{(2)},\quad (\D\otimes\Id)\D(e)=e_{(1)(1)}\otimes e_{(1)(2)}\otimes e_{(2)},\quad e\in\CE.
\end{equation}
With this notation, the co-associativity property simply reads $e_{(1)(1)}\otimes e_{(1)(2)}\otimes e_{(2)}=e_{(1)}\otimes e_{(2)(1)}\otimes e_{(2)(2)}$, and the previous decomposition is
\begin{equation}
\rho_{\bsu}^{(r+1,\bsn)}(e)=\rho_{u_1}^{(1,n_1)}(e_{(1)})\otimes \rho_{\tilde{\bsu}}^{(r,\tilde{\bsn})}(e_{(2)}).
\end{equation}
The induction hypothesis at order $r$ is
\begin{center}
\begin{tikzpicture}[scale=1]
\draw[postaction={on each segment={mid arrow=black}}] (0,0) -- (1,0) -- (1.7,0.7);
\draw[postaction={on each segment={mid arrow=black}}] (1,-1) -- (1,0);
\node[below right,scale=0.7] at (1,0) {$\Phi^{(r,\bsn)}$};
\node[left,scale=0.7] at (0,0) {$\rho_{\bsu}^{(r,\bsn)}(e_{(2)})$};
\node[below,scale=0.7] at (1,-1) {$\rho_{\bsv}^{(0,r)}(e_{(1)})$};
\node[,scale=1] at (2.5,0) {$=$};
\draw[postaction={on each segment={mid arrow=black}}] (3+0,0) -- (3+1,0) -- (3+1.7,0.7);
\draw[postaction={on each segment={mid arrow=black}}] (3+1,-1) -- (3+1,0);
\node[below right,scale=0.7] at (3+1,0) {$\Phi^{(r,\bsn)}$};
\node[above right,scale=0.7] at (3+1.7,.7) {$\rho_{\bsu'}^{(r,\bsn')}(e)$};
\end{tikzpicture}
\end{center}

We first need a lemma regarding the intertwining property of the operator $A^{(r,\bsn)}$, namely
\begin{center}
\begin{tikzpicture}[scale=1]
\node[ellipse,draw,scale=.7] (M1) at (0,0) {$A^{(r,n)}$};
\draw[postaction={on each segment={mid arrow=black}}] (-1,0) -- (M1);
\draw[postaction={on each segment={mid arrow=black}}] (0,-1) -- (M1);
\draw[postaction={on each segment={mid arrow=black}}] (M1) -- (0,1);
\draw[postaction={on each segment={mid arrow=black}}] (M1) -- (.7,.7);
\node[left,scale=.7] at (-1,0) {$\rho_u^{(1,n)}(e_{(2)})$};
\node[below,scale=.7] at (0,-1) {$\rho_\bsv^{(0,r)}(e_{(1)})$};
\node[scale=1] at (1.5,0) {$=$}; 
\def\offset{3}
\node[ellipse,draw,scale=.7] (M2) at (\offset+0,0) {$A^{(r,n)}$};
\draw[postaction={on each segment={mid arrow=black}}] (\offset-1,0) -- (M2);
\draw[postaction={on each segment={mid arrow=black}}] (\offset+0,-1) -- (M2);
\draw[postaction={on each segment={mid arrow=black}}] (M2) -- (\offset+0,1);
\draw[postaction={on each segment={mid arrow=black}}] (M2) -- (\offset+.7,.7);
\node[above,scale=.7] at (\offset,1) {$\rho_{\bsv'}^{(0,r-1)}(e_{(2)})$};
\node[above right,scale=.7] at (\offset+.7,.7) {$\rho_{u'}^{(1,n+1)}(e_{(1)})$};
\end{tikzpicture}
\end{center}
To show this, we can decompose the operator $A^{(r,n)}$ into a product of an operator $M^{(r-1,n)}$ and an AFS intertwiner in the horizontal Fock module, i.e. $A^{(r,n)}[u,\bsv,\bsv']=M^{(r-1,n)}[u,\tilde{\bsv},\bsv']\Phi^{(1,n)}[u_r,v_1]$, with $\tilde{\bsv}=(v_2,\dots,v_r)$ and
\be
 M^{(r,n)}[u,\bsv,\bsv']:\CF_{\bsv}^{(0,r)}\otimes \CF_u^{(1,n)}\to \CF_{\bsv'}^{(0,r)}\otimes \CF^{(1,n)}_{u'},
 \hspace{30pt} u'=\prod_{\a=1}^r\dfrac{v_\a}{v'_\a},
\ee
\begin{multline}
M^{(r,n)}[u,\bsv,\bsv']=\sum_{\l^{(1)},\dots,\l^{(r)}}\prod_{\a=1}^r a_{\l^{(\a)}}\equskip\sum_{\mu^{(1)},\dots,\mu^{(r)}}\prod_{\b=1}^r a_{\mu^{(\b)}}\  \left(\dket{\l^{(1)}}\otimes\cdots\otimes\dket{\l^{(r)}}\right)\left(\dbra{\mu^{(1)}}\otimes\cdots\otimes\dbra{\mu^{(r)}}\right)\\
\otimes \Phi_{\l^{(1)}}[u_1,v_1]\Phi_{\mu^{(1)}}^\ast[u_{1},v'_{1}]\cdots\Phi_{\l^{(r)}}[u_r,v_r]\Phi_{\mu^{(r)}}^\ast[u_r,v'_r]
\end{multline}
The operator $M^{(r,n)}$ has been studied extensively (see e.g. \cite{Bourgine2017}), and it is known to satisfy the following intertwining relation,
\begin{center}
\begin{tikzpicture}[scale=1]
\node[ellipse,draw,scale=.7] (M1) at (0,0) {$M^{(r,n)}$};
\draw[postaction={on each segment={mid arrow=black}}] (-1,0) -- (M1);
\draw[postaction={on each segment={mid arrow=black}}] (0,-1) -- (M1);
\draw[postaction={on each segment={mid arrow=black}}] (M1) -- (0,1);
\draw[postaction={on each segment={mid arrow=black}}] (M1) -- (1,0);
\node[left,scale=.7] at (-1,0) {$\rho_u^{(1,n)}(e_{(2)})$};
\node[below,scale=.7] at (0,-1) {$\rho_\bsv^{(0,r)}(e_{(1)})$};
\node[scale=1] at (1.5,0) {$=$}; 
\def\offset{3}
\node[ellipse,draw,scale=.7] (M2) at (\offset+0,0) {$M^{(r,n)}$};
\draw[postaction={on each segment={mid arrow=black}}] (\offset-1,0) -- (M2);
\draw[postaction={on each segment={mid arrow=black}}] (\offset+0,-1) -- (M2);
\draw[postaction={on each segment={mid arrow=black}}] (M2) -- (\offset+0,1);
\draw[postaction={on each segment={mid arrow=black}}] (M2) -- (\offset+1,0);
\node[above,scale=.7] at (\offset,1) {$\rho_{\bsv'}^{(0,r)}(e_{(1)})$};
\node[above right,scale=.7] at (\offset+1,0) {$\rho_{u'}^{(1,n)}(e_{(2)})$};
\end{tikzpicture}
\end{center}
To deduce the intertwining property of the operator $A^{(r,n)}$, we decompose the vertical representation as $\rho_{\bsv}^{(0,r)}(e)=\rho_{v_1}^{(0,1)}(e_{(1)})\otimes\rho_{\tilde{\bsv}}^{(0,r-1)}(e_{(2)})$ and use the co-associativity to show that

\begin{center}
\includegraphics[width=15cm]{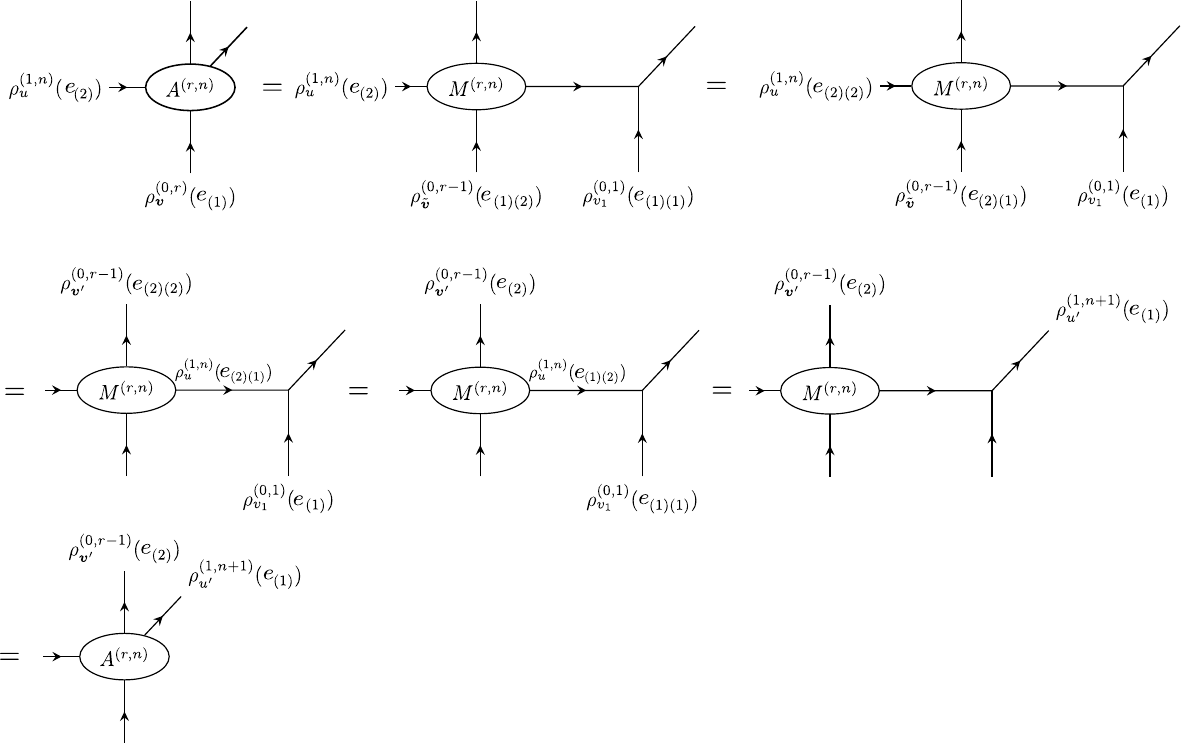}
\end{center}

The intertwiner relation can be proved recursively using a similar procedure,
\begin{center}
\includegraphics[width=17cm]{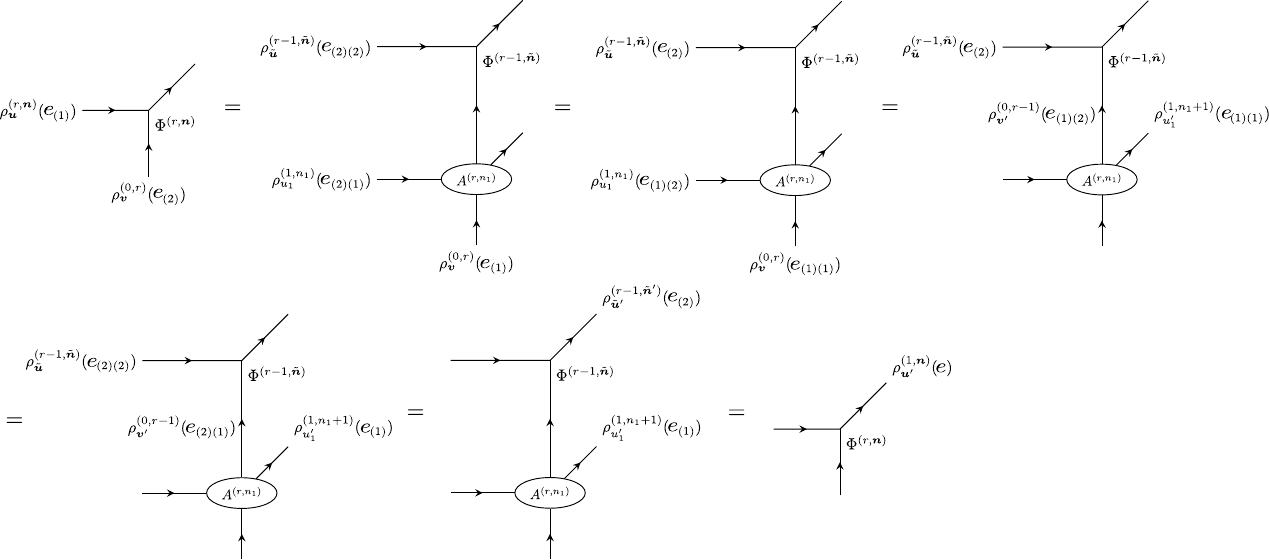}
\end{center}

\paragraph{Critical values.} Under certain specializations of the internal vertical weights $\bsw$, the sum over partitions $\l^{(\a,\b)}$ in \eqref{exp_Phi_rbsn} is restricted to finite sums over tuples of sub-partitions of $\bl=(\l^{(1)},\dots,\l^{(r)})$. We distinguish two cases corresponding respectively to (I) the restriction by the left $\l_{\a,\b}\subseteq\l_{\a+1,\b}$, and (II) the restriction by the right $\l_{\a,\b}\subseteq\l_{\a+1,\b+1}$. Let us start with the restriction (I), which follows from the presence of the diagrams
\begin{center}
\begin{tikzpicture}[scale=1.2]
\draw[postaction={on each segment={mid arrow=black}}] (-.7,-.7) -- (0,0) -- (2,0) -- (2.7,.7);
\draw[postaction={on each segment={mid arrow=black}}] (0,0) -- (0,1);
\draw[postaction={on each segment={mid arrow=black}}] (2,-1) -- (2,0);
\node[above left,scale=0.7] at (2,0) {$\Phi_{\a,\b}$};
\node[below right,scale=0.7] at (0,0) {$\Phi^{\ast}_{\a,\b}$};
\node[below left,scale=0.7] at (-.7,-.7) {$\CF^{(1,n_\a+1)}_{u'_{\a,\b+1}}$};
\node[above,scale=0.7] at (0,1) {$\CF^{(0,1)}_{w_{\a-1,\b}}$};
\node[below,scale=0.7] at (2,-1) {$\CF^{(0,1)}_{w_{\a,\b}}$};
\node[above right,scale=0.7] at (2.7,.7) {$\CF^{(1,n_\a+1)}_{u'_{\a,\b}}$};
\end{tikzpicture}
\end{center}
The vertical matrix elements of this operator is given by the following product of AFS vertex operators,
\begin{equation}
\Phi_{\l^{(\a,\b)}}^{(1,n_\a)}[u_{\a,\b},w_{\a,\b}]\Phi_{\l^{(\a-1,\b)}}^{(1,n_\a)\ast}[u_{\a,\b},w_{\a-1,\b}]=\dfrac{N_{\l^{(\a-1,\b)},\l^{(\a,\b)}}(\g w_{\a-1,\b}/w_{\a,\b})}{\CG(w_{\a-1,\b}/(\g w_{\a,\b}))}:\Phi_{\l^{(\a,\b)}}^{(1,n_\a)}[u_{\a,\b},w_{\a,\b}]\Phi_{\l^{(\a-1,\b)}}^{(1,n_\a)\ast}[u_{\a,\b},w_{\a-1,\b}]:
\end{equation}
Thus, if internal vertical weights take the value $w_{\a,\b}^{(I)}=\g^{r-\a}v_\b$, the argument of the Nekrasov factors becomes $\g w_{\a-1,\b}/w_{\a,\b}=q_3$, and the summation over internal partitions is restricted to $\l^{(\a-1,\b)}\subseteq\l^{(\a,\b)}$. As a result, if $\l^{(\a)}=\vac$ for $\a=1,\dots,r-1$, all internal partitions are frozen to $\l^{(\a,\b)}=\vac$, and the expression of the vertex operator factorizes,
\begin{multline}
\Phi_{\vac,\dots,\vac,\l^{(r)}}^{(r,\bsn)}[\bsu,\bsv,\bsw^{(I)}]\\
=t_{\l^{(r)}}^{(1,n_1)}[u_1,v_r] :\Phi_{\vac}^{(r,\bsn)}[\bsu,\bsv,\bsw^{(I)}]\,\mathe^{\sum_{k>0}\sum_{\sAbox\in\l^{(r)}}\frac{(v_r\chi_\sAbox)^k}{k}(1-q_1^k)p_k^{(1)}}\mathe^{-\sum_{k>0}\sum_{\sAbox\in\l^{(r)}}(v_r\chi_\sAbox)^{-k}(1-q_2^k)\frac{\p}{\p p_k^{(1)}}}:
\end{multline}
with
\begin{equation}
\Phi_{\vac,\dots,\vac}^{(r,\bsn)}[\bsu,\bsv,\bsw^{(I)}]=\CG(1)^{-r(r-1)/2}\mathe^{-\sum_{k>0}\frac1{k(1-q_2^k)}\sum_{\a=1}^r\left(v_{r-\a+1}^k+(1-q_3^k)\sum_{\b=\a+1}^{r}v_{r-\b+1}^k\right)p_k^{(\a)}}\mathe^{\sum_{k>0}\frac{q_3^{-k}}{1-q_1^k}\sum_{\a=1}^r v_{r-\a+1}^{-k}\frac{\p}{\p p_k^{(\a)}}}
\end{equation}

Another possibility is to consider the following diagrams,
\begin{center}
\begin{tikzpicture}[scale=1.2]
\draw[postaction={on each segment={mid arrow=black}}] (-1,0) -- (0,0) -- (.7,.7) -- (1.7,.7);
\draw[postaction={on each segment={mid arrow=black}}] (0,-1) -- (0,0);
\draw[postaction={on each segment={mid arrow=black}}] (.7,.7) -- (.7,1.7);
\node[below right,scale=0.7] at (0,0) {$\Phi_{\a,\b}$};
\node[left,scale=0.7] at (.7,.7) {$\Phi^{\ast}_{\a-1,\b}$};
\node[right,scale=0.7] at (1.7,.7) {$\CF^{(1,n_\a)}_{u_{\a,\b-1}}$};
\node[below,scale=0.7] at (0,-1) {$\CF^{(0,1)}_{w_{\a,\b}}$};
\node[above,scale=0.7] at (.7,1.7) {$\CF^{(0,1)}_{w_{\a-1,\b-1}}$};
\node[left,scale=0.7] at (-1,0) {$\CF^{(1,n_\a)}_{u_{\a,\b}}$};
\end{tikzpicture}
\end{center}
The vertical matrix elements of this operator are given by the following product of AFS vertex operators,
\begin{multline}
\Phi_{\l^{(\a,\b-1)}}^{(1,n_\a)\ast}[u_{\a,\b-1},w_{\a-1,\b-1}]\Phi_{\l^{(\a,\b)}}^{(1,n_\a)}[u_{\a,\b},w_{\a,\b}]\\=\dfrac{N_{\l^{(\a,\b)},\l^{(\a,\b-1)}}(\g w_{\a,\b}/w_{\a-1,\b-1})}{\CG(w_{\a,\b}/(\g w_{\a-1,\b-1}))}:\Phi_{\l^{(\a,\b-1)}}^{(1,n_\a)\ast}[u_{\a,\b-1},w_{\a-1,\b-1}]\Phi_{\l^{(\a,\b)}}^{(1,n_\a)}[u_{\a,\b},w_{\a,\b}]:.
\end{multline}
Just like before, if internal vertical weights take the values $w_{\a,\b}^{(II)}=\g^{r-\a}v_{r-\a+\b}$, the argument of the Nekrasov factors becomes $\g w_{\a,\b}/w_{\a-1,\b-1}=1$, and the summation over internal partitions is restricted to $\l^{(\a-1,\b-1)}\subseteq\l^{(\a,\b)}$. As a result, if $\l^{(\a)}=\vac$ for $\a=2,\dots,r$, all internal partitions are frozen to $\l^{(\a,\b)}=\vac$, and the expression of the vertex operator factorizes,
\begin{multline}
\Phi_{\l^{(1)},\vac,\dots,\vac}^{(r,\bsn)}[\bsu,\bsv,\bsw^{(II)}]=t_{\l^{(1)}}^{(1,n_1)}[u_1,v_1] \prod_{\b=2}^r\prod_{\sAbox\in\l^{(1)}}\dfrac{1-q_3v_\b/(v_1\chi_\sAbox)}{1-v_\b/(v_1\chi_\sAbox)}\\
\times:\Phi_{\vac}^{(r,\bsn)}[\bsu,\bsv,\bsw^{(II)}]\,\mathe^{\sum_{k>0}\sum_{\sAbox\in\l^{(1)}}\frac{(v_1\chi_\sAbox)^k}{k}(1-q_1^k)p_k^{(1)}}\mathe^{-\sum_{k>0}\sum_{\sAbox\in\l^{(r)}}(v_1\chi_\sAbox)^{-k}(1-q_2^k)\frac{\p}{\p p_k^{(1)}}}:
\end{multline}
with\footnote{The expression simplifies using the properties 
\be
\label{prop_Phi_vac}
:\Phi_\vac(v)\Phi_\vac^\ast(\g v):\,=\mathe^{-\sum_{k>0}\frac{v^k}{k}\frac{1-q_3^k}{1-q_2^k}J_{-k}},\hspace{30pt} :\Phi_\vac(v)\Phi_\vac^\ast(\g^{-1} v):\,=\mathe^{-\sum_{k>0}\frac{v^{-k}}{k}\frac{1-q_3^{-k}}{1-q_1^k}J_{k}}.
\ee}
\be
\begin{aligned}
\Phi_{\vac,\dots,\vac}^{(r,\bsn)}[\bsu,\bsv,\bsw^{(II)}]&=\G_{II}(\bsv)\,\mathe^{-\sum_{k>0}\frac1{k(1-q_2^k)}\sum_{\a=1}^r\left[v_\a^k+(1-q_3^k)\sum_{\b=\a+1}^{r}v_\b^k\right]p_k^{(\a)}}\mathe^{\sum_{k>0}\frac{q_3^{-k}}{1-q_1^k}\sum_{\a=1}^r v_\a^{-k}\frac{\p}{\p p_k^{(\a)}}}.\\
\G_{II}(\bsv)&=\CG(1)^{-r(r-1)/2}\prod_{\superp{\a,\b}{\a<\b}}\dfrac{\CG(v_\b/(q_3 v_\a))}{\CG(v_\b/v_\a)}.
\end{aligned}
\ee
Finally, the factors $\CG(1)^{-r(r-1)/2}$ in the case of $\bsw^{(I)}$, and $\G_{II}(\bsv)$ in the case of $\bsw^{(II)}$ are eliminated by choosing the normalization \eqref{norm_VO} of the vertex operators.

\subsection{Vertex operator \texorpdfstring{$\Phi^{(r,\bsn)\ast}$}{Phi(r,n)*}}
\begin{figure}
\begin{center}
\raisebox{-0.5\height}{\begin{tikzpicture}[scale=1]
\draw[postaction={on each segment={mid arrow=black}}] (-2.7,.3) -- (-2,1) -- (-1,1) -- (-.3,1.7) -- (.7,1.7) -- (1.4,2.4) -- (2.4,2.4);
\draw[postaction={on each segment={mid arrow=black}}] (-2.05,-1.05) -- (-1,0) -- (0,0) -- (0.7,.7) -- (2.4,.7);
\draw[postaction={on each segment={mid arrow=black}}] (-1.4,-2.4) -- (0,-1) -- (2.4,-1);
\draw[postaction={on each segment={mid arrow=black}}] (-1,0) -- (-1,1);
\draw[postaction={on each segment={mid arrow=black}}] (.7,.7) -- (.7,1.7);
\draw[postaction={on each segment={mid arrow=black}}] (0,-1) -- (0,0);
\draw[postaction={on each segment={mid arrow=black}}] (-2,1) -- (-2,3.4);
\draw[postaction={on each segment={mid arrow=black}}] (-.3,1.7) -- (-.3,3.4);
\draw[postaction={on each segment={mid arrow=black}}] (1.4,2.4) -- (1.4,3.4);
\node[above,scale=0.7] at (-2,3.4) {$\CF^{(0,1)}_{v_3}$};
\node[above,scale=0.7] at (-.3,3.4) {$\CF^{(0,1)}_{v_2}$};
\node[above,scale=0.7] at (1.4,3.4) {$\CF^{(0,1)}_{v_1}$};
\node[right,scale=0.7] at (2.4,-1) {$\CF^{(1,n_1)}_{u_1}$};
\node[right,scale=0.7] at (2.4,.7) {$\CF^{(1,n_2)}_{u_2}$};
\node[right,scale=0.7] at (2.4,2.4) {$\CF^{(1,n_3)}_{u_3}$};
\node[below left,scale=0.7] at (-1.4,-2.4) {$\CF^{(1,n_1+1)}_{u_1'}$};
\node[below left,scale=0.7] at (-2.05,-1.05) {$\CF^{(1,n_2+1)}_{u_2'}$};
\node[below left,scale=0.7] at (-2.7,.3) {$\CF^{(1,n_3+1)}_{u_3'}$};
\node[below right,scale=0.7] at (0,-1) {$\Phi_{1,1}^\ast$};
\node[above,scale=0.7] at (0,0) {$\Phi_{2,1}$};
\node[below right,scale=0.7] at (-1,0) {$\Phi_{2,2}^\ast$};
\node[below right,scale=0.7] at (.7,.7) {$\Phi_{2,1}^\ast$};
\node[above,scale=0.7] at (-1,1) {$\Phi_{3,2}$};
\node[above,scale=0.7] at (.7,1.7) {$\Phi_{3,1}$};
\node[below right,scale=0.7] at (-2,1) {$\Phi_{3,3}^\ast$};
\node[below right,scale=0.7] at (.7-1,1.7) {$\Phi_{3,2}^\ast$};
\node[below right,scale=0.7] at (1.4,2.4) {$\Phi_{3,1}^\ast$};
\end{tikzpicture}}
\hspace{30pt}
\raisebox{-0.5\height}{\begin{tikzpicture}[scale=1]
\draw[postaction={on each segment={mid arrow=black}}] (0,0) -- (1,0) -- (1.7,0.7);
\draw[postaction={on each segment={mid arrow=black}}] (1,-1) -- (1,0);
\node[above,scale=0.7] at (0.9,0) {$\Phi_{\a,\b}$};
\node[above,scale=0.7] at (0,0) {$\CF^{(1,n_\a)}_{u_{\a,\b}}$};
\node[right,scale=0.7] at (1,-0.5) {$\CF^{(0,1)}_{w_{\a-1,\b}}$};
\node[above,scale=0.7] at (1.7,0.7) {$\CF^{(1,n_\a+1)}_{u_{\a,\b}'}$};
\end{tikzpicture}}
\hspace{5mm}
\raisebox{-0.5\height}{\begin{tikzpicture}[scale=1]
\draw[postaction={on each segment={mid arrow=black}}] (-0.7,-0.7) -- (0,0) -- (0,1);
\draw[postaction={on each segment={mid arrow=black}}] (0,0) -- (1,0);
\node[left,scale=0.7] at (0,0) {$\Phi^{\ast}_{\a,\b}$};
\node[below,scale=0.7] at (1,0) {$\CF^{(1,n_\a)}_{u_{\a,\b-1}}$};
\node[right,scale=0.7] at (0,0.5) {$\CF^{(0,1)}_{w_{\a,\b}}$};
\node[below,scale=0.7] at (-0.7,-0.7) {$\CF^{(1,n_\a+1)}_{u_{\a,\b}'}$};
\end{tikzpicture}}
\end{center}
\caption{Diagram associated to $\Phi^{(r,\bsn)\ast}$ showing the choice of labels for weights and intermediate vertex operators (right)}
\label{Fig_def_Phi_ast}
\end{figure}

We take the conventions of Figure~\ref{Fig_def_Phi_ast} for the labeling of AFS vertex operators and internal weights. With these conventions, that slightly differ from those used for $\Phi_\bl^{(r,\bsn)}$, we have the internal weights $w_{\a,\b}$ with $\a=1,\dots,r$ and $\b=1,\dots,\a$, and such that $v_\a=w_{r,\a}$. For the horizontal weights, we have $u_\a=u_{\a,0}$ and $u'_\a=u'_{\a,\a}$, with
\begin{equation}
 u'_{\a,\b}
 =-\g u_{\a,\b}w_{\a-1,\b}
 =-\g u_{\a,\b-1}w_{\a,\b}\implies
 u_{\a,\b}
 =u_\a\prod_{k=1}^\b\dfrac{w_{\a,k}}{w_{\a-1,k}},
\end{equation}
which fixes all the remaining weights and leads to the relations \eqref{rel_weights}.

The vertical component of the vertex operator is expressed as a sum over the realizations of $r(r-1)/2$ partitions $\l^{(\a,\b)}$ with $1\leq \a\leq r-1$, $1\leq \b\leq \a$ of vertex operators acting on the tensor product of $r$ bosonic Fock spaces,
\be
 \Phi_\bl^{(r,\bsn)\ast}[\bsu,\bsv,\bsw]=\sum_{\{\l^{(\a,\b)}\}}\prod_{\a,\b}a_{\l^{(\a,\b)}}\
 \Phi_{1,1}^\ast\otimes\Phi_{2,1}^\ast\Phi_{2,1}\Phi_{2,2}^\ast\otimes\cdots\otimes
 \Phi_{r,1}^\ast\Phi_{r,1}\cdots\Phi_{r,r-1}^\ast\Phi_{r,r-1}\Phi_{r,r}^\ast,
\ee
where we used the following shortcut notations for the AFS vertex operators,
\be
 \Phi_{\a+1,\b} = \Phi_{\l^{(\a,\b)}}^{(1,n_{\a+1})}[u_{\a+1,\b},w_{\a,\b}],
 \hspace{30pt}
 \Phi_{\a,\b}^\ast = \Phi_{\l^{(\a,\b)}}^{(1,n_\a)\ast}[u_{\a,\b-1},w_{\a,\b}],
\ee
and the understanding that $\l^{(r,\a)}=\l^{(\a)}$ are the fixed partitions of the external vertical modules, $\bl=(\l^{(1)},\dots,\l^{(r)})$. 

\paragraph{Critical values.} Taking the normal-ordering of AFS vertex operators, we observe that there are two sets of critical values for the internal weights that will restrict the summation over the partitions $\l^{(\a,\b)}$, namely $w_{\a,\b}=\g w_{\a+1,\b}$, in which case $\l_{\a,\b}\subseteq\l_{\a+1,\b}$, and $w_{\a,\b}=\g w_{\a+1,\b+1}$, in which case $\l_{\a,\b}\subseteq \l_{\a+1,\b+1}$. We will denote these two sets of critical values by $\bsw^{(I)}$ and $\bsw^{(II)}$, respectively, where we have $w^{(I)}_{\a,\b}=\g^{r-\a}v_\b$ and $w^{(II)}_{\a,\b}=\g^{r-\a}v_{r+\b-\a}$.

Let us consider
\begin{multline}
\Phi_{\vac,\dots,\vac,\l^{(r)}}^{(r,\bsn)\ast}[\bsu,\bsv,\bsw^{(I)}]=\Phi_\vac^{\ast}(\g^{r-1}v_1)\otimes \Phi_\vac^\ast(\g^{r-2}v_1)\Phi_\vac(\g^{r-1}v_1)\Phi_\vac^\ast(\g^{r-2}v_2)\otimes\cdots\\
\cdots\otimes\Phi_\vac^\ast(v_1)\Phi_\vac(\g v_1)\cdots\Phi_\vac^\ast(v_{r-1})\Phi_\vac(\g v_{r-1})\Phi_{\l^{(r)}}^{(1,n_r)\ast}[u_{r,r-1},v_r].
\end{multline}
After normal-ordering, we find
\begin{multline}
 \Phi_{\vac,\dots,\vac,\l^{(r)}}^{(r,\bsn)\ast}[\bsu,\bsv,\bsw^{(I)}]
 = t_{\l^{(r)}}^{(1,n_r)}[u_{r,r-1},v_r]
 \prod_{\a=1}^{r-1}\prod_{\sAbox\in\l^{(r)}}\dfrac{1-q_3^{-1}v_r\chi_\sAbox/v_\a}
 {1-v_r\chi_\sAbox/v_\a}:\Phi_{\vac}^{(r,\bsn)\ast}[\bsv,\bsw^{(I)}]\\
 \times \mathe^{-\sum_{k>0}\g^{-rk}(1-q_1^k)\sum_{\sAbox\in\l^{(r)}}
 \frac{(q_3 v_r\chi_\sAbox)^k}{k}p_k(\bsx^{(r)})}
 \mathe^{\sum_{k>0}\g^{rk}(1-q_2^k)\sum_{\sAbox\in\l^{(r)}}(v_r\chi_\sAbox)^{-k}
 \frac{\p}{\p p_k(\bsx^{(r)})}}:
\end{multline}
with
\be
 \Phi_{\vac,\dots,\vac}^{(r,\bsn)\ast}[\bsv,\bsw^{(I)}]
 =\G_I(\bsv)\, \mathe^{\sum_{k>0}\frac{\g^{rk}}{k(1-q_2^k)}
 \sum_{\a=1}^r q_3^{-(\a-1)k} v_\a^k p_k^{(\a)}}
 \mathe^{-\sum_{k>0}\frac{\g^{-rk}}{1-q_1^k}\sum_{\a=1}^r q_3^{(\a-1)k}
 \left(v_\a^{-k}+(1-q_3^{-k})\sum_{\b=1}^{\a-1}v_\b^{-k}\right)\frac{\p}{\p p_k^{(\a)}}}.
\ee
The coefficient $\G_I(\bsv)$ arises from the normal-ordering of the vacuum components of the vertex operators, it is given by
\be
 \G_I(\bsv) = \CG(1)^{-r(r-1)/2}
 \prod_{\superp{\a,\b=1}{\a<\b}}^r\dfrac{\CG(v_\b/v_\a)}{\CG(v_\b/(q_3v_\a))}
\ee
In the same way, we have
\begin{multline}
\Phi_{\l^{(1)},\vac,\dots,\vac}^{(r,\bsn)\ast}[\bsu,\bsv,\bsw^{(II)}]=\Phi_\vac^{\ast}(\g^{r-1}v_r)\otimes \Phi_\vac^\ast(\g^{r-2}v_{r-1})\Phi_\vac(\g^{r-1}v_r)\Phi_\vac^\ast(\g^{r-2}v_r)\otimes\cdots\\
\cdots\otimes\Phi_{\l^{(1)}}^{(1,n_r)\ast}[u_1,v_1]\Phi_\vac(\g v_{2})\Phi_\vac^\ast(v_2)\cdots\Phi_\vac(\g v_r)\Phi_\vac^\ast(v_r),
\end{multline}
and so, after normal-ordering,
\begin{multline}
\Phi_{\l^{(1)},\vac,\dots,\vac}^{(r,\bsn)\ast}[\bsu,\bsv,\bsw^{(II)}]
 = t_{\l^{(1)}}^{(1,n_1)\ast}[u_1,v_1]:\Phi_{\vac}^{(r,\bsn)\ast}[\bsv,\bsw^{(II)}]\\
\times \mathe^{-\sum_{k>0}\g^{-rk}(1-q_1^k)\sum_{\sAbox\in\l^{(1)}}\frac{(q_3 v_1\chi_\sAbox)^k}{k}p_k^{(r)}}\mathe^{\sum_{k>0}\g^{rk}(1-q_2^k)\sum_{\sAbox\in\l^{(1)}}(v_1\chi_\sAbox)^{-k}\frac{\p}{\p p_k^{(r)}}}:
\end{multline}
with
\begin{multline}
 \Phi_{\vac,\dots,\vac}^{(r,\bsn)\ast}[\bsv,\bsw^{(II)}] = \G_{II}(\bsv)\,
 \mathe^{\sum_{k>0}\frac{\g^{rk}}{k(1-q_2^k)}
 \sum_{\a=1}^rq_3^{-(\a-1)k}v_{r-\a+1}^k p_k^{(\a)}} \\
 \times\mathe^{-\sum_{k>0}\frac{\g^{-rk}}{1-q_1^k}\sum_{\a=1}^r q_3^{(\a-1)k}
 \left(v_{r-\a+1}^{-k}+(1-q_3^{-k})\sum_{\b=1}^{\a-1}v_{r-\b+1}^{-k}\right)
 \frac{\p}{\p p_k^{(\a)}}}.
\end{multline}
We note that the vacuum component $\Phi_{\vac}^{(r,\bsn)\ast}[\bsv,\bsw^{(II)}]$ corresponds to $\Phi_{\vac}^{(r,\bsn)\ast}[\bsv,\bsw^{(I)}]$ with the replacement $v_\a\to v_{r-\a+1}$. We find $\G_{II}(\bsv)=\CG(q_3^{-1})^{-r(r-1)/2}$. Like before, the normalization factors $\G_I(\bsv)$ and $\G_{II}(\bsv)$ are eliminated by choosing the normalization \eqref{norm_VO} for the vertex operators.

\section{Generalized Macdonald polynomials with two variables}\label{app_finite}
In this section, we study the GMP at level $r=2$, with only two variables. To do so, we restrict ourselves to partitions of length one, and so $\l$, $\mu$ will be positive integer throughout this section.

\paragraph{Macdonald operator.} The derivation of the Macdonald operator from the zero mode $x_0^+$ follows from the deformation of the integration contour,
\footnote{Once again we omit to indicate the representation $\rho_1^{(1,0)}$ here.}
\be
 x^+_0 = \oint_0\dfrac{\mathd z}{2\pi\mathi z} \eta^+(z).
\ee
Let $\pi_N:\Lambda\to\Lambda_N$ be the restriction map from infinitely many variables $\bsx$ to finitely many variables $\{x_i\}_{i=1}^N$. Then, for $f\in\Lambda$, we have
\be
\begin{aligned}
 \pi_N(x_0^+f(\bsx))
 &= \pi_N\left(\oint_0\dfrac{\mathd z}{2\pi\mathi z}\mathe^{\sum_{k>0}\frac{z^k}{k}(1-t^{-k})p_k(\bsx)}\mathe^{-\sum_{k>0}z^{-k}(1-q^k)\frac{\p}{\p p_k(\bsx)}} f(\bsx) \right)\\
 &= \oint_0\dfrac{\mathd z}{2\pi\mathi z}\pi_N\left(\mathe^{\sum_{k>0}\frac{z^k}{k}(1-t^{-k})p_k(\bsx)}f(\bsx-(1-q)z^{-1}) \right)\\
 &= \oint_0\dfrac{\mathd z}{2\pi\mathi z}
 \prod_{i=1}^N\frac{1-t^{-1}zx_i}{1-zx_i}
 \pi_N(f(\bsx-(1-q)z^{-1}))\\
 &=-\oint_\infty\dfrac{\mathd z}{2\pi\mathi z} \prod_{i=1}^N \dfrac{1-t^{-1}zx_i}{1-zx_i}
 \pi_N(f(\bsx-(1-q)z^{-1}))\\
 &\quad+(1-t^{-1})\sum_{i=1}^N \prod_{j\neq i}\dfrac{x_i-t^{-1}x_j}{x_i-x_j}
 \pi_N(f(\bsx-(1-q)x_i))\\
 &=t^{-N}\pi_N(f(\bsx))
 -t^{-N}(1-t)\sum_{i=1}^N \prod_{j\neq i}\dfrac{tx_i-x_j}{x_i-x_j}
 q^{x_i\p_{x_i}}\pi_N(f(\bsx)) \\
 &=t^{-N}\left(1-(1-t)D^1_N\right)\pi_N(f(\bsx))
\end{aligned}
\ee
where we deformed the contour from being around $z=0$ to the sum of contours around $z=\infty$ and $z=x_i^{-1}$, $i=1,\dots,N$. Moreover, we used the identity $\pi_N(f(\bsx-(1-q)x_i))=q^{x_i\p_{x_i}}\pi_N(f(\bsx))$. Here $D^1_N$ is the standard Macdonald operator on $\Lambda_N$,
\be
 D^1_N=\sum_{i=1}^N \prod_{j\neq i}\frac{tx_i-x_j}{x_i-x_j}q^{x_i\p_{x_i}},
\ee
diagonal on Macdonald polynomials $P_\l(x_1,\dots,x_N)$ with eigenvalues $\frac{1-t^N\me_\l}{1-t}=\sum_{i=1}^N q^{\l_i}t^{N-i}$.

Repeating this operation for the operator at level two, we find the following expression for the operator $X_0^+$ corresponding to the finite variable version of $u_2^{-1}\rho_{u_1,u_2}^{(2,\bsn_0)}(x_0^+)$, 
\begin{align}
\begin{split}
X_0^+&=(Qt^{-N_1}+t^{-N_2})+(1-t^{-1})Q\sum_{i=1}^{N_1}\prod_{j\neq i}\dfrac{x_i-t^{-1}x_j}{x_i-x_j}\ q^{x_i\p_{x_i}}\\
&+(1-t^{-1})\sum_{i=1}^{N_2}\prod_{j\neq i}\dfrac{y_i-t^{-1}y_j}{y_i-y_j}\times\prod_{j=1}^{N_1}\dfrac{(y_i-t^{-1}x_j)(y_i-tq^{-1}x_j)}{(y_i-x_j)(y_i-q^{-1}x_j)}\ q^{y_i\p_{y_i}}\\
&+\k_{13}\sum_{i=1}^{N_1}\prod_{j\neq i}\dfrac{(x_i-t^{-1}x_j)(x_i-tq^{-1}x_j)}{(x_i-x_j)(x_i-q^{-1}x_j)}\times\prod_{j=1}^{N_2}\dfrac{x_i-t^{-1}y_j}{x_i-y_j}\ V_\bsy(x_i)\\
&-\k_{13}\sum_{i=1}^{N_1}\prod_{j\neq i}\dfrac{(x_i-tx_j)(x_i-qt^{-1}x_j)}{(x_i-x_j)(x_i-qx_j)}\times\prod_{j=1}^{N_2}\dfrac{x_i-qt^{-1}y_j}{x_i-qy_j} V_\bsy(q^{-1}x_i).
\end{split}
\end{align}
with $\k_{13}=-(1-t)(1-q/t)/(1-q)$. In this expression, we introduced the following operator, at the moment only well defined in the limit $N_2\to\infty$,
\begin{equation}
V_\bsy(z)=\mathe^{-\sum_{k>0} z^k(1-q^k)\frac{\p}{\p p_k(\bsy)}}.
\end{equation} 
Finding the action of this operator on symmetric polynomials with finitely many variables requires some careful analysis. Restricting ourselves to the case $N_1=N_2=1$, the expression of the operator $x_0^+$ simplifies to
\begin{align}
\begin{split}
X_0^+&=(1+Q)t^{-1}+(1-t^{-1})\left(Qq^{x\p_{x}}+\dfrac{(y-t^{-1}x)(y-tq^{-1}x)}{(y-x)(y-q^{-1}x)}\ q^{y\p_{y}}\right)\\
&+\k_{13}\left(\dfrac{x-t^{-1}y}{x-y}V_y(x)-\dfrac{x-qt^{-1}y}{x-qy}V_y(q^{-1}x)\right),
\end{split}
\end{align}
where the action of the operator $V_y(x)$ remains to be specified.

In order to determine the action of $V_y(x)$ on the basis elements, we first restrict ourselves to partitions of length one, but keep an infinite number of variables. For a partition with a single column $\l=[n]$, the Macdonald polynomial is
\begin{equation}
P_{[n]}(\bsx)=\dfrac{(q;q)_n}{(t,q)_n}g_n(\bsx),
\end{equation} 
where the symmetric functions $g_n(\bsx)$, depending on $(q,t)$, are the coefficient $z^n$ in the expansion \cite[Ch.VI, \textsection2, (2.8)]{Macdonald},
\begin{equation}
\Pi(\bsx|z)=\mathe^{\sum_{k>0}\frac{z^k}{k}\frac{1-t^k}{1-q^k}p_k(\bsx)}=\sum_{n\geq0}z^n g_n(\bsx).
\end{equation}
Acting on this generating function with the operator $V_{\bsy}(x)$, we find
\begin{equation}
V_{\bsy}(x)\Pi(\bsy|z)=\dfrac{1-xz}{1-txz}\Pi(\bsy|z).
\end{equation} 
Expanding both sides in powers of $z$, we deduce the action of $V_\bsy(x)$ on the symmetric functions $g_n(\bsy)$,
\begin{equation}
V_{\bsy}(x)g_n(\bsy)=g_n(\bsy)-(1-t)\sum_{k=0}^{n-1} t^{n-1-k}x^{n-k}g_{k}(\bsy).
\end{equation} 
We note that we have a triangular action. Coming back to Macdonald symmetric functions, we find
\begin{equation}
V_{\bsy}(x)P_{[n]}(\bsy)=P_{[n]}(\bsy)-(1-t)\sum_{k=0}^{n-1} t^{n-1-k}x^{n-k}\dfrac{(q^{k+1};q)_{n-k}}{(tq^k;q)_{n-k}}P_{[k]}(\bsy).
\end{equation}
At this stage, we can finally restrict ourselves to a single variable,\footnote{The expression for $g_n(y)$ follows from the expansion of the kernel, using the $q$-binomial identity
\begin{equation}
\Pi(y|z)=\mathe^{\sum_{k>0}\frac{(yz)^k}{k}\frac{1-t^k}{1-q^k}}=\dfrac{(tyz;q)_\infty}{(yz;q)_\infty}=\sum_{n=0}^\infty(yz)^n\dfrac{(t;q)_n}{(q;q)_n}.
\end{equation}}
\begin{equation}\label{expr_gn}
g_n(y)=y^n\dfrac{(t;q)_n}{(q;q)_n}\implies P_{[n]}(y)=y^n,
\end{equation} 
and so, the operator $V_{y}(x)$ acts as follows on the basis $y^n$ of $\mC[y]$
\begin{equation}
V_y(x)y^n=y^n-(1-t)\sum_{k=0}^{n-1} t^{n-1-k}x^{n-k}\dfrac{(q^{k+1};q)_{n-k}}{(tq^k;q)_{n-k}}y^k.
\end{equation}

To derive the action of the generalized Macdonald operator $X_0^+$ on the basis of $\mC[x]\otimes\mC[y]$, it is easier to first consider the action on $x^n\Pi(y|z)$.
Using
\begin{equation}
 V_{y}(x)\Pi(y|z)=\dfrac{1-xz}{1-txz}\Pi(y|z),\hspace{30pt}
 \Pi(qy|z)=\dfrac{1-yz}{1-tyz}\Pi(y|z),
\end{equation}
we find after simplification
\begin{equation}
X_0^+ x^n\Pi(y|z)=\left[Q\me_{[n]}+t^{-1}+(1-t^{-1})\dfrac{(1-xz)(1-yz)(1-t^2q^{-1}xz)}{(1-txz)(1-tyz)(1-tq^{-1}xz)}\right] x^n\Pi(y|z),
\end{equation}
with $\me_{[n]}=t^{-1}+(1-t^{-1})q^n$. Observe that all poles of the form $x-\a y$ disappear, and so the r.h.s. is indeed a series in $z$ of polynomials in the variables $x,y$. Expanding both sides in $z$, and taking into account the $q$-Pochhammer factor in \eqref{expr_gn}, we end up with
\begin{align}
\begin{split}
&X_0^+x^m y^n=(Q\me_{[n]}+\me_{[m]})x^m y^n+\sum_{k=1}^n\a_{n,k} x^{m+k}y^{n-k},\\
\text{with}\quad &\a_{n,k}=\k_{13} (1-t^{-1})(1-q^{-k})q^{n-k}t^{k}\dfrac{(q^{n-k+1};q)_k}{(tq^{n-k};q)_k}.
\end{split}
\end{align}
We notice that the coefficients $\a_{n,k}$ are independent of $m$. The operator $X_0^+$ leaves invariant the $n+1$ dimensional subspace $V_{m,n}$ of $\mC[x]\otimes\mC[y]$ generated by the monomials $(x^{n+m},x^{n+m-1}y^1,\dots,x^{m}y^n)$, and acts as an upper triangular matrix on this basis
\begin{equation}
\left.X_0^+\right|_{V_{m,n}}=
\begin{pmatrix}
Q\me_{[n+m]}+\me_\vac & \a_{1,1} & \a_{2,2}  &\cdots & \a_{n,n}\\
0 & Q\me_{[n+m-1]}+\me_{[1]} & \a_{2,1}  &\cdots & \a_{n,n-1}\\
0 & 0 & Q\me_{[n+m-2]}+\me_{[2]}  & & \a_{n,n-2}\\
\vdots & \vdots   &  & \ddots  & \vdots\\
0 & 0 & \cdots  & Q\me_{[m+1]}+\me_{[n-1]} & \a_{n,1}\\
0 & 0 & \cdots  & 0 & Q\me_{[m]}+\me_{[n]}
\end{pmatrix}
\end{equation}
We observe that only the matrix elements on the diagonal depend on $Q$.

\paragraph{Examples.} The subspaces $V_{m,0}$ are one dimensional, which corresponds to the eigenvectors $P_{[m],\vac}(x,y|Q)=x^m$ of $X_0^+$. The subspaces $V_{m,1}$ are two dimensional, and we find the pair of eigenvectors
\begin{equation}
P_{[m+1],\vac}(x,y|Q)=x^{m+1},\quad P_{[m],[1]}(x,y|Q)=x^m y+\dfrac{1-tq^{-1}}{1-Qq^{m}}x^{m+1}.
\end{equation} 
The first eigenvector also belongs to the subspace $V_{m+1,0}$ and so it was already known, but the second one is new. Next, subspaces $V_{m,2}$ are three dimensional, the first two eigenvectors belong respectively to $V_{m+2,0}$ and $V_{m+1,1}$, but the third one is new,
\begin{align}
\begin{split}
 P_{[m],[2]}(x,y|Q) &= x^m y^2
 +(1-t)(1-tq^{-1})\dfrac{(1+q)}{(1-qt)(1-q^{m-1}Q)}x^{m+1}y \\
 &+(1-tq^{-1})\dfrac{1-t-tq+t^2q^{-1}+tq^m(1-q^{-2})Q}{(1-qt)(1-q^m Q)(1-q^{m-1}Q)}x^{m+2}.
\end{split}
\end{align}
It is instructive to compare these expressions to the generalized Macdonald polynomials with infinite number of variables $P_{\l,\mu}(\bsx,\bsy|Q)$. In the case of partitions with single columns $\l=[m]$, $\mu=[n]$, the GMP reduce indeed to the expressions of $P_{[m],[n]}(x,y|Q)$ found above when specialized to $\bsx=(x,0,\dots)$ and $\bsy=(y,0,\dots)$. Note, however, that the property of Macdonald functions $P_\l(\bsx)=0$ for $\ell(\l)>1$ and $\bsx=(x,0,\dots)$ is no longer true for generalized Macdonald functions. For instance, the specialization of $P_{\vac,[1^2]}(\bsx,\bsy|Q)$ produces a linear combination of $P_{[2],\vac}(x,y|Q)$ and $P_{[1],[1]}(x,y|Q)$.

\section{Pieri rules from Mukad\'e operator}\label{app:Mukade}

The Mukad\'e operator has been introduced in \cite{Fukuda:2019ywe} in the proof of the invariance under the choice of preferred direction for refined topological strings amplitudes. The starting point of the construction is the operator
\begin{equation}
\mukT[\bsu,\bsv,w]:\CF_{w}^{(0,1)}\otimes \CF^{(r,0)}_\bsu\to \CF_{w'}^{(0,1)}\otimes \CF_{\bsv}^{(r,0)},\hspace{30pt} w'=w\prod_{\a=1}^r\dfrac{u_\a}{v_\a}.
\end{equation} 
defined as a product of AFS vertex operators,
\begin{equation}
\mukT[\bsu,\bsv,w]=\sum_{\{\nu^{(\a)}\}_{\a=0}^r}\prod_{\a=0}^{r} a_{\nu^{(\a)}}\ \dket{\nu^{(r)}}\dbra{\nu^{(0)}}\otimes \Phi_{\nu^{(1)}}^{(1,0)\ast}[v_1,w'_1]\Phi_{\nu^{(0)}}^{(1,0)}[u_1,w_1]\otimes\cdots\otimes  \Phi_{\nu^{(r)}}^{(1,0)\ast}[v_r,w'_r]\Phi_{\nu^{(r-1)}}^{(1,0)}[u_r,w_r],
\end{equation} 
where we used the notation $w'_\a=w_{\a+1}=w_\a(u_\a/v_\a)$ for the intermediate weights obtained inductively from $w_1=w$. The summation is done over the realizations of $(r+1)$ partitions $\nu^{(\a)}$, $\a=0,\dots,r+1$ corresponding to the two external vertical states and the trace over $r-1$ internal vertical modules. This operator is associated to the diagram in Figure~\ref{fig:mukade},
\begin{figure}[!ht]
\begin{center}
\begin{tikzpicture}[scale=1, font=\small]
\def\xN{2.4}
\def\yN{6}
\node[scale=1] at (2,4.2) {$\vdots$};
\draw[postaction={on each segment={mid arrow=black}}] (-1,0) -- (0,0) -- (0.7,0.7) -- (1.7,0.7);
\draw[postaction={on each segment={mid arrow=black}}] (-.3,1.7) -- (.7,1.7) -- (1.4,2.4) -- (2.4,2.4);
\draw[postaction={on each segment={mid arrow=black}}] (0,-1) -- (0,0);
\draw[postaction={on each segment={mid arrow=black}}] (.7,.7) -- (.7,1.7);
\draw[postaction={on each segment={mid arrow=black}}] (1.4,2.4) -- (1.4,3.4);
\draw[postaction={on each segment={mid arrow=black}}] (\xN-1,\yN) -- (\xN,\yN) -- (\xN+.7,\yN+.7) -- (\xN+1.7,\yN+.7);
\draw[postaction={on each segment={mid arrow=black}}] (\xN,\yN-1) -- (\xN,\yN);
\draw[postaction={on each segment={mid arrow=black}}] (\xN+.7,\yN+.7) -- (\xN+.7,\yN+1.7);
\node[left] at (-1,0) {$u_1$};
\node[left] at (-.3,1.7) {$u_2$};
\node[left] at (\xN-1,\yN) {$u_r$};
\node[right] at (1.7,.7) {$v_1$};
\node[right] at (2.4,2.4) {$v_2$};
\node[right] at (\xN+1.7,\yN+.7) {$v_r$};
\node[below] at (0,-1) {$w_1$};
\node[above] at (\xN+.7,\yN+1.7) {$w'_r$};
\node[right] at (.7,.7+.5) {$w_2=w'_1$};
\node[right] at (1.4,2.9) {$w_3=w'_2$};
\node[right] at (\xN,\yN-.5) {$w_r$};
\end{tikzpicture}
\end{center}
\caption{Ladder diagram corresponding to the intertwiner $\mukT[\bsu,\bsv,w]$.}
\label{fig:mukade}
\end{figure}
which, by construction, is an intertwiner satisfying the relation
\be\label{intw_CT}
 \mukT[\bsu,\bsv,w]\left(\rho_w^{(0,1)}\otimes\rho_\bsu^{(r,\bsn_0)}\right)\D(e)
 = \left(\rho_{w'}^{(0,1)}\otimes\rho_\bsv^{(r,0)}\right)\D'(e)
 \mukT[\bsu,\bsv,w],\quad \forall e\in\CE.
\ee 
The proof follows from the same arguments as those developed in the previous section, they will not be reproduced here.

The Mukad\'e operator $\mukT^V[\bsu,\bsv,w]:\CF_{\bsu}^{(r,0)}\to\CF_{\bsv}^{(r,0)}$ is obtained as the (normalized) vacuum matrix element of the operator $\mukT[\bsu,\bsv,w]$ in the vertical module, i.e.
\be
 \mukT^V[\bsu,\bsv,w] :=
 \dfrac{\left(\dbra{\vac}\otimes\mathrm{Id}\right)
 \mukT[\bsu,\bsv,w]\left(\dket{\vac}\otimes\mathrm{Id}\right)}
 {\left(\dbra{\vac}\otimes\bra{\vec{\vac}}\right)
 \mukT[\bsu,\bsv,w]\left(\dket{\vac}\otimes\ket{\vec{\vac}}\right)}
\ee
where $\dket{\vac}$ is the vacuum state in $\CF_{w}^{(0,1)}$ and $\ket{\vec{\vac}}$ is the vacuum in $\CF_{\bsu}^{(r,0)}$, and similarly for their duals. Since we are only interested in the projection of the external vertical legs to the vacuum, the result is independent of the choice of normalization of vertical states.

One of the key results in \cite{Fukuda:2019ywe} is the computation of the matrix elements of this operator in the generalized Macdonald basis.
Let us denote as $\ket{P_\blam(\bsu)}\in\CF^{(r,0)}_\bsu$ the vector in the Fock module corresponding to the multi-symmetric function $P_\blam(\bsx^\bullet|\bsu)\in\Lambda^{\otimes r}$, and $\bra{P_\blam(\bsu)}$ the dual vector corresponding to the linear map $\left\langle P_{\l^{(r)},\dots,\l^{(1)}}(\bsx^\bullet|u_r,\dots,u_1),-\right\rangle_\mathrm{Z}$. Then we have the following.
\begin{theorem}[Fukuda--Ohkubo--Shiraishi, \cite{Fukuda:2019ywe}]
The matrix elements of the Mukad\'e operator are\footnote{We inserted in the r.h.s.\ the extra factor $\g^{-r|\bl|}\g^{\sum_\a(\a-1)(|\mu^{(\a)}|-|\l^{(\a)}|)}$ that appeared to be missing in the original formula. It might be due to a different choice of conventions.}
\label{thm:FOS}
\begin{multline}
\label{eq:matrix-element-mukade}
 \frac{\bra{P_{\bl}(\bsv)} \mukT^V[\bsu,\bsv,w] \ket{P_{\bmu}(\bsu)}}
 {\langle P_{\bl}(\bsv)| P_{\bl}(\bsv)\rangle}
 = \g^{\sum_\a(\a-1)(|\mu^{(\a)}|-|\l^{(\a)}|)}
 \dfrac{w^{|\bl|-|\bmu|}q_3^{-|\bmu|}}
 {\CC_{\bl}^{(-)}(\bsv)\,\CC_{\bmu}^{(+)}(\bsu)} \\
 \times \Big[\prod_{\a=1}^r b_{\l^{(\a)}} v_\a^{|\l^{(\a)}|}u_\a^{|\mu^{(\a)}|}g_{\l^{(\a)}}g_{\mu^{(\a)}}\Big]
 \prod_{\a,\b=1}^r\tN_{\l^{(\a)},\mu^{(\b)}}(\g v_\a/u_\b).
\end{multline}
with $g_\l$ defined as in \eqref{expr_gl}, and
\be
\begin{aligned}
 \CC_{\bmu}^{(+)}(\bsu) &= t^{-|\bmu|}
 \left[ \prod_{\a=1}^r \g^{(\a-1)|\mu^{(\a)}|}
 \frac{u_\a^{|\mu^{(\a)}|}g_{\mu^{(\a)}}}{P_{\mu^{(\a)}}(\sp_\vac)}\right]
 \times
 \prod_{\superp{\a,\b=1}{\a<\b}}^r \tN_{\mu^{(\a)},\mu^{(\b)}}(u_\a/u_\b),\\
 \CC_{\bl}^{(-)}(\bsv) &=
 \left[\prod_{\a=1}^r \g^{-(\a-1)|\l^{(\a)}|}
 \frac{v_\a^{|\l^{(\a)}|}g_{\l^{(\a)}}}{P_{\l^{(\a)}}(\sp_\vac)}\right] \times
 \prod_{\superp{\a,\b=1}{\a>\b}}^r \tN_{\l^{(\a)},\l^{(\b)}}(v_\a/v_\b).
\end{aligned}
\ee
\end{theorem}
It is also convenient to rewrite the matrix element in the $\tP_\bl$ basis, so that
\be
\label{eq:matrix-element-mukade-tilde}
 \frac{\bra{\tP_{\bl}(\bsv)} \mukT^V[\bsu,\bsv,w] \ket{\tP_{\bmu}(\bsu)}}
 {\langle \tP_{\bl}(\bsv)| \tP_{\bl}(\bsv)\rangle}
 = q^{|\bmu|}w^{|\bl|-|\bmu|}
 \frac{\prod_{\a=1}^r \tilde{b}_{\l^{(\a)}}
 \prod_{\a,\b=1}^r\tN_{\l^{(\a)},\mu^{(\b)}}(\g v_\a/u_\b)}
 { \prod_{\a>\b} \tN_{\l^{(\a)},\l^{(\b)}}(v_\a/v_\b)
 \prod_{\a<\b} \tN_{\mu^{(\a)},\mu^{(\b)}}(u_\a/u_\b) }.
\ee

In order to derive Pieri rules for generalized Macdonald functions, we first need to prove the following specialization formulas for the Mukad\'e operator.
\begin{lemma}
\label{lem:mukade}
Under the specialization of weights $\bsv=\gamma^{\mp1}\bsu$, we have
\be
\label{sp_mukV}
 \mukT^V[\bsu,\g^{\mp1}\bsu,w]
 = \rho_\bsu^{(r,\bsn_0)}\left(\mathe^{\pm\sum_{k>0}\frac{\g^{rk/2}w^{\pm k}}
 {(1-q_1^{\pm k})(1-q_2^{\pm k})}a_{\mp k}}\right).
\ee
\end{lemma}
\begin{proof}
The specialization formulas \eqref{sp_mukV} follow from the normal-ordering of the Mukad\'e operator, using \eqref{NO_Phi}
\begin{multline}
 \left(\dbra{\vac}\otimes\mathrm{Id}\right)
 \mukT[\bsu,\bsv,w]\left(\dket{\vac}\otimes\mathrm{Id}\right)
 =\sum_{\{\nu^{(\a)}\}_{\a=1}^{r}}\prod_{\a=1}^{r-1} a_{\nu^{(\a)}}\times\prod_{\a=1}^r\dfrac{N_{\nu^{(\a-1)},\nu^{(\a)}}(\g w_\a/w_\a')}{\CG(w_\a/(\g w'_\a))}\\
\times :\Phi_{\nu^{(1)}}^{(1,0)\ast}[v_1,w'_1]\Phi_{\nu^{(0)}}^{(1,0)}[u_1,w_1]:\otimes\cdots\otimes  :\Phi_{\nu^{(r)}}^{(1,0)\ast}[v_r,w'_r]\Phi_{\nu^{(r-1)}}^{(1,0)}[u_r,w_r]:,
\end{multline}
with $\nu^{(0)}=\nu^{(r)}=\vac$. The product of functions $\CG(w_\a/(\g w'_\a))$ will be eliminated by our choice of normalization, they will not play role in the analysis.

We start by considering the case $\bsv=\g^{-1}\bsu$. Examining the first factor in the product, namely (see \eqref{N_vac})
\begin{equation}
N_{\vac,\nu^{(1)}}(\g w_1/w_1')=\prod_{\sAbox\in\nu^{(1)}}\left(1-\g\dfrac{w_1}{w'_1}\chi_\sAbox^{-1}\right),
\end{equation} 
we observe that when $w'_1=\g w_1$, this factor vanishes unles $\nu^{(1)}=\vac$. Since this zero cannot be compensated for by another singularity, all terms in the sum with $\nu^{(1)}\neq\vac$ are vanishing. Then, the second factor reads 
\begin{equation}
N_{\nu^{(1)},\nu^{(2)}}(\g w_2/w_2')=N_{\vac,\nu^{(2)}}(\g w_2/w_2')=\prod_{\sAbox\in\nu^{(2)}}\left(1-\g\dfrac{w_2}{w'_2}\chi_\sAbox^{-1}\right),
\end{equation}
The previous argument can be repeated, showing inductively that under the specialization $w'_\a=\g w_\a$ (i.e. $v_\a=\g^{-1}u_\a$), the only non-vanishing term corresponds to $\nu^{(1)}=\dots=\nu^{(r-1)}=\vac$, and the Mukad\'e operator takes the factorized form
\begin{equation}
\mukT^V[\bsu,\g^{-1}\bsu,w]=\,:\Phi_{\vac}^\ast(\g w_1)\Phi_{\vac}(w_1):\otimes\cdots\otimes  :\Phi_{\vac}^\ast(\g w_r)\Phi_{\vac}(w_r):
\end{equation} 
Using the property \eqref{prop_Phi_vac}, and inductively $w_\a=\g^{\a-1} w$, we find
\begin{equation}
 \mukT^V[\bsu,\g^{-1}\bsu,w]
 = \rho_\bsu^{(r,\bsn_0)}\left(\mathe^{\sum_{k>0}\frac{\g^{rk/2}w^k}{(1-q_1^k)(1-q_2^k)}a_{-k}}\right)
 = \mathe^{-\sum_{k>0}\frac{(1-q_3^k)w^k}{k(1-q_2^k)}\sum_{\a=1}^r p_k^{(\a)}}.
\end{equation} 

The argument for the second specialization formula is similar. Examining the last Nekrasov factor, namely
\begin{equation}
N_{\nu^{(r-1)},\vac}(\g w_r/w'_r)=\prod_{\sAbox\in \nu^{(r-1)}}\left(1-\g^{-1}\dfrac{w_r}{w'_r}\chi_{\sAbox}\right).
\end{equation} 
we observe that this factor vanishes when $w'_r=\g^{-1}w_r$, unless $\nu^{(r-1)}=\vac$. Inductively, we show that under the specialization $w'_\a=\g^{-1} w_\a$ (i.e. $v_\a=\g u_\a$), the only non-vanishing term corresponds to $\nu^{(r-1)}=\nu^{(r-2)}=\dots=\nu^{(1)}=\vac$, and so the Mukad\'e operator takes again a factorized form,
\begin{equation}
\mukT^V[\bsu,\g\bsu,w]=\,:\Phi_{\vac}^\ast(\g^{-1} w_1)\Phi_{\vac}(w_1):\otimes\cdots\otimes :\Phi_{\vac}^\ast(\g^{-1} w_r)\Phi_{\vac}(w_r):
\end{equation} 
Using the property \eqref{prop_Phi_vac}, and inductively $w_\a=\g^{1-\a} w$, we find
\be
 \mukT^V[\bsu,\g\bsu,w]
 = \rho_\bsu^{(r,\bsn_0)}\left(\mathe^{-\sum_{k>0}\frac{\g^{rk/2}w^{-k}}{(1-q_1^{-k})(1-q_2^{-k})}a_{k}}\right)
 = \mathe^{\sum_{k>0}\frac{(1-q_3^k)(q_3w)^{-k}}{(1-q_1^k)}
 \sum_{\a=1}^r q_3^{(\a-1)k} \frac{\partial}{\partial p_k^{(\a)}}}.
\ee
\end{proof}

\paragraph{Generalized Pieri rules.} Using the result of Lemma~\ref{lem:mukade}, we can now expand the specializations of the Mukad\'e operator in powers of $w^{\pm1}$.
We deduce that, at first order in $w^{\pm1}$, the matrix elements of the Mukad\'e operators are proportional to those of the operators $\rho_\bsu^{(r,\bsn_0)}(a_{\pm1})$,
\be
 \rho_\bsu^{(r,\bsn_0)}(a_{\pm 1}) = \mp\g^{-r/2}(1-q_1^{\mp1})(1-q_2^{\mp1})
 \oint\frac{\mathd w}{2\pi\mathi w} w^{\pm1} \mukT^V[\bsu,\g^{\pm1}\bsu,w].
\ee
Combining with Theorem~\ref{thm:FOS}, this gives a way to derive Pieri and dual Pieri rules for the generalized Macdonald basis.
\begin{proposition}
\label{prop:pieriApp}
The matrix elements of the operators $\rho_\bsu^{(r,\bsn_0)}(a_{\pm1})$ in the generalized Macdonald basis are given by
\be
\label{mat_apm1}
\begin{aligned}
 \dfrac{\bra{\tP_{\bmu-\sAbox}(\bsu)}\rho_\bsu^{(r,\bsn_0)}(a_1)\ket{\tP_{\bmu}(\bsu)}}
 {\langle \tP_\bmu(\bsu)|\tP_\bmu(\bsu)\rangle}
 &= -\g^{-r/2}(1-q_2)(1-q_3)\tilde{\psi}^\ast_{\mu^{(\a)}}(\Abox)
 \prod_{\b=1}^{\a-1}\Psi_{\mu^{(\b)}}(u_\a\chi_\sAbox/u_\b),\\
 \dfrac{\bra{\tP_{\bmu+\sAbox}(\bsu)}\rho_\bsu^{(r,\bsn_0)}(a_{-1})\ket{\tP_{\bmu}(\bsu)}}
 {\langle \tP_\bmu(\bsu)|\tP_\bmu(\bsu)\rangle}
 &= -\g^{-r/2}(1-q_1)(1-q_3)\tilde{\psi}_{\mu^{(\a)}}(\Abox)
 \prod_{\b=\a+1}^r\Psi_{\mu^{(\b)}}(u_\a\chi_\sAbox/u_\b),
\end{aligned}
\ee
where $\Abox\in R(\mu^{(\a)})$ (resp. $\Abox\in A(\mu^{(\a)})$, and the coefficients $\tilde{\psi}_{\mu^{(\a)}}(\Abox)$, $\tilde{\psi}^\ast_{\mu^{(\a)}}(\Abox)$ are independent of the level $r$ of the representation. Thus, these coefficients have the same expression for any level, and can be determined from the representation of level one. They will depend on the explicit choice of normalization for the Macdonald symmetric functions. These two formulas allow us to write down the action of $\rho^{(r,\bsn_0)}(a_{\pm1})$ for arbitrary $r$, knowing the action for $r=1$.
\end{proposition}
\begin{proof}
Expanding the operators in \eqref{sp_mukV} in powers of $w^{\pm1}$,
we can use Theorem~\ref{thm:FOS} to write the matrix elements of the operators $\rho_\bsu^{(r,\bsn_0)}(a_{\pm1})$ as
\be
\begin{aligned}
 \dfrac{\bra{\tP_{\bl}(\g\bsu)}\rho_\bsu^{(r,\bsn_0)}(a_1)\ket{\tP_{\bmu}(\bsu)}}
 {\langle \tP_\bmu(\bsu)|\tP_\bmu(\bsu)\rangle}
 =& -\d_{|\bl|,|\bmu|-1}\g^{-r/2}(1-t)(1-q^{-1})
 \prod_{\a=1}^r\frac{\tilde{b}_{\l^{(\a)}}}
 {\tilde{b}_{\mu^{(\a)}}} \\
 & \times\prod_{\a>\b}\frac{\tN_{\mu^{(\a)},\mu^{(\b)}}(u_\a/u_\b)}
 {\tN_{\l^{(\a)},\l^{(\b)}}(u_\a/u_\b)}
 \times\prod_{\a,\b=1}^r\frac{\tN_{\l^{(\a)},\mu^{(\b)}}(q_3u_\a/u_\b)}
 {\tN_{\mu^{(\a)},\mu^{(\b)}}(u_\a/u_\b)},\\
 \dfrac{\bra{\tP_{\bl}(\g^{-1}\bsu)}\rho_\bsu^{(r,\bsn_0)}(a_{-1})\ket{\tP_{\bmu}(\bsu)}}
 {\langle \tP_\bmu(\bsu)|\tP_\bmu(\bsu)\rangle}
 =&\, \d_{|\bl|,|\bmu|+1}\g^{-r/2}(1-t^{-1})(1-q) q^{|\bmu|-|\blam|}\\
 & \times\prod_{\a<\b}\frac{\tN_{\l^{(\a)},\l^{(\b)}}(u_\a/u_\b)}
 {\tN_{\mu^{(\a)},\mu^{(\b)}}(u_\a/u_\b)}
 \times\prod_{\a,\b=1}^r\frac{\tN_{\l^{(\a)},\mu^{(\b)}}(u_\a/u_\b)}
 {\tN_{\l^{(\a)},\l^{(\b)}}(u_\a/u_\b)},
\end{aligned}
\ee
where the Kronecker $\d$ appears from the identification of the power of $w$.
We deduce that $a_{-1}$ adds boxes to $\bmu$ and $a_1$ removes them. Moreover, we will use that $\bra{\tP_{\bl}(\g^{\pm1}\bsu)}=\bra{\tP_{\bl}(\bsu)}$ since the basis vectors are invariant under a simultaneous rescaling of all the weights.

Let us assume that $\a$ is the label such that $\l^{(\a)}=\mu^{(\a)}\mp\Abox$. Then we can use the variation formulas \eqref{prop_tN} to simplify the ratio of Nekrasov factors and we obtain \eqref{mat_apm1}.
\end{proof}

\bibliographystyle{utphys}
\bibliography{Refined_IH}

\end{document}